\documentclass[11pt,reqno]{amsproc}
\usepackage[colorinlistoftodos,shadow]{todonotes}
\usepackage[T1]{fontenc}
\usepackage{amsmath,amscd,amssymb}
\usepackage{cite}
\usepackage{color}
\usepackage{graphicx}
\usepackage{psfrag,epsfig}
\usepackage{multirow}
\usepackage{rotating}
\usepackage[letterpaper,margin=0.75in]{geometry}
\usepackage{subfigure}
\usepackage{enumerate}
\usepackage{comment}
\usepackage{lineno}
\usepackage{algorithmic}
\usepackage{algorithm}
\usepackage[small]{caption}
\usepackage{bbold}
\usepackage[debug=false, colorlinks=true, pdfstartview=FitV, linkcolor=blue, citecolor=blue, urlcolor=blue]{hyperref}
\newtheorem{theorem}{Theorem}[section]

\numberwithin{equation}{section}
\usepackage{chemarr}
\usepackage[version=3]{mhchem}
\newcommand*\patchAmsMathEnvironmentForLineno[1]{%
  \expandafter\let\csname old#1\expandafter\endcsname\csname #1\endcsname
  \expandafter\let\csname oldend#1\expandafter\endcsname\csname end#1\endcsname
  \renewenvironment{#1}%
     {\linenomath\csname old#1\endcsname}%
     {\csname oldend#1\endcsname\endlinenomath}}%
\newcommand*\patchBothAmsMathEnvironmentsForLineno[1]{%
  \patchAmsMathEnvironmentForLineno{#1}%
  \patchAmsMathEnvironmentForLineno{#1*}}%
\AtBeginDocument{%
\patchBothAmsMathEnvironmentsForLineno{equation}%
\patchBothAmsMathEnvironmentsForLineno{align}%
\patchBothAmsMathEnvironmentsForLineno{flalign}%
\patchBothAmsMathEnvironmentsForLineno{alignat}%
\patchBothAmsMathEnvironmentsForLineno{gather}%
\patchBothAmsMathEnvironmentsForLineno{multline}%
}

\title[Learning the State of Reactive-Mixing]{Physics-Informed Machine Learning Models for Predicting the Progress of Reactive-Mixing}
\author[M.~K.~Mudunuru and S. Karra]{M.~K.~Mudunuru$^{*}$ and S.~Karra \\
{\small Computational Earth Science Group (EES-16), Earth and Environmental Sciences Division, Los Alamos National Laboratory, Los Alamos, NM 87545, USA.} \\
}
\thanks{$^*$Corresponding author, \texttt{maruti@lanl.gov}}
\date{\today}
\begin{document}
\maketitle
%
%
%
%
\section*{ABSTRACT}
This paper presents a physics-informed machine learning (ML) framework to construct reduced-order models (ROMs) for reactive-transport quantities of interest (QoIs) based on high-fidelity numerical simulations. 
QoIs include species decay, product yield, and degree of mixing. 
The ROMs for QoIs are applied to quantify and understand how the chemical species evolve over time. 
First, high-resolution datasets for constructing ROMs are generated by solving anisotropic reaction-diffusion equations using a non-negative finite element formulation for different input parameters. 
Non-negative finite element formulation ensures that the species concentration is non-negative (which is needed for computing QoIs) on coarse computational grids even under high anisotropy. 
The reactive-mixing model input parameters are a time-scale associated with flipping of velocity, a spatial-scale controlling small/large vortex structures of velocity, a perturbation parameter of the vortex-based velocity, anisotropic dispersion strength/contrast, and molecular diffusion.
Second, random forests, F-test, and mutual information criterion are used to evaluate the importance of model inputs/features with respect to QoIs.
We observed that anisotropic dispersion strength/contrast is the most important feature and time-scale associated with flipping of velocity is the least important feature. 
Third, Support Vector Machines (SVM) and Support Vector Regression (SVR) are used to construct ROMs based on the model inputs. 
Then, SVR-ROMs are used to predict scaling of QoIs. 
Qualitatively, SVR-ROMs are able to describe the trends observed in the scaling law associated with QoIs. 
Moreover, it is observed that SVR-ROMs accuracy (which is the $R^2$-score) in predicting scaling QoIs is greater than 0.9. 
Meaning that, we have a good quantitative prediction of reaction-diffusion system state using SVR-ROMs. 
Fourth, the scaling law's exponent dependence on model inputs/features are evaluated using $k$-means clustering.
Based on clustering analysis, we infer that incomplete mixing occurs at high anisotropic contrast and the system is well-mixed for low anisotropic contrast. 
Finally, in terms of the computational cost, the proposed SVM-ROMs and SVR-ROMs are $\mathcal{O}(10^7)$ times faster than running a high-fidelity numerical simulation for evaluating QoIs. 
This makes the proposed ML-ROMs attractive for reactive-transport sensing and real-time monitoring applications as they are significantly faster yet provide reasonably good predictions.
\newline
\newline
\textbf{KEYWORDS:}~Reduced-order modeling, 
machine learning, 
reaction-diffusion equations,  
anisotropic dispersion, 
non-negativity, 
finite element method, 
incomplete mixing, 
scaling laws.
%
%

\section{INTRODUCTION}
\label{Sec:S1_ROM_Intro}
The efficiency of many hydrogeological applications such as reactive-transport, contaminant remediation, and chemical degradation vastly depends on the macroscopic mixing occurring in porous media \cite{Willingham2008,gelhar1993stochastic,fetter1999contaminant,vengosh2014critical,mudunuru2012framework}. 
In the case of remediation activities, enhancement and control of the macroscopic mixing is fundamentally  influenced by the structure of flow field, which is impacted by pumping/extraction activities, heterogeneity, dispersion, and anisotropy of the flow medium \cite{dentz2011mixing,chiogna2014helicity,neupauer2014chaotic,2015_Mudunuru_etal_ASME,cirpka2015transverse,ye2018effect}. 
However, the relative importance of these hydrogeological parameters to understand mixing process is not well studied. 
This is partially because to understand and quantify mixing, one needs to perform multiple model runs of high-fidelity numerical simulations for various subsurface model inputs. 

Typically, reactive-transport models are computationally expensive as high-fidelity numerical simulations of these models take hours to complete on several thousands of processors.
The associated degrees of freedom (e.g., species concentrations at the mesh nodes or cell centers) to be solved at each time-step is approximately between $\mathcal{O}(10^6)$ to $\mathcal{O}(10^9)$ \cite{chang2017large}.
As a result, such computationally-intensive models may not be feasible for real-time calculations of certain 
important Quantities of Interests (QoIs) such as decay of chemical species, product yield, and degree of mixing.
Moreover, traditional numerical formulations (e.g., standard Galerkin, two-point finite volume, etc.) for diffusion-type equations can produce unphysical solutions for concentration of chemical species \cite{Ciarlet_Raviart_CMAME_1973_v2_p17,Liska_Shashkov_CiCP_2008_v3_p852,droniou2014finite}.
The numerical solution quality is worse when anisotropy dominates \cite{2015_Mudunuru_etal_ASME}.
To overcome this problem, different types of numerical techniques were proposed in literature \cite{Liska_Shashkov_CiCP_2008_v3_p852, castillo2013mimetic,2013_Nakshatrala_Mudunuru_Valocchi_JCP_v253_p278_p307,da2014mimetic,droniou2014finite,mudunuru2016enforcing,
mudunuru2017mesh,chang2017large,shashkov2018conservative}.
Herein, we employ a novel non-negative finite element method (FEM) \cite{2013_Nakshatrala_Mudunuru_Valocchi_JCP_v253_p278_p307,chang2017large} to produce high-fidelity non-negative numerical solutions. 
This non-negative FEM provides physically meaningful concentrations for reaction-diffusion equations even under high anisotropic constrast. 
Hence, the QoIs derived from species concentration obtained using non-negative FEM are always physically meaningful.

There is a pressing need to develop computationally efficient and reasonably accurate models for the QoIs, which take 
seconds to minutes to compute on a laptop or on a edge computing device (e.g., Raspberry Pi, Google's edge TPU, etc.) attached to a sensor. 
An attractive way to construct computationally efficient models (e.g., reduced-order models) is through machine learning (ML). 
Reduced-Order Models (ROMs) (aka surrogate models) are similar to analytical solutions or lookup tables \cite{schilders2008model,koziel2013surrogate,quarteroni2014reduced,keiper2018reduced}.
However, the methods to construct ROMs are different compared to the approaches to develop analytical solutions. 
In general, ROMs are constructed from either high-fidelity numerical simulations or experimental data or through a combination of both.
Various ML techniques (e.g., ensemble methods, kernel-based methods, shallow and deep neural networks, etc.) exist in literature \cite{salah2018machine,brunton2019data,zhu2019physics,wang2019reduced,tripathy2018deep,wang2018deep} to construct ROMs.
These approaches can substantially improve our capabilities to develop fast, accurate, and robust ROMs to predict contaminant fate and transport under natural conditions and during remediation activities.
In this paper, we use support vector machines to develop ROMs and the reason to use them  is explained below.

Support vector machines (SVMs) and support vector regression (SVRs) are a set of kernel-based supervised ML methods used for classification and regression \cite{evgeniou1999support,cristianini2000introduction,scholkopf2001learning}. 
SVMs are used for categorical variable prediction while SVRs are used for continuous variable prediction. 
Given a known relationship/classification, SVMs help to find/identify the class that the data belongs. 
In a similar fashion, given a set of continuous higher-dimensional inputs and outputs, SVR tries to find the best relation that maps these inputs (e.g., features) to outputs (e.g., QoIs) using kernels.
A main advantage of SVMs and SVRs is that the decision functions use only a subset of training dataset/points called support vectors (SV), which makes them memory efficient. 
Such a succinct representation of the decision functions in terms of support vectors makes SVMs and SVRs models effective even in high dimensional spaces and also in cases where the number of dimensions is greater than number of samples. 
Other advantage of SVMs and SVRs is their versatility in employing different kernel functions. 
For example, common kernels (e.g., linear, polynomial, radial basis, sigmoidal, etc.) and custom/user-defined 
kernels can be specified in the decision function. 
In this paper (e.g., see the paragraph at end of Sec.~\ref{Sec:S2_ROM_GE}), we choose the kernel based on the underlying physics of the problem. 
SVMs and SVRs have a nice mathematical formalism from which they are derived, making them more attractive. 
In a nutshell, the SVM and SVR decision functions are the result of solving constrained optimization problem that has a unique minimizer. 
A disadvantage of SVMs and SVRs is that if the number of features is much greater than the number of samples then over-fitting occurs.
To avoid over-fitting, adding a penalty term/regularization term in the optimization problem is crucial.
In results section, we investigate the importance of various parameters (e.g, penality, kernel coefficients, tolerance, etc.) in classifying and predicting the state of reactive-mixing.

In our previous work (by Vesselinov et. al. \cite{vesselinov2019unsupervised}), we have utilized a structure-preserving unsupervised ML approach to discover hidden features in the high-fidelity reactive-mixing simulation data.
This approach provided physical insights on reactive-mixing progress under low and highly anisotropic conditions.
It also provided information on processes (e.g., longitudinal dispersion, transverse dispersion, molecular dispersion) that contribute to the formation of product at early and later stages of time.
However, the previous work did not estimate the relative importance of various input parameters in predicting key QoIs and efficiently emulate the high-fidelity model outputs.
This is an important problem that we are attempting to tackle in this paper.
There is no existing methodology or closed-form mathematical solutions that can predict how the variation of reactive-mixing input parameters (e.g., anisotropic contrast, molecular diffusion) impacts the decay, production, and mixing of chemical species.
Our goal is two fold:~(1) quantify the influences of the input parameters on species mixing, decay of reactants, and formation of the reaction product, and (2) efficiently and accurately predict the state of reactive-mixing at different stages of time through reduced-order or surrogate modeling. 

The main contribution of this study is to develop physics-informed machine learning models to predict the state and progress of reactive-mixing. 
The models for reactive-mixing QoIs are built on high-fidelity numerical simulations that respect the underlying physics.
An advantage of the proposed ML method is that it is approximately $\mathcal{O} \left(10^7 \right)$ times faster than running a high-fidelity reactive-mixing simulation.
Moreover, for QoI calculations, ROMs provide another added benefit, which is data compression.
The compression ratio is approximately $\mathcal{O} \left(10^{-4} \right)$ (as the ROMs can be thought as a compressed version of post-processed high-fidelity numerical simulation data).
With minimal training data (e.g., 5\%), the accuracy of our SVM-ROMs and SVR-ROMs on unseen data is greater than 90\%.
Due to the increase in computational savings with accurate predictions and high data compression, our ML methodology is ideal for usage in comprehensive uncertainty quantification studies which require 1000s of forward model runs.

The paper is organized as follows:~Sec.~\ref{Sec:S2_ROM_GE} describes the governing equations for a fast irreversible bimolecular reaction-diffusion system. 
It provides details on the non-negative numerical formulation for computing concentration of chemical species.
High-fidelity reactive-mixing simulation datasets are generated by performing a series of model runs using a parallel non-negative FEM solver with different model input parameters representing uncertainties in the underlying reaction-diffusion processes.
This dataset is used for constructing and testing ROMs using the ML framework described in Sec.~\ref{Sec:S3_ROM_Framework}.
We also present estimates and inequalities on the QoIs, which inform us that the species decay, mix, and produce in an exponential fashion.
Moreover, these inequalities on QoIs inform the ML framework that a suitable kernel to model QoIs needs to be an exponential function.
Sec.~\ref{Sec:S4_ROM_Results} provides results on the predictive capabilities of SVM-ROMs and SVR-ROMs on the mixing state of the reaction-diffusion system.
It quantifies the important features that influence reactive-mixing.
Computational efficiency, data compression, and accuracy of SVM-ROMs and SVR-ROMs towards predicting the state of reactive-mixing at different stages of time are discussed as well. 
Finally, conclusions are drawn in Sec.~\ref{Sec:S5_ROM_Conclusions}.

\section{DIFFUSION EQUATIONS AND A NON-NEGATIVE NUMERICAL FORMULATION}
\label{Sec:S2_ROM_GE}
Let $\Omega \subset \mathbb{R}^{nd}$ be an open bounded domain, where $nd$ denotes the number of spatial dimensions. 
The boundary is denoted by $\partial \Omega$, which is assumed to be piecewise smooth. 
Let $\overline{\Omega}$ be the set closure of $\Omega$ and let a spatial point $\boldsymbol{x} \in \overline{\Omega}$, the divergence and gradient operators with respect to $\boldsymbol{x}$ are denoted by $\mathrm{div}[\bullet]$ and $\mathrm{grad}[\bullet]$, respectively. 
Let $\boldsymbol{n}(\boldsymbol{x})$ be the unit outward normal to $\partial \Omega$. 
Let $t \in \, ]0, \mathcal{I}[$ denote the time, where $\mathcal{I}$ is the length of time of interest.
$]\bullet[$ denotes an open interval (i.e., the end points are not included) and $[\bullet]$ denotes a closed interval\footnote{
Note that the governing equations are posed on $\Omega \times ]0, \mathcal{I}[$ and the initial condition is posed on the $\overline{\Omega}$.}.

Consider the fast bimolecular reaction where species $A$ and species $B$ react irreversibly to give the product $C$ according to:
\begin{align}
  \label{Eqn:Bimolecular_fast_reaction}
    n_{A} \, A \, + \, n_{B} \, B \longrightarrow n_{C} \, C
\end{align}
where $n_A$, $n_B$, and $n_C$ are (positive) stoichiometric coefficients.

The governing equations for fast bimolecular reaction Eq.~\eqref{Eqn:Bimolecular_fast_reaction} without any non-reactive volumetric source/sink are given by: 
\begin{subequations}
  \label{Eqn:DRs_for_A_B_C}
  \begin{align}
    \label{Eqn:DRs_for_A}
    &\frac{\partial c_A}{\partial t} - \mathrm{div}[\boldsymbol{D}
    (\boldsymbol{x},t) \, \mathrm{grad}[c_A]] =  - 
    n_{\small{A}} \, k_{AB} c_A c_B \quad \mathrm{in} \; \Omega \times ]0, 
    \mathcal{I}[ \\ 
    \label{Eqn:DRs_for_B} 
    &\frac{\partial c_B}{\partial t} - \mathrm{div}[\boldsymbol{D}
    (\boldsymbol{x},t) \, \mathrm{grad}[c_B]] =  - 
    n_{\small{B}} \, k_{AB} c_A c_B \quad \mathrm{in} \; \Omega \times ]0, 
    \mathcal{I}[ \\
    \label{Eqn:DRs_for_C} 
    &\frac{\partial c_C}{\partial t} - \mathrm{div}[\boldsymbol{D}
    (\boldsymbol{x},t) \, \mathrm{grad}[c_C]] = +
    n_{\small{C}} \, k_{AB} c_A c_B \quad \mathrm{in} \; \Omega \times ]0, 
    \mathcal{I}[ \\
    \label{Eqn:DRs_for_Dirchlet}
    &c_i(\boldsymbol{x},t) = c^{\mathrm{p}}_i(\boldsymbol{x},t) \quad 
    \mathrm{on} \; \Gamma^{\mathrm{D}}_{i} \times ]0, 
    \mathcal{I}[ \quad (i = A, \, B, \, C) \\
    \label{Eqn:DRs_for_Neumann}
    & \left(-\boldsymbol{D} (\boldsymbol{x},t) \, \mathrm{grad}
    [c_i] \right) \bullet \boldsymbol{n}(\boldsymbol{x}) = h^
    {\mathrm{p}}_i(\boldsymbol{x},t) \quad \mathrm{on} \; 
    \Gamma^{\mathrm{N}}_{i} \times ]0, \mathcal{I}[ 
    \quad (i = A, \, B, \, C) \\
    \label{Eqn:DRs_for_IC}
    &c_i(\boldsymbol{x},t=0) = c^{0}_i(\boldsymbol{x}) \quad 
    \mathrm{in} \; \overline{\Omega} \quad (i = A, \, B, \, C)
  \end{align}
\end{subequations}
where $c_i$ is the molar concentration of $i$-th chemical species, $\boldsymbol{D}(\boldsymbol{x},t)$ is the anisotropic dispersion tensor, and $k_{AB}$ is the bilinear reaction rate coefficient.
$c^{\mathrm{p}}_i(\boldsymbol{x},t)$ and $h^{\mathrm{p}}_i(\boldsymbol{x},t)$ are the prescribed molar concentration and flux on Dirichlet and Neumann boundaries $\Gamma^{\mathrm{D}}_{i}$ and $\Gamma^{\mathrm{N}}_{i}$, respectively. 
Additionally, $c^{0}_i(\boldsymbol{x})$ is the initial concentration of $i$-th chemical species.

Using the non-negative linear transformations \cite{2013_Nakshatrala_Mudunuru_Valocchi_JCP_v253_p278_p307},
\begin{subequations}
  \label{Eqn:Definitions_of_F_G}
  \begin{align}
    \label{Eqn:Definitions_of_F}
    c_F &:= c_A + \left( \frac{n_A}{n_C} \right) c_C \\ 
    \label{Eqn:Definitions_of_G}
    c_G &:= c_B + \left( \frac{n_B}{n_C} \right) c_C 
  \end{align}
\end{subequations}
Eqs.~\eqref{Eqn:DRs_for_A}--\eqref{Eqn:DRs_for_IC} can be transformed to 
\begin{subequations}
  \begin{align}
    \label{Eqn:Diffusion_for_F}
    &\frac{\partial c_F}{\partial t} - \mathrm{div}[\boldsymbol{D}
    (\boldsymbol{x},t) \, \mathrm{grad}[c_F]] = 0
    \quad \mathrm{in} \; \Omega \times ]0, \mathcal{I}[ \\
    \label{Eqn:Diffusion_for_Dirchlet_F}
    &c_F(\boldsymbol{x},t) = c_F^{\mathrm{p}}(\boldsymbol{x},t) := 
    c^{\mathrm{p}}_A(\boldsymbol{x},t) + \left( \frac{n_A}{n_C} 
    \right) c^{\mathrm{p}}_C(\boldsymbol{x},t) \quad \mathrm{on} 
    \; \Gamma^{\mathrm{D}} \times ]0, \mathcal{I}[ \\
    \label{Eqn:Diffusion_for_Neumann_F}
    & \left(-\boldsymbol{D} (\boldsymbol{x},t) \, \mathrm{grad}[c_F]
    \right) \bullet \boldsymbol{n}(\boldsymbol{x}) =  h^{\mathrm{p}}
    _F(\boldsymbol{x},t) := h^{\mathrm{p}}_A(\boldsymbol{x},t) + 
    \left( \frac{n_A}{n_C} \right) h^{\mathrm{p}}_C(\boldsymbol{x},
    t) \quad \mathrm{on} \; \Gamma^{\mathrm{N}} \times ]0, 
    \mathcal{I}[ \\
    \label{Eqn:Diffusion_for_IC_F}
    &c_F(\boldsymbol{x},t=0) = c^{0}_F(\boldsymbol{x}) := 
    c^{0}_A(\boldsymbol{x}) + \left( \frac{n_A}{n_C} \right) 
    c^{0}_C(\boldsymbol{x}) \quad \mathrm{in} \; \overline{\Omega}
  \end{align}
\end{subequations} 
and 
\begin{subequations}
  \begin{align}
    \label{Eqn:Diffusion_for_G} 
    &\frac{\partial c_G}{\partial t} - \mathrm{div}[\boldsymbol{D}
    (\boldsymbol{x},t) \, \mathrm{grad}[c_G]] = 0
    \quad \mathrm{in} \; \Omega \times ]0, \mathcal{I}[ \\
    \label{Eqn:Diffusion_for_Dirchlet_G}
    &c_G(\boldsymbol{x},t) = c^{\mathrm{p}}_G(\boldsymbol{x},t) := 
    c^{\mathrm{p}}_B (\boldsymbol{x},t) + \left( \frac{n_B}{n_C} 
    \right) c^{\mathrm{p}}_C (\boldsymbol{x},t) \quad \mathrm{on} 
    \; \Gamma^{\mathrm{D}} \times ]0, \mathcal{I}[ \\
    \label{Eqn:Diffusion_for_Neumann_G}
    & \left(-\boldsymbol{D} (\boldsymbol{x},t) \, \mathrm{grad}[c_G] 
    \right) \bullet \boldsymbol{n}(\boldsymbol{x}) = h^{\mathrm{p}}_
    G(\boldsymbol{x},t) := h^{\mathrm{p}}_B(\boldsymbol{x},t) + 
    \left( \frac{n_B}{n_C} \right) h^{\mathrm{p}}_C(\boldsymbol{x},
    t) \quad \mathrm{on} \; \Gamma^{\mathrm{N}} \times ]0, 
    \mathcal{I}[ \\
    \label{Eqn:Diffusion_for_IC_G}
    &c_G(\boldsymbol{x},t=0) = c^{0}_G(\boldsymbol{x}) := c^{0}_B
    (\boldsymbol{x}) + \left( \frac{n_B}{n_C} \right) c^{0}_
    C(\boldsymbol{x}) \quad \mathrm{in} \; \overline{\Omega} 
  \end{align}
\end{subequations}

Since we are considering fast bimolecular reactions, we will assume that the species $A$ and $B$ cannot co-exist at the same location $\boldsymbol{x}$ at any instant of time.
Based on this assumption, once quantities $c_F(\boldsymbol{x},t)$ and $c_G(\boldsymbol{x},t)$ are calculated, one can then evaluate $c_A(\boldsymbol{x},t)$, $c_B(\boldsymbol{x},t)$, and $c_C(\boldsymbol{x},t)$ values, via:
\begin{subequations}
  \label{Eqn:Fast_A_B_C}
  \begin{align}
    \label{Eqn:Fast_A}
    &c_A(\boldsymbol{x},t) = \mathrm{max} \left[c_F(\boldsymbol{x},t) 
      - \left(\frac{n_A}{n_B}\right) c_G(\boldsymbol{x},t), \, 0 
      \right] \\
    \label{Eqn:Fast_B}
    &c_B(\boldsymbol{x},t) = \left( \frac{n_B}{n_A} \right) \; 
    \mathrm{max}\left[- c_F(\boldsymbol{x},t) + 
      \left(\frac{n_A}{n_B} \right) c_G(\boldsymbol{x},t), \, 0 \right] \\
    \label{Eqn:Fast_C}
    &c_C(\boldsymbol{x},t) = \left( \frac{n_C}{n_A} \right) \; 
    \left(c_F(\boldsymbol{x},t) - c_A(\boldsymbol{x},t) \right)
  \end{align}
\end{subequations}

\subsection{Non-negative single-field formulation}
\label{SubSec:S2_ROM_NNSF}
Let the total time interval $[0, \mathcal{I}]$ be discretized into $N$ non-overlapping sub-intervals: 
\begin{align}
  [0, \mathcal{I}] = \bigcup_{n = 1}^{N} [t_{n-1}, t_{n}]
\end{align}
where $t_0 = 0$ and $t_N = \mathcal{I}$. 
Assuming a uniform time step $\Delta t$, we define:
\begin{align}
  c^{(n)}_F(\boldsymbol{x}) &:= c_F(\boldsymbol{x},t_n)
\end{align}

First we will discretize with respect to time, using the backward Euler scheme, and so Eqs.~\eqref{Eqn:Diffusion_for_F}--\eqref{Eqn:Diffusion_for_IC_F} lead to:
\begin{subequations}
  \begin{align}
    \label{Eqn:Decay_Diffusion_for_F}
    &\left(\frac{1}{\Delta t} \right) c^{(n + 1)}_F(\boldsymbol{x})  
    - \mathrm{div} \left[\boldsymbol{D}(\boldsymbol{x}) \, \mathrm{grad}
    \left[c^{(n + 1)}_F(\boldsymbol{x}) \right] \right] = 
    f_F(\boldsymbol{x},t_{n + 1}) + \left( \frac{1}{\Delta t} \right)
    c^{(n)}_F(\boldsymbol{x}) \quad \mathrm{in} \; \Omega \\
    \label{Eqn:Decay_Diffusion_for_Dirchlet_F}
    &c^{(n + 1)}_F(\boldsymbol{x}) = c_F^{\mathrm{p}}(\boldsymbol{x},
    t_{n + 1}) := c^{\mathrm{p}}_A(\boldsymbol{x},t_{n + 1}) + 
    \left( \frac{n_A}{n_C} \right) c^{\mathrm{p}}_C(\boldsymbol{x},t_{n + 1}) 
    \quad \mathrm{on} \; \Gamma^{\mathrm{D}} \\
    \label{Eqn:Decay_Diffusion_for_Neumann_F}
    &\boldsymbol{n}(\boldsymbol{x}) \bullet \boldsymbol{D} (\boldsymbol{x}) \, 
    \mathrm{grad} \left[c^{(n + 1)}_F(\boldsymbol{x}) \right] =  
    h^{\mathrm{p}}_F(\boldsymbol{x},t_{n + 1}) := 
    h^{\mathrm{p}}_A(\boldsymbol{x},t_{n + 1}) + \left( \frac{n_A}{n_C} 
    \right) h^{\mathrm{p}}_C(\boldsymbol{x},t_{n + 1}) \quad \mathrm{on} 
    \; \Gamma^{\mathrm{N}} \\
    \label{Eqn:Decay_Diffusion_for_IC_F}
    &c_F(\boldsymbol{x},t_0) = c^{0}_F(\boldsymbol{x}) := 
    c^{0}_A(\boldsymbol{x}) + \left( \frac{n_A}{n_C} \right) 
    c^{0}_C(\boldsymbol{x}) \quad \mathrm{in} \; \overline{\Omega}
  \end{align}
\end{subequations}

Note that one can extent the proposed framework to other time-stepping schemes by following the non-negative procedures given in Reference \cite{nakshatrala2016numerical}.
The standard finite element Galerkin formulation for Eqs.~\eqref{Eqn:Decay_Diffusion_for_F}--\eqref{Eqn:Decay_Diffusion_for_IC_F} is given as follows:~Find $c^{(n + 1)}_F(\boldsymbol{x}) \in \mathcal{C}^t_F$ with 
\begin{align}
  \mathcal{C}^t_F &:= \left\{c_F(\bullet ,t) \in H^{1}(\Omega) \; \big| \; 
  c_F(\boldsymbol{x},t) = c^{\mathrm{p}}_F(\boldsymbol{x},t) \; \mathrm{on} \; 
  \Gamma^{\mathrm{D}}\right\}
\end{align}
such that we have
\begin{align}
  \label{Eqn:Transient_single_field_formulation_diffusion}
  \mathcal{B}^{t} \left(w;c^{(n + 1)}_F \right) = L^{t}_{F}(w) 
  \quad \forall w(\boldsymbol{x}) \in \mathcal{W}
\end{align}
where the bilinear form and linear functional are, respectively, defined as 
\begin{subequations}
  \begin{align}
    \mathcal{B}^{t} \left(w;c^{(n + 1)}_F \right) &:= 
    \frac{1}{\Delta t} \int \limits_{\Omega} w(\boldsymbol{x}) \; c^{(n + 1)}_F(\boldsymbol{x}) \; 
    \mathrm{d} \Omega + \int \limits_{\Omega} \mathrm{grad}[w(\boldsymbol{x})] 
    \bullet \boldsymbol{D}(\boldsymbol{x}) \; \mathrm{grad}[c^{(n + 1)}_F(\boldsymbol{x})] 
    \; \mathrm{d} \Omega \\
    L^{t}_{F}(w) &:= \frac{1}{\Delta t}\int \limits_{\Omega}  w(\boldsymbol{x}) \; 
    c^{(n)}_F(\boldsymbol{x}) \; \mathrm{d} \Omega + 
    \int \limits_{\Gamma^{\mathrm{N}}} w(\boldsymbol{x}) \; 
    h^{\mathrm{p}}_F(\boldsymbol{x},t_{n + 1}) \; \mathrm{d} \Gamma 
  \end{align}
\end{subequations}
and
\begin{align}
  \mathcal{W} &:= \left\{w(\boldsymbol{x}) \in H^{1}(\Omega) \; \big| \; 
  w(\boldsymbol{x}) = 0 \; \mathrm{on} \; 
  \Gamma^{\mathrm{D}}\right\}
\end{align}
A similar formulation can be written for the invariant $c_G$.
 
The Galerkin formulation Eq.~\eqref{Eqn:Transient_single_field_formulation_diffusion} can be re-written as a minimization problem as follows:
\begin{align}
  \label{Eqn:Minimzation_Problem_Statement_DiffusionDecay}
  \mathop{\mathrm{minimize}}_{c^{(n + 1)}_F(\boldsymbol{x}) \in 
  \mathcal{C}^t_F} & \quad \frac{1}{2} \mathcal{B}^{t} 
  \left(c^{(n + 1)}_F; c^{(n + 1)}_F \right) - L^t_{F} 
  \left(c^{(n + 1)}_F \right)
\end{align}
However, it has been shown that the Galerkin formulation leads to negative concentrations for product $C$ 
\cite{2013_Nakshatrala_Mudunuru_Valocchi_JCP_v253_p278_p307}, which is unphysical. 
Instead, we will solve the minimization problem Eq.~\eqref{Eqn:Minimzation_Problem_Statement_DiffusionDecay} with the constraint that $c_F$ is non-negative, i.e.,
\begin{align}
    c_F \geq 0
\end{align}
A similar minimization problem with non-negative constraint for $c_G$ can be written as follows:
\begin{subequations}
  \begin{align}
    \label{Eqn:Minimzation_Problem_Statement_DiffusionDecay_G}
    \mathop{\mathrm{minimize}}_{c^{(n + 1)}_G(\boldsymbol{x}) \in 
    \mathcal{C}^t_G} & \quad \frac{1}{2} \mathcal{B}^{t} 
    \left(c^{(n + 1)}_G; c^{(n + 1)}_G \right) - L^t_{G} 
    \left(c^{(n + 1)}_G \right) \\
    \mbox{subject to} &\quad c_G \geq 0
  \end{align}
\end{subequations}

Upon discretizing Eqs.~\eqref{Eqn:Minimzation_Problem_Statement_DiffusionDecay}--\eqref{Eqn:Minimzation_Problem_Statement_DiffusionDecay_G} using lower-order finite elements, one arrives at:
\begin{subequations}
  \label{Eqn:NonNegative_Solver_Invariant_F_DiffusionDecay}
    \begin{align}
      \mathop{\mbox{minimize}}_{\boldsymbol{c}^{(n + 1)}_F \in 
      \mathbb{R}^{ndofs}} & \quad \frac{1}{2}  \left \langle 
      \boldsymbol{c}^{(n + 1)}_F; {\boldsymbol{K}} 
      \boldsymbol{c}^{(n + 1)}_F \right \rangle - 
      \frac{1}{\Delta t}  \left \langle 
      \boldsymbol{c}^{(n + 1)}_F; \boldsymbol{c}^{(n)}_F
      \right \rangle \\
      \mbox{subject to} & \quad \boldsymbol{0}
      \preceq \boldsymbol{c}^{(n + 1)}_F  \\
      \mathop{\mbox{minimize}}_{\boldsymbol{c}^{(n + 1)}_G \in 
      \mathbb{R}^{ndofs}} & \quad \frac{1}{2}  \left \langle 
      \boldsymbol{c}^{(n + 1)}_G; {\boldsymbol{K}} 
      \boldsymbol{c}^{(n + 1)}_G \right \rangle - 
      \frac{1}{\Delta t}  \left \langle 
      \boldsymbol{c}^{(n + 1)}_G; \boldsymbol{c}^{(n)}_G
      \right \rangle \\
      \mbox{subject to} & \quad \boldsymbol{0}
      \preceq \boldsymbol{c}^{(n + 1)}_G 
    \end{align}
\end{subequations}
where $\boldsymbol{c}^{(n + 1)}_F$, $\boldsymbol{c}^{(n + 1)}_G$ are the nodal concentration vectors for the invariant $c_F$ and $c_G$ at time level $t_{n + 1}$, respectively.
The coefficient matrix ${\boldsymbol{K}}$ is positive definite and hence a unique global minimizer exists.

\subsection{Reaction tank problem}
\label{SubSec:Reaction_Tank_Problem}
Figure \ref{Fig:ROM_BVPs} provides a pictorial description of the initial boundary value problem. 
The computational domain is a square with $L = 1$. 
On the sides of the domain, zero flux boundary conditions are enforced. 
For all the chemical species, the non-reactive volumetric source $f_i(\boldsymbol{x}, t)$ is equal to zero. 
Initially, species $A$ and species $B$ are segregated. 
Species $A$ is placed in the left half of the domain while species $B$ is placed in the right half. 
The stoichiometric coefficients are taken as $n_A = 1$, $n_B = 1$, and $n_C = 1$. 
The total time of interest is taken as $\mathcal{I} = 1$. 
The dispersion tensor is taken from the subsurface literature \cite{Pinder_Celia} and is given as follows:
\begin{align}
  \label{Eqn:Aniso_Diff_Tensor_Lit}
  \boldsymbol{D}_{\mathrm{subsurface}}
  (\boldsymbol{x}) = D_{m} \boldsymbol{I} + 
  \alpha_{T} \|\boldsymbol{v}\| \boldsymbol{I} + 
  \frac{\alpha_L - \alpha_T}{\|\boldsymbol{v}\|} 
  \boldsymbol{v} \otimes \boldsymbol{v}
\end{align}
where $D_m$ is the molecular diffusivity, $\alpha_L$ is the longitudinal diffusivity, $\alpha_T$ is the transverse diffusivity, $\boldsymbol{I}$ is the identity tensor, $\otimes$ is the tensor product, $\boldsymbol{v}$ is the velocity vector field, and $\| \bullet \|$ is the Frobenius norm. 
We have neglected advection and the model velocity field is used to define the diffusivity tensor through the following stream function \cite{2002_Adrover_etal_CCE_v26_p125_p139,2009_Tsang_PRE_v80_p026305,mudunuru2016enforcing}:
\begin{align}
  \label{Eqn:Div_Free_Stream_Function}
  \psi(\boldsymbol{x},t) = 
  \begin{cases}
    \frac{1}{2 \pi \kappa_f} \left( \sin(2 \pi \kappa_f x) 
    - \sin(2 \pi \kappa_f y) + v_0 \cos(2 \pi \kappa_f y)
    \right) &\quad \mathrm{if} \; \nu T \leq t < \left( \nu 
    + \frac{1}{2} \right) T  \\
    \frac{1}{2 \pi \kappa_f} \left( \sin(2 \pi \kappa_f x) 
    - \sin(2 \pi \kappa_f y) - v_0 \cos(2 \pi \kappa_f y)
    \right) &\quad \mathrm{if} \; \left( \nu + \frac{1}{2} 
    \right) T \leq t < \left( \nu + 1 \right) T
  \end{cases}
\end{align}
where $\nu = 0, 1, 2, 3, \cdots$ is an integer. $\kappa_fL$ and $T$ are characteristic scales of the flow field. 
Using Eq.~\eqref{Eqn:Div_Free_Stream_Function}, the divergence free velocity field components are given as follows:
\begin{align}
  \label{Eqn:Vel_x}
  \mathrm{v}_{x}(\boldsymbol{x},t) = -\frac{\partial 
  \psi}{\partial \mathrm{y}} = 
  \begin{cases}
    \cos(2 \pi \kappa_f y) + v_o \sin(2 \pi \kappa_f y) 
    &\quad \mathrm{if} \; \nu T \leq t < \left( \nu + 
    \frac{1}{2} \right) T  \\
    \cos(2 \pi \kappa_f y) &\quad \mathrm{if} \; 
    \left( \nu + \frac{1}{2} \right) T \leq t < 
    \left( \nu + 1 \right) T
  \end{cases}
\end{align}
\begin{align}
  \label{Eqn:Vel_y}
  \mathrm{v}_{y}(\boldsymbol{x},t) = +\frac{\partial 
  \psi}{\partial \mathrm{y}} = 
  \begin{cases}
    \cos(2 \pi \kappa_f x) &\quad \mathrm{if} \; 
    \nu T \leq t < \left( \nu + \frac{1}{2} \right) T \\
    \cos(2 \pi \kappa_f x) + v_o \sin(2 \pi \kappa_f x) 
    &\quad \mathrm{if} \; \left( \nu + \frac{1}{2} \right) 
    T \leq t < \left( \nu + 1 \right) T
  \end{cases}
\end{align}

In Eqs.~\eqref{Eqn:Vel_x}--\eqref{Eqn:Vel_y}, the input parameter $T$ controls the flipping of the velocity field from clockwise to anti-clockwise direction.
$v_0$ is the perturbation parameter of the underlying vortex-based flow field.
Higher values of $v_0$ make the vortices skewed (ellipses) while lower values of $v_0$ correspond to circular vortex structures in the velocity field.
$\frac{\alpha_L}{\alpha_T}$ controls the anisotropic dispersion contrast.
Lower values of $\frac{\alpha_L}{\alpha_T}$ means low anisotropy and higher values of $\frac{\alpha_L}{\alpha_T}$ means high anisotropy.
$\kappa_fL$ corresponds to small-scale/large-scale vortex structures in the flow field.
Figure \ref{Fig:Large_Small_Scales_VelField} shows the contours and streamlines of the velocity field given by Eqs.~\eqref{Eqn:Vel_x}--\eqref{Eqn:Vel_y} for different values of $v_0$ and $\kappa_fL$. 
Small-scale and large-scale vortices are observed when $\kappa_fL$ is high and low, respectively.
One can also observe that by varying the perturbation parameter $v_0$, the location of these vortices is not highly altered.

\subsection{Upper and Lower Bounds on Quantities of Interests (QoIs)}
\label{SubSec:QoIs}
For reactive-transport applications, the following are the quantities of interests:
\begin{itemize}
  \item Species decay/production, which can be analyzed by calculating average of concentration `$\mathfrak{c}_i$' and average of square of concentration `$\mathfrak{c}_i$'.  
    These quantities are given as follows:
    \begin{align}
      \label{Eqn:Avg_Conc}
      \mathfrak{c}_i &:= \frac{\langle c_i(t) \rangle}{
      \mathrm{max}\left[\langle c_i(t) \rangle \right]}, 
      \quad \mathrm{where} \; \left \langle c_i(t) \right 
      \rangle = \int \limits_{\Omega} c_i(\mathbf{x},t) \, 
      \mathrm{d} \Omega \\
      \label{Eqn:Avg_Sq_Conc}
      \mathbb{c}_i &:= \frac{\langle c^2_i \rangle}{
      \mathrm{max}\left[\langle c^2_i \rangle \right]}, 
      \quad \mathrm{where} \; \left \langle c^2_i (t) 
      \right \rangle = \int \limits_{\Omega} c^2_i(
      \mathbf{x},t) \, \mathrm{d} \Omega
    \end{align}
  \item Degree of mixing, which can be analyzed by calculating the variance of concentration `$\sigma^2_i$' is given as follows:
    \begin{align}
      \label{Eqn:Degree_of_Mixing}
      \sigma^2_i := \frac{\langle c^2_i \rangle - \langle 
      c_i \rangle^2}{\mathrm{max} \left[\langle c^2_i \rangle 
      - \langle c_i \rangle^2 \right]}
    \end{align}
\end{itemize}
Note that the values for $\mathfrak{c}_i$, $\mathbb{c}_i$, and $\sigma^2_i$ are between 0 and 1 $\forall \quad i = A, B, C, F, G$.
To this end, let $H^{m}(\Omega)$ denote the standard Sobolev space for a given non-negative integer $m$ \cite{Evans_PDE}. 
The associated standard inner product and norm are denoted by $(\bullet,\bullet)_m$ and $\| \bullet \|_m$, respectively. 
Note that similar analysis can be performed for other reactive volumetric sources \cite{pao2012nonlinear}.
We also assume that $\Omega$ is simply connected.
Meaning that, there are no holes in the domain.
This assumption is needed for Poincar\'{e}-Friedrichs inequality to be valid \cite{Bochev_Gunzburger}.

\begin{theorem}[\texttt{Estimates on QoIs for invariant-F and invariant-G}]
  \label{Thm:Estimates_F_and_G}
  If $c_i(\mathbf{x},t) = 0$ on $\partial \Omega \, \times \, ]0, \mathcal{I}[$, $\Omega$ is simply connected, and
  $\left( \boldsymbol{D}(\boldsymbol{x},t) \, \mathrm{grad} [c_i(\mathbf{x},t)] \right) 
  \bullet \boldsymbol{n}(\boldsymbol{x})$ is bounded above and below by constants 
  $h_{i,lb}$ and $h_{i,ub}$, then
  \begin{align}
    \label{Eqn:Thm_AvgCFG_Bound}
    \frac{h_{i,lb} \, \mathrm{meas}(\partial \Omega)}{\mathrm{max}\left[\langle c_i(t) 
    \rangle \right]} \leq \frac{d\mathfrak{c}_i}{dt} \leq 
    \frac{h_{i,ub} \, \mathrm{meas}(\partial \Omega)}{\mathrm{max}\left[\langle c_i(t) 
    \rangle \right]} \quad \forall i = F, G
  \end{align}
  and
  \begin{align}
    \label{Eqn:Thm_AvgSqCFG_Bound}
    -\frac{2\lambda_{max}}{\mathrm{max}\left[\langle c^2_i \rangle \right]} 
    \| \mathrm{grad}[c_i] \|^2_0 \leq \frac{d\mathbb{c}_i}{dt} 
    \leq -\frac{2\lambda_{min}}{C^2_{pf}} \mathbb{c}_i \quad \forall i = F, G
  \end{align}
  where $\lambda_{\mathrm{min}}$ and $\lambda_{\mathrm{max}}$ are 
  the minimum and maximum eigenvalue of $\boldsymbol{D}(\boldsymbol{x},t)$ 
  in $\overline{\Omega} \, \times \, [0, \mathcal{I}]$. $C_{pf}$ is the 
  Poincar\'{e}-Friedrichs inequality constant.
\end{theorem}
\begin{proof}
  From Eqs.~\ref{Eqn:Diffusion_for_F}--\ref{Eqn:Diffusion_for_IC_F}, 
  as the volumetric source/sink term is equal to zero, by employing 
  divergence theorem, we have 
  \begin{align}
    \label{Eqn:Thm_AvgCFG_Bound_1}
    \frac{d}{dt} \int \limits_{\Omega} c_F \, \mathrm{d} \Omega = 
    \int \limits_{\partial \Omega} \boldsymbol{D} \, \mathrm{grad}[c_F]  
    \bullet \boldsymbol{n} \, \mathrm{d} \Gamma \implies h_{F,lb} \, 
    \mathrm{meas}(\partial \Omega) \leq  \frac{d}{dt} \int \limits_{\Omega} 
    c_F \, \mathrm{d} \Omega \leq h_{F,ub} \, \mathrm{meas}(\partial 
    \Omega)
  \end{align}
  
  Dividing Eq.~\ref{Eqn:Thm_AvgCFG_Bound_1} with $\mathrm{max}\left[\langle c_F(t) 
  \rangle \right]$, we get the desired inequality given by Eq.~\ref{Eqn:Thm_AvgCFG_Bound}.
  This completes the first part of the proof.
  For the second part of the proof, multiplying Eq.~\ref{Eqn:Diffusion_for_F} with $c_F$, 
  using the divergence theorem, and the identity that $\mathrm{div} \left[c_F \boldsymbol{D} 
  \mathrm{grad}[c_F] \right] = \boldsymbol{D} \mathrm{grad}[c_F] \bullet\mathrm{grad}[c_F] 
  + c_F \mathrm{div}\left[\boldsymbol{D} \mathrm{grad}[c_F] \right]$, we obtain the following 
  relationship
  \begin{align}
    \label{Eqn:Thm_AvgCFG_Bound_2}
    \frac{1}{2} \frac{d}{dt} \int \limits_{\Omega} c^2_F \, \mathrm{d} \Omega = 
    \int \limits_{\partial \Omega} c_F \boldsymbol{D} \, \mathrm{grad}[c_F]  
    \bullet \boldsymbol{n} \, \mathrm{d} \Gamma - \int \limits_{\Omega} 
    \boldsymbol{D} \, \mathrm{grad}[c_F] \bullet \mathrm{grad}[c_F] \, 
    \mathrm{d} \Omega
  \end{align}
  
  As we have assummed $c_F(\mathbf{x},t) = 0$ on $\partial \Omega \, 
  \times \, ]0, \mathcal{I}[$, $\displaystyle \int \limits_{\partial 
  \Omega} c_F \boldsymbol{D} \, \mathrm{grad}[c_F]  \bullet \boldsymbol{n} 
  \, \mathrm{d} \Gamma = 0$. Appealing to spectral theorem for symmetric 
  and positive definite matrices \cite{Golub}, we have $\lambda_{min} \| \mathrm{grad}[c_F] 
  \|_0 \leq \| \boldsymbol{D}^{1/2} \mathrm{grad}[c_F] \|_0 \leq \lambda_{max} 
  \| \mathrm{grad}[c_F] \|_0$. Using these relationships, Eq.~\ref{Eqn:Thm_AvgCFG_Bound_2} 
  reduces to 
  \begin{align}
    \label{Eqn:Thm_AvgCFG_Bound_3}
    -2\lambda_{max} \| \mathrm{grad}[c_F] \|^2_0 \leq \frac{d}{dt} 
    \int \limits_{\Omega} c^2_F \, \mathrm{d} \Omega \leq -2\lambda_{min} 
    \| \mathrm{grad}[c_F] \|^2_0
  \end{align}
  
  The Poincar\'{e}-Friedrichs inequality \cite{Bochev_Gunzburger} states 
  that:~there exists a constant $C_{pf} > 0$ such that we have $\| u \|_0 
  \leq C_{pf} \| \mathrm{grad}[u] \|_0 \quad \forall \, u \in H^{1}_0(\Omega)$.
  Invoking this inequality on Eq.~\ref{Eqn:Thm_AvgCFG_Bound_3} gives the 
  following result
  \begin{align}
    \label{Eqn:Thm_AvgCFG_Bound_4}
    -2\lambda_{max} \| \mathrm{grad}[c_F] \|^2_0 \leq \frac{d}{dt} 
    \int \limits_{\Omega} c^2_F \, \mathrm{d} \Omega \leq -2\lambda_{min} 
    \| \mathrm{grad}[c_F] \|^2_0 \leq -\frac{2 \lambda_{min}}{C^2_{pf}} 
    \| c_F \|^2_0
  \end{align}
  
  Dividing Eq.~\ref{Eqn:Thm_AvgCFG_Bound_4} with $\mathrm{max}\left[\langle c^2_i \rangle \right]$,
  we get Eq.~\ref{Eqn:Thm_AvgSqCFG_Bound}. 
  This completes the second part of the proof.
  Repeating the same process for $c_G$, we get inequalities given by 
  Eqs.~\ref{Eqn:Thm_AvgCFG_Bound}--\ref{Eqn:Thm_AvgSqCFG_Bound} for invariant-$G$.
\end{proof}

\begin{theorem}[\texttt{Estimates on average of concentration $\mathfrak{c}_i$}]
  \label{Thm:Estimates_AvgConc_ABC}
  Based on the assumptions outlined in Thm.~\ref{Thm:Estimates_F_and_G}, 
  the quantity $\frac{d\mathfrak{c}_i}{dt}$ is bounded as follows
  \begin{align}
    \label{Eqn:Thm_AvgAB_Bound}
    &\frac{1}{\mathrm{max}\left[\langle c_i(t) \rangle \right]} 
    \left(h_{i,lb} \, \mathrm{meas}(\partial \Omega)- n_i k_{AB} \, 
    \mathrm{meas}(\Omega) \, \mathrm{max}\left[c_F^{\mathrm{max}}, 
    c_G^{\mathrm{max}} \right]^2 \right) \leq \nonumber \\
    &\frac{h_{i,lb} \, \mathrm{meas}(\partial \Omega)}{\mathrm{max}
    \left[\langle c_i(t) \rangle \right]} - \left( n_i k_{AB} \, 
    \mathrm{max}\left[c_F^{\mathrm{max}}, c_G^{\mathrm{max}} \right] 
    \right) \mathfrak{c}_i 
    \leq \frac{d\mathfrak{c}_i}{dt} \leq 
    \frac{h_{i,ub} \, \mathrm{meas}(\partial \Omega) }{\mathrm{max} 
    \left[\langle c_i(t) \rangle \right]} \quad \forall i = A, B \\
    \label{Eqn:Thm_AvgC_Bound}
    &\frac{1}{\mathrm{max}\left[\langle c_i(t) \rangle \right]} 
    \left(h_{i,lb} \, \mathrm{meas}(\partial \Omega)- k_{AB} \, 
    \mathrm{meas}(\Omega) \, \mathrm{max}\left[c_F^{\mathrm{max}}, 
    c_G^{\mathrm{max}} \right] \, \left(n_A c_G^{\mathrm{max}} + 
    n_B c_F^{\mathrm{max}} \right) \right) \leq \nonumber \\
    &\frac{h_{i,lb} \, \mathrm{meas}(\partial \Omega)}{\mathrm{max}
    \left[\langle c_i(t) \rangle \right]} - k_{AB} \left(n_A 
    c_G^{\mathrm{max}} + n_B c_F^{\mathrm{max}} \right) \mathfrak{c}_i 
    \leq \frac{d\mathfrak{c}_i}{dt} \leq \nonumber \\
    &\frac{h_{i,ub} \, \mathrm{meas}(\partial \Omega) + n_C k_{AB} \, \mathrm{meas}
    (\Omega) \, c_F^{\mathrm{max}} c_G^{\mathrm{max}}}{\mathrm{max} 
    \left[\langle c_i(t) \rangle \right]} + \frac{n_A n_B k_{AB}}{n_C}
    \mathrm{max}\left[c_F^{\mathrm{max}}, c_G^{\mathrm{max}} \right] 
    \mathfrak{c}_i \leq \nonumber \\
    &\frac{h_{i,ub} \, \mathrm{meas}(\partial \Omega) + n_C k_{AB} \, \mathrm{meas}
    (\Omega) \, c_F^{\mathrm{max}} c_G^{\mathrm{max}} + \frac{n_A n_B k_{AB}}{n_C}
    \, \mathrm{meas}(\Omega) \, \mathrm{max}\left[c_F^{\mathrm{max}}, 
    c_G^{\mathrm{max}} \right]^2}{\mathrm{max}\left[\langle c_i(t) \rangle \right]} \quad \forall i = C
  \end{align} 
\end{theorem}
\begin{proof}
  Integrating Eq.~\ref{Eqn:DRs_for_A} and Eq.~\ref{Eqn:DRs_for_C}, and employing divergence theorem, we have
  \begin{align}
    \label{Eqn:Thm:Estimates_AvgConc_ABC_1}
    &\frac{d}{dt} \int \limits_{\Omega} c_A \, \mathrm{d} \Omega = 
    \int \limits_{\partial \Omega} \boldsymbol{D} \, \mathrm{grad}[c_A]  
    \bullet \boldsymbol{n} \, \mathrm{d} \Gamma - \int \limits_{\Omega} 
    n_A k_{AB} c_A c_B \, \mathrm{d} \Omega \nonumber \\
    &\implies h_{A,lb} \, \mathrm{meas}(\partial \Omega) - \int \limits_{\Omega} 
    n_A k_{AB} c_A c_B \, \mathrm{d} \Omega \leq  \frac{d}{dt} \int \limits_{\Omega} 
    c_A \, \mathrm{d} \Omega \leq h_{A,ub} \, \mathrm{meas}(\partial 
    \Omega) - \int \limits_{\Omega} n_A k_{AB} c_A c_B \, \mathrm{d} \Omega\\
    \label{Eqn:Thm:Estimates_AvgConc_ABC_2}
    &\frac{d}{dt} \int \limits_{\Omega} c_C \, \mathrm{d} \Omega = 
    \int \limits_{\partial \Omega} \boldsymbol{D} \, \mathrm{grad}[c_C]  
    \bullet \boldsymbol{n} \, \mathrm{d} \Gamma + \int \limits_{\Omega} 
    n_C k_{AB} c_A c_B \, \mathrm{d} \Omega \nonumber \\
    &\implies h_{C,lb} \, \mathrm{meas}(\partial \Omega) + \int \limits_{\Omega} 
    n_C k_{AB} c_A c_B \, \mathrm{d} \Omega \leq  \frac{d}{dt} \int \limits_{\Omega} 
    c_C \, \mathrm{d} \Omega \leq h_{C,ub} \, \mathrm{meas}(\partial 
    \Omega) + \int \limits_{\Omega} n_C k_{AB} c_A c_B \, \mathrm{d} \Omega
  \end{align}
  
  As $0 \leq \mathrm{min}\left[c_F^{\mathrm{min}}, c_G^{\mathrm{min}} \right] 
  \leq c_i \leq \mathrm{max}\left[c_F^{\mathrm{max}}, c_G^{\mathrm{max}} 
  \right] \quad \forall i = A, B, C$, from Eq.~\ref{Eqn:Thm:Estimates_AvgConc_ABC_1}, 
  we have the following relationship
  \begin{align}
    \label{Eqn:Thm:Estimates_AvgConc_ABC_3}
    &h_{A,lb} \, \mathrm{meas}(\partial \Omega)- n_A k_{AB} \, 
    \mathrm{meas}(\Omega) \mathrm{max}\left[c_F^{\mathrm{max}}, 
    c_G^{\mathrm{max}} \right]^2 \leq         
    h_{A,lb} \, \mathrm{meas}(\partial \Omega) - n_A k_{AB} \mathrm{max}
    \left[c_F^{\mathrm{max}}, c_G^{\mathrm{max}} \right] \int \limits_{\Omega} 
    c_A \, \mathrm{d} \Omega \leq \nonumber \\
    &\frac{d}{dt} \int \limits_{\Omega} 
    c_A \, \mathrm{d} \Omega \leq h_{A,ub} \, \mathrm{meas}(\partial 
    \Omega) - n_A k_{AB} \mathrm{min}\left[c_F^{\mathrm{min}}, 
    c_G^{\mathrm{min}} \right] \int \limits_{\Omega} c_A \, 
    \mathrm{d} \Omega \leq \nonumber \\
    &h_{A,ub} \, \mathrm{meas}(\partial \Omega)- n_A k_{AB} \, 
    \mathrm{meas}(\Omega) \mathrm{min}\left[c_F^{\mathrm{min}}, 
    c_G^{\mathrm{min}} \right]^2 \leq h_{A,ub} \, \mathrm{meas}(\partial 
    \Omega)
  \end{align}
  
  Dividing Eq.~\ref{Eqn:Thm:Estimates_AvgConc_ABC_3} with $\mathrm{max}\left[\langle c_A(t) 
  \rangle \right]$, we get the desired inequality given by Eq.~\ref{Eqn:Thm_AvgAB_Bound} 
  for $c_A$. Repeating the same process for species $B$, we get the inequality given by 
  Eq.~\ref{Eqn:Thm_AvgAB_Bound} for $c_B$. This completes the first part of the proof.
  
  For the second part of the proof, from Eqs.~\ref{Eqn:Definitions_of_F}--\ref{Eqn:Definitions_of_G}, 
  we have the following inequalities on $\left(c_A, c_B \right)_0$
  \begin{align}
    \label{Eqn:Thm:Estimates_AvgConc_ABC_4}
    &\left(c_A, c_B \right)_0 = \left(c_F- \frac{n_A}{n_C} c_C, c_G - 
    \frac{n_B}{n_C} c_C \right)_0 = \nonumber \\
    &\left(c_F, c_G \right)_0 - \frac{n_B}{n_C}\left(c_F, c_C \right)_0 
    - \frac{n_A}{n_C}\left(c_G, c_C \right)_0 
    + \frac{n_A n_B}{n^2_C}\left(c_C, c_C \right)_0 \\
    \label{Eqn:Thm:Estimates_AvgConc_ABC_5}
    & - \left(\frac{n_B}{n_C}c_F + \frac{n_A}{n_C}c_G , c_C \right)_0 
    \leq \left(c_A, c_B \right)_0 \leq \left(c_F, c_G \right)_0
    + \frac{n_A n_B}{n^2_C}\left(c_C, c_C \right)_0 = 
    \left(c_F, c_G \right)_0 + \frac{n_A n_B}{n^2_C} \| c_C \|^2_0
  \end{align}
  
  Invoking Cauchy-Schwarz inequality on $(c_F, c_G)_0$, $\| c_F \|_0$, $\| c_G \|_0$, 
  and $\| c_C \|^2_0$, we have following inequalities
  \begin{align}
    \label{Eqn:Thm:Estimates_AvgConc_ABC_6}
    &\left(c_F, c_G \right)_0 \leq \| c_F \|_0 \| c_G \|_0 \\
    \label{Eqn:Thm:Estimates_AvgConc_ABC_7}
    &0 \leq c_i^{\mathrm{min}} \, \sqrt{\mathrm{meas}(\Omega)} \leq \| c_i \|_0 \leq 
    c_i^{\mathrm{max}} \, \sqrt{\mathrm{meas}(\Omega)} \quad \forall i = F, G \\
    \label{Eqn:Thm:Estimates_AvgConc_ABC_8}
    &0 \leq \mathrm{meas}(\Omega) \, \mathrm{min}\left[c_F^{\mathrm{min}}, 
    c_G^{\mathrm{min}} \right]^2 \leq \mathrm{min}\left[c_F^{\mathrm{min}}, 
    c_G^{\mathrm{min}} \right] \left \langle c_C(t) \right \rangle \leq 
    \| c_C \|^2_0 \leq \nonumber \\
    &\mathrm{max}\left[c_F^{\mathrm{max}}, c_G^{\mathrm{max}} \right] \left 
    \langle c_C(t) \right \rangle \leq \mathrm{meas}(\Omega) \, \mathrm{max}
    \left[c_F^{\mathrm{max}}, c_G^{\mathrm{max}} \right]^2
  \end{align}
  
   From Eq.~\ref{Eqn:Thm:Estimates_AvgConc_ABC_5} and 
   Eqs.~\ref{Eqn:Thm:Estimates_AvgConc_ABC_6}--\ref{Eqn:Thm:Estimates_AvgConc_ABC_8}, we have 
   \begin{align}
     \label{Eqn:Thm:Estimates_AvgConc_ABC_9}
     & - \frac{1}{n_C} \, \mathrm{meas}(\Omega) \, \mathrm{max} 
     \left[c_F^{\mathrm{max}}, c_G^{\mathrm{max}} \right] \, 
     \left(n_A c_G^{\mathrm{max}} + n_B c_F^{\mathrm{max}} \right) 
     \leq - \frac{1}{n_C} \left(n_A c_G^{\mathrm{max}} + n_B 
     c_F^{\mathrm{max}} \right) \left \langle c_C(t) \right 
     \rangle \leq \nonumber \\
     &\left(c_A, c_B \right)_0 \leq \mathrm{meas}(\Omega) \, 
     c_F^{\mathrm{max}} c_G^{\mathrm{max}} + \frac{n_A n_B}{n^2_C}
     \mathrm{max}\left[c_F^{\mathrm{max}}, c_G^{\mathrm{max}} \right] 
     \left \langle c_C(t) \right \rangle \leq \nonumber \\
     & \mathrm{meas}
     (\Omega) \, c_F^{\mathrm{max}} c_G^{\mathrm{max}} + \frac{n_A n_B}{n^2_C}
     \, \mathrm{meas}(\Omega) \, \mathrm{max}\left[c_F^{\mathrm{max}}, 
     c_G^{\mathrm{max}} \right]^2
   \end{align}
   
   From Eq.~\ref{Eqn:Thm:Estimates_AvgConc_ABC_2} and Eq.~\ref{Eqn:Thm:Estimates_AvgConc_ABC_9}, 
   we have
   \begin{align}
     \label{Eqn:Thm:Estimates_AvgConc_ABC_10}
     & h_{C,lb} \, \mathrm{meas}(\partial \Omega) - k_{AB} \, 
     \mathrm{meas}(\Omega) \, \mathrm{max} \left[c_F^{\mathrm{max}}, 
     c_G^{\mathrm{max}} \right] \, \left(n_A c_G^{\mathrm{max}} 
     + n_B c_F^{\mathrm{max}} \right) \leq \nonumber \\
     &h_{C,lb} \, \mathrm{meas}(\partial \Omega) - k_{AB} 
     \left(n_A c_G^{\mathrm{max}} + n_B c_F^{\mathrm{max}} 
     \right) \left \langle c_C(t) \right \rangle \leq \nonumber \\
     & \frac{d}{dt} \int \limits_{\Omega} c_C \, \mathrm{d} \Omega 
     \leq h_{C,ub} \, \mathrm{meas}(\partial \Omega) + 
     n_C k_{AB} \, \mathrm{meas}(\Omega) \, c_F^{\mathrm{max}} 
     c_G^{\mathrm{max}} + \frac{n_A n_B k_{AB}}{n_C} \, \mathrm{max}
     \left[c_F^{\mathrm{max}}, c_G^{\mathrm{max}} \right] 
     \left \langle c_C(t) \right \rangle \leq \nonumber \\
     & h_{C,ub} \, \mathrm{meas}(\partial \Omega) + 
     n_C k_{AB} \, \mathrm{meas}(\Omega) \, c_F^{\mathrm{max}} 
     c_G^{\mathrm{max}} + \frac{n_A n_B k_{AB}}{n_C} \, \mathrm{meas}
     (\Omega) \, \mathrm{max}\left[c_F^{\mathrm{max}}, 
     c_G^{\mathrm{max}} \right]^2
   \end{align}
   
   Dividing Eq.~\ref{Eqn:Thm:Estimates_AvgConc_ABC_10} with $\mathrm{max}\left[\langle c_C(t) 
  \rangle \right]$, we get the desired inequality given by Eq.~\ref{Eqn:Thm_AvgC_Bound} 
  for $c_C$.
\end{proof}

\begin{theorem}[\texttt{Estimates on average of square of concentration $\mathbb{c}_i$}]
  \label{Thm:Estimates_AvgSqConc_ABC}
  Based on the assumptions outlined in Thm.~\ref{Thm:Estimates_F_and_G}, 
  if $c_i(\mathbf{x},t) = 0$ on $\partial \Omega \, \times \, ]0, \mathcal{I}[$,
  then the quantity $\frac{d\mathbb{c}_i}{dt}$ is bounded as follows
  \begin{align}
    \label{Eqn:Thm_AvgSqAB_Bound}
    &-\frac{2}{\mathrm{max}\left[\langle c^2_i \rangle \right]} 
    \left(\lambda_{max} \| \mathrm{grad}[c_i] \|^2_0 + n_i k_{AB} \, 
    \mathrm{meas}(\Omega) \, \mathrm{max}\left[c_F^{\mathrm{max}}, 
    c_G^{\mathrm{max}} \right]^3 \right) \leq \frac{d\mathbb{c}_i}{dt} 
    \leq \nonumber \\
    &-\frac{2\lambda_{min}}{C^2_{pf}} \mathbb{c}_i \leq
    -\frac{2\lambda_{min}}{\mathrm{max}\left[\langle c^2_i \rangle \right] \, C^2_{pf}} \, \mathrm{meas}(\Omega) \, 
    \mathrm{min}\left[c_F^{\mathrm{min}}, c_G^{\mathrm{min}} 
    \right]^2 \leq 0 \quad \forall i = A, B \\
    \label{Eqn:Thm_AvgSqC_Bound}
    &-\frac{2 \lambda_{max}}{\mathrm{max}\left[\langle c^2_i 
    \rangle \right]} \| \mathrm{grad}[c_i] \|^2_0 \leq 
    \frac{d\mathbb{c}_i}{dt} \leq -\frac{2\lambda_{min}}{C^2_{pf}} 
    \mathbb{c}_i + \frac{2 n_C k_{AB} \, \mathrm{meas}(\Omega) \, 
    \mathrm{max}\left[c_F^{\mathrm{max}}, c_G^{\mathrm{max}} 
    \right]^3 }{\mathrm{max}\left[\langle c^2_i \rangle \right]} 
    \leq \nonumber \\
    &-\frac{2\lambda_{min}}{\mathrm{max}\left[\langle c^2_i \rangle \right] C^2_{pf}} \, \mathrm{meas}(\Omega) \, 
    \mathrm{min}\left[c_F^{\mathrm{min}}, c_G^{\mathrm{min}} \right]^2 + 
    \frac{2 n_C k_{AB} \, \mathrm{meas}(\Omega) \, \mathrm{max} 
    \left[c_F^{\mathrm{max}}, c_G^{\mathrm{max}} \right]^3 }{\mathrm{max} 
    \left[\langle c^2_i \rangle \right]} \quad \forall i = C
  \end{align}
\end{theorem}
\begin{proof}
  Herein, we follow a similar proof procedure as outlined 
  in Thm.~\ref{Thm:Estimates_F_and_G} (where we derived Eq.~\ref{Eqn:Thm_AvgSqCFG_Bound}). 
  Multiplying Eq.~\ref{Eqn:DRs_for_A} with $c_A$, 
  using the divergence theorem, using $c_A(\mathbf{x},t) = 0$ 
  on $\partial \Omega \, \times \, ]0, \mathcal{I}[$, and the 
  identity that $\mathrm{div} \left[c_A \boldsymbol{D} \mathrm{grad}[c_A] 
  \right] = \boldsymbol{D} \mathrm{grad}[c_A] \bullet\mathrm{grad}[c_A] 
  + c_A \mathrm{div}\left[\boldsymbol{D} \mathrm{grad}[c_A] \right]$, 
  we obtain the following relationship

  \begin{align}
    \label{Eqn:Thm:Estimates_AvgSqConc_ABC_1}
    \frac{1}{2} \frac{d}{dt} \int \limits_{\Omega} c^2_A \, \mathrm{d} 
    \Omega = - \int \limits_{\Omega} \boldsymbol{D} \, \mathrm{grad}[c_A] 
    \bullet \mathrm{grad}[c_A] \, \mathrm{d} \Omega - n_A k_{AB} \int 
    \limits_{\Omega} c^2_A c_B \, \mathrm{d} \Omega
  \end{align}
  
  From the spectral theorem, the Poincar\'{e}-Friedrichs inequality, and the 
  relationship that $0 \leq \mathrm{min}\left[c_F^{\mathrm{min}}, 
  c_G^{\mathrm{min}} \right] \leq c_i \leq \mathrm{max}\left[c_F^{\mathrm{max}}, 
  c_G^{\mathrm{max}} \right] \quad \forall i = A, B, C$, we have
  
  \begin{align}
    \label{Eqn:Thm:Estimates_AvgSqConc_ABC_2}
    & - \lambda_{max} \| \mathrm{grad}[c_A] \|^2_0 - n_A k_{AB} \, 
    \mathrm{meas}(\Omega) \, \mathrm{max}\left[c_F^{\mathrm{max}}, 
    c_G^{\mathrm{max}} \right]^3 \leq \frac{1}{2} \frac{d}{dt} 
    \int \limits_{\Omega} c^2_A \, \mathrm{d} \Omega \leq 
    - \lambda_{min} \| \mathrm{grad}[c_A] \|^2_0\nonumber \\
    & -\frac{\lambda_{min}}{C^2_{pf}} \| c_A \|^2_0 \leq
    -\frac{\lambda_{min}}{C^2_{pf}} \, \mathrm{meas}(\Omega) \, 
    \mathrm{min}\left[c_F^{\mathrm{min}}, c_G^{\mathrm{min}} 
    \right]^2 \leq 0 
  \end{align}
  
  Dividing Eq.~\ref{Eqn:Thm:Estimates_AvgSqConc_ABC_2} with 
  $\mathrm{max}\left[\langle c^2_A \rangle \right]$,
  we get Eq.~\ref{Eqn:Thm_AvgSqAB_Bound}. Repeating the same
  process for $c_B$, we get Eq.~\ref{Eqn:Thm_AvgSqAB_Bound}.
  This completes the first part of the proof.
  
  For the second part of the proof, we have 
  \begin{align}
    \label{Eqn:Thm:Estimates_AvgSqConc_ABC_3}
    &\frac{1}{2} \frac{d}{dt} \int \limits_{\Omega} c^2_C \, \mathrm{d} 
    \Omega = - \int \limits_{\Omega} \boldsymbol{D} \, \mathrm{grad}[c_C] 
    \bullet \mathrm{grad}[c_C] \, \mathrm{d} \Omega + n_C k_{AB} \int 
    \limits_{\Omega} c_A c_B c_C \, \mathrm{d} \Omega
  \end{align}
  
  Invoking the spectral theorem and the Poincar\'{e}-Friedrichs inequality, we have
  \begin{align}
    \label{Eqn:Thm:Estimates_AvgSqConc_ABC_4}
    & - \lambda_{max} \| \mathrm{grad}[c_C] \|^2_0 \leq 
    \frac{1}{2} \frac{d}{dt} \int \limits_{\Omega} c^2_C \, \mathrm{d} 
    \Omega \leq - \lambda_{min} \| \mathrm{grad}[c_C] \|^2_0 + n_C k_{AB} 
    \int \limits_{\Omega} c_A c_B c_C \, \mathrm{d} \Omega \leq \nonumber \\
    &-\frac{\lambda_{min}}{C^2_{pf}} \| c_C \|^2_0 + n_C k_{AB} \, 
    \mathrm{meas}(\Omega) \, \mathrm{max}\left[c_F^{\mathrm{max}}, 
    c_G^{\mathrm{max}} \right]^3 \leq \nonumber \\
    &-\frac{\lambda_{min}}{C^2_{pf}} \, \mathrm{meas}(\Omega) \, 
    \mathrm{min}\left[c_F^{\mathrm{min}}, c_G^{\mathrm{min}} \right]^2 + 
    n_C k_{AB} \, \mathrm{meas}(\Omega) \, \mathrm{max} 
    \left[c_F^{\mathrm{max}}, c_G^{\mathrm{max}} \right]^3 \leq \nonumber \\
    &n_C k_{AB} \, \mathrm{meas}(\Omega) \, \mathrm{max} 
    \left[c_F^{\mathrm{max}}, c_G^{\mathrm{max}} \right]^3
  \end{align}
  
  Dividing Eq.~\ref{Eqn:Thm:Estimates_AvgSqConc_ABC_4} with 
  $\mathrm{max}\left[\langle c^2_C \rangle \right]$,
  we get Eq.~\ref{Eqn:Thm_AvgSqC_Bound}, which completes the proof.  
\end{proof}

\begin{theorem}[\texttt{Estimates on degree of mixing $\sigma^2_i$}]
  \label{Thm:Estimates_DoM_ABC}
  There exist constants $\mathcal{A}_i$, $\mathcal{B}_i$, and $\mathcal{M}_i$ such that $0 \leq \sigma^2_i \leq \mathcal{A}_i + \mathcal{M}_i \, e^{-\mathcal{B}_it} \leq 1 \quad \forall i = A, B, C, F, G$.
\end{theorem}
\begin{proof}
  From Eq.~\ref{Eqn:Degree_of_Mixing}, we have $0 \leq \sigma^2_i \leq 
  \frac{\langle c^2_i \rangle}{\mathrm{max} \left[\langle c^2_i \rangle \right]}$.
  That is, $0 \leq \sigma^2_i \leq \mathbb{c}_i$. From Eqs.~\ref{Eqn:Thm_AvgSqCFG_Bound}, 
  \ref{Eqn:Thm_AvgSqAB_Bound}, and \ref{Eqn:Thm_AvgSqC_Bound}, we have the following 
  inequalities
  \begin{align}
    \label{Eqn:Decay_DoM_ABFG}
    &\frac{d\mathbb{c}_i}{dt} \leq -\frac{2\lambda_{min}}{C^2_{pf}} 
    \mathbb{c}_i \quad \forall i = A, B, F, G \\
    \label{Eqn:Decay_DoM_C}
    &\frac{d\mathbb{c}_C}{dt} \leq -\frac{2\lambda_{min}}{C^2_{pf}} 
    \mathbb{c}_C + \frac{2 n_C k_{AB} \, \mathrm{meas}(\Omega) \, 
    \mathrm{max}\left[c_F^{\mathrm{max}}, c_G^{\mathrm{max}} 
    \right]^3 }{\mathrm{max}\left[\langle c^2_C \rangle \right]}
  \end{align}
  On integrating Eq.~\ref{Eqn:Decay_DoM_ABFG} and Eq.~\ref{Eqn:Decay_DoM_C}, we get the following relationships
  \begin{align}
    \label{Eqn:DoM_Eq_ABFG}
    &\mathbb{c}_i \leq \frac{1}{\mathcal{B}_i} \left(\mathbb{A}_i +
    e^{-\mathcal{B}_it} \right) \quad \forall i = A, B, F, G \\
    \label{Eqn:DoM_Eq_C}
    &\mathbb{c}_C \leq \frac{1}{\mathcal{B}_i} \left(\frac{2 n_C k_{AB} 
    \, \mathrm{meas}(\Omega) \, \mathrm{max}\left[c_F^{\mathrm{max}}, 
    c_G^{\mathrm{max}} \right]^3}{\mathrm{max}\left[\langle c^2_C 
    \rangle \right]} \pm \mathbb{A}_i \pm e^{-\mathcal{B}_it} \right)
  \end{align}
  where $\mathbb{A}_i$ is the integration constant, $\mathcal{B}_i = \frac{2 
  \lambda_{min}}{C^2_{pf}}$, and $\mathcal{M}_i = \pm \frac{1}{\mathcal{B}_i}$.
  For $i = A, B, F, G$, we have $\mathcal{A}_i = \frac{\mathbb{A}_i}{\mathcal{B}_i}$ and for $i = C$, we have $\mathcal{A}_i = \frac{2 n_C k_{AB} \, \mathrm{meas}(\Omega) \, 
  \mathrm{max}\left[c_F^{\mathrm{max}}, c_G^{\mathrm{max}} \right]^3 \pm 
  \mathbb{A}_i}{\mathrm{max}\left[\langle c^2_C \rangle \right] \, \mathcal{B}_i}$ 
  such that $\sigma^2_C$ is non-negative. 
  From Eq.~\ref{Eqn:DoM_Eq_ABFG} and Eq.~\ref{Eqn:DoM_Eq_C}, it is clear that 
  $0 \leq \sigma^2_i \leq \mathcal{A}_i + \frac{1}{\mathcal{B}_i} \, e^{-\mathcal{B}_it} 
  \quad \forall i = A, B, C, F, G$. 
  The Poincar\'{e}-Friedrichs inequality constant $C_{pf}$ can be choosen such that 
  $0 \leq \sigma^2_i \leq 1$, which completes the proof.
\end{proof}

From Thm.~\ref{Thm:Estimates_F_and_G}--\ref{Thm:Estimates_DoM_ABC}, most of the times, we can expect that species $A$ and $B$ decay and mix in an exponential manner with respect to time.
Also, we can expect that species $C$ is produced and mixed in an exponential fashion (which forms the basis for choosing SVM and SVR kernels).
Note Thm.~\ref{Thm:Estimates_F_and_G}--\ref{Thm:Estimates_DoM_ABC} provide lower and upper bounds on the species decay and mixing rates.
For the above QoIs, we would like to construct reduced-order models (ROMs) using machine learning (ML). 
In the next section, we shall provide a ML framework to construct ROMs for QoIs given by Eqs.~\eqref{Eqn:Avg_Conc}--\eqref{Eqn:Degree_of_Mixing}.

\section{MACHINE LEARNING FRAMEWORK}
\label{Sec:S3_ROM_Framework}
In this section, we shall present a ML framework to preprocess data, perform feature analysis, identify important features, and construct ML-ROMs.

\subsection{Pre-processing and feature analysis}
\label{SubSec:QoIs_FeatureAnalysis}
The features that correspond to QoIs include $T$, $\log[\frac{\alpha_L}{\alpha_T}]$, $\kappa_fL$, $\log[v_o]$, and $D_m$.
Pre-processing and standardizing the features are needed for kernel-based ML methods (as SVMs and SVRs are sensitive to feature transformations).
This is because SVMs and SVRs are based on distance-metrics or similarities (e.g. in form of scalar product) between training samples.
Feature scaling is performed to normalize data, which ensures that priority is not given to a particular feature during ML-ROM construction.
There are various types of feature scaling for ML.
Examples include standardization (or $Z$-score normalization) with zero mean and unit variance, Min-Max scaling or scaling features to a given range, quantile transforms, power transforms, and feature binarization.
As inputs (e.g., $T$, $\alpha_L$, $\alpha_T$, $\kappa_fL$, $v_o$, and $D_m$) to the high-fidelity reactive-mixing model are non-negative, we use a Min-Max feature scaling.
This feature scaling estimator scales and translates each feature individually such that it is in the range of $[0,1]$. 
There are various advantages of this feature scaling for our problem.
The main advantage is that Min-Max feature scaling provides easier and physically meaningful interpretation relating inputs to non-negative QoIs.
Another advantage is that having features in this bounded range will result in smaller standard deviations.
This can suppress the effect of outliers.
That is, Min-Max feature scaling is robust to very small standard deviations of features, which makes the ROM development and predictions sensible.

To understand the role of each input parameter or feature, we perform feature importance using F-test, mutual information (MI) criteria, and Random Forests (RF) methods. 
F-test is an univariate statistical test, which has a F-distribution under the null hypothesis.
The main advantage is that it provides the significance of each feature in improving the ML-ROM.
But there are some drawbacks of using F-Test for selecting features. 
F-Test only captures linear relationships between features and QoIs. 
In F-test, a highly correlated feature is given higher score and less correlated features are given lower score.
Note that correlation is highly misleading because it does not capture strong non-linear relationships.
As a result, we use MI criteria and RF methods to identify features that are related to QoIs in non-linear fashion.
Mutual Information between a feature and QoI is a non-negative value that measures the dependence of a feature to QoI.
MI value is equal to zero if and only if a feature and QoI are independent.
If MI value is non-zero and corresponds to a high value, it mean there is a higher dependency between that feature and QoI.
MI methods can capture any kind of statistical dependency.
However, they are computationally intensive and require more training samples for accurate estimation of a non-linear relationship, as these methods are non-parametric.
RF is an ensemble learning method that combines the predictions of several decision trees in order to enhance generalizability and robustness over a single tree-based model.
They are one of the successful ML methods because they provide good predictions, have low overfitting, and are easy to interpret. 
Decision Trees keep the most important features near the root of the tree.
Hence, constructing a decision tree involves calculating the best predictive feature.
By averaging the prediction estimates of several decision trees, RF method reduces the variance of such an estimate and uses it for calculating feature importance.
As a result, RF method provides the relative importance of each feature with respect to the predictability of QoI.

For the sake of illustration purposes, we have divided the mixing state of the system in to four different classes.
Each class provides information on the degree of mixing.
\texttt{Class-1} corresponds to well-mixed/strong mixed system. 
For this class, $\sigma^2_i$ is between 0.0 to 0.25.
As time progress the system tends to mix well and the value for degree of mixing tends to 0.0.
\texttt{Class-2} represents moderate mixing. 
For this class, $\sigma^2_i$ is between 0.25 to 0.5.
\texttt{Class-3} represents weak mixing. 
For this class, $\sigma^2_i$ is between 0.5 to 0.75.
These two classes provide details on the mixing state at moderate times or when anisotropy is dominant.
\texttt{Class-4} corresponds to confined/ultra-weak mixing. 
For this class, $\sigma^2_i$ is between 0.75 to 1.0.
This state corresponds to system that is segregated, which happens during the initial times.
Note that, we can define more classes (or states of mixing) if needed.
Then, the classification ROMs have to be re-trained to predict the new classes. 
In addition to classifying the state of mixing, we construct surrogate models for a continuous value of degree of mixing.
In the next subsection, we provide details on the construction of ML-ROMs for both classification and regression.

\subsection{Mathematical formulation for SVR-ROMs and SVM-ROMs}
\label{SubSec:S3_SVM_SVR_ML}
Given training inputs (e.g., data vectors or samples) $\boldsymbol{x}_i \in \mathbb{R}^{p}$ where $i = 1,2, \cdots, n$ and a target vector $\boldsymbol{y} \in \mathbb{R}^n$, then $\varepsilon$-insensitive SVR is formulated as minimization of the following functional (primal problem) \cite{scholkopf2001learning}:
\begin{align}
  \label{Eqn:Math_SVR}
  \mathop{\mathrm{minimize}}_{\boldsymbol{w}, b, \boldsymbol{\zeta}, 
  \boldsymbol{\zeta^{*}} } & \quad \frac{1}{2} \langle \boldsymbol{w}; 
  \boldsymbol{w} \rangle + \mathcal{P} \sum_{i=1}^{\mathrm{n}} \left( 
  \zeta_i + \zeta^{*}_i \right) \\
  \label{Eqn:Math_SVR_Constraints}
  \mbox{subject to} & \quad 
  \begin{cases}
    y_i - \boldsymbol{w} \bullet \boldsymbol{\phi}(\boldsymbol{x}_i) - b 
    \leq \varepsilon + \zeta_i \\
    \boldsymbol{w} \bullet \boldsymbol{\phi}(\boldsymbol{x}_i) + b - y_i
    \leq \varepsilon + \zeta^{*}_i \\
    \zeta_i, \zeta^{*}_i \geq 0, \quad \forall i = 1,2, \cdots, n
  \end{cases}
\end{align}
where $\boldsymbol{w}$ is the weight vector or normal vector to the separating hyperplane. 
$\mathcal{P}$ is the regularization parameter to prevent over-fitting. 
It is a positive number that controls the penalty imposed on training data that lie outside the $\varepsilon$-band. 
Moreover, $\mathcal{P}$ determines the trade-off between the SVR-ROM complexity and the degree to which deviations larger than $\varepsilon$ are tolerated in the constrained optimization problem. 
The parameter $\varepsilon$ controls the width of the $\varepsilon$-insensitive zone. 
The value of $\varepsilon$ can affect the number of support vectors used to construct the regression decision function. 
Bigger $\varepsilon$-value means that fewer support vectors are selected. 
On the other hand, bigger $\varepsilon$-values results in `more flat' estimates, meaning that, predictions of SVR-ROM are less sensitive to perturbations in the model inputs. 
$b$ is the bias term. 
$\zeta_i$ and $\zeta^{*}_i$ are non-negative slack variables to measure the deviation of training data outside $\varepsilon$-insensitive band. 
$y_i$ is the $i$-th element of the target vector $\boldsymbol{y}$. 
$\phi(\boldsymbol{x}_i)$ is non-linear mapping of training data $\boldsymbol{x}_i$ into a higher dimensional feature space such that one can perform a linear regression on the transformed data points $\phi(\boldsymbol{x}_i)$. 

Figure \ref{Fig:Pic_Support_Vectors} provides a pictorial description of maximum margin hyperplane, support vectors, and non-linear mapping function from input space to feature space. 
The top left and top right figure shows an example of clustered data and infinitely-possible hyperplanes separating the data.
Among the infinitely-possible scenarios, we choose a hyperplane which maximizes the margin (as shown in the middle figure). 
This results in a classifier or a regression model that has a better generalization with lower empirical risk (which is defined as the average loss of an estimator for a finite set of data drawn from a bigger dataset). 
The training of maximum-margin hyperplane classifier is achieved by solving a linearly constrained optimization problem. 
The middle figure also shows the support vectors. 
These are data points that lie close to the decision hyperplane/surface. 
They have direct bearing on the optimal location of the surface. 
Note that the SVM-ROMs and SVR-ROMs depend only on these support vectors and all other points are ignored. 
The bottom left figure shows an example of data that is not linearly separable. 
So we map the data on to a higher dimensional space and then to use a linear classifier in the higher dimensional space, which is shown in bottom right figure. 
This is the general idea behind the non-linear SVMs and SVRs. 
That is, we map the original feature space to certain higher-dimensional feature space via a non-linear transformation $\Phi:\boldsymbol{x} \rightarrow \phi(\boldsymbol{x})$. 
This transformation ensures that the underlying data is linearly separable and is performed in ways that the resultant classifier/regression generalize well.

The above formulation given by Eqs.~\ref{Eqn:Math_SVR}--\ref{Eqn:Math_SVR_Constraints} can be transformed into the dual optimization problem.
Its solution $\mathbb{y}_{\mathrm{{\tiny svr}}}$ (SVR-ROMs), which is also called as $\varepsilon$-insensitive SVR decision function, is given as follows:
\begin{align}
  \label{Eqn:Math_SVR_Decision_Function}
  \mathbb{y}_{\mathrm{{\tiny svr}}}(\boldsymbol{x}) 
  = \sum_{j=1}^{\mathrm{n_{{\tiny Supp}}}} \left(\alpha_j - 
  \alpha_j^* \right) \mathcal{K}(\boldsymbol{x}, 
  \mathbb{x}_j) + b,  \quad \mathrm{such \, \, that} 
  \; \; 0 \leq \alpha_j, \alpha_j^* \leq \mathcal{P}, \; \; 
  \forall j = 1,2, \cdots , \mathrm{n_{{\tiny Supp}}}
\end{align}
where $\alpha_j, \alpha_j^*$ are the Lagrange multipliers enforcing the $\varepsilon$-insensitive loss function, $\mathrm{n_{{\tiny Supp}}}$ is the number of support vectors, $\mathbb{x}_j$ are the support vectors, and $\mathcal{K}(\boldsymbol{x}, \mathbb{x}_j)$ is the kernel function. 
This function computes the inner product between mapped vectors in feature space $\phi(\boldsymbol{x})$.
In this paper, we use Radial Basis Function (RBF) kernel for both SVRs and SVMs.
This is because the QoIs decrease or increase in an exponential fashion (e.g., see Thm.~\ref{Thm:Estimates_DoM_ABC}). 
The RBF kernel uses Gaussian form given by:
\begin{align}
  \label{Eqn:RBF_Kernel}
  \mathcal{K}(\boldsymbol{x}, \mathbb{x}_j) = 
  \exp\left(-\gamma \| \boldsymbol{x} - \mathbb{x}_j 
  \|^2 \right)
\end{align}
where $\gamma$ is a user defined parameter, which corresponds to the smoothness of the decision surface. 
The output of the RBF kernel depends upon the Euclidean distance between the test vector $\boldsymbol{x}$ and the support vector $\mathbb{x}_j$, respectively. 
$\|\boldsymbol{x} - \mathbb{x}_j \|$ is a similarity measure, which tell us how far is a given input/feature from a support vector.

Similar to SVR-ROMs, the constrained optimization (primal) problem for a two-class SVM-ROMs is as follows:
Given a training data $\boldsymbol{x}_i \in \mathbb{R}^{p}$ where $i = 1,2, \cdots, n$ and labeled data $\boldsymbol{y} \in \{+1, -1 \}^n$, then 
\begin{align}
  \label{Eqn:Math_SVM}
  \mathop{\mathrm{minimize}}_{\boldsymbol{w}, b, \boldsymbol{\zeta}} 
  & \quad \frac{1}{2} \langle \boldsymbol{w}; \boldsymbol{w} \rangle + 
  \mathcal{P} \sum_{i=1}^{\mathrm{n}} \zeta_i \\
  \label{Eqn:Math_SVM_Constraints}
  \mbox{subject to} & \quad 
  \begin{cases}
    y_i \left(\boldsymbol{w} \bullet \boldsymbol{\phi}(\boldsymbol{x}_i) + b \right)
    \geq 1 - \zeta_i \\
    \zeta_i \geq 0, \quad \forall i = 1,2, \cdots, n
  \end{cases}
\end{align}
The solution to Eqs.~\ref{Eqn:Math_SVM}--\ref{Eqn:Math_SVM_Constraints} is $\mathbb{y}_{\mathrm{{\tiny svm}}}$ (SVR-ROMs), which is given as follows:
\begin{align}
  \label{Eqn:Math_SVM_Decision_Function}
  \mathbb{y}_{\mathrm{{\tiny svm}}}(\boldsymbol{x}) 
  = \displaystyle \mathrm{sgn} \left[ \sum_{j=1}^
  {\mathrm{n_{{\tiny Supp}}}} y_j \alpha_j \mathcal{K}
  (\boldsymbol{x}, \mathbb{x}_j) + b \right] \quad 
  \mathrm{such \, \, that} \; \; 0 \leq \alpha_j \leq 
  \mathcal{P}, \; \; \forall j = 1,2, \cdots , 
  \mathrm{n_{{\tiny Supp}}}
\end{align}

However, as we have a multi-class classification problem (where the degree of mixing is classified into four different classes), we use one-vs-one (OvO) strategy.
Using this, the multi-class classification problem is transformed into multiple binary classification problems.
Additional parameters and constraints are added to the optimization problem given by Eqs.~\ref{Eqn:Math_SVM}--\ref{Eqn:Math_SVM_Constraints} to handle the separation of the different classes.
For a 4-way multi-class problem, we train 6 binary classifiers.
Each classifier receives data related to a pair of classes from the original training data set.
Then, it will learn to distinguish these pair of classes. 
During prediction, all the 6 binary classifiers are applied to an unseen test sample.
The class that gets the highest number of `+1' predictions gets predicted by the combined classifier \cite{scholkopf2001learning}.

In this paper, the SVM-ROMs and SVR-ROMs are constructed and implemented using the ML algorithms provided in the \textsf{scikit-learn} python software package \cite{2011_scikit_learn}. 
Internally, \textsf{scikit-learn} uses \textsf{libsvm} and \textsf{liblinear} libraries to handle all the optimization calculations \cite{2011_Chang_Lin_LibSVM,2008_Fan_etal_Liblinear}. 
However, it should be noted that the constrained optimization problem and the decision function given by Eqs.~\eqref{Eqn:Math_SVR}--\eqref{Eqn:Math_SVR_Constraints} and \eqref{Eqn:Math_SVR_Decision_Function} can produce negative values for QoIs (even though the RBF kernel is non-negative). 
This is because $\alpha_j$ can be less than $\alpha_j^*$. 
In such a case, we clip the negative values. 
On the other hand, in the dual problem, one can enforce the constraint that the weights/Lagrange multipliers `$\alpha_j \geq \alpha_j^*$' and the bias `$b \geq 0$' to ensure non-negativity for QoIs when constructing SVMs and SVRs using non-negative kernels. 
This involves changing the underlying algorithms in \textsf{libsvm} and \textsf{liblinear} software packages, which is beyond the scope of the current paper.

We note that ML methods based on SVMs and SVRs are powerful and attractive to construct ROMs.
This is because they need less training data and provide decent prediction capabilities.
Moreover, they can quickly generate a good ML-model, if training data size is low (e.g., see Tab.~\ref{Tab:SVM_SVR_ROMs_Accuracy_2} for time taken to construct SVM-ROMs and SVR-ROMs).
Another advantage of SVMs and SVRs methods is that they have a unique global minimizer \cite{2011_Chang_Lin_LibSVM,
2008_Fan_etal_Liblinear}.
Hence, robust solvers that exist to solve convex quadratic programming (CQP) can be applied to solve Eqs.~\ref{Eqn:Math_SVR}--\ref{Eqn:Math_SVR_Constraints} and Eqs.~\ref{Eqn:Math_SVM}--\ref{Eqn:Math_SVM_Constraints}.
However, these SVM and SVR methods are time consuming to solve if training data is very large.
Their computational time and corresponding storage requirements increase rapidly with the number of training samples.
In our case, the computational cost of \textsf{libsvm}-based CQP solver is typically between $\mathcal{O}(n_{features} \times n^2_{samples})$ and $\mathcal{O}(n_{features} \times n^3_{samples})$ \cite{2011_Chang_Lin_LibSVM,
2008_Fan_etal_Liblinear}.
Moreover, the kernel selection is challenging if the underlying physics is unknown.
Algorithm \ref{Algo:PIML_Transient_Bimolecular} summarizes the steps involved in the implementation of the proposed ML methodology for the transient fast irreversible bimolecular diffusive-reactive systems.

\begin{algorithm}
  \caption{Overview of physics-informed machine learning framework for reactive-mixing}
  \label{Algo:PIML_Transient_Bimolecular}
  {\small
  \begin{algorithmic}[1]
    \STATE \textbf{Input for non-negative FEM simulations:}~Time step $\Delta t$; total time of interest $\mathcal{I}$; stoichiometric coefficients; initial and boundary conditions for the chemical species $A$, $B$, and $C$; finite element mesh parameters; model parameters for anisotropic reaction-diffusion equation $\frac{\alpha_L}{\alpha_T}$, $D_m$, $v_0$, $\kappa_fL$, and $T$.
    %
    \STATE \textbf{Input for feature analysis:}~Number of neighbors to use for MI estimation, the number of trees in the forest, criterion to measure the quality of a split for each node during the construction of the tree, minimum number of training samples required to split an internal node in a tree, minimum number of training samples required to be at a leaf node, and maximum number of features considered when looking for the best split in a tree construction.
    
    \STATE \textbf{Input for ML-ROMs construction:}~Penalty parameter to prevent over-fitting $\mathcal{P}$, RBF kernel coefficient $\gamma$, and tolerance for stopping criterion $\varepsilon$. 
    
    \STATE Calculate initial and boundary conditions for the non-negative invariants $c_F$ and $c_G$ using Eqs.~\eqref{Eqn:Definitions_of_F}--\eqref{Eqn:Definitions_of_G}.
    %
    \STATE Solve for invariant concentrations $c_F$ and $c_G$ for all times.
    
    \STATE Calculate $c_F^{\mathrm{min}}$, $c_F^{\mathrm{max}}$, $c_G^{\mathrm{min}}$ and $c_G^{\mathrm{max}}$ based on discrete maximum principles and the non-negative constraint.
    
    \FOR{$n = 0, 1, \cdots, \left(\texttt{NumTimeSteps} - 1 \right)$}
    %
    \STATE Call optimization-based diffusion with decay solver to obtain $\boldsymbol{c}_F^{(n+1)}$:    
    \begin{align*}
      \mathop{\mbox{minimize}}_{\boldsymbol{c}^{(n + 1)}_F \in 
        \mathbb{R}^{ndofs}} & \quad \frac{1}{2}  \left \langle 
      \boldsymbol{c}^{(n + 1)}_F; \boldsymbol{K} 
      \boldsymbol{c}^{(n + 1)}_F \right \rangle - \left \langle 
      \boldsymbol{c}^{(n + 1)}_F; \boldsymbol{f}^{(n + 1)}_F  
      \right \rangle - \frac{1}{\Delta t}  \left \langle 
      \boldsymbol{c}^{(n + 1)}_F; \boldsymbol{c}^{(n)}_F
      \right \rangle \\
      \mbox{subject to} & \quad c_F^{\mathrm{min}} \boldsymbol{1} \preceq 
      \boldsymbol{c}^{(n + 1)}_F \preceq c_F^{\mathrm{max}} \boldsymbol{1} 
    \end{align*}
    \STATE Call optimization-based diffusion with decay solver to obtain $\boldsymbol{c}_G^{(n+1)}$:    
    \begin{align*}
      \mathop{\mbox{minimize}}_{\boldsymbol{c}^{(n + 1)}_G \in 
        \mathbb{R}^{ndofs}} & \quad \frac{1}{2}  \left \langle 
      \boldsymbol{c}^{(n + 1)}_G; \boldsymbol{K} 
      \boldsymbol{c}^{(n + 1)}_G \right \rangle - \left \langle 
      \boldsymbol{c}^{(n + 1)}_G; \boldsymbol{f}^{(n + 1)}_G  
      \right \rangle - \frac{1}{\Delta t}  \left \langle 
      \boldsymbol{c}^{(n + 1)}_G; \boldsymbol{c}^{(n)}_G
      \right \rangle \\
      \mbox{subject to} & \quad c_G^{\mathrm{min}} \boldsymbol{1} \preceq 
      \boldsymbol{c}^{(n + 1)}_G \preceq c_G^{\mathrm{max}} \boldsymbol{1} 
    \end{align*}
    \STATE where $\boldsymbol{K}$ is the symmetric positive definite coefficient matrix, $ndofs$ is the number of degrees of freedom, $\langle \bullet;\bullet \rangle$ represents the standard inner-product on Euclidean spaces, $\boldsymbol{1}$ denotes a vector of ones of size $ndofs \times 1$, and the symbol $\preceq$ represents the component-wise inequality for vectors. 
    \ENDFOR
    %
    \STATE Get non-negative solutions for $c_A$, $c_B$, and $c_C$ using Eqs.~\eqref{Eqn:Fast_A}--\eqref{Eqn:Fast_C}.
    %
    \STATE From $c_A$, $c_B$, and $c_C$, calculate QoIs from Eqs.~\eqref{Eqn:Avg_Conc}--\eqref{Eqn:Degree_of_Mixing}.   
    %
    \STATE Construct and preprocess the features for QoIs. 
    Scale the processed features to a range of $[0,1]$. 
    Features that correspond to QoIs include $T$, $\log[\frac{\alpha_L}{\alpha_T}]$, $\kappa_fL$, $\log[v_o]$, and $D_m$.    
    %
    \STATE Perform the feature importance using F-test, MI-criteria, and Random Forests.
    %
    \STATE Construct the SVM-ROMs and SVR-ROMs for QoIs on training dataset by solving the following minimization problem:
      \begin{align}
        \mathop{\mathrm{minimize}}_{\boldsymbol{w}, b, \boldsymbol{\zeta}, 
        \boldsymbol{\zeta^{*}} } & \quad \frac{1}{2} \langle \boldsymbol{w}; 
        \boldsymbol{w} \rangle + \mathcal{P} \sum_{i=1}^{\mathrm{n}} \left( 
        \zeta_i + \zeta^{*}_i \right) \\
        \mbox{subject to} & \quad 
        \begin{cases}
          y_i - \boldsymbol{w} \bullet \boldsymbol{\phi}(\boldsymbol{x}_i) - b 
          \leq \varepsilon + \zeta_i \\
          \boldsymbol{w} \bullet \boldsymbol{\phi}(\boldsymbol{x}_i) + b - y_i
          \leq \varepsilon + \zeta^{*}_i \\
          \zeta_i, \zeta^{*}_i \geq 0, \quad \forall i = 1,2, \cdots n
        \end{cases}
     \end{align}
    %
    \STATE The explicit function form for SVM-ROMs and SVR-ROMs for each chemical species is given as follows:
      \begin{itemize}
        \item \textbf{SVM-ROMs:}~$\mathbb{y}_{\mathrm{{\tiny svm}}}(\boldsymbol{x}) 
          = \displaystyle \mathrm{sgn} \left[ \sum_{j=1}^{\mathrm{n_{{\tiny Supp}}}} y_j \alpha_j 
          \mathcal{K}(\boldsymbol{x}, \mathbb{x}_j) + b \right]$
        \item \textbf{SVR-ROMs:}~$\mathbb{y}_{\mathrm{{\tiny svr}}}(\boldsymbol{x}) 
          = \displaystyle \sum_{j=1}^{\mathrm{n_{{\tiny Supp}}}} \left(\alpha_j - 
          \alpha_j^* \right) \mathcal{K}(\boldsymbol{x}, 
          \mathbb{x}_j) + b$
        \item \textbf{SVM and SVR Kernels:}~$\mathcal{K}(\boldsymbol{x}, \mathbb{x}_j) = 
          \exp\left(-\gamma \| \boldsymbol{x} - \mathbb{x}_j 
          \|^2 \right)$
      \end{itemize}
      $\mathrm{such \, \, that} \; \; 0 \leq \alpha_j, \alpha_j^* \leq \mathcal{P}, \; \; \forall j = 1,2, \cdots \mathrm{n_{{\tiny Supp}}}$
  \end{algorithmic}
  }
\end{algorithm}

\section{RESULTS}
\label{Sec:S4_ROM_Results}
The proposed ML framework was applied to analyze 2315 numerical simulations for predicting the QoIs using Algorithm \ref{Algo:PIML_Transient_Bimolecular} outlined in Sec.~\ref{Sec:S3_ROM_Framework}.
Each high-fidelity numerical simulation is obtained for a different set of reaction-diffusion model input parameters.
The varied input parameters are:~longitudinal-to-transverse anisotropic dispersion ratio $\frac{\alpha_L}{\alpha_T}$, molecular diffusivity $D_m$, perturbation parameter of the underlying vortex-based velocity field $v_0$, and velocity field characteristic scales $\kappa_fL$ and $T$.
The values corresponding to the varied input parameters are:~$v_0 = \left[1, 10^{-1}, 10^{-2}, 10^{-3}, 10^{-4} \right]$, $\frac{\alpha_L}{\alpha_T} = \left[1, 10^{1}, 10^{2}, 10^{3}, 10^{4} \right]$, $D_m = \left[10^{-8}, 10^{-1}, 10^{-2}, 10^{-3} \right]$, $\kappa_fL = \left[2, 3, 4, 5 \right]$, and $T = \left[1 \times 10^{-4}, 2 \times 10^{-4}, 3 \times 10^{-4}, 4 \times 10^{-4}, 5 \times 10^{-4} \right]$.
In our finite element simulations, $\alpha_L = 1$ and $\alpha_T$ is varied accordingly.
The model domain is discretized by low-order structured triangular finite elements with 81 nodes on each side.
The corresponding triangular finite element mesh is of size $81 \times 81$.
Analysis is performed using a uniform time step equal to 0.001.
The end time for finite element simulation is equal to 1.0.
That is, the total number of time steps is equal to 1000.
Non-negative finite element method described in Sec.~\ref{Sec:S2_ROM_GE} ensures that the obtained concentrations are non-negative and satisfy the discrete maximum principle.

Examples obtained from non-negative FEM simulations and proposed ML framework are presented in Figs.~\ref{Fig:Contours_F_Difftimes}--\ref{Fig:SVR_Scaling_Law}.
Figs.~\ref{Fig:Contours_F_Difftimes}--\ref{Fig:Contours_G_Difftimes} show the concentration of invariant-$F$ and invariant-$G$ at times $t = 0.1, \, 0.5,$ and $1.0$ by varying the input parameter $\kappa_fL$. 
Other input parameters which are held constant are $\frac{\alpha_L}{\alpha_T} = 10^{4}$, $v_o = 10^{-1}$, $T = 1 \times 10^{-4}$, and $D_m = 10^{-3}$. 
$\frac{\alpha_L}{\alpha_T} = 10^{4}$ correspond to high anisotropy and $D_m = 10^{-3}$ correspond to low molecular diffusivity.
High anisotropic contrast means less mixing across the streamlines.
At $t = 1.0$ (which is the final time of the FEM simulation), we can observed that for higher values of $\kappa_fL$, invariant-$F$ and $G$ spreads across the entire domain.
For low values of $\kappa_fL$, for example when $\kappa_fL = 2$, there are certain regions in the domain where invariant-$F$ and $G$ doesn't diffuse. 
Fig.~\ref{Fig:Contours_F_Difftimes}(c,f) and Fig.~\ref{Fig:Contours_G_Difftimes}(c,f) show distinctive regions where $c_F$ and $c_G$ is equal to zero and one.
These regions are roughly located at the center of the vortex structures.
From these figures, one can conclude that to enhance species mixing for low molecular diffusivity and high anisotropic dispersion contrast, we need the input parameter $\kappa_fL$ to be high.
This means, we need the system to be chaotic (more vortices) to enhance mixing when anisotropy is high.
    
Figs.~\ref{Fig:Contours_A_Difftimes}--\ref{Fig:Contours_C_Difftimes} show the concentration of species $A$, $B$, and $C$ at times $t = 0.1, \, 0.5,$ and $1.0$. 
Note that these non-negative concentrations are derived from $c_F$ and $c_G$ through Eqs.~\eqref{Eqn:Fast_A}--\eqref{Eqn:Fast_C}.
From Fig.~\ref{Fig:Contours_A_Difftimes}, it is clear that species $A$ is not consumed in its entirety for $t \in [0,1]$. 
Moreover, at lower values of $\kappa_fL$, a considerable amount of species $A$ remains in the left half of the domain compared to higher values of $\kappa_fL$.
Similar to species $A$, from Fig.~\ref{Fig:Contours_B_Difftimes} it is clear that species $B$ is not consumed in its entirety for lower values of $\kappa_fL$ resulting in incomplete mixing.
This is because of low molecular diffusivity and very high anisotropic dispersion.
If there was no anisotropy, towards the end of simulation time, we will have uniform mixing/reaction along the longitudinal and transverse directions \cite{2009_Tsang_PRE_v80_p026305}. 
However, anisotropy hinders uniform mixing (which is what these figures summarize). 
As a result, towards the end of the FEM simulation time, we see incomplete/preferential mixing.
For higher $\kappa_fL$ values, we see more of species $B$ being consumed in the right half of the domain. 
From Fig.~\ref{Fig:Contours_C_Difftimes}, it is clear that for smaller values of $\kappa_fL$, less amount of product $C$ is formed.
As $\kappa_fL$ increases, we can see that the product formation increases and the number of vortex regions with zero concentration decrease (even under high anisotropy). 
This is because the velocity field contains a lot of small-scale vortices that are spread across the domain. 
These small-scale vortex structures enhance mixing even under low molecular diffusivity and high anisotropy.

Fig.~\ref{Fig:SpecA_Data} shows the average of concentration $\mathfrak{c}_i$, average of square of concentration $\mathbb{c}_i$, and degree of mixing $\sigma^2_i$ as a function of time for all three species $A, B,$ and $C$. 
Plots are shown for all the realizations, which are depicted in light blue, light magenta, and light green colors. 
The solid blue, magenta, and green lines are the average of all the realization at each time level, which are scaling QoIs. 
Note that all QoIs are normalized and lie between $[0,1]$.
From these figures, it is clear that $\ln[\mathfrak{c}_i] \propto t$, $\ln[\mathbb{c}_i] \propto t$, and $\ln[\sigma^2_i] \propto t$. 
This means that the QoIs decrease (for reactants) or increase (for product) with time in an exponential manner, approximately. 
Also, it can be observed from Fig.~\ref{Fig:SpecA_Data}(a,b,d,e) that for $t \in [0.2, 1.0]$, the average of concentration and average of square of concentration of reactants decrease at a much faster rate, resulting in high product yield. 
This is because at later time, the effect of molecular diffusion comes into play resulting in enhanced mixing.
A system is well-mixed if the degree of mixing is closer to 0.0 and segregated if the degree of mixing is closer to 1.0.
From Fig.~\ref{Fig:SpecA_Data}(g,h,i), it clear that the chaotic nature of the flow field and molecular diffusion make the reactive-diffusive system progress towards a uniformly mixed state.

We performed a simple exponential analysis on QoIs \cite{istratov1999exponential} to see the scaling exponent dependence on input parameters.
For example, Fig.~\ref{Fig:Clustering_Analysis} provides the relationship between scaling exponent and input parameters for degree of mixing for species $A$, $B$, and $C$.
The exponential coefficients are obtained through a naive fit by an exponential function.
Lower value for scaling exponent means the system has mixed well while higher value implies the system is segregated or mixed incompletely.
Then, $k$-means clustering is performed on these scaling exponents to identify the features that demarcate high and low mixing states. 
$k$-means clustering aims to partition a set of observations into $k$ clusters in which each observation belongs to the cluster with the nearest mean.
This helps the user to understand the natural grouping or structure in a dataset.
Elbow method \cite{2011_Everitt_etal} is used to identify the appropriate number of clusters (which is the $k$ in $k$-means clustering).
This method looks at the percentage of variance explained as a function of the number of clusters.
Percentage of variance explained is the ratio of the variance between the cluster groups to the total variance.
For our case, we found $k = 4$ to be optimal as it satisfies the elbow criterion \cite{2011_Everitt_etal}. 
Meaning that, beyond $k = 4$ the marginal gain in the percentage of variance is low.
Fig.~\ref{Fig:Clustering_Analysis} provides the clustering details for this value of $k$.
Each color represents a cluster.
Among the various sub-figures in Fig.~\ref{Fig:Clustering_Analysis}, scaling exponent vs. $\log[\frac{\alpha_L}{\alpha_T}]$ (see Fig.~\ref{Fig:Clustering_Analysis}(d,e,f)) for species $A$, $B$, and $C$ provide the most useful and interesting insight.
Other sub-figures do not provide much insight into the relationship between scaling exponent vs. remaining input parameters (which are $T$, $\kappa_fL$, $\log[v_o]$, and $D_m$).
A main inference from Fig.~\ref{Fig:Clustering_Analysis}(d,e,f) is that for lower values of $\log{[\frac{\alpha_L}{\alpha_T}]}$, the scaling exponent is low ($\approx -20$). 
Meaning that, enhanced mixing occurs when $\alpha_L \approx \alpha_T$, which is in accordance with the physics of reactive-mixing \cite{2009_Tsang_PRE_v80_p026305,2002_Adrover_etal_CCE_v26_p125_p139}.
Reaction-diffusion system tends to be more diffusive along streamlines if anisotropy is low.
As the ratio of $\frac{\alpha_L}{\alpha_T}$ increases, we see incomplete/preferential mixing due to high anisotropy.

Quantitatively, to have detailed grasp and contribution of each input parameter in understanding mixing process, we perform feature importance.
Three different types of feature selection methods are studied for consistency. 
These include F-test, MI criteria, and RF.
Entire simulation data is used as there is no QoI prediction involved in using F-test and MI criteria methods.
For RF method, 1620 simulations (which is approx. 70\%) are used for training and remaining 695 simulations are used for testing the feature importance.
All training samples are given equal weights during tree construction.
Samples are not bootstrap when building trees.
Out-of-bag samples are used to estimate the generalization accuracy (which is $R^2$-score) of feature selection using RF.
In our case, the generalized $R^2$-score for testing the RF model for feature importance is 0.88.
The inputs values for these feature selection methods are as follows:~Number of neighbors to use for MI estimation is equal to 3.
For RF, analysis is performed for different number of trees in the forest. 
These are equal to 5, 100, and 250. 
Gini criterion \cite{james2013introduction,zhou2012ensemble} is used to measure the quality of a split for each node during the construction of the tree.
Minimum number of training samples required to split an internal node in a tree is equal to 2.
Minimum number of training samples required to be at a leaf node is equal to 1.
Maximum number of features considered when looking for the best split in a tree construction is equal to 4.

Figs.~\ref{Fig:RF_DOM_Feature_Importance}--\ref{Fig:MIC_Feature_Importance} show the feature importances using F-test, MI criteria, and RF.
Fig.~\ref{Fig:RF_DOM_Feature_Importance} shows the use of forests of trees, which fit a number of randomized decision trees on different sub-samples of the data to evaluate the relative importance of model inputs. 
The color bars indicated the feature importances of the RF, along with their inter-tree variability.
To summarize, all the three methods consistently tell us that $\log{[\frac{\alpha_L}{\alpha_T}]}$ is the most important feature and $T$ being the least important.
After $\log{[\frac{\alpha_L}{\alpha_T}]}$, next important features are in the following order $D_m$, $\kappa_fL$, $\log{\left[v_0 \right]}$, and $T$.
F-test shows that $\kappa_fL$ and $\log{\left[v_0 \right]}$ are not that relevant (see Fig.~\ref{Fig:Ftest_Feature_Importance}).
From Fig.~\ref{Fig:MIC_Feature_Importance}, MI-criteria tells us that a steep decrease in the feature importance from $\log{[\frac{\alpha_L}{\alpha_T}]}$ to $D_m$ is observed, which is agreement with RF and F-test. 
However, there is a gradual decrease in feature importance from $D_m$ to $T$, which is in contrast with feature importances from RF and F-test.
To conclude, feature importance is in accordance with the physics of mixing as $\log{[\frac{\alpha_L}{\alpha_T}]}$ and $\kappa_fL$ mainly controls the spread of chemical species during initial stages. 
At later times, molecular diffusivity $D_m$ becomes the major mechanism to enhance mixing between reactants $A$ and $B$ resulting in higher product yield.

Based on the feature importance results, we compare the predictive capabilities of SVM-ROMs and SVR-ROMs constructed based on Algorithm \ref{Algo:PIML_Transient_Bimolecular}.
Two different set of ROMs are constructed for QoIs classification and regression.
First, set of ROMs are based all features and second set of ROMs are built on top three features (which are $\log{[\frac{\alpha_L}{\alpha_T}]}$, $D_m$, and $\kappa_fL$).
Analysis is performed for different values of SVM-ROM and SVR-ROM parameters and different sizes of training data.
Values assumed for ROM construction in Eqns.~\ref{Eqn:Math_SVR}--\ref{Eqn:RBF_Kernel} are as follows:~Penalty parameter $\mathcal{P} = [1, 10, 10^2, 10^3, 10^4]$, RBF kernel coefficients $\gamma = [0.1, 0.01, 0.001, 0.0001]$, and tolerance for stopping criterion $\varepsilon = [0.1, 0.01, 0.001, 0.0001]$.
For SVM-ROMs based on OVR, all the four classes that represent mixing states are given equal weights.
ROMs are trained on different sizes of simulation data. 
The total size of the input data is $2,315,000 \times 6$ and the total size of output data is $2,315,000 \times 1$.
Three different ROMs are constructed using 1\%, 5\%, and 30\% of input data, which correspond to approximately 23, 115, and 694 simulation data.
Tables.~\ref{Tab:SVM_SVR_ROMs_Accuracy_1}--\ref{Tab:SVM_SVR_ROMs_Accuracy_2} summarizes the ensemble average prediction accuracy of SVM-ROMs and SVR-ROMs on different sizes of training data and different ROM construction parameters.

Table.~\ref{Tab:SVM_SVR_ROMs_Accuracy_1} also shows the number of support vectors contained in the SVM and SVR-ROMs.
Support vectors tell us about the storage requirements of the model on a laptop.
Note that only a fraction of entire simulation data points are the support vectors.
These support vectors form the basis for the ROMs.
They are obtained by solving the minimization problem given by Eqns.~\ref{Eqn:Math_SVR}--\ref{Eqn:Math_SVR_Constraints}.
However, the storage requirements and model complexity of SVM and SVR-ROMs increase very quickly with the number of training samples.
In addition, the ROM construction time also increases rapidly with increase in size of training data.
This may be seen as a downside of SVMs and SVR. 
But from Table.~\ref{Tab:SVM_SVR_ROMs_Accuracy_2} it is evident that the classification and prediction accuracy of various ROMs outweighs the drawbacks.
For instance, less training data ($\leq 5\%$) is needed to get decent prediction capability ($R^2$-score greater than 0.85).
Moreover, the methodology to construct a ROM has firm physical and mathematical underpinnings.
It should be noted that these ROMs are constructed on a laptop. 
The time to construct ROMs can be accelerated by either using a GPU or multi-core HPC machines, which is beyond the scope of current paper.
The corresponding processor speed for the laptop on which ROMs are developed is 3.1 GHz, total number of cores utilized is 2, and maximum memory used in 16 GB.

The computational cost to run a high-fidelity numerical simulation and post-process the solution to obtain QoIs is approximately 26 mins on a single core and 18 mins on a eight core processor (Intel(R) Xeon(R) CPU E5-2695 v4 \@ 2.10GHz).
The data produced from each simulation is of size 2.2 GB.
For all realizations, the output from FEM simulations is approximately 5.1 TB. 
The resulting QoIs from post-processing the simulation data is of size 1.2 GB.
The maximum size to store SVM and SVR-ROMs is approximately 100 KB (which includes the support vectors).
For QoI calculations, this corresponds to a compression ratio close to $8.3 \times 10^{-5}$ (as these ROMs can be thought as a compressed version of post-processed FEM data).
In addition, the computational time to compute QoIs using SVM and SVR-ROMs in the $\mathcal{O}(10^{-3})$ seconds, which is another advantage of using these ROMs in addition to their predictive capability.
Hence, the ROMs provide QoIs approximately $10^{7}$ faster than running a high-fidelity numerical simulation.

Fig.~\ref{Fig:SVR_vs_Unseen_data} shows the prediction of SVR-ROMs on randomly selected test data.
The ground truth obtained from high-resolution numerical simulations are shown as markers/points.
The color band denotes the upper and lower prediction estimates obtained from SVR-ROM ensembles by varying $\mathcal{P}$, $\gamma$, and $\varepsilon$.
The solid line is the ensemble average of SVR-ROM predictions.
From this figure, it is clear that the SVR-ROMs consistently predict various QoIs with high accuracy.
The $R^2$-score for ROMs prediction on these unseen data is greater than 0.9.
Fig.~\ref{Fig:SVR_Scaling_Law} shows the prediction of SVR-ROMs for scaling QoIs, which is an ensemble average of QoIs across all realizations.
As mentioned previously, the ground truth is shown as markers.
SVR-ROM predictions are shown by solid line and color band, respectively.
From these figures, it clear that SVR-ROMs are able to predict the scaling law trends of QoIs for species $A$ and $B$, accurately. 
Quantitatively, SVR-ROMs deviate when predicting scaling QoIs for species $C$. 
However, qualitatively, SVR-ROMs are able to predict the trends observed in scaling QoIs for species $C$. 
To summarize, the proposed ML method to construct ROMs has high prediction accuracy, lower computational time for prediction with minimal loss of information, and substantially reduces the data storage requirements.

\begin{table}
  \centering
	\caption{Summary of the ROM construction using SVMs and SVRs using all features and top three features.
	\label{Tab:SVM_SVR_ROMs_Accuracy_1}}
	\resizebox{\textwidth}{!}{\begin{tabular}{|c|c|c|c|c|c|c|} \hline
	  \multirow{2}{*}{{\small Description}} & 
	  \multicolumn{2}{|c|}{{\small No. of simulations (\% of input data)}} & 
	  \multicolumn{2}{|c|}{{\small Size of samples}} & 
	  \multicolumn{2}{|c|}{{\small Support Vectors (SV)}} \\ 
	  \cline{2-7}
	  & Training ROMs & Testing ROMs & Training & Testing & No. of SV & \% of training data used as SV \\ \hline
	  {\small ROM-1} & {\small 23 (1\%)} & {\small 2292 (99\%)} & {\small 23,150} & {\small 2,291,850} & {\small 1158 points} & {\small 5\% of 23 simulations} \\ 
	  {\small ROM-2} & {\small 115 (5\%)} & {\small 2200 (95\%)} & {\small 115,750} & {\small 2,199,250} & {\small 48,136 points}  & {\small 42\% of 115 simulations} \\ 
	  {\small ROM-3} & {\small 694 (30\%)}  & {\small 1621 (70\%)} & {\small 694,500} & {\small 1,620,500} & {\small 212,212 points} & {\small 30\% of 694 simulations} \\ 
	  \hline
	\end{tabular}}
\end{table} 

\begin{table}
  \centering
	\caption{Summary of performance of proposed ROMs for classifying and predicting mixing states using all features and top three features. 
    $F_1$-score for classification and $R^2$-score for regression.
	ROM construction time is also provided. 
	\label{Tab:SVM_SVR_ROMs_Accuracy_2}}
	 \resizebox{\textwidth}{!}{\begin{tabular}{|c|c|c|c|c|c|c|} \hline
	  \multirow{2}{*}{{\small Description}} & 
	  \multicolumn{2}{|c|}{{\small $F_1$/$R^2$-score based on all features}} &
	  \multicolumn{2}{|c|}{{\small $F_1$/$R^2$-score based on top three features}} &
	  \multicolumn{2}{|c|}{{\small ROM construction time}} \\
	  \cline{2-7}
	  & Classification & Regression & Classification & Regression & Classification & Regression \\ \hline
	  {\small ROM-1} & {\small 0.86} & {\small 0.85} & {\small 0.84} & {\small 0.84} & {\small 110 seconds} & {\small 120 seconds}\\ 
	  {\small ROM-2} & {\small 0.91} & {\small 0.90} & {\small 0.90} & {\small 0.88} & {\small 1.1 hours} & {\small 1.15 hours}\\ 
	  {\small ROM-3} & {\small 0.93} & {\small 0.91} & {\small 0.91} & {\small 0.90} & {\small 14.8 hours} & {\small 15.1 hours}\\ 
	  \hline
	\end{tabular}}
\end{table}

\section{CONCLUSIONS}
\label{Sec:S5_ROM_Conclusions}
Identifying the important features that influence the progress of reactive-mixing and accurately predicting its state is important to reactive-transport applications.
This includes accounting for species decay and product formation over time.
In this paper, our results demonstrate the applicability of SVM-ROMs and SVR-ROMs to efficiently and accurately emulate the mixing state of the reaction-diffusion system under high anisotropy.
We have presented estimates (e.g., upper and lower bounds) on species decay, product formation, and degree of mixing.
From these estimates, we inferred that the species decay and mix in an exponential fashion with respect to time.
This inference formed the basis for choosing radial basis function kernels for developing SVM-ROMs and SVR-ROMs for classifying and predicting the progress of reactive-mixing.
The ROMs are constructed based on high-fidelity reactive-mixing model output data generated using a parallel non-negative finite element solver.
The numerical solution procedure ensures that all the computed QoIs are non-negative even under high anisotropy.
Feature importance is performed using F-test, mutual information criteria, and random forests methods to identify dominant features that contribute to the mixing process.
We identified that that anisotropic dispersion strength/contrast is the most important feature and time-scale associated with flipping of velocity is the least important feature.
The characteristic spatial-scale associated of the vortex-based velocity field is another important feature that influences reactive-mixing.
We observed that as the spatial-scale $\kappa_fL$ increases, the product formation increases even under high anisotropy. 
This was because the model velocity field consisted of a lot of small-scale vortices that are spread across the domain. 
These small-scale vortex structures enhanced mixing even under low molecular diffusivity and high anisotropy.
We also demonstrated the predictive capability of ROMs under limited data.
The ROMs showed that good prediction ($\geq 90\%$) can be obtained with less than 5\% of training data.
This is because the RBF kernel in SVM-ROMs and SVR-ROMs was able to capture the important aspects of the reactive-mixing.
Moreover, these ROMs were approximately $10^7$ times faster than running a high-fidelity numerical simulation to compute QoIs.
This makes our the physics-informed machine learning ROMs ideal for usage in comprehensive uncertainty quantification studies.
Our future work includes developing ROMs based on the recent advances in physics-informed/constrained deep learning methods \cite{raissi2018hidden1,wang2018variational,raissi2019physics,zhu2019physics,yang2019adversarial,geneva2019modeling}.

\section*{ACKNOWLEDGMENTS}
This research was funded by the LANL Laboratory Directed Research and Development Early Career Award 20150693ECR.
MKM gratefully acknowledges the support of LANL Chick-Keller Postdoctoral Fellowship through Center for Space and Earth Sciences (CSES).
MKM thanks Prof. Kalyana Nakshatrala for his input and feedback on the paper.
LANL is operated by Triad National Security, LLC, for the National Nuclear Security Administration of U.S. Department of Energy (Contract No. 89233218CNA000001).
Additional information regarding the simulation datasets can be obtained from the corresponding author.

\bibliographystyle{unsrt}
\bibliography{Master_References/ML_mixing_references}


\begin{figure}
  \centering
  \subfigure[Problem description]
    {\includegraphics[width = 0.65\textwidth]
    {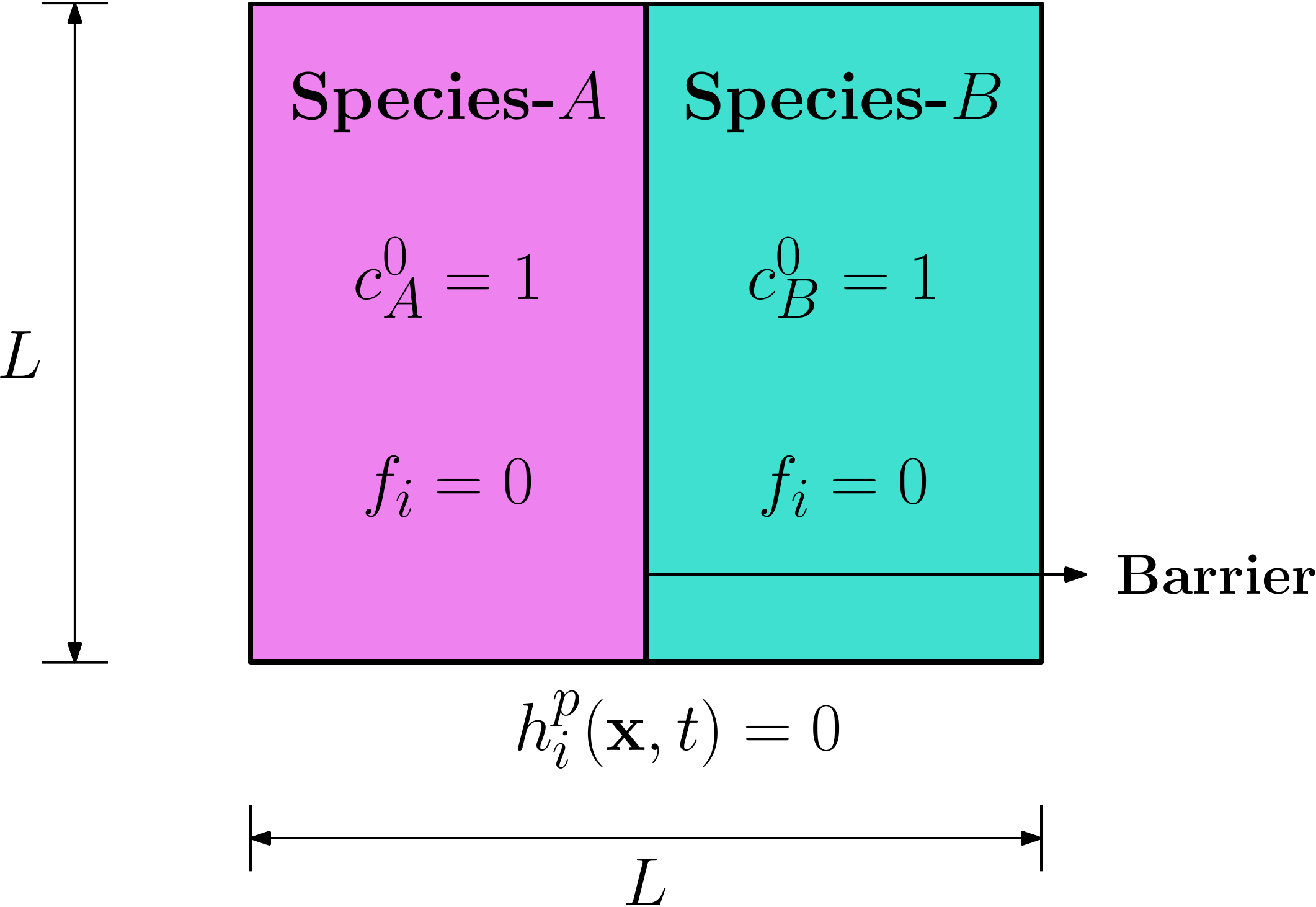}}
  \vspace{0.1in}
  \subfigure[Species $A$:~Initial condition]
    {\includegraphics[width = 0.35\textwidth]
    {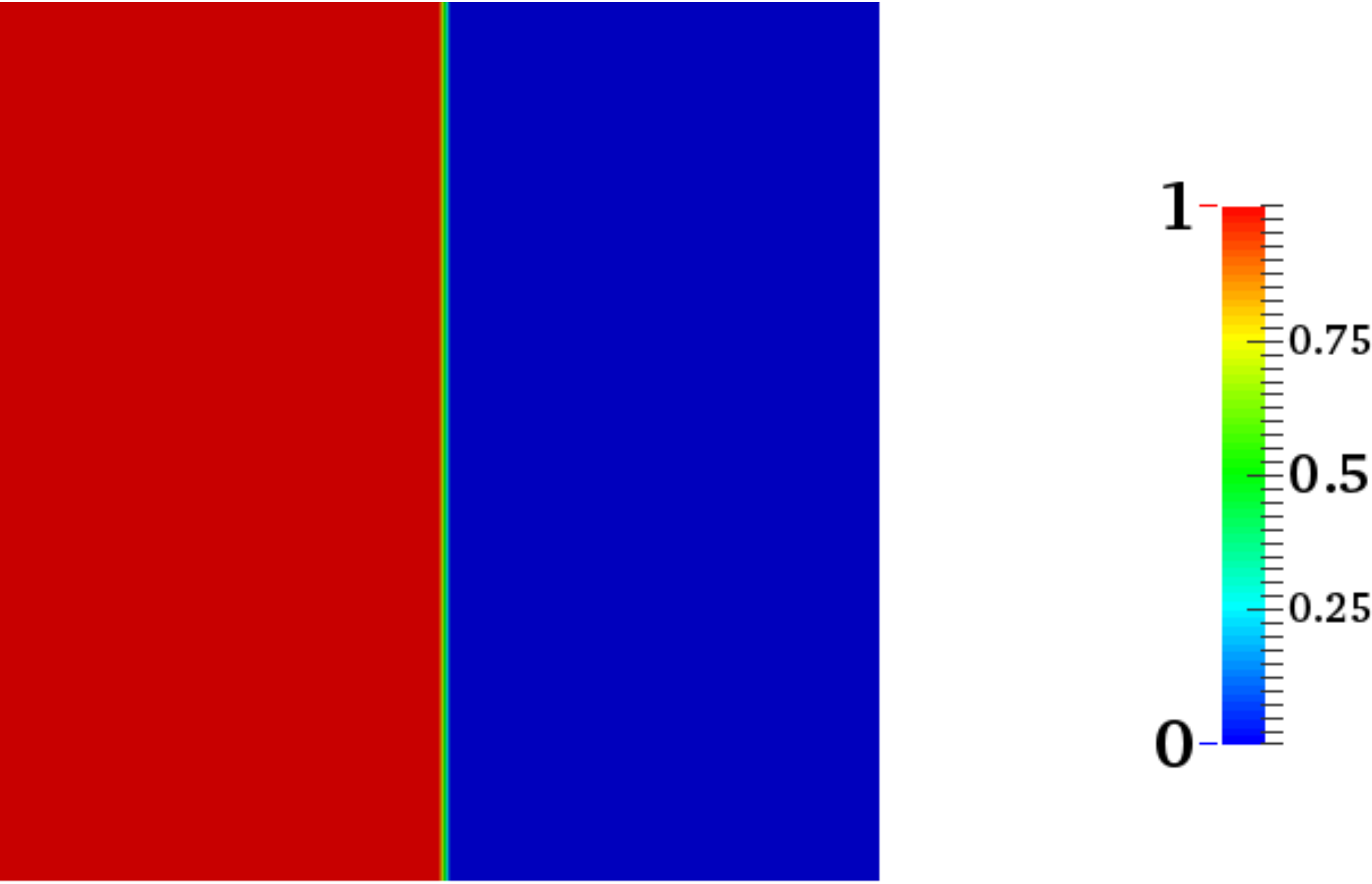}}
  \hspace{0.1in}
  \subfigure[Species $B$:~Initial condition]
    {\includegraphics[width = 0.35\textwidth]
    {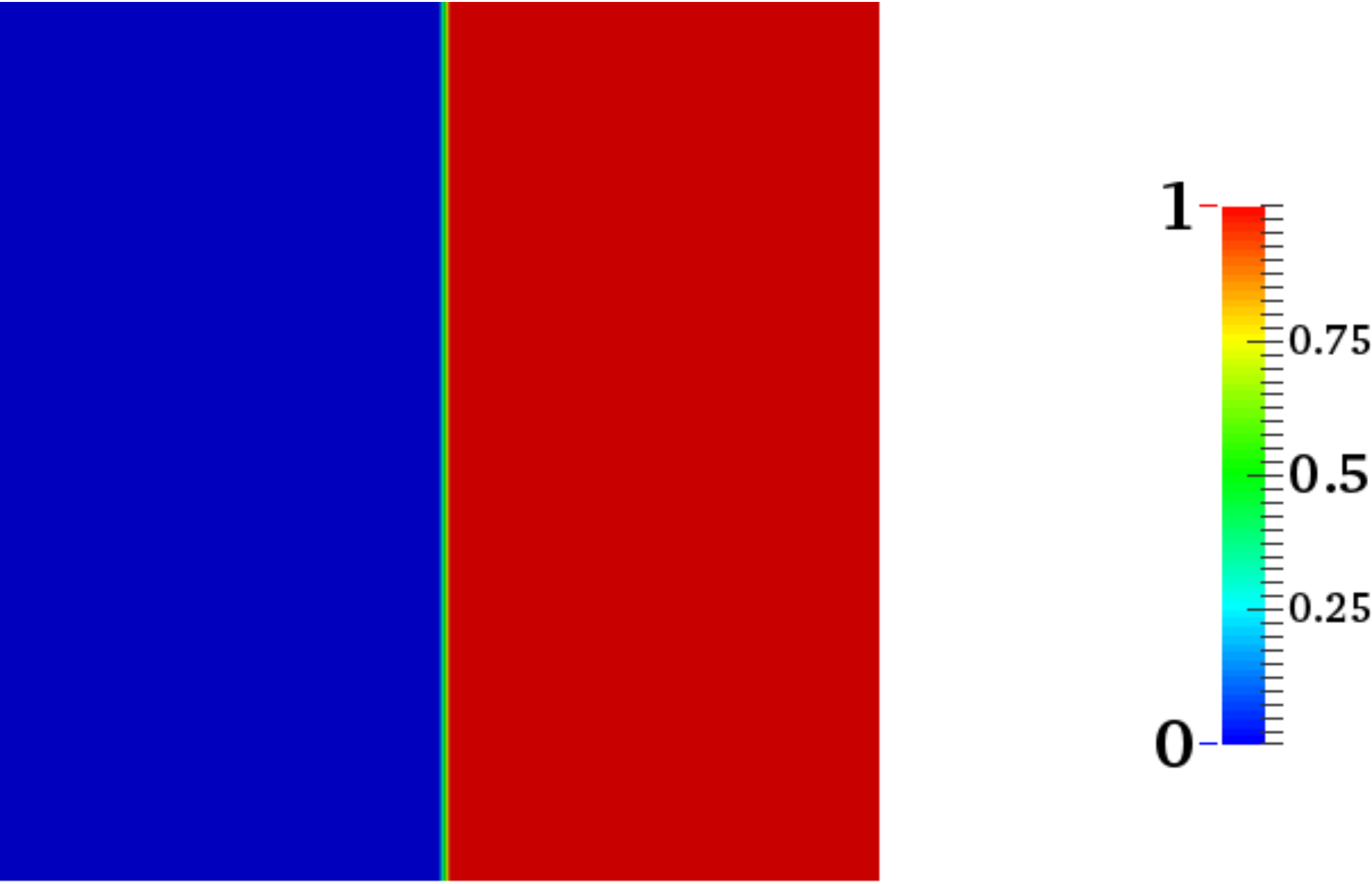}}
  \caption{\textsf{\textbf{Species mixing due to perturbed 
    vortex-based flow:}}~A pictorial description of the 
    initial boundary value problem. The length of the domain 
    `$L$' is assumed to be equal to 1. Zero flux boundary 
    conditions are enforced on all sides of the domain 
    ($h^{\mathrm{p}}_i(\boldsymbol{x},t) = 0.0$). Initially species $A$ 
    is on the left part of the domain while the species $B$ 
    is on the right part of the domain. They are initially 
    separated by a barrier and are allowed to mix after 
    time $t > 0$ according to the velocity field given 
    by Eqs.~\eqref{Eqn:Vel_x}--\eqref{Eqn:Vel_y}. 
    Their initial concentrations are both equal to 1, 
    respectively as shown in bottom two figures. The 
    volumetric sources for all the chemical species 
    are assumed to be equal to zero.
  \label{Fig:ROM_BVPs}}
\end{figure}

\begin{figure}
  \centering
  \subfigure[$\nu T \leq t < \left( \nu + \frac{1}{2} \right) T$]
    {\includegraphics[width = 0.75\textwidth]
    {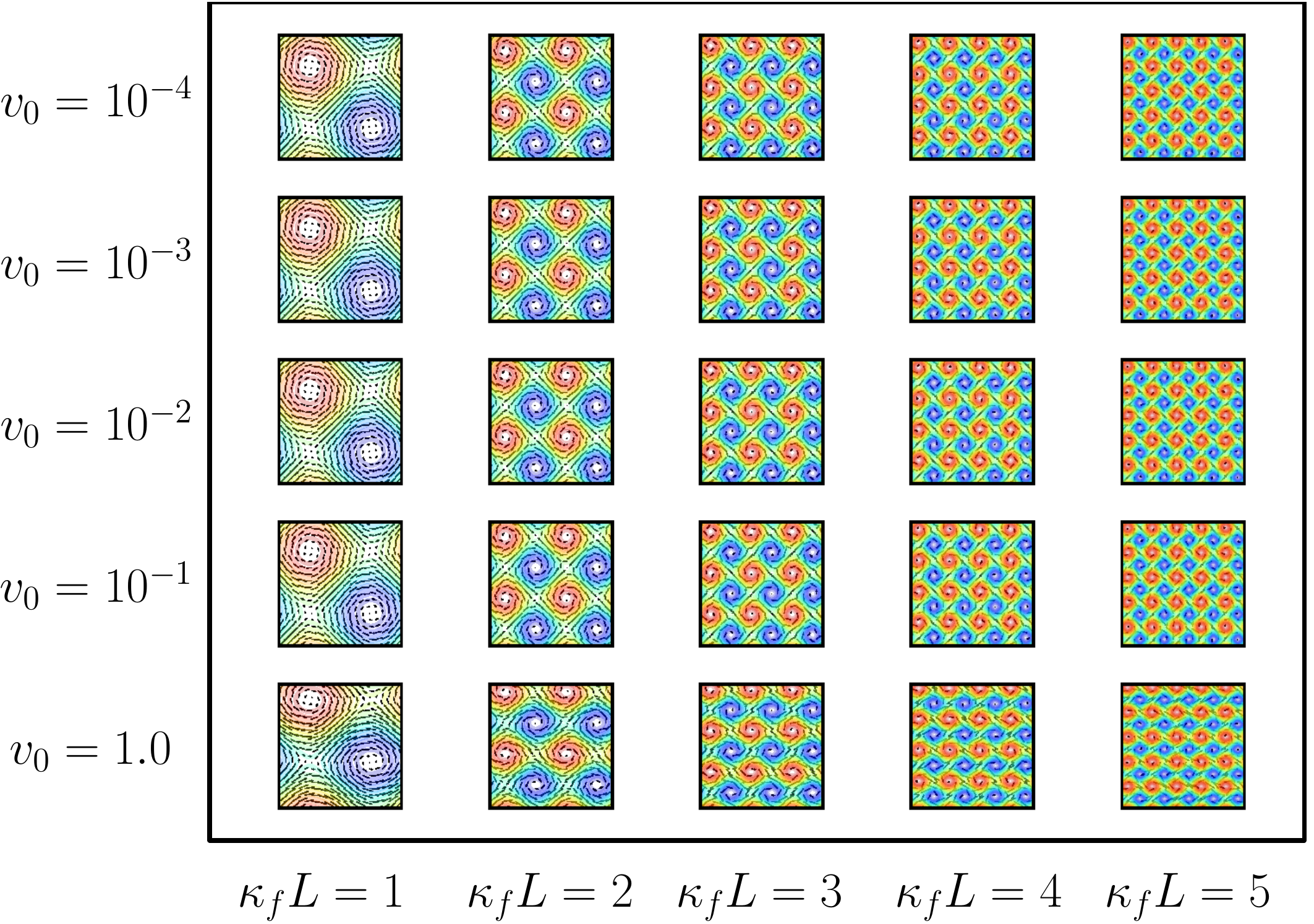}}
  \subfigure[$\left( \nu + \frac{1}{2} \right) T \leq t < \left( 
    \nu + 1 \right) T$]
    {\includegraphics[width = 0.75\textwidth]
    {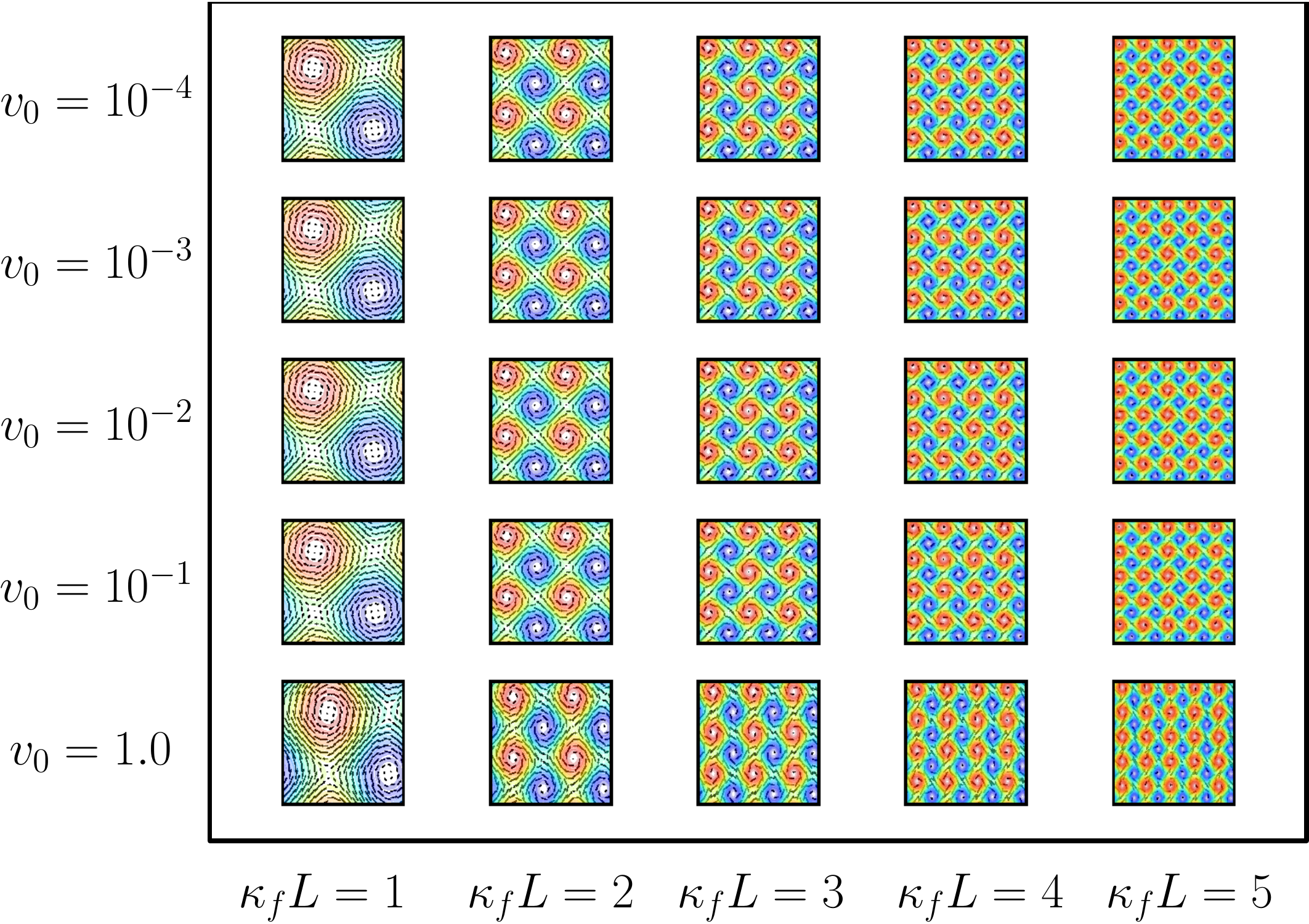}}
  \caption{\textsf{\textbf{Large-scale and small-scale structures 
    of velocity field:}}~The top and bottom figures show the 
    streamlines and associated velocity vectors for different 
    values of $v_0$ and $\kappa_fL$. For larger values of 
    $\kappa_fL$, we have more number of vortices. This means 
    we have a lot of small-scale structures. For low values 
    of $\kappa_fL$, we have less number of vortices.
  \label{Fig:Large_Small_Scales_VelField}}
\end{figure}

\begin{figure}
  \centering
  \includegraphics[width = 0.7\textwidth]
    {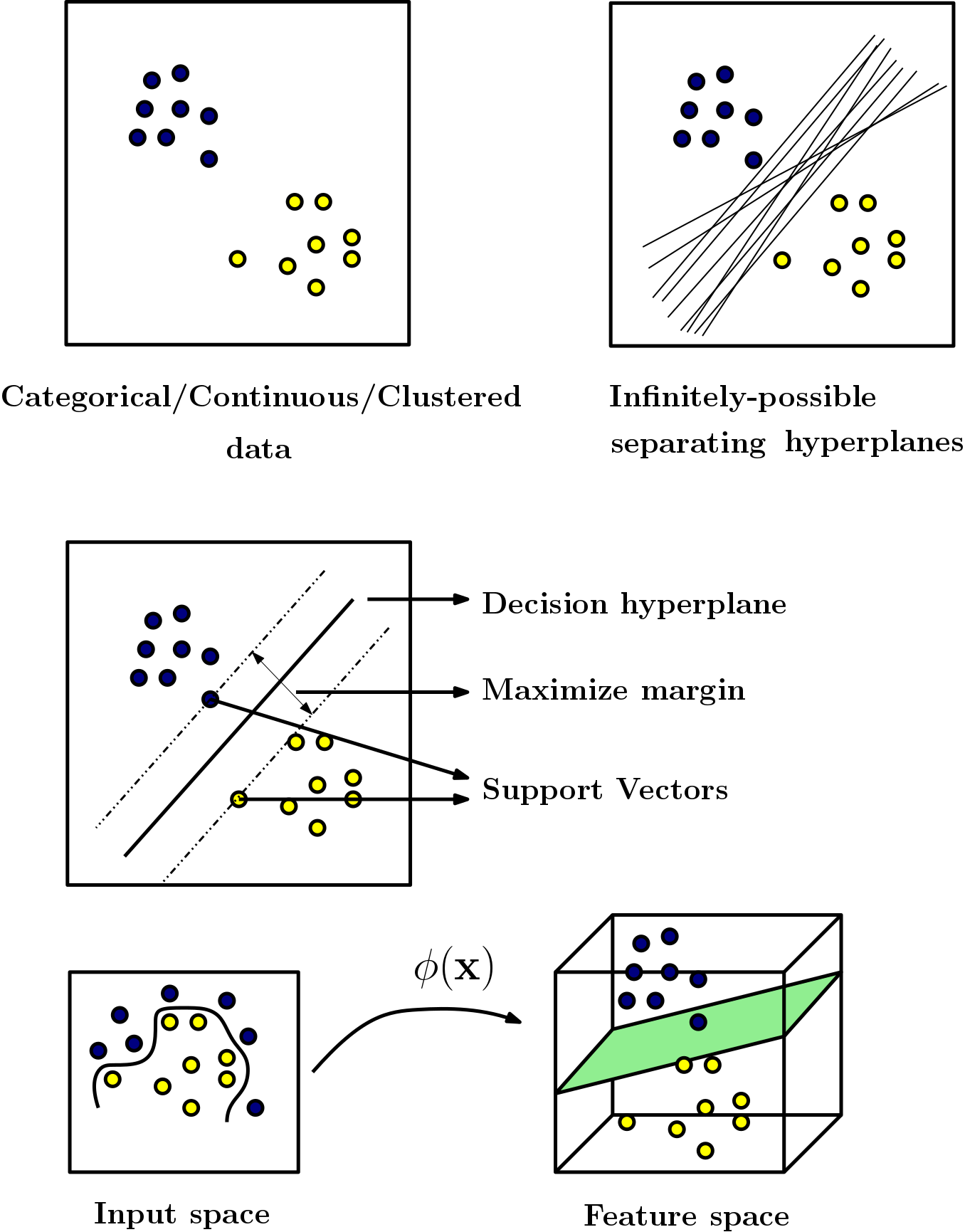}
  \caption{\textsf{\textbf{Pictorial description of Support 
    Vectors:}}~This figure provides a pictorial description 
    of how SVMs and SVRs work.
  \label{Fig:Pic_Support_Vectors}}
\end{figure}

\begin{figure}
  \centering
  \subfigure[$\kappa_fL = 2$ and $t = 0.1$]
    {\includegraphics[clip=true,width = 0.37\textwidth]
    {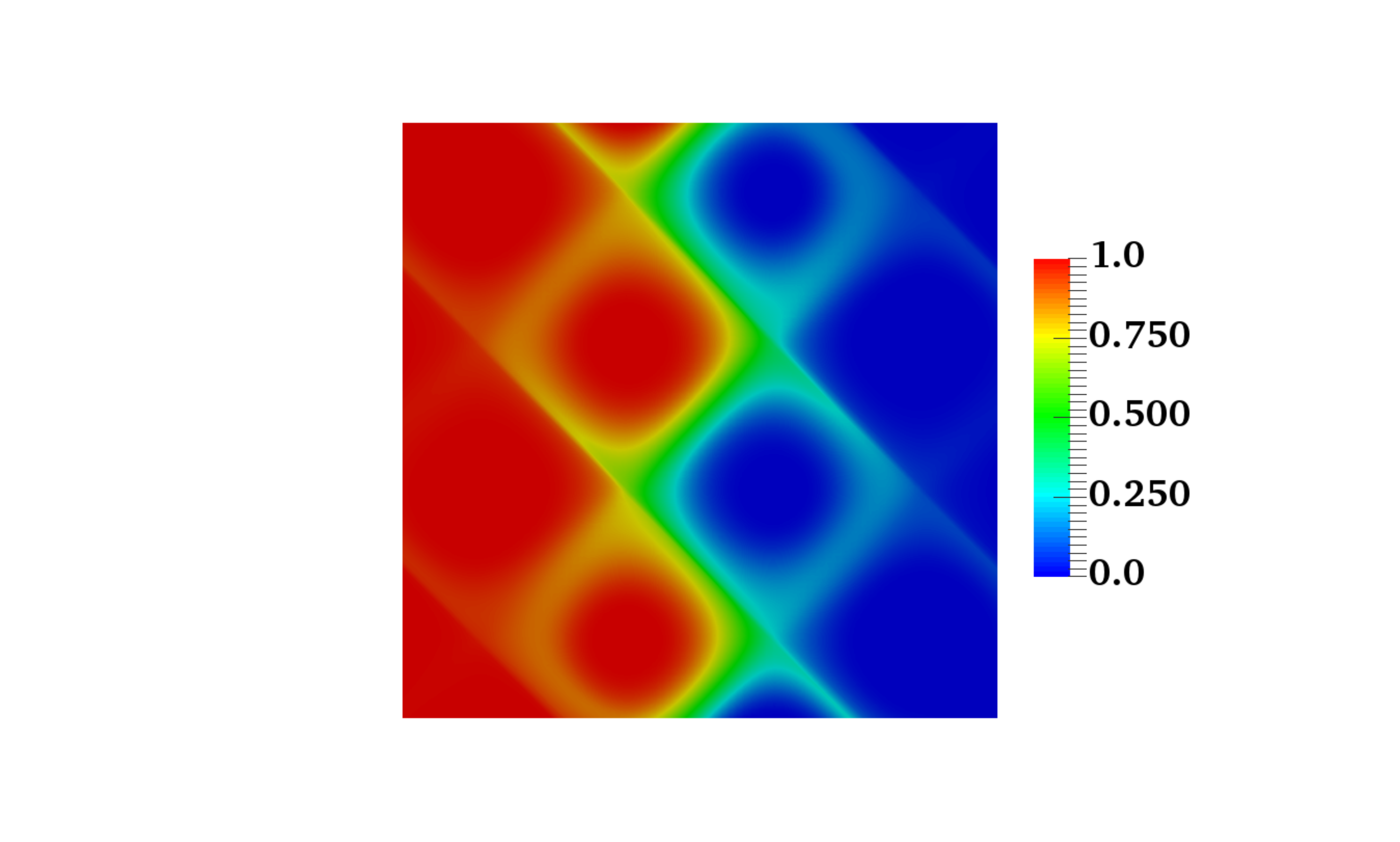}}
  \hspace{-0.5in}
  \subfigure[$\kappa_fL = 2$ and $t = 0.5$]
    {\includegraphics[clip=true,width = 0.37\textwidth]
    {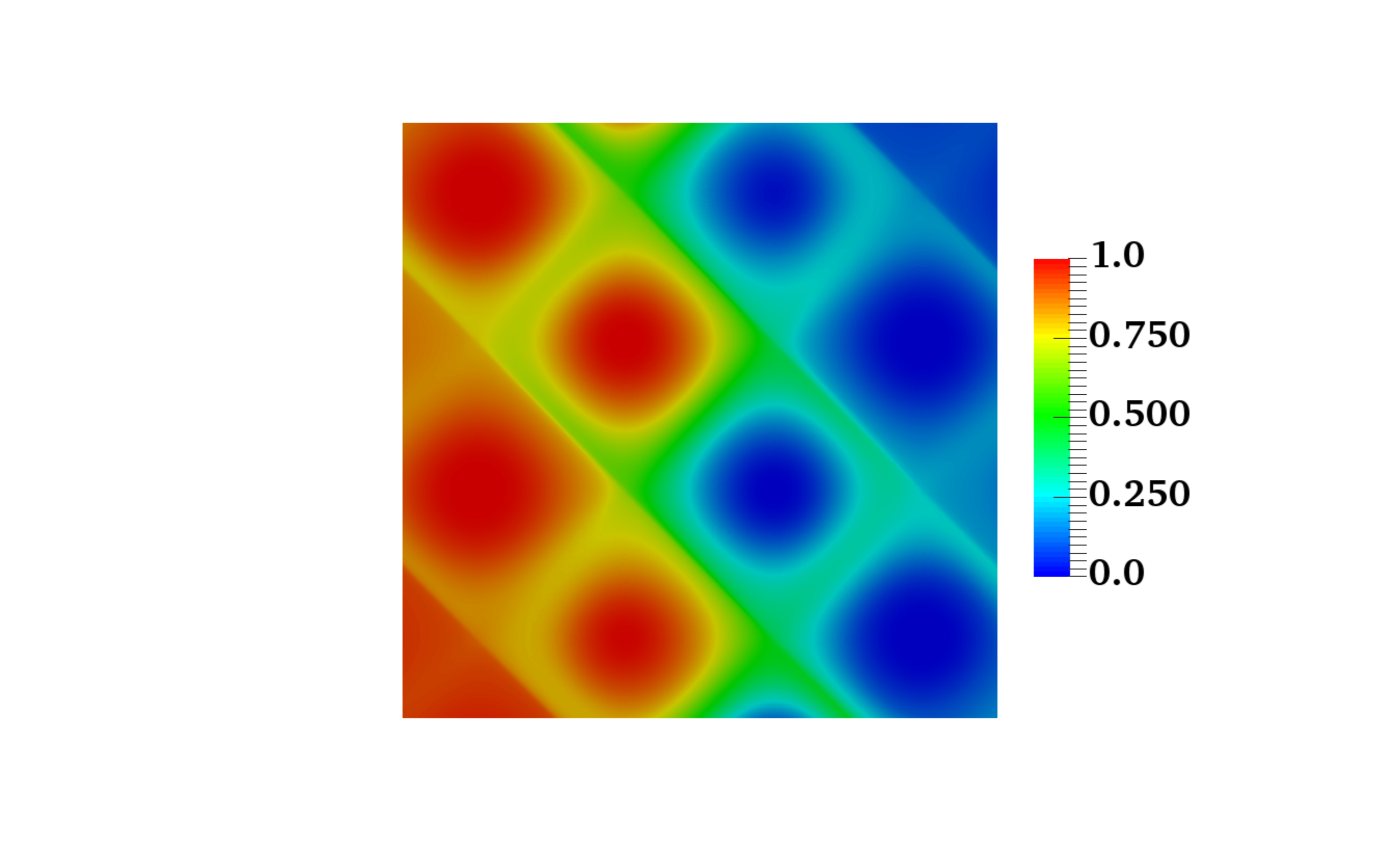}}
  \hspace{-0.5in}
  \subfigure[$\kappa_fL = 2$ and $t = 1.0$]
    {\includegraphics[clip=true,width = 0.37\textwidth]
    {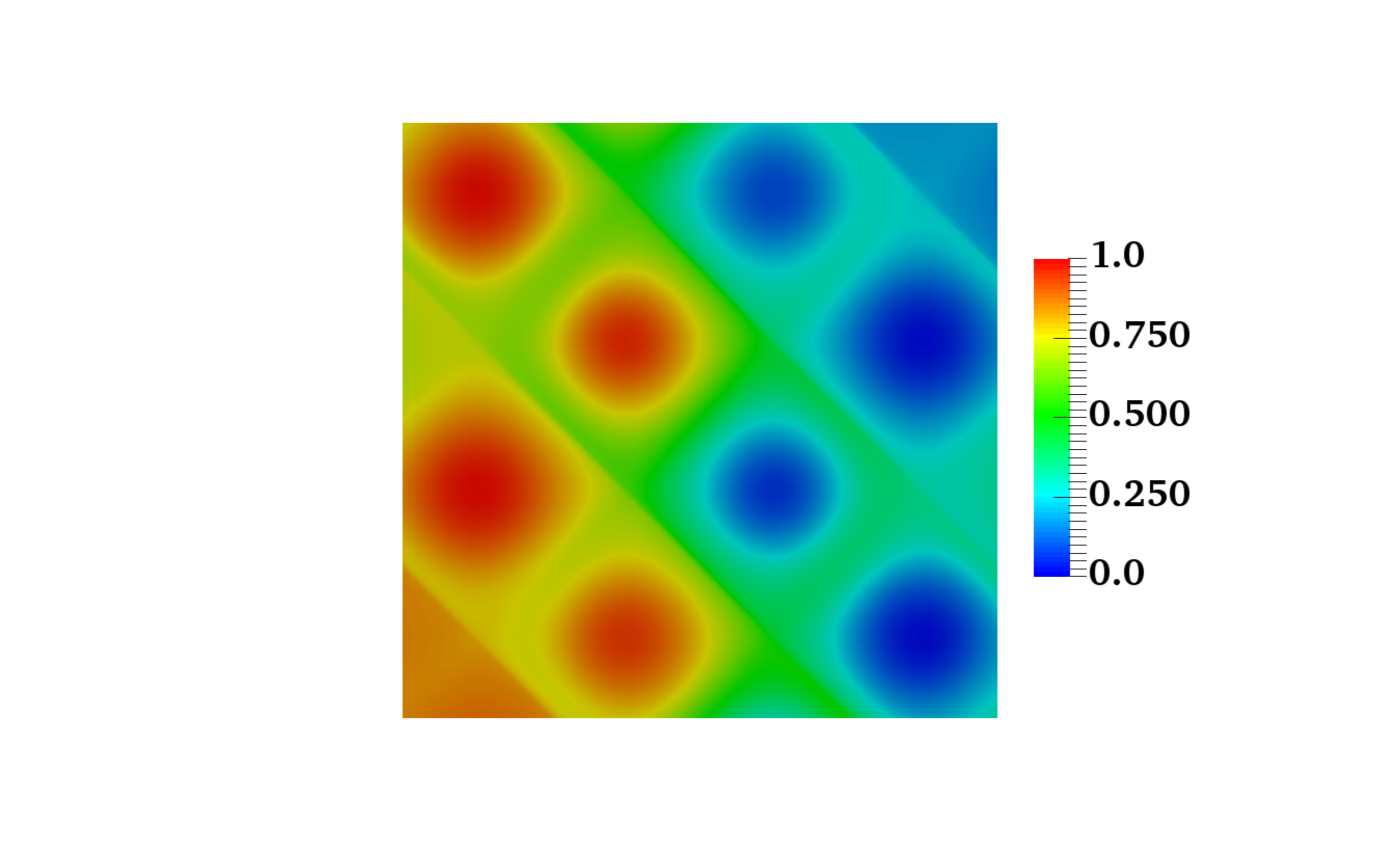}}
  \subfigure[$\kappa_fL = 3$ and $t = 0.1$]
    {\includegraphics[clip=true,width = 0.37\textwidth]
    {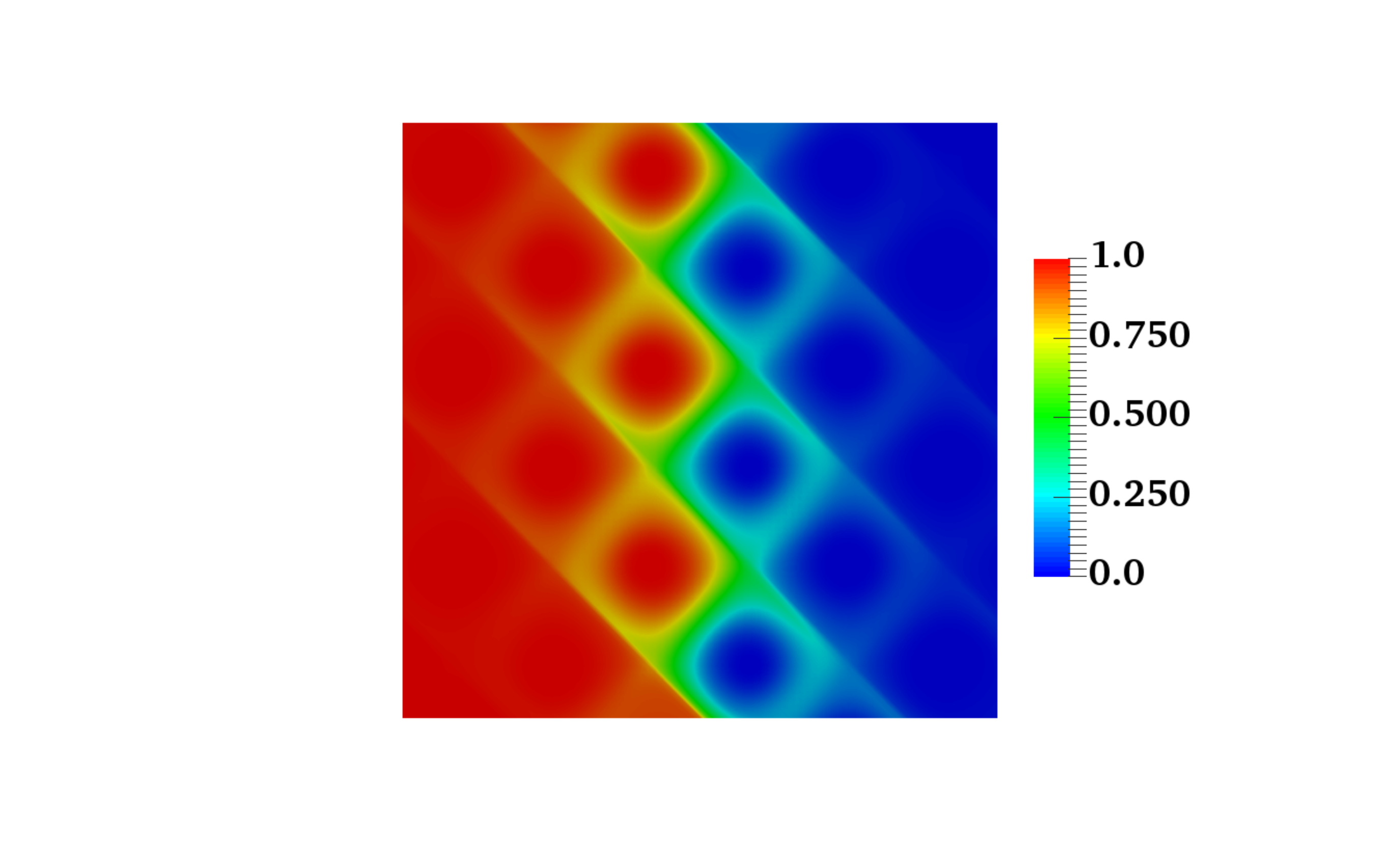}}
  \hspace{-0.5in}
  \subfigure[$\kappa_fL = 3$ and $t = 0.5$]
    {\includegraphics[clip=true,width = 0.37\textwidth]
    {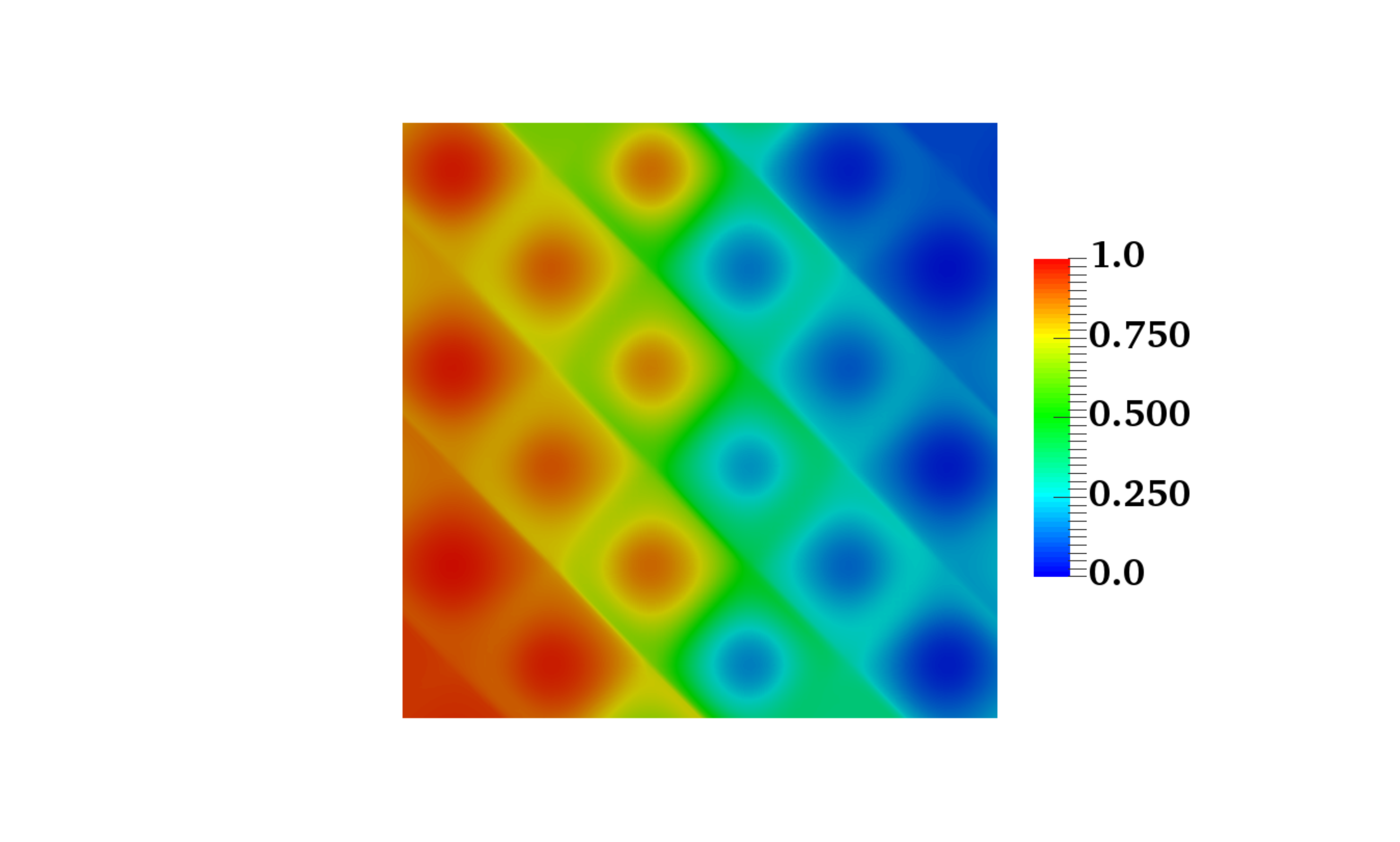}}
  \hspace{-0.5in}
  \subfigure[$\kappa_fL = 3$ and $t = 1.0$]
    {\includegraphics[clip=true,width = 0.37\textwidth]
    {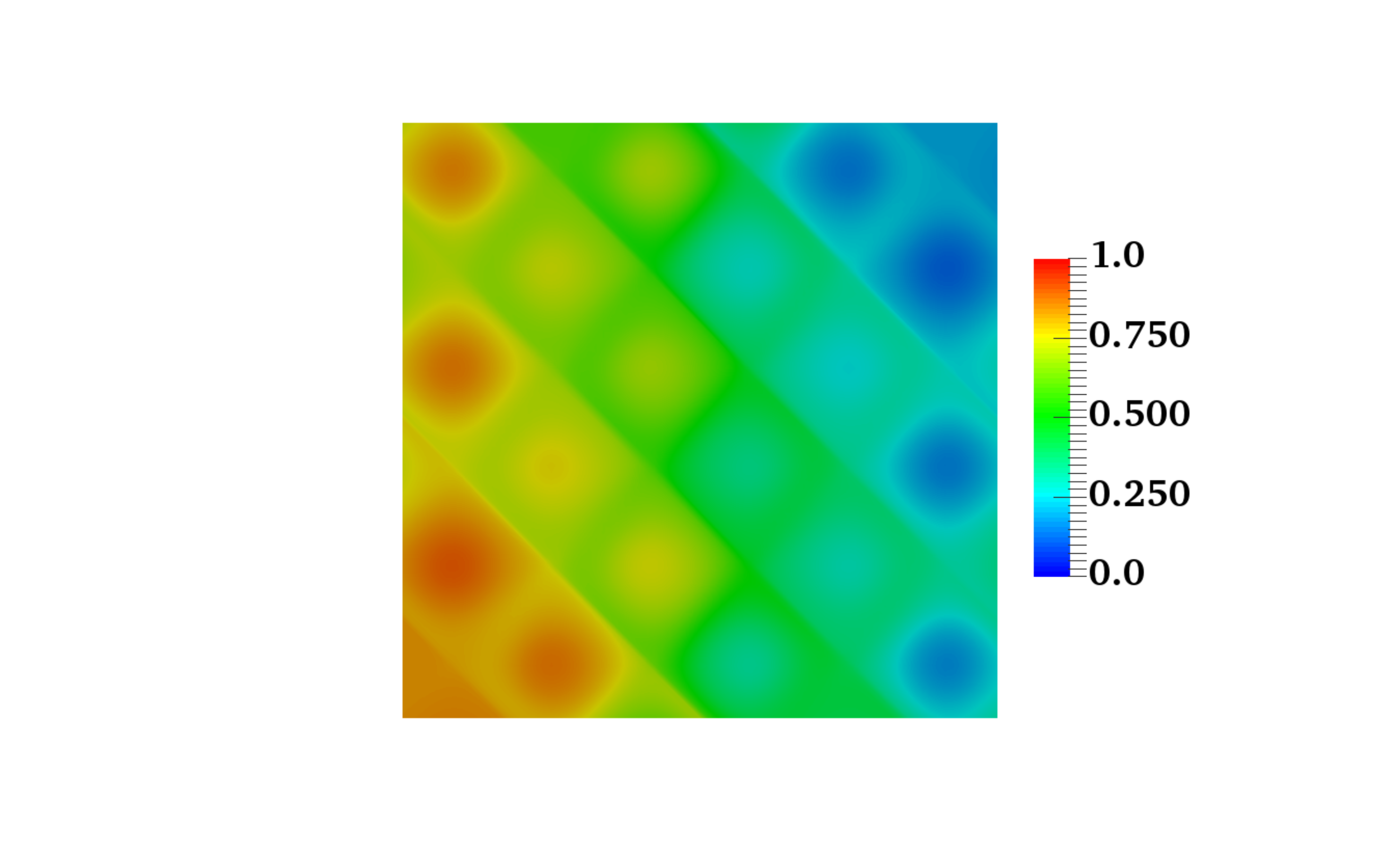}}
  \subfigure[$\kappa_fL = 4$ and $t = 0.1$]
    {\includegraphics[clip=true,width = 0.37\textwidth]
    {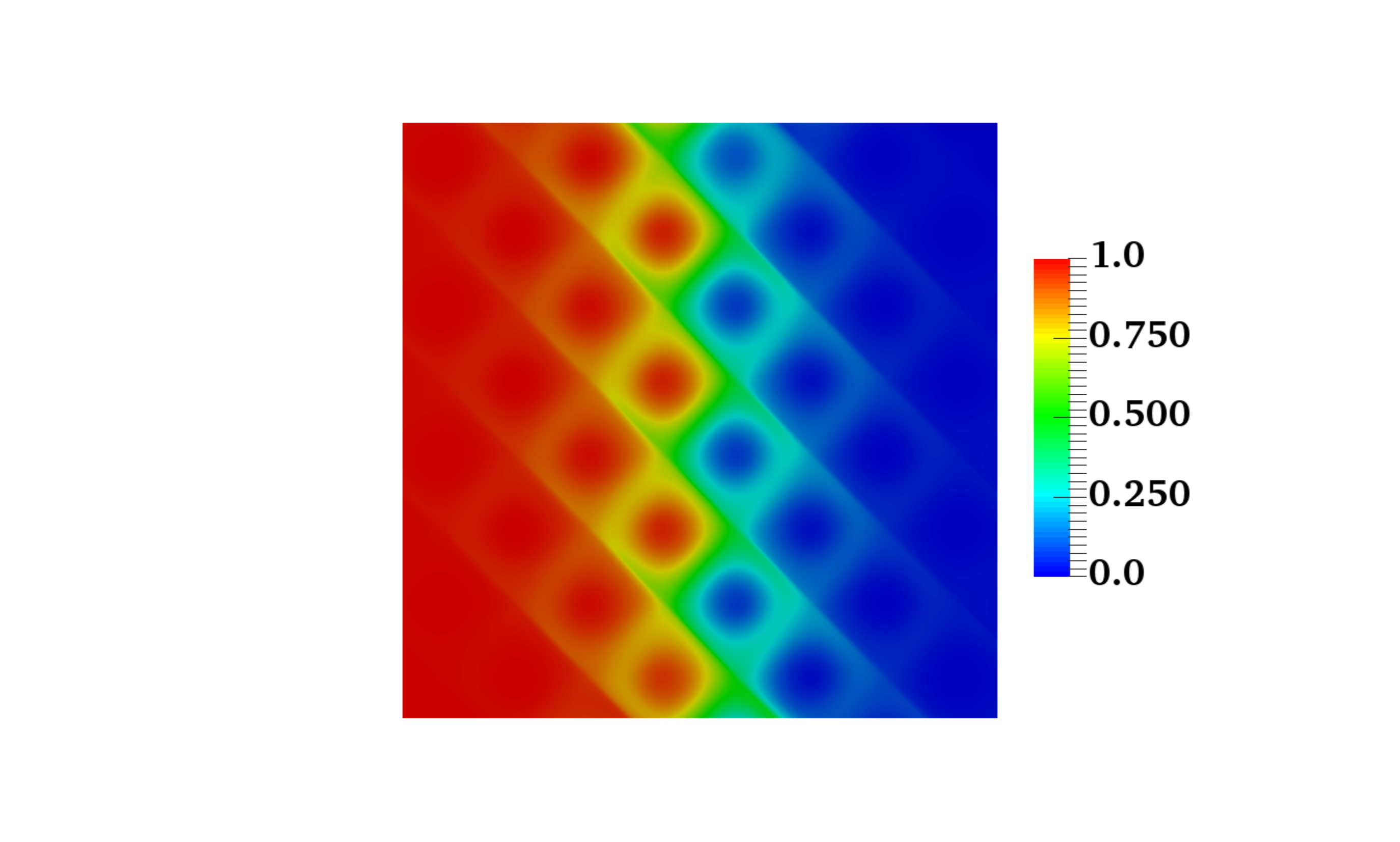}}
  \hspace{-0.5in}
  \subfigure[$\kappa_fL = 4$ and $t = 0.5$]
    {\includegraphics[clip=true,width = 0.37\textwidth]
    {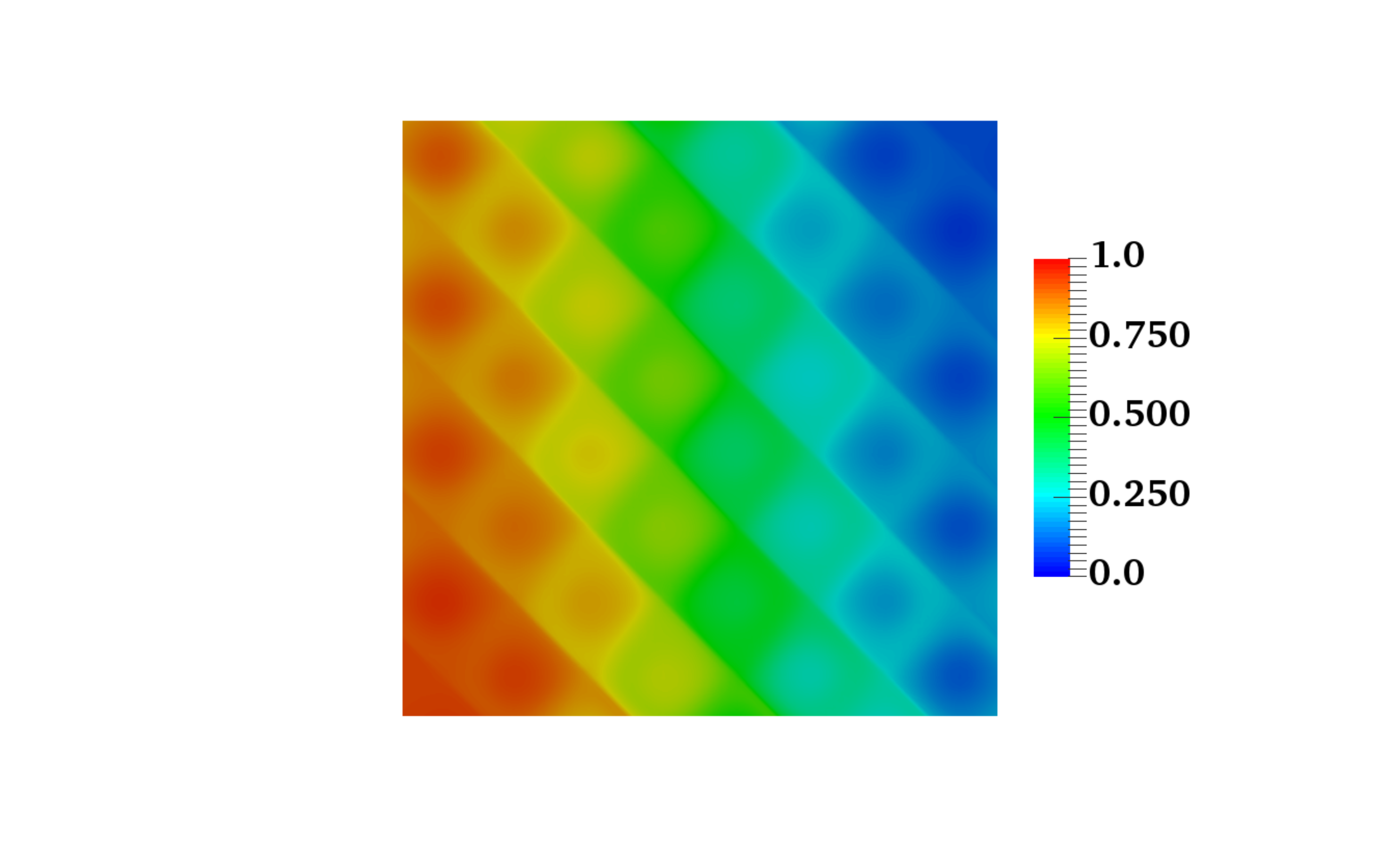}}
  \hspace{-0.5in}
  \subfigure[$\kappa_fL = 4$ and $t = 1.0$]
    {\includegraphics[clip=true,width = 0.37\textwidth]
    {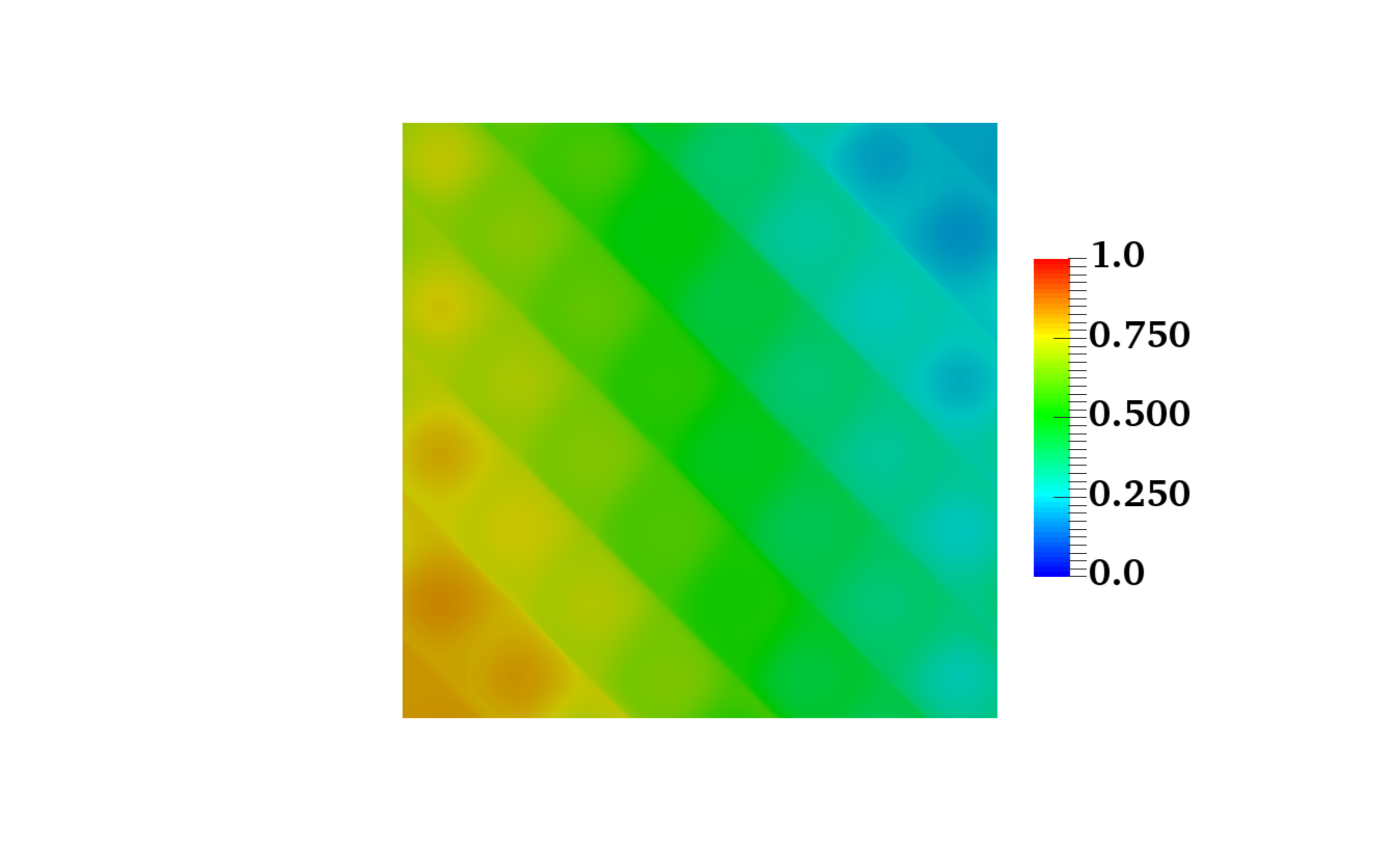}}
  \subfigure[$\kappa_fL = 5$ and $t = 0.1$]
    {\includegraphics[clip=true,width = 0.37\textwidth]
    {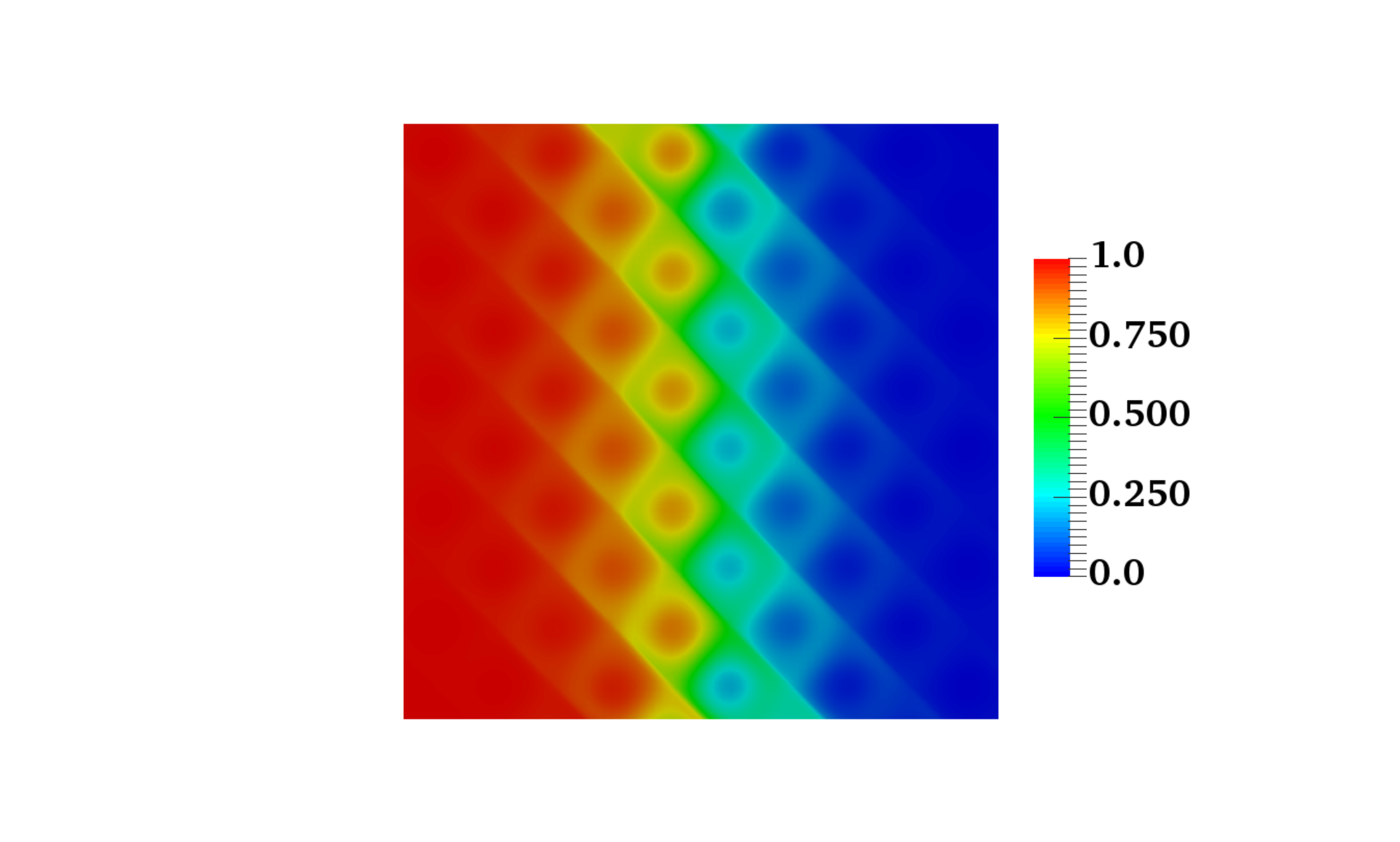}}
  \hspace{-0.5in}
  \subfigure[$\kappa_fL = 5$ and $t = 0.5$]
    {\includegraphics[clip=true,width = 0.37\textwidth]
    {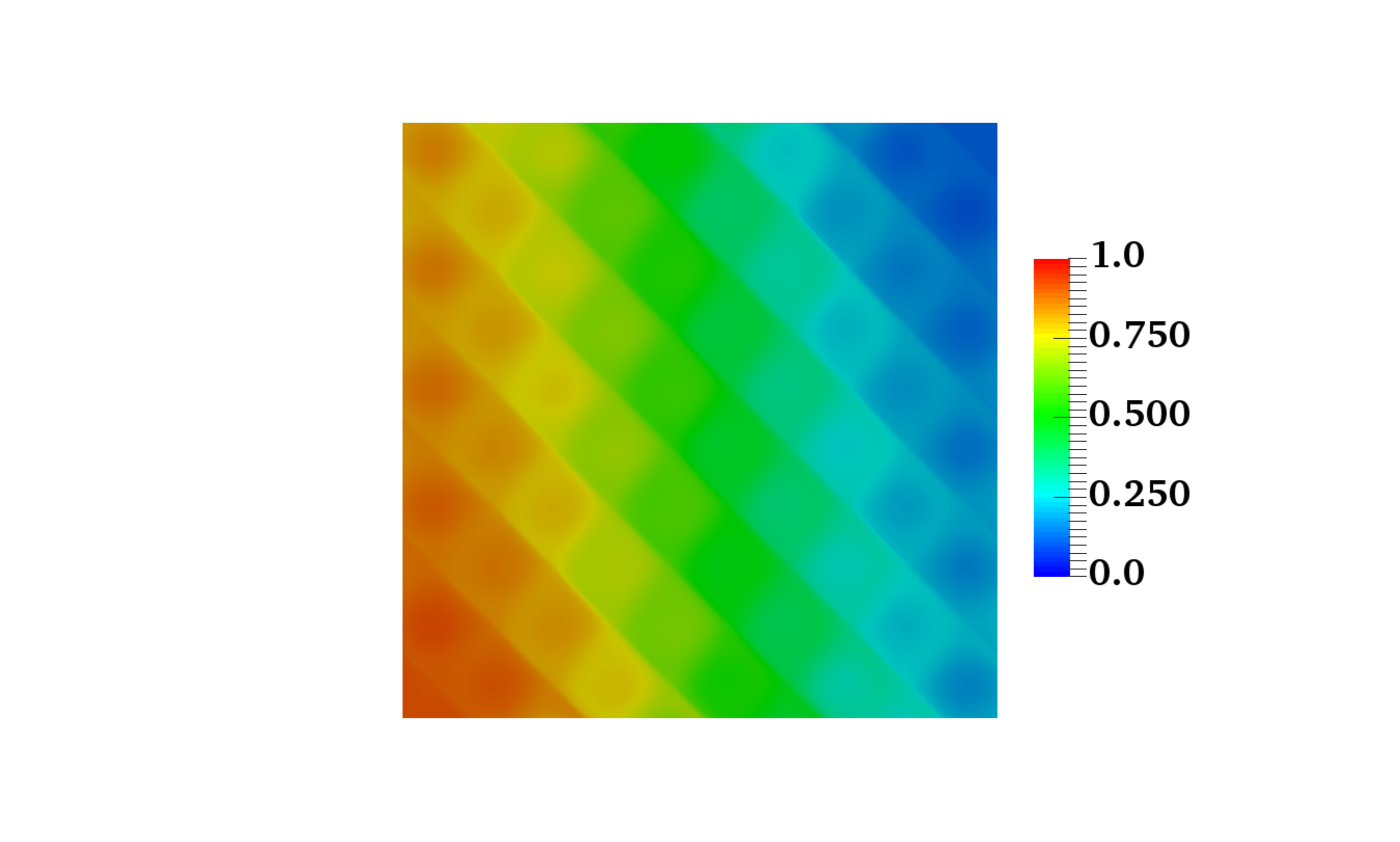}}
  \hspace{-0.5in}
  \subfigure[$\kappa_fL = 5$ and $t = 1.0$]
    {\includegraphics[clip=true,width = 0.37\textwidth]
    {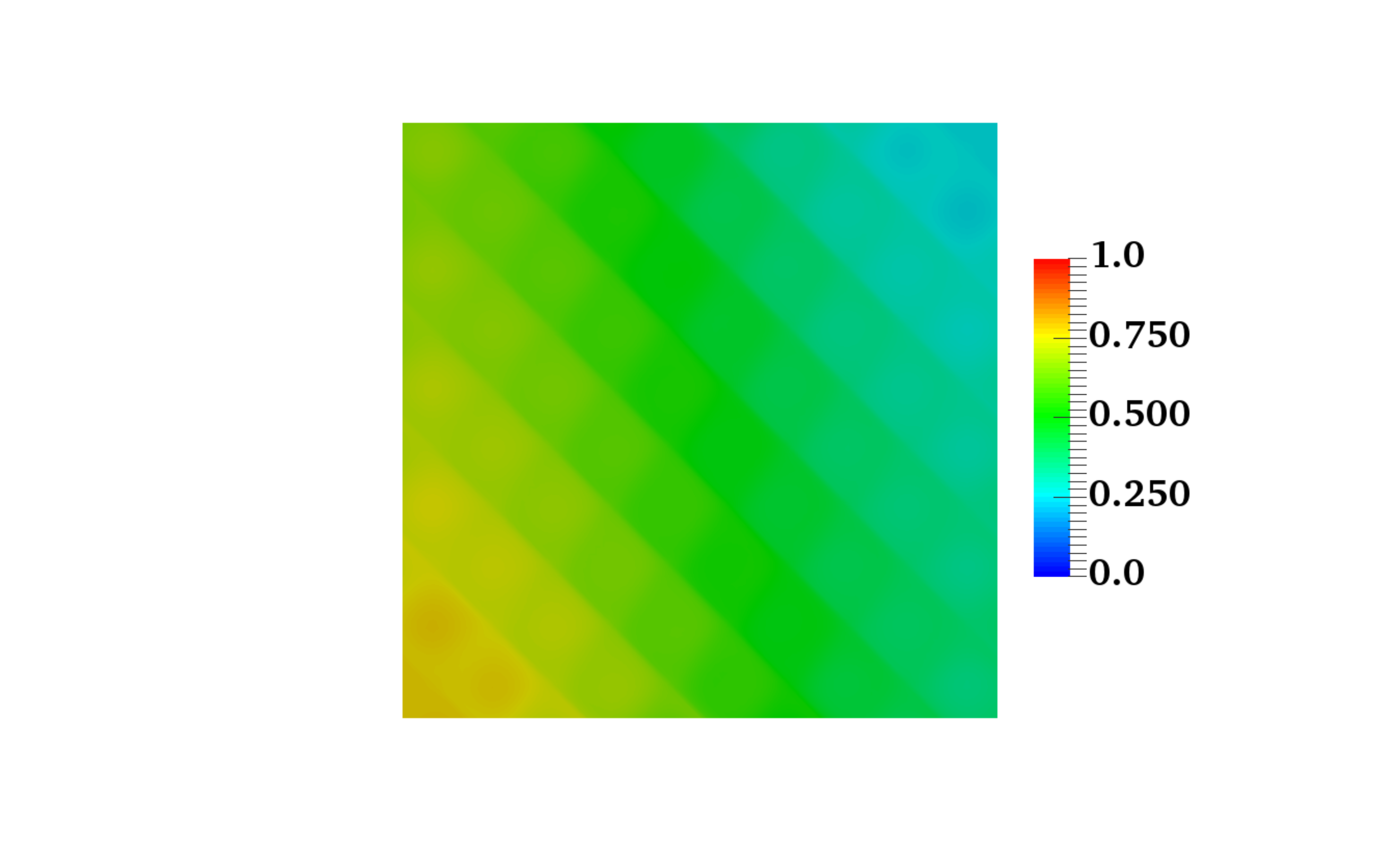}}
  \caption{\textsf{\textbf{Concentration contours of invariant-$F$:}}~These figures show the concentration of invariant-$F$ at times $t = 0.1, \, 0.5$, and $1.0$ for various values of $\kappa_fL$.
  \label{Fig:Contours_F_Difftimes}}
\end{figure}

\begin{figure}
  \centering
  \subfigure[$\kappa_fL = 2$ and $t = 0.1$]
    {\includegraphics[clip=true,width = 0.37\textwidth]
    {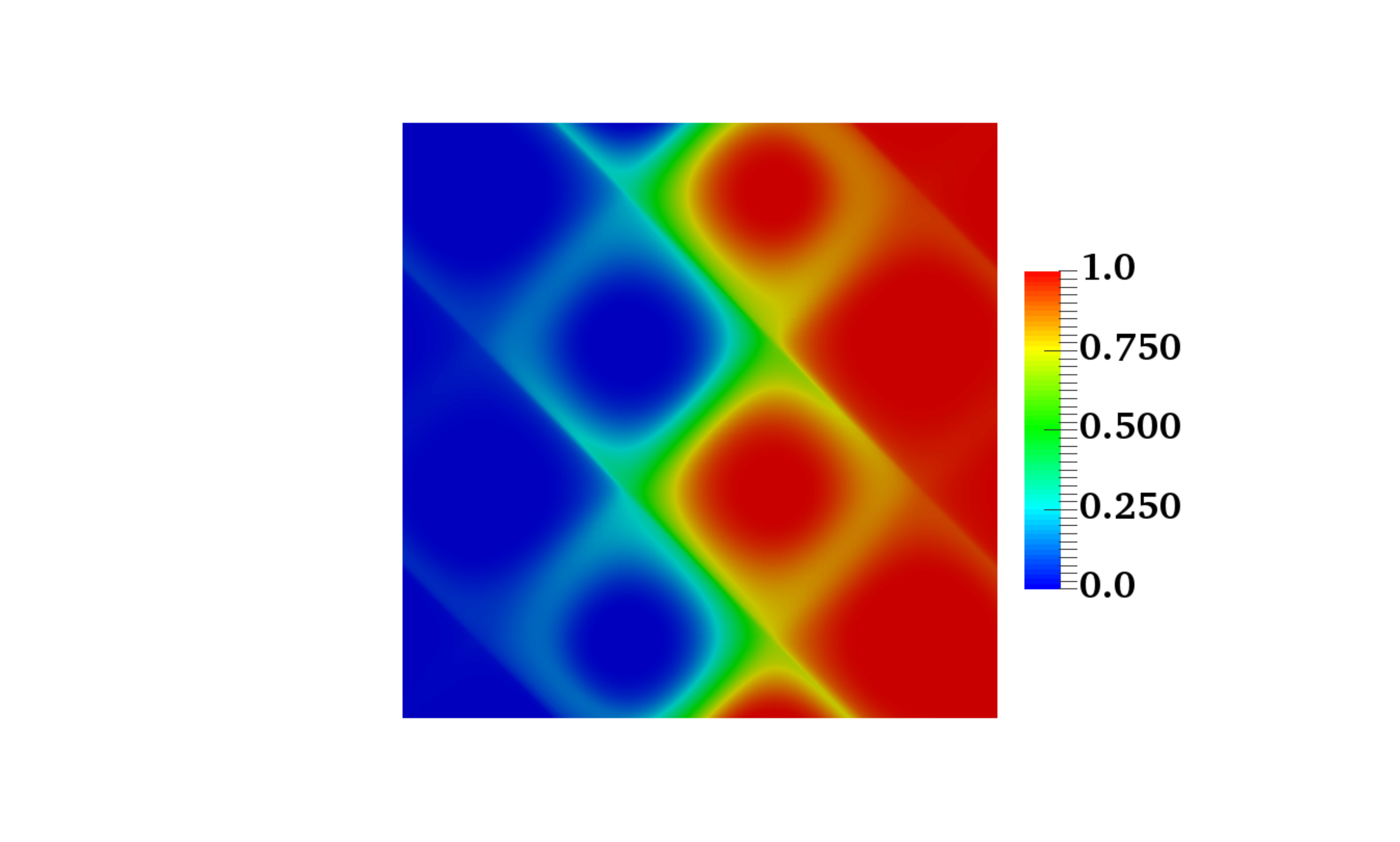}}
  \hspace{-0.5in}
  \subfigure[$\kappa_fL = 2$ and $t = 0.5$]
    {\includegraphics[clip=true,width = 0.37\textwidth]
    {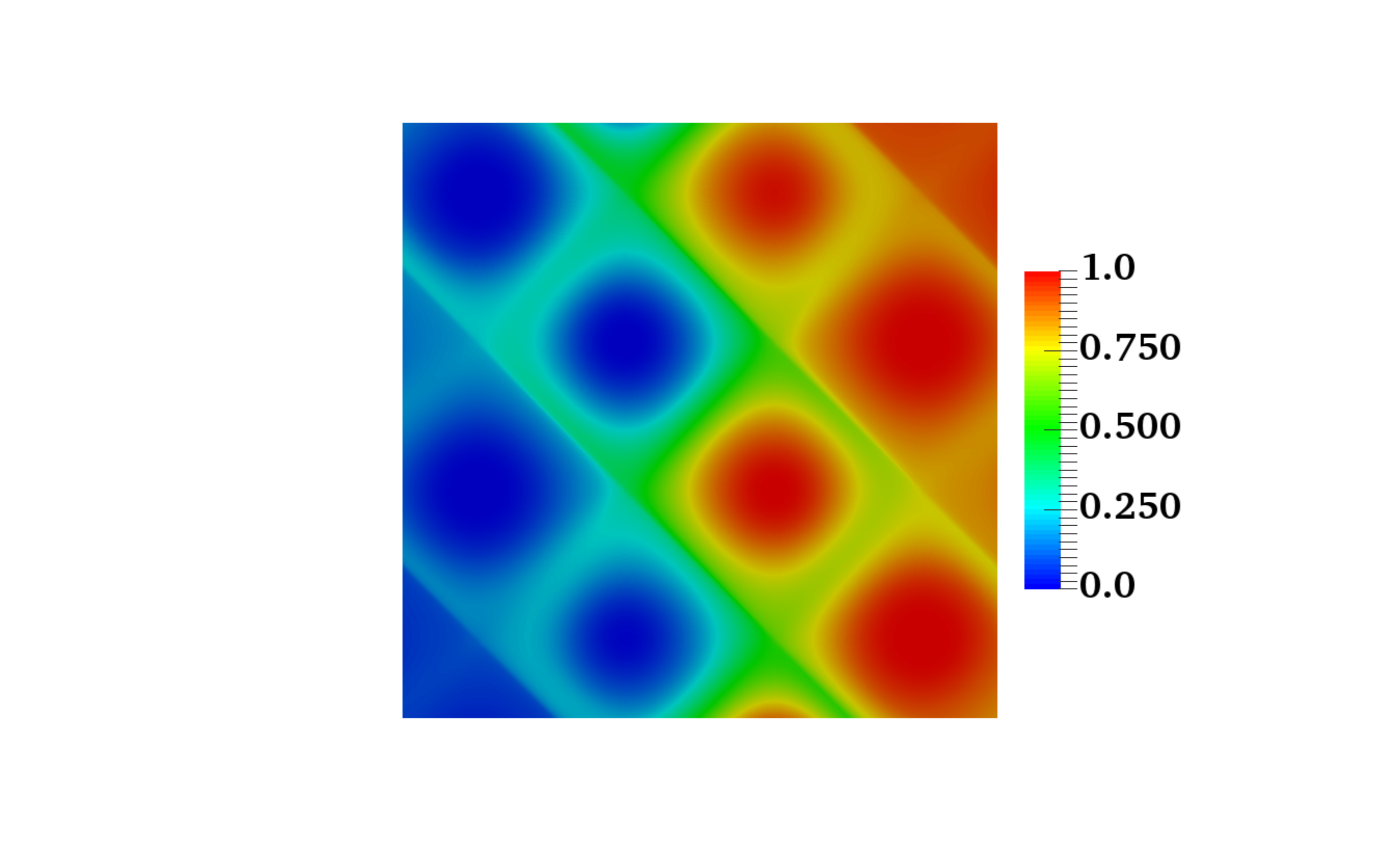}}
  \hspace{-0.5in}
  \subfigure[$\kappa_fL = 2$ and $t = 1.0$]
    {\includegraphics[clip=true,width = 0.37\textwidth]
    {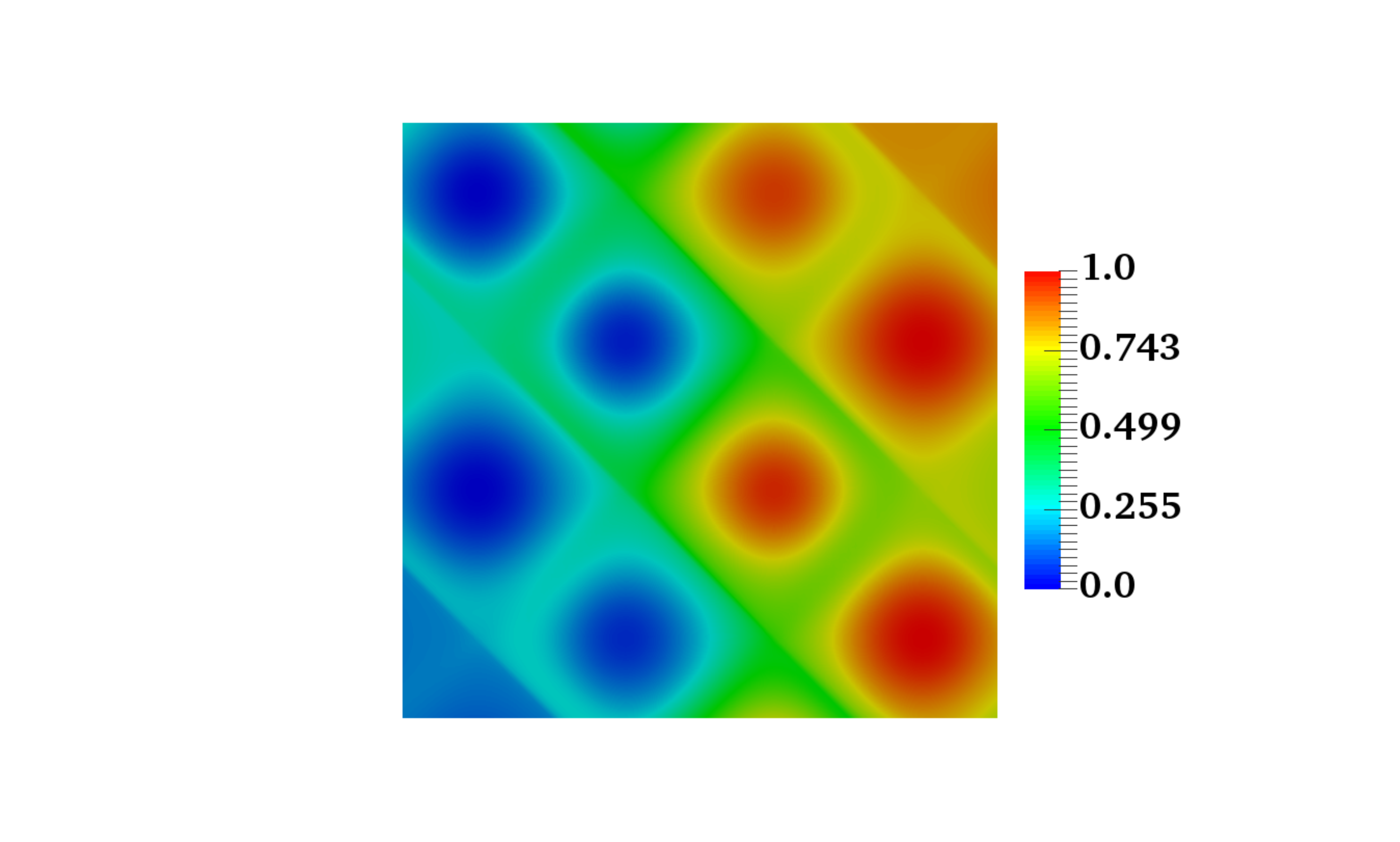}}
  \subfigure[$\kappa_fL = 3$ and $t = 0.1$]
    {\includegraphics[clip=true,width = 0.37\textwidth]
    {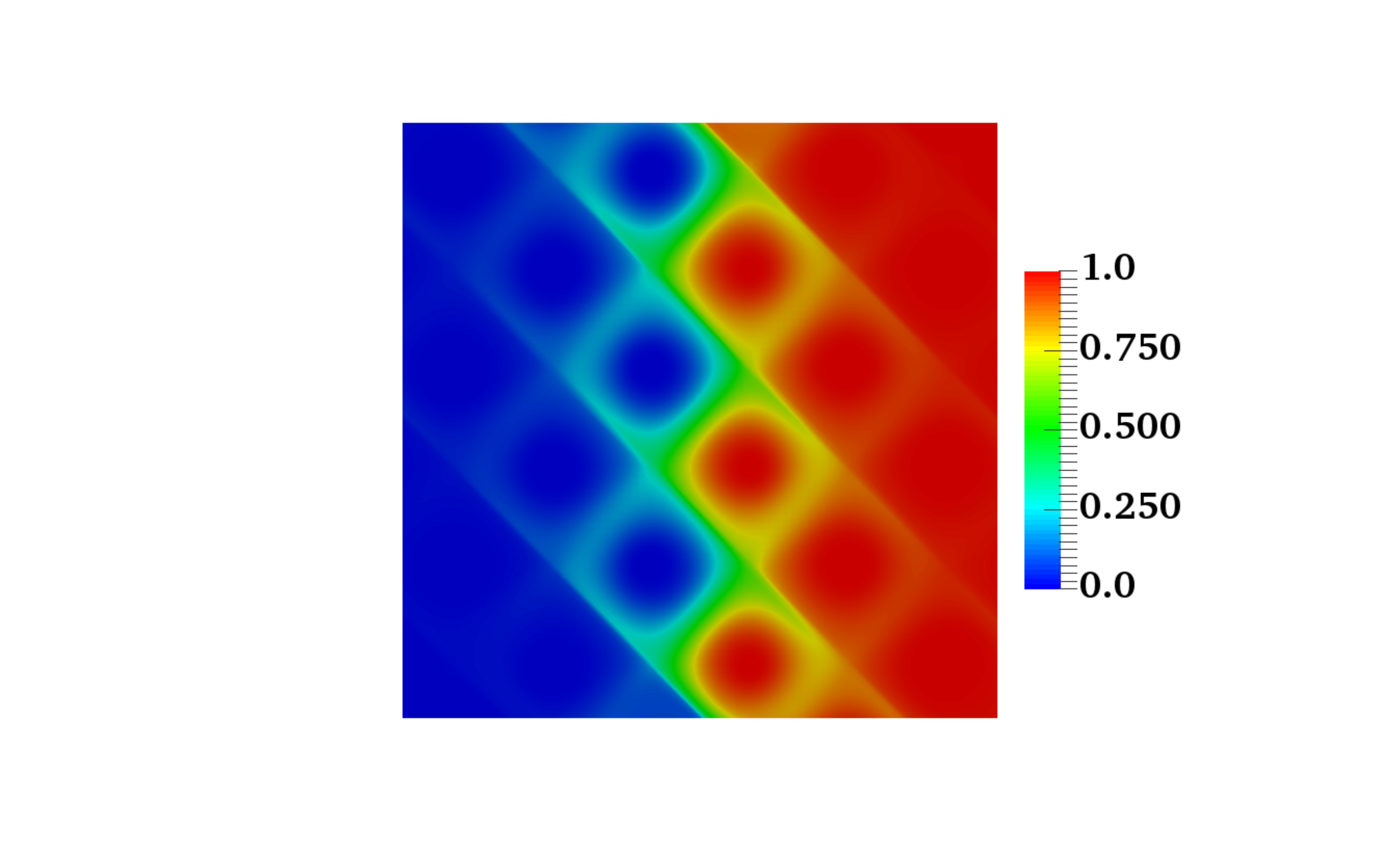}}
  \hspace{-0.5in}
  \subfigure[$\kappa_fL = 3$ and $t = 0.5$]
    {\includegraphics[clip=true,width = 0.37\textwidth]
    {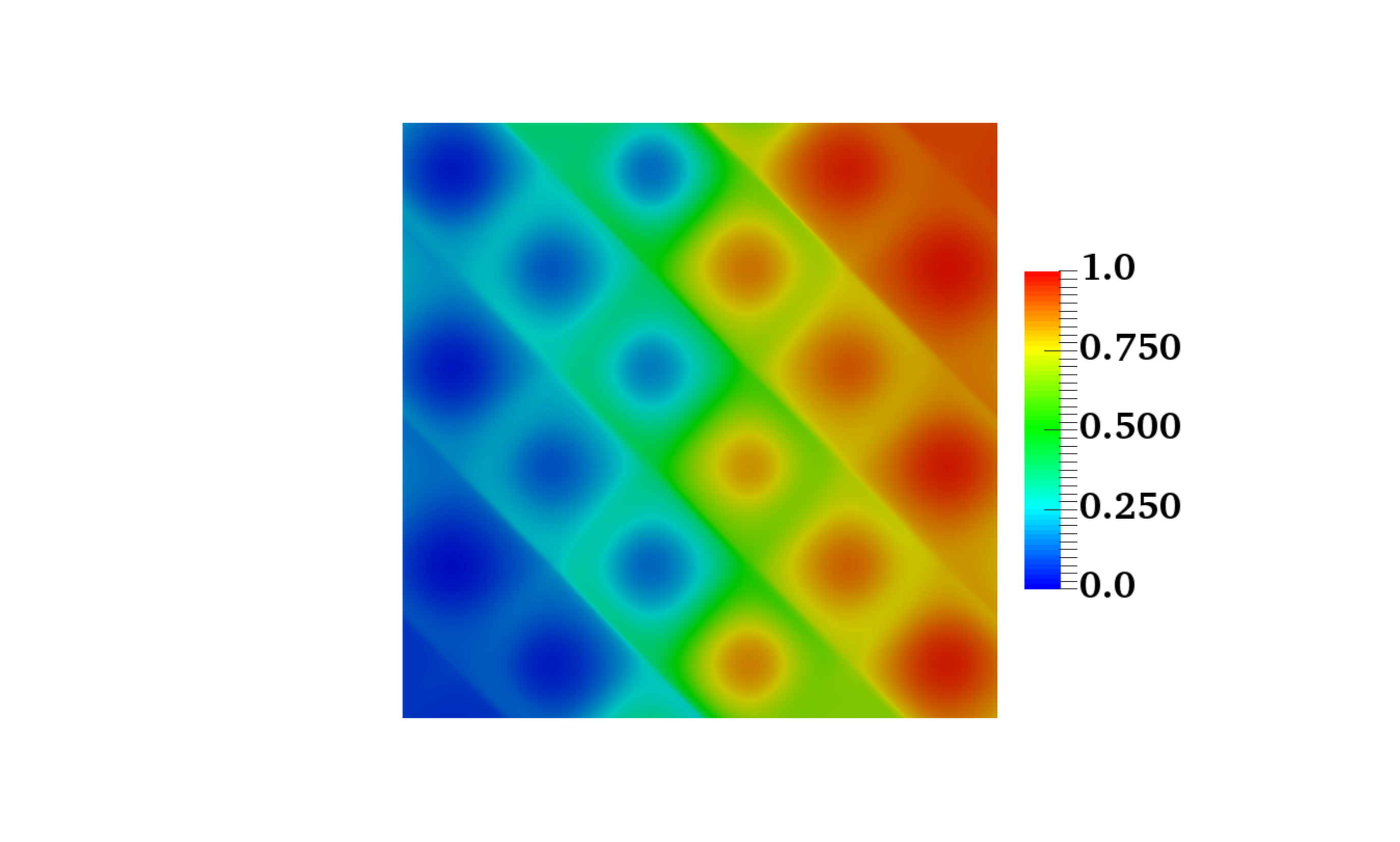}}
  \hspace{-0.5in}
  \subfigure[$\kappa_fL = 3$ and $t = 1.0$]
    {\includegraphics[clip=true,width = 0.37\textwidth]
    {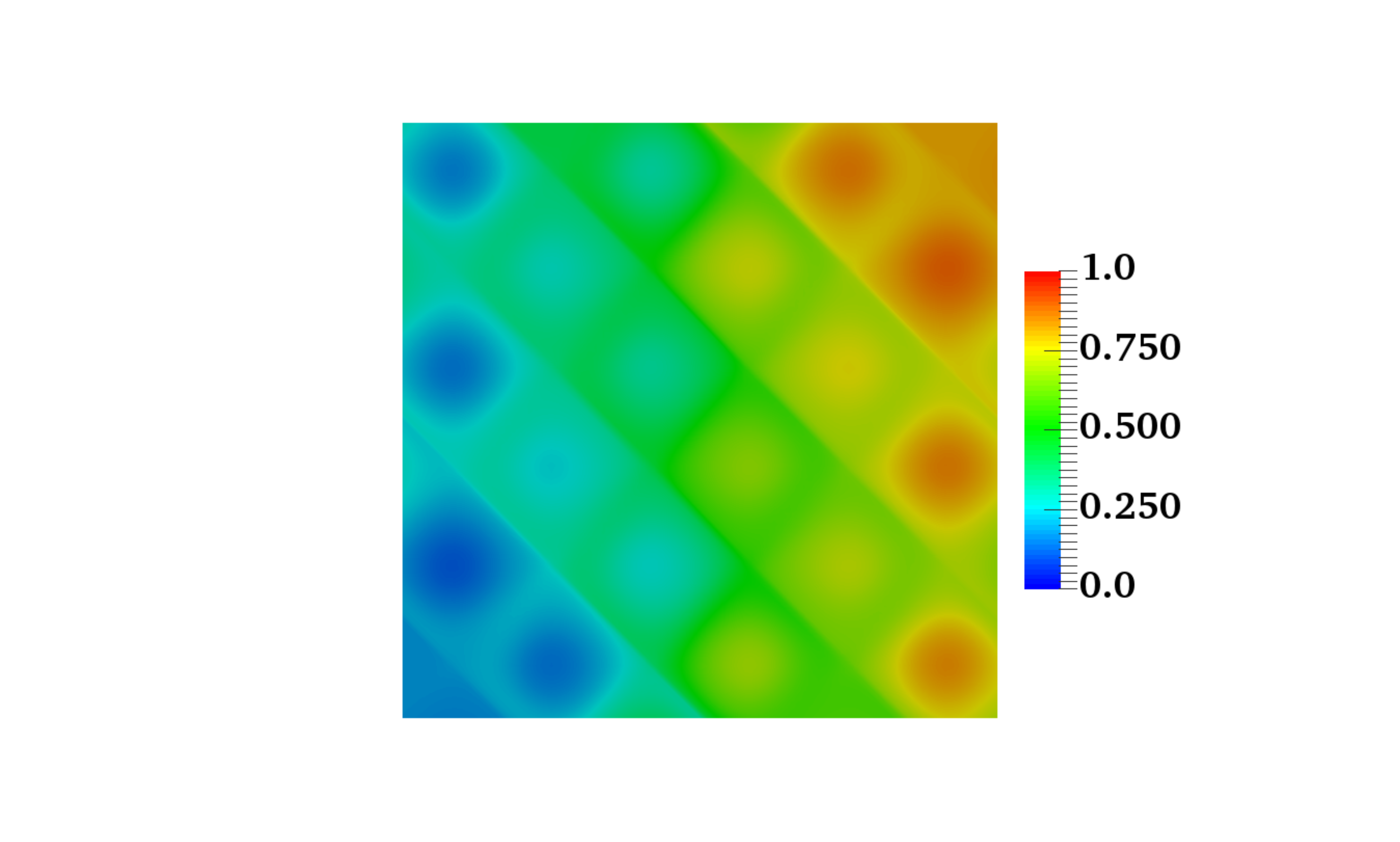}}
  \subfigure[$\kappa_fL = 4$ and $t = 0.1$]
    {\includegraphics[clip=true,width = 0.37\textwidth]
    {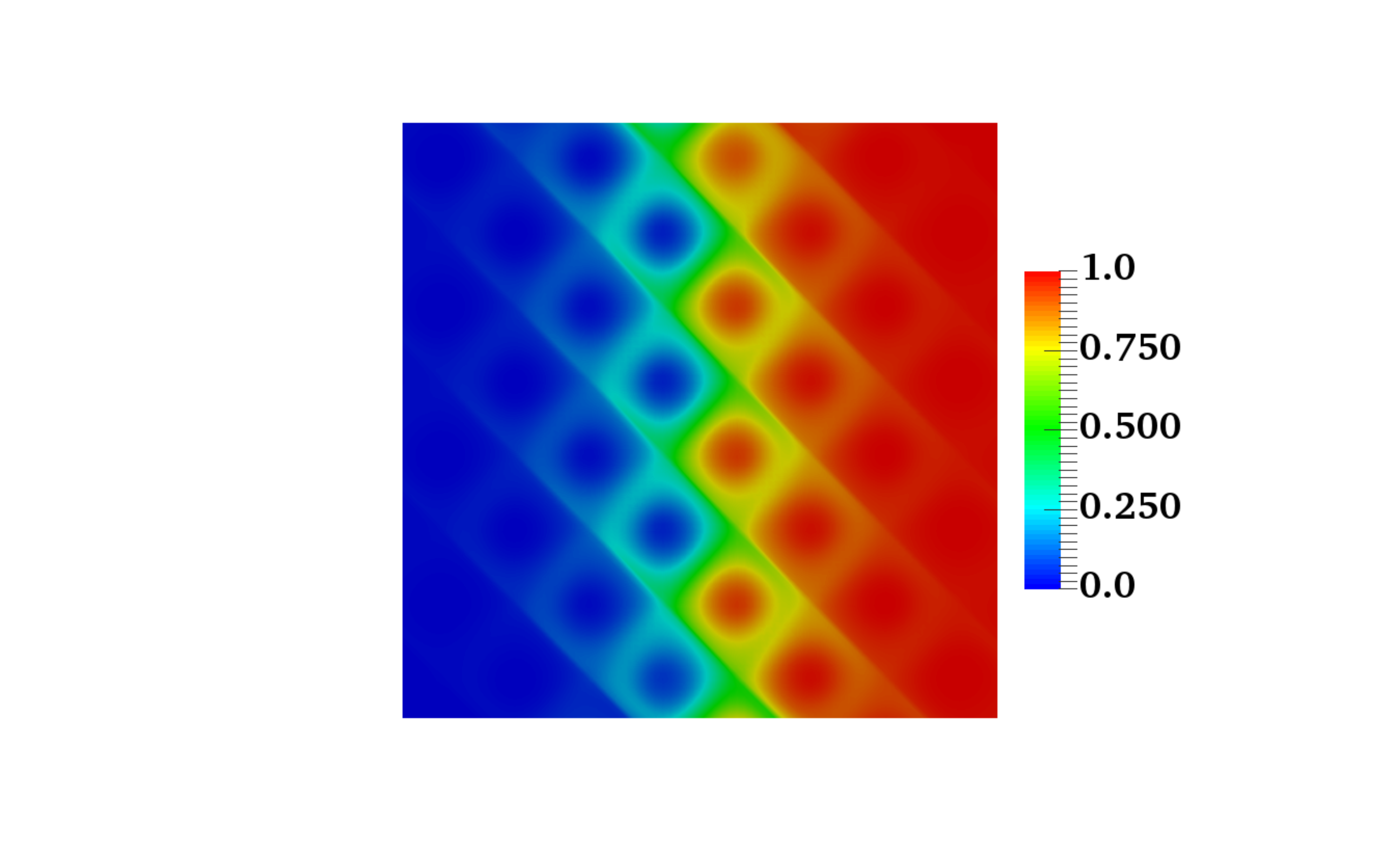}}
  \hspace{-0.5in}
  \subfigure[$\kappa_fL = 4$ and $t = 0.5$]
    {\includegraphics[clip=true,width = 0.37\textwidth]
    {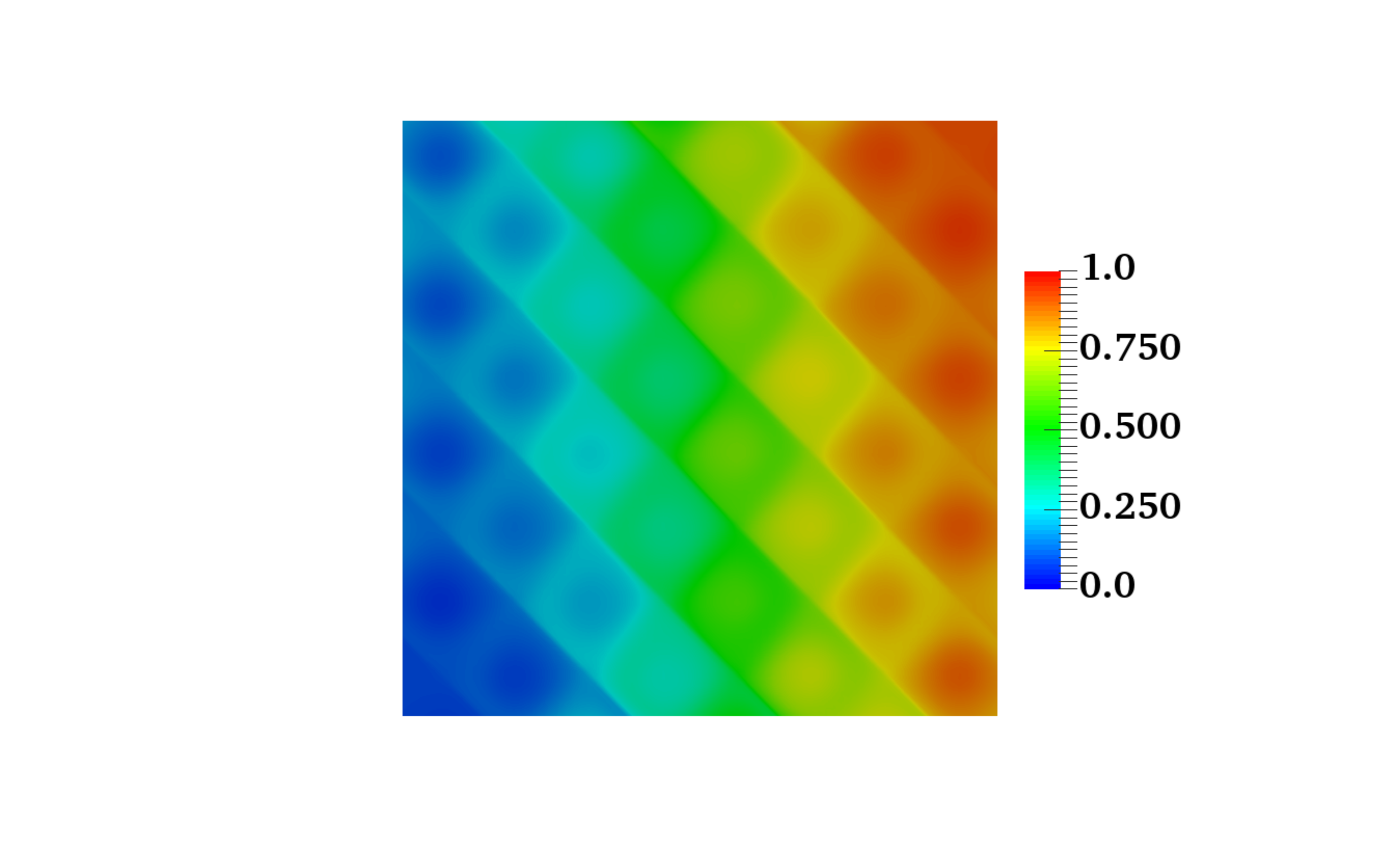}}
  \hspace{-0.5in}
  \subfigure[$\kappa_fL = 4$ and $t = 1.0$]
    {\includegraphics[clip=true,width = 0.37\textwidth]
    {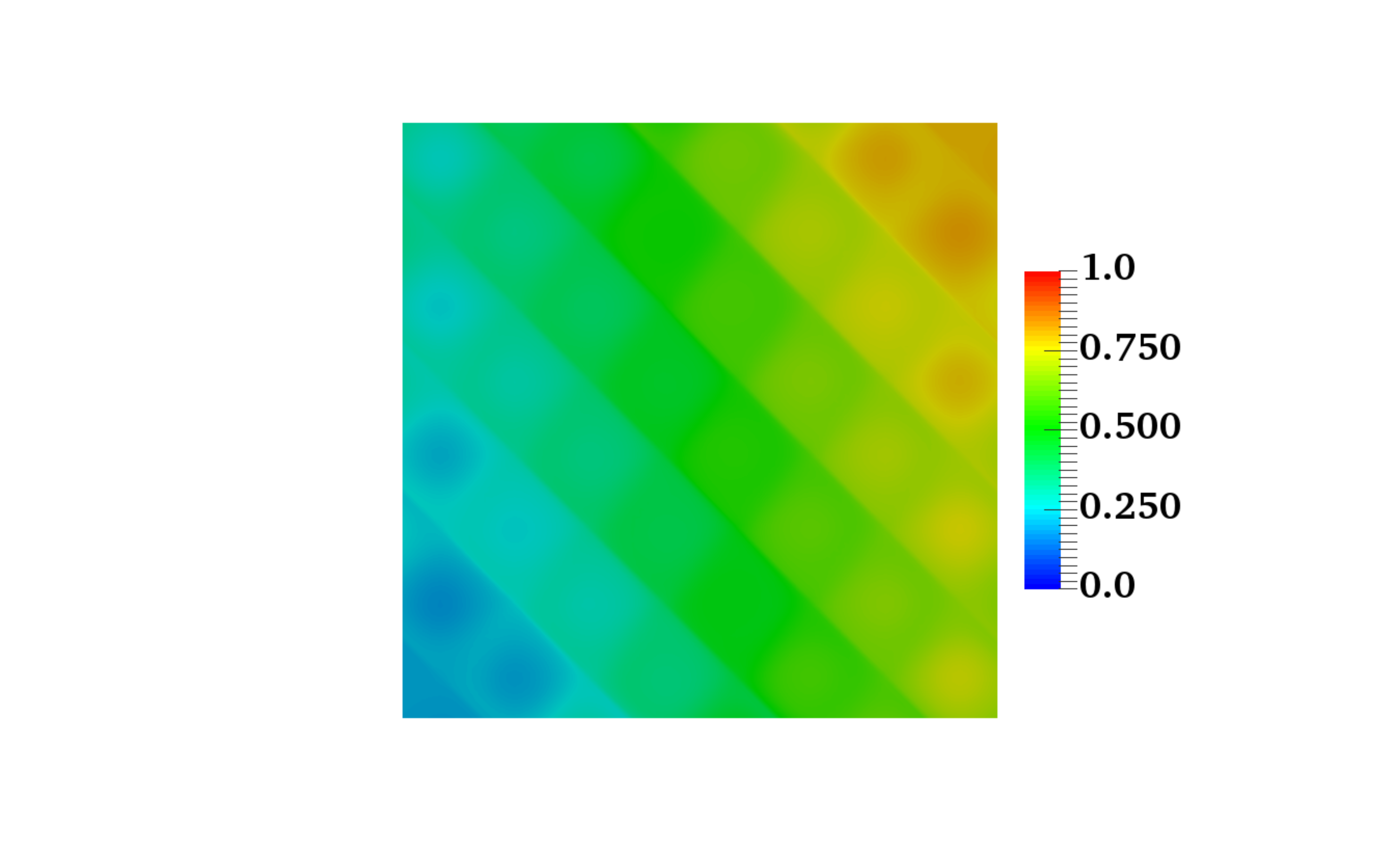}}
  \subfigure[$\kappa_fL = 5$ and $t = 0.1$]
    {\includegraphics[clip=true,width = 0.37\textwidth]
    {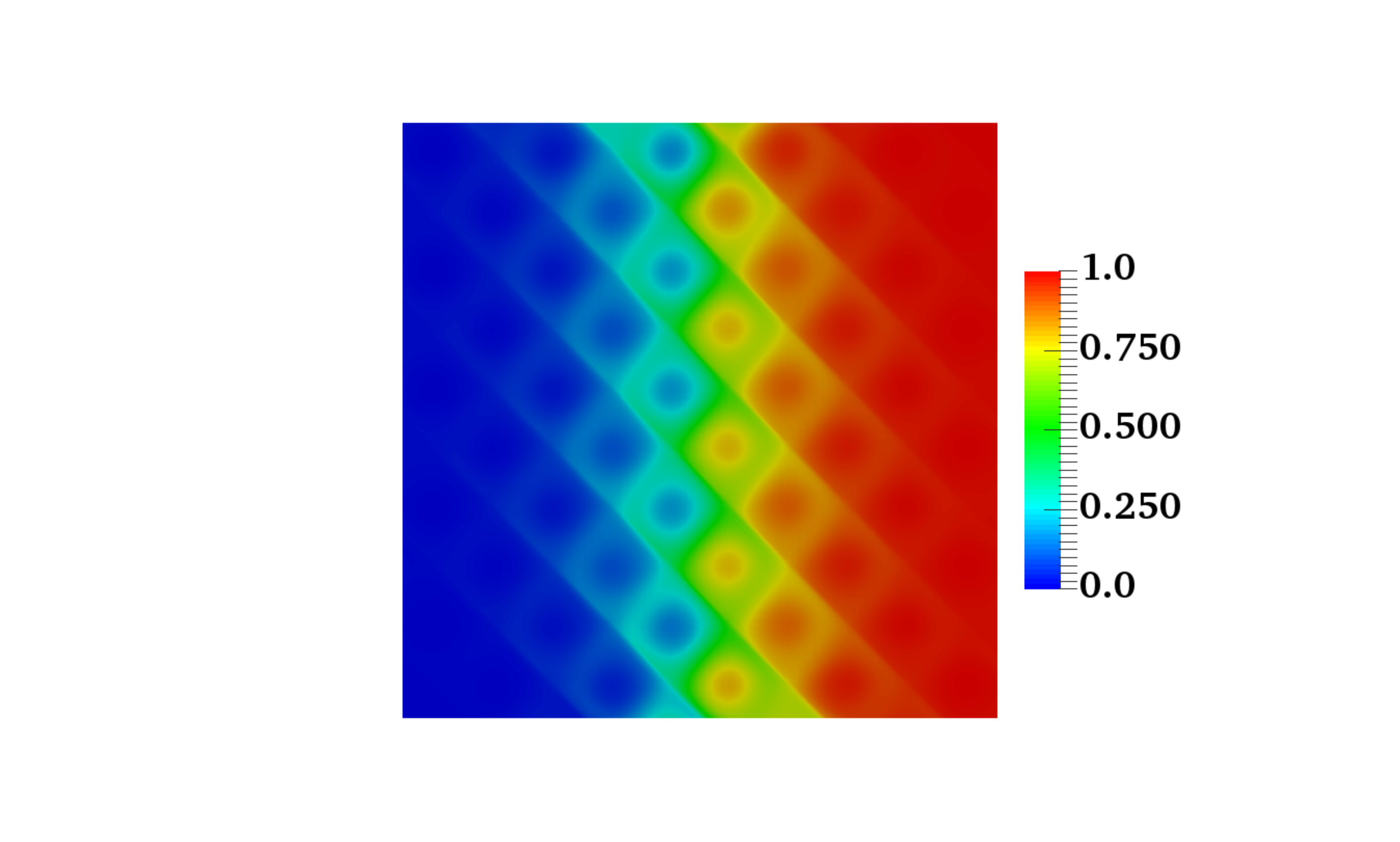}}
  \hspace{-0.5in}
  \subfigure[$\kappa_fL = 5$ and $t = 0.5$]
    {\includegraphics[clip=true,width = 0.37\textwidth]
    {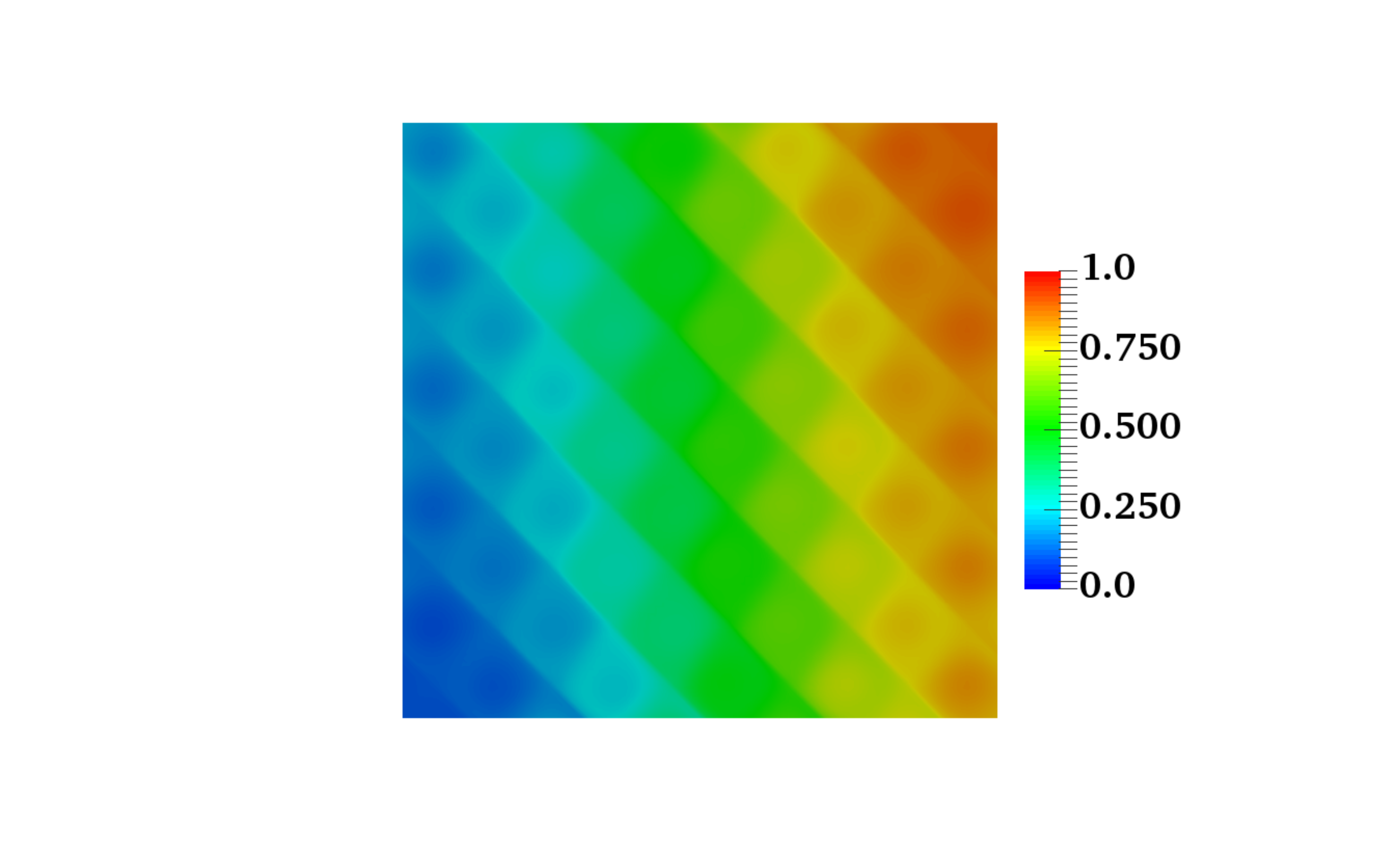}}
  \hspace{-0.5in}
  \subfigure[$\kappa_fL = 5$ and $t = 1.0$]
    {\includegraphics[clip=true,width = 0.37\textwidth]
    {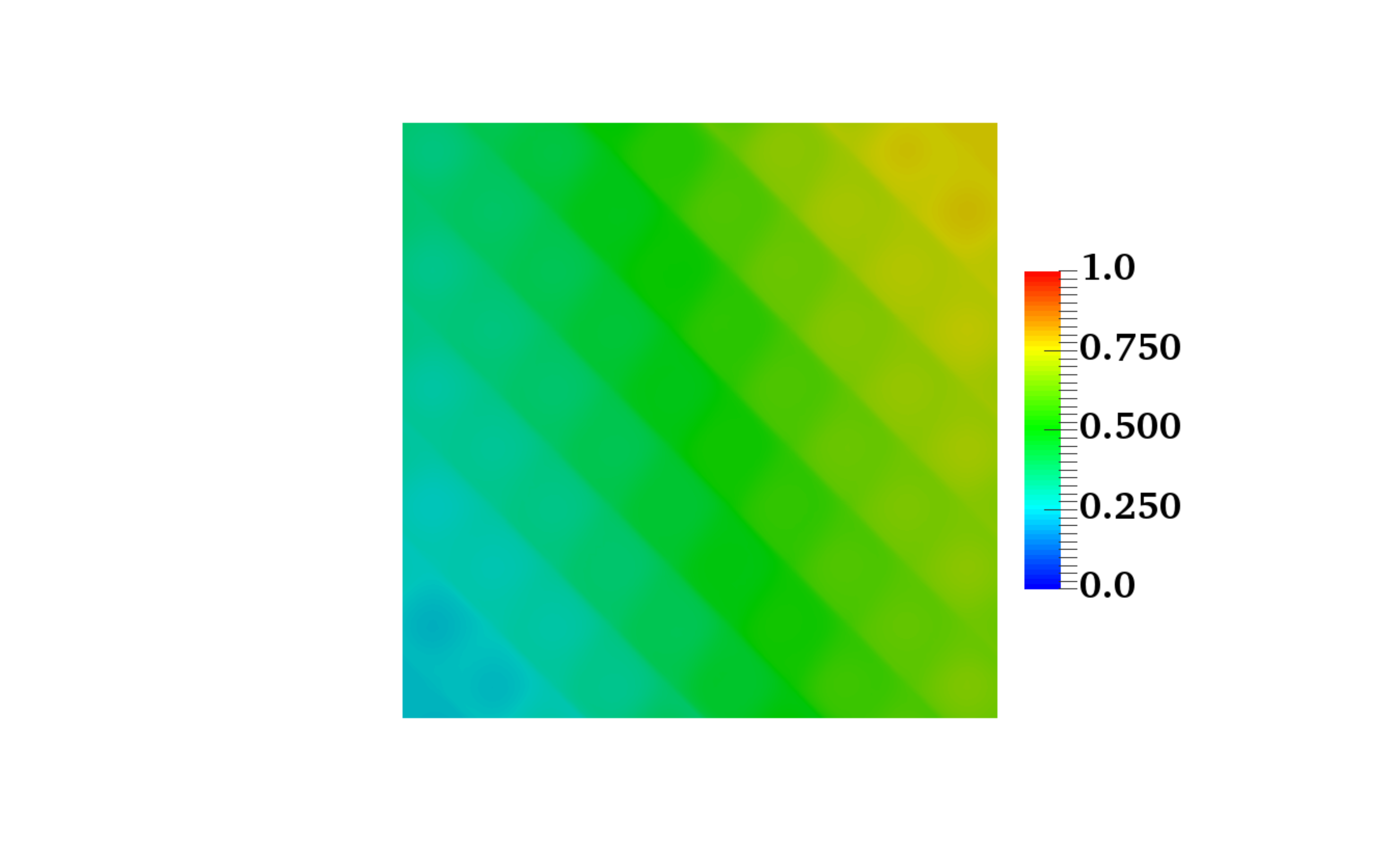}}
  \caption{\textsf{\textbf{Concentration contours of invariant-$G$:}}~These figures show the concentration of invariant-$G$ at times $t = 0.1, \, 0.5$, and $1.0$ for various values of $\kappa_fL$.
  \label{Fig:Contours_G_Difftimes}}
\end{figure}

\begin{figure}
  \centering
  \subfigure[$\kappa_fL = 2$ and $t = 0.1$]
    {\includegraphics[clip=true,width = 0.37\textwidth]
    {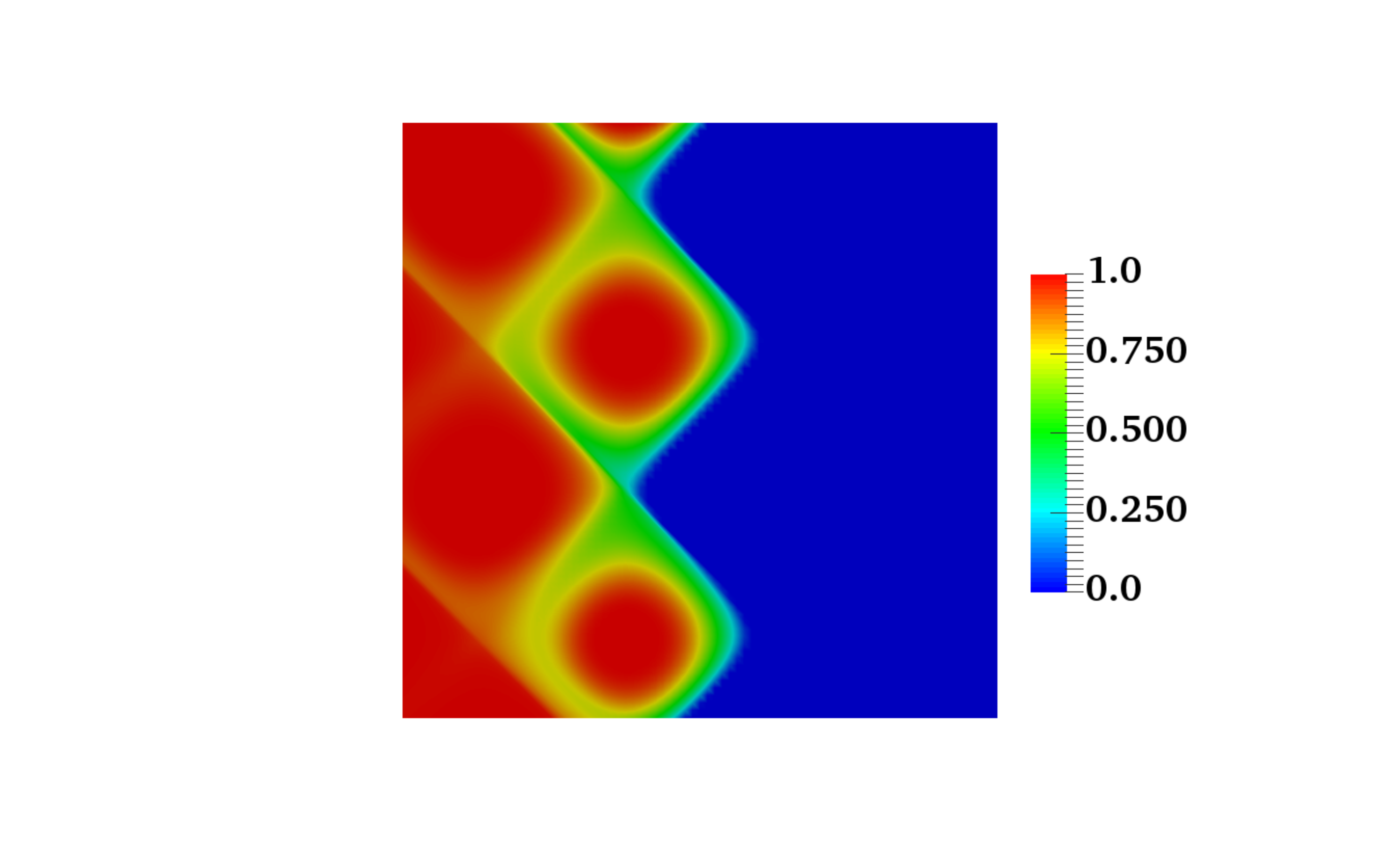}}
  \hspace{-0.5in}
  \subfigure[$\kappa_fL = 2$ and $t = 0.5$]
    {\includegraphics[clip=true,width = 0.37\textwidth]
    {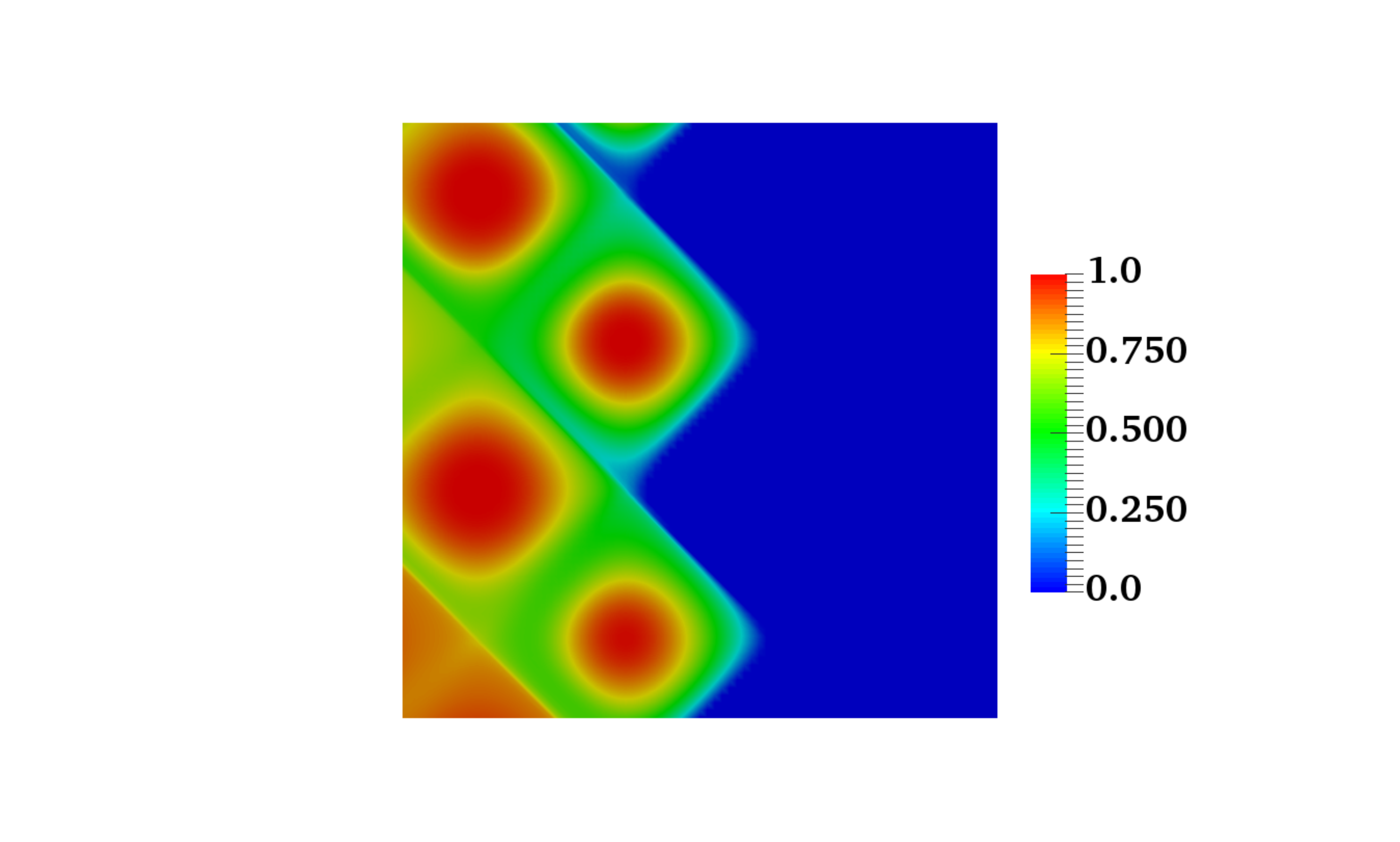}}
  \hspace{-0.5in}
  \subfigure[$\kappa_fL = 2$ and $t = 1.0$]
    {\includegraphics[clip=true,width = 0.37\textwidth]
    {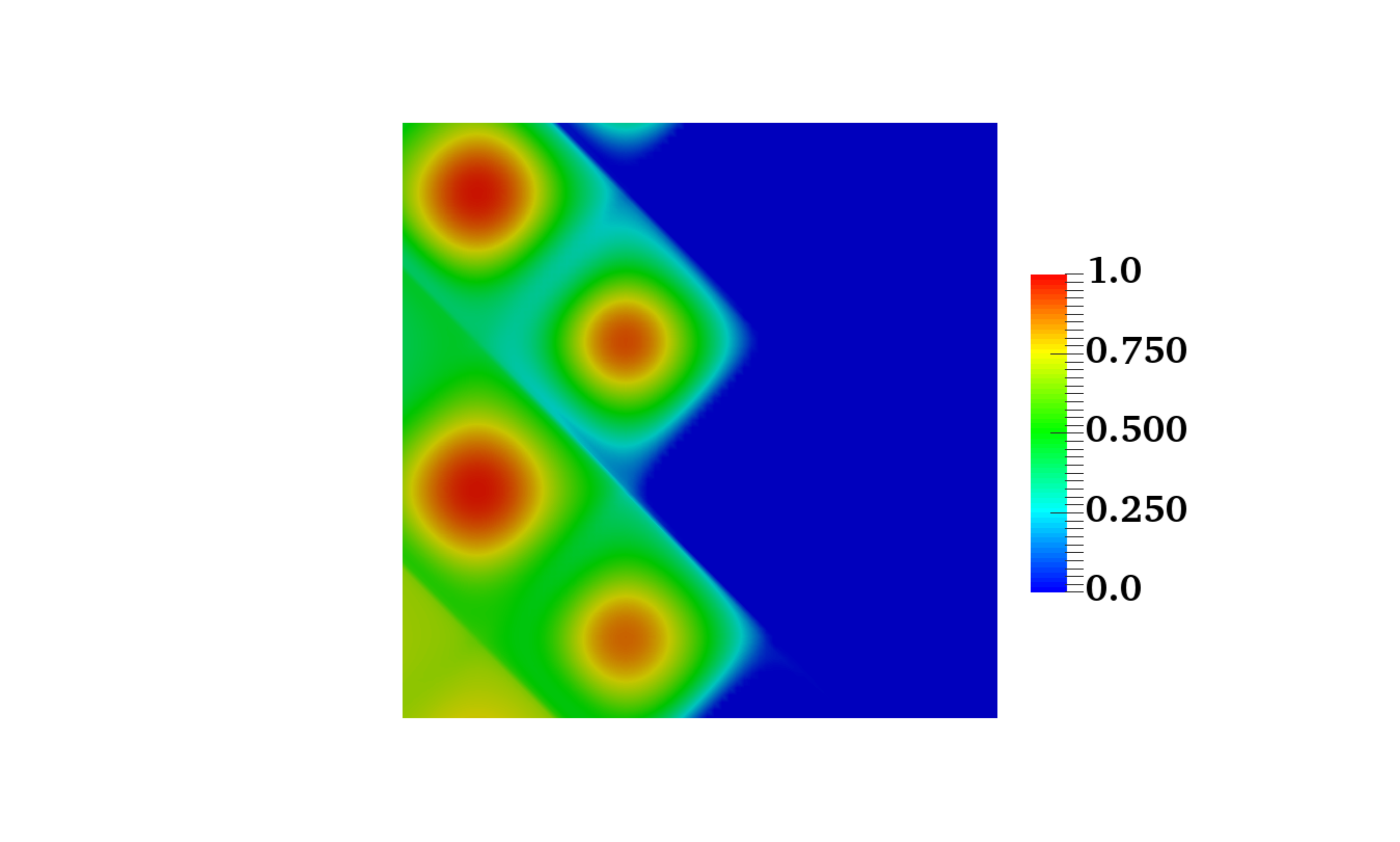}}
  \subfigure[$\kappa_fL = 3$ and $t = 0.1$]
    {\includegraphics[clip=true,width = 0.37\textwidth]
    {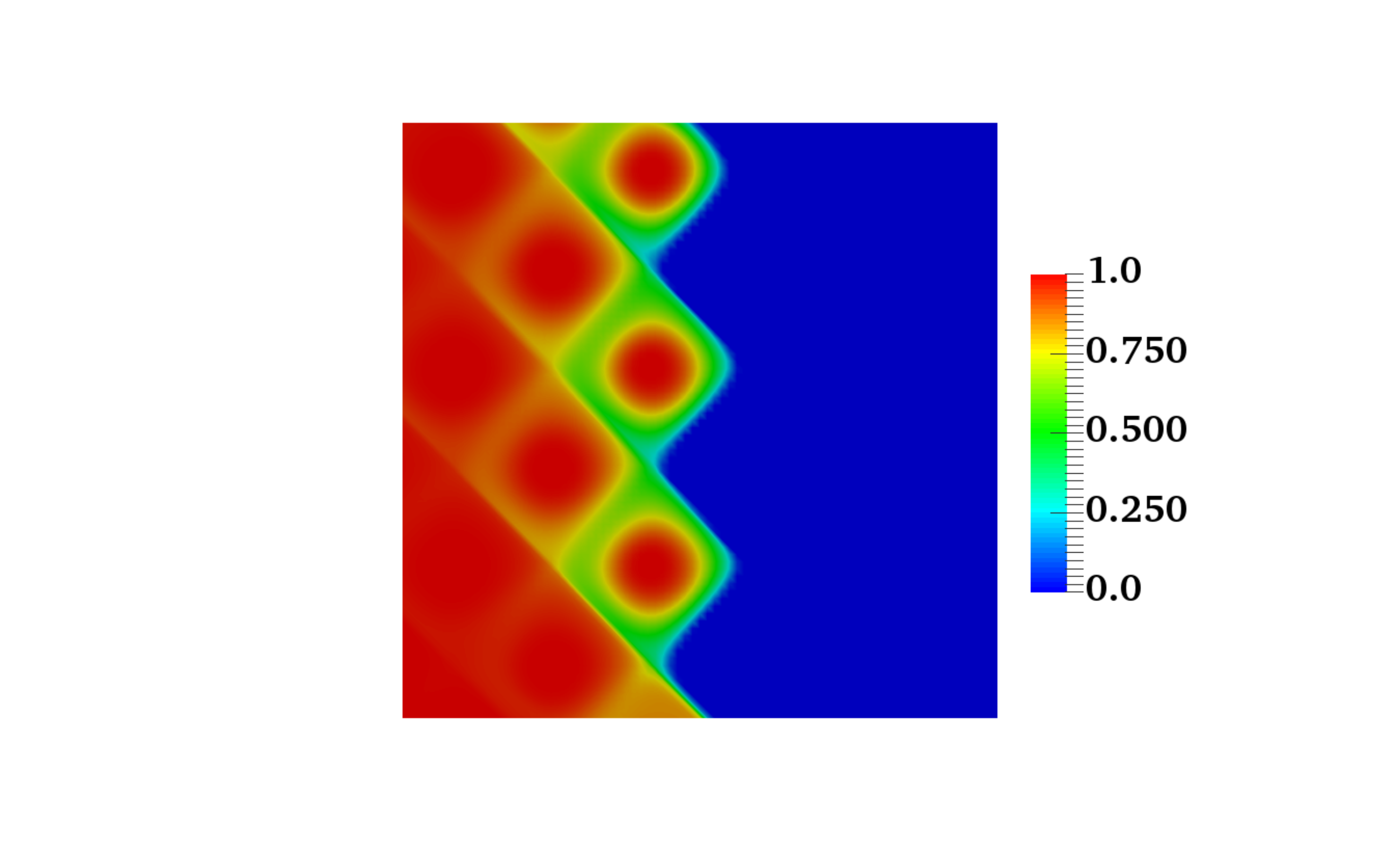}}
  \hspace{-0.5in}
  \subfigure[$\kappa_fL = 3$ and $t = 0.5$]
    {\includegraphics[clip=true,width = 0.37\textwidth]
    {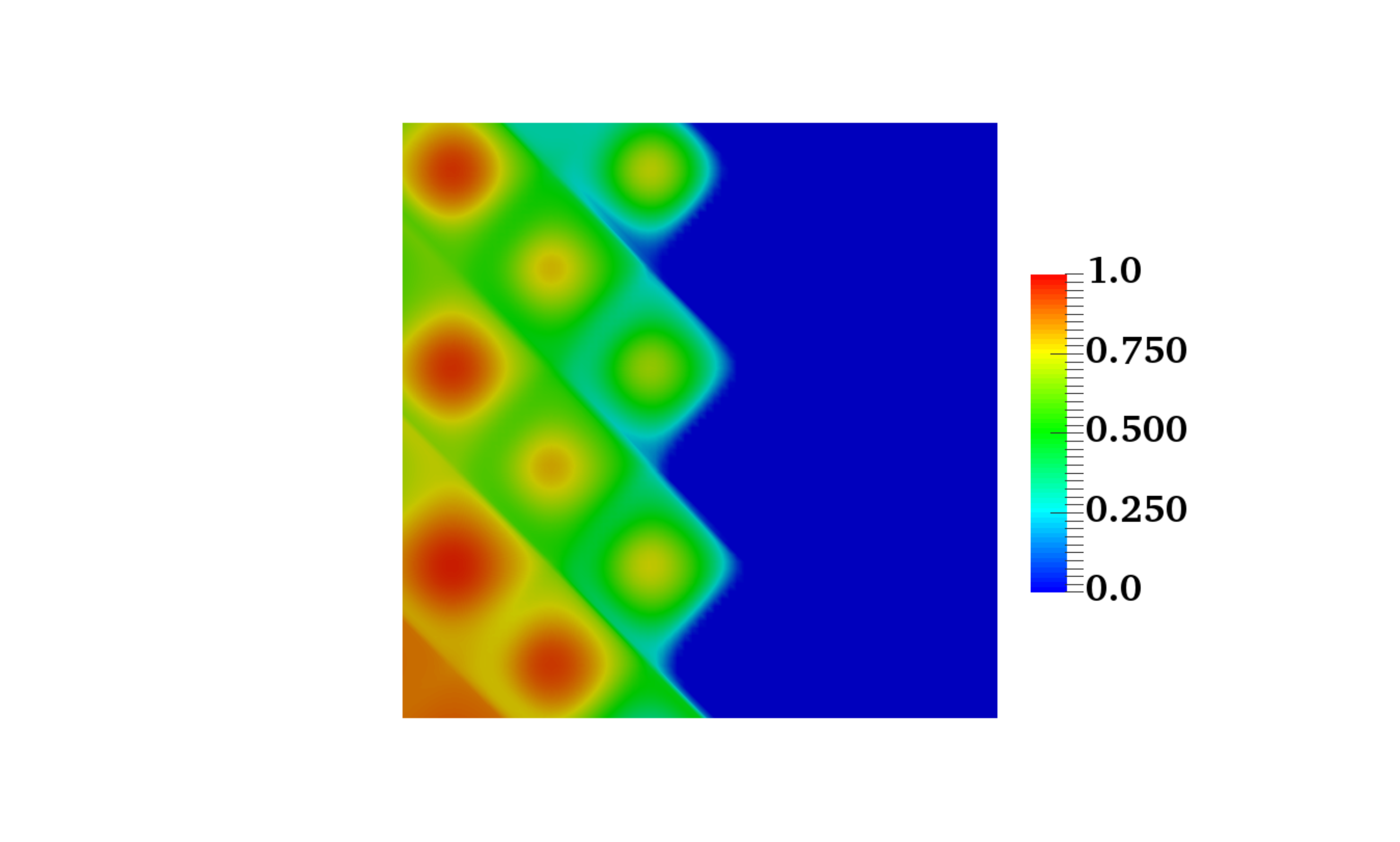}}
  \hspace{-0.5in}
  \subfigure[$\kappa_fL = 3$ and $t = 1.0$]
    {\includegraphics[clip=true,width = 0.37\textwidth]
    {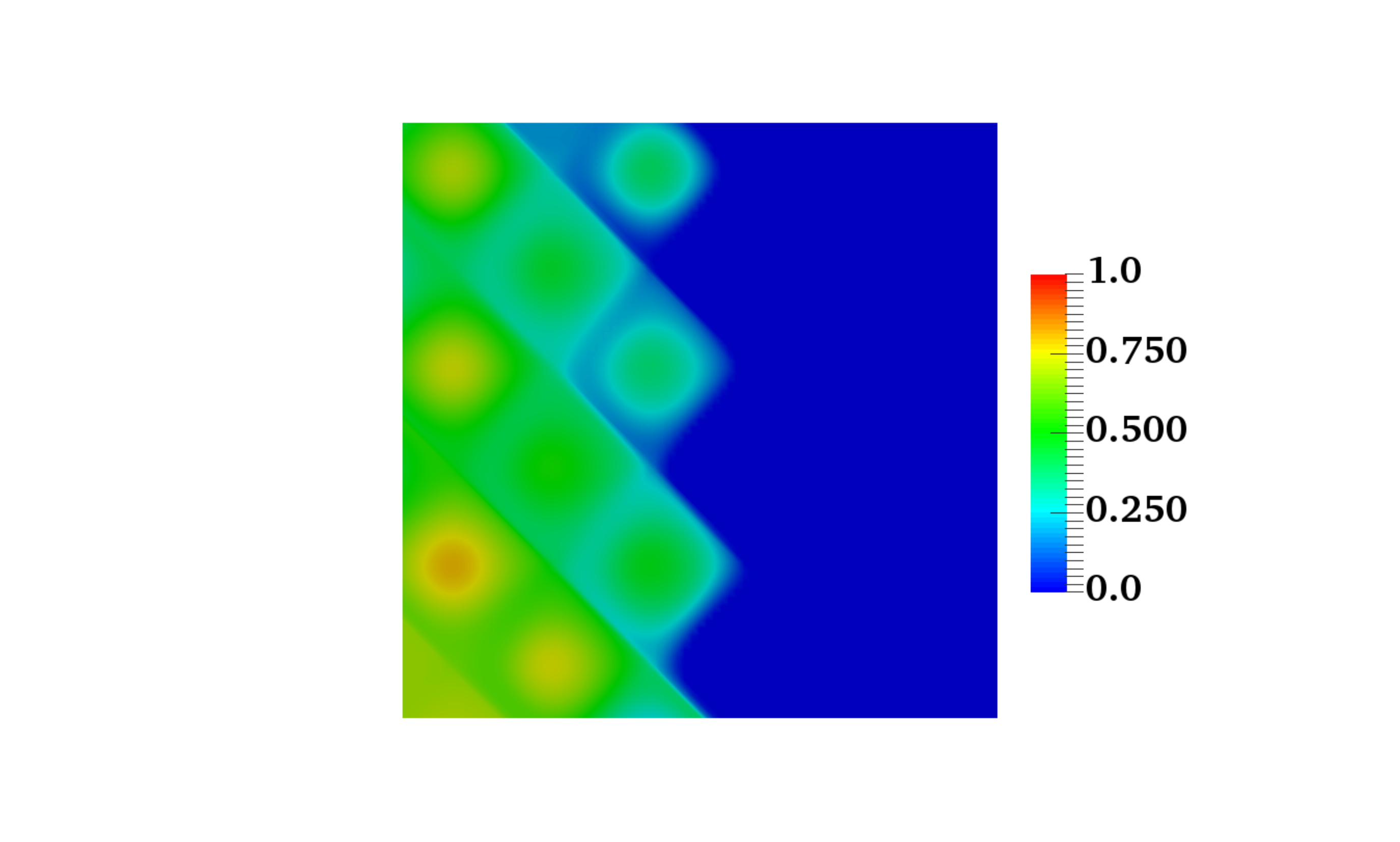}}
  \subfigure[$\kappa_fL = 4$ and $t = 0.1$]
    {\includegraphics[clip=true,width = 0.37\textwidth]
    {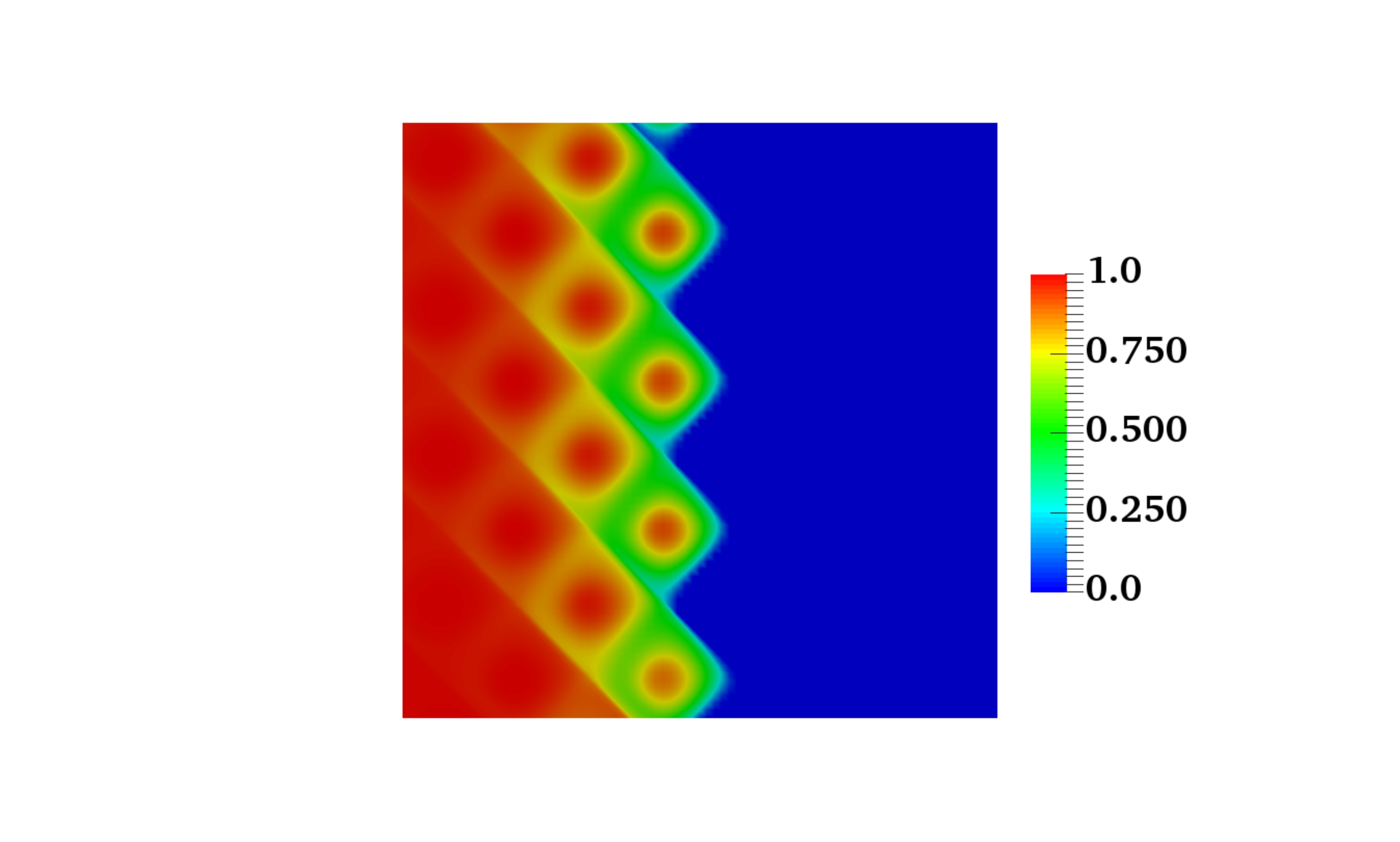}}
  \hspace{-0.5in}
  \subfigure[$\kappa_fL = 4$ and $t = 0.5$]
    {\includegraphics[clip=true,width = 0.37\textwidth]
    {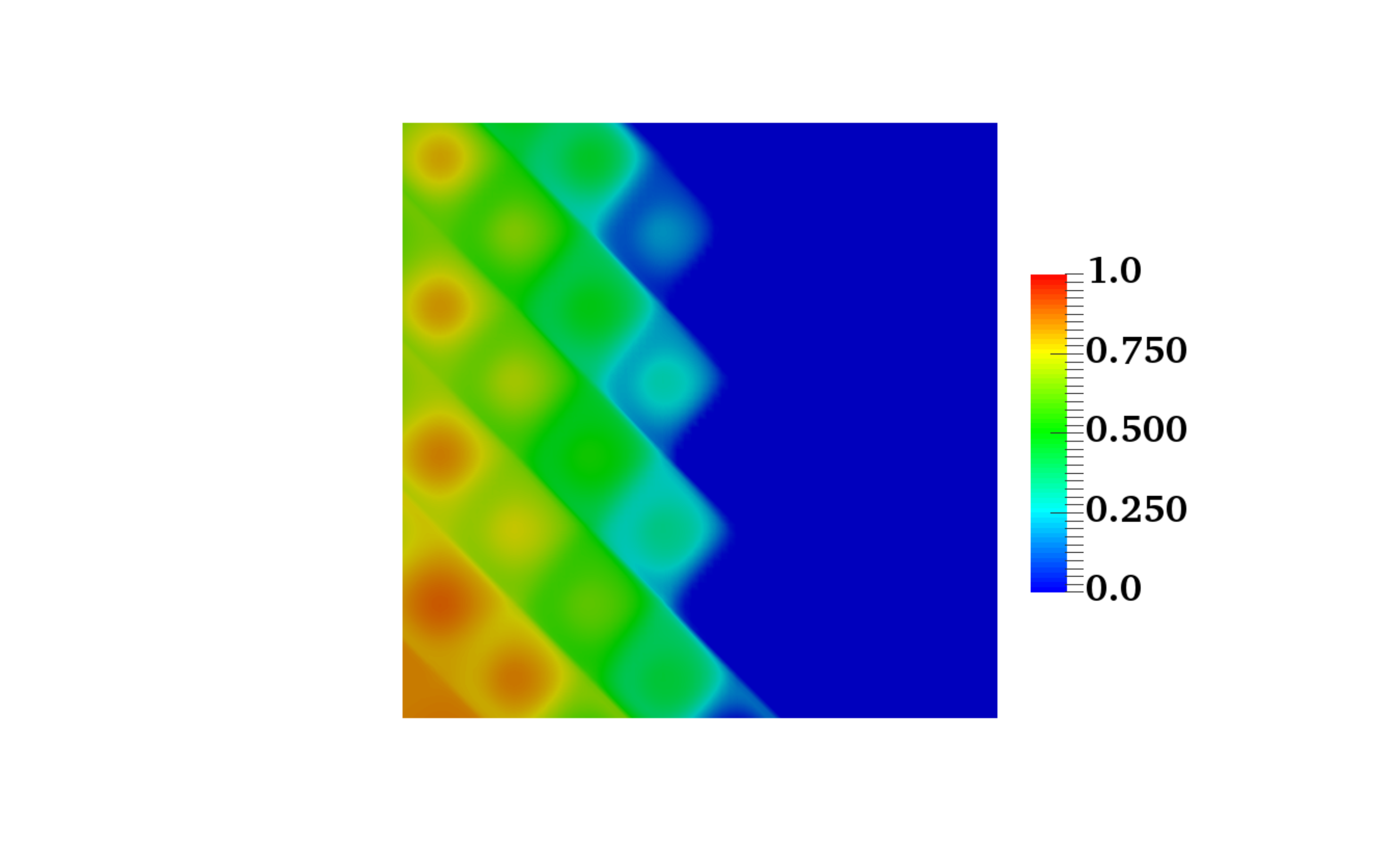}}
  \hspace{-0.5in}
  \subfigure[$\kappa_fL = 4$ and $t = 1.0$]
    {\includegraphics[clip=true,width = 0.37\textwidth]
    {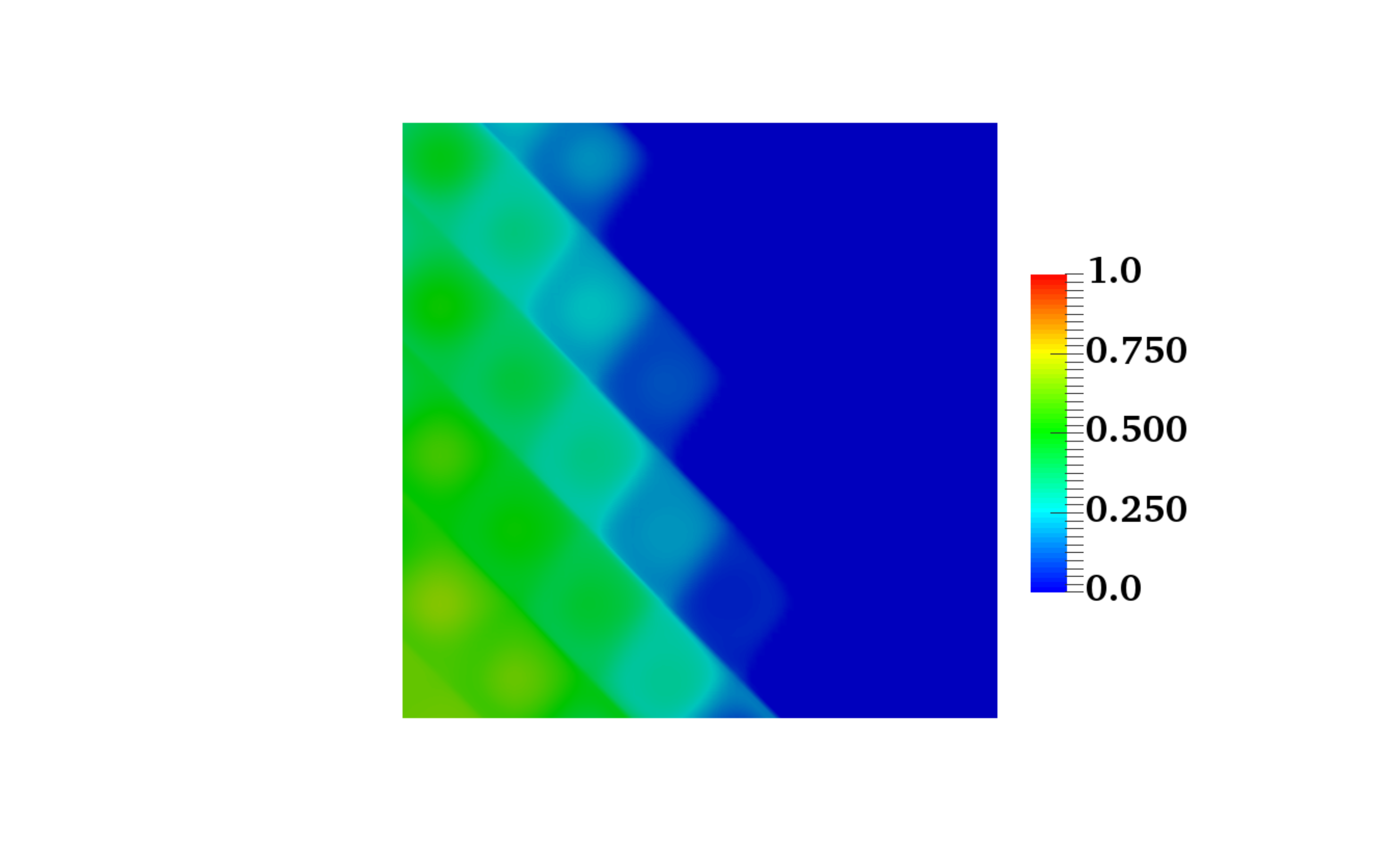}}
  \subfigure[$\kappa_fL = 5$ and $t = 0.1$]
    {\includegraphics[clip=true,width = 0.37\textwidth]
    {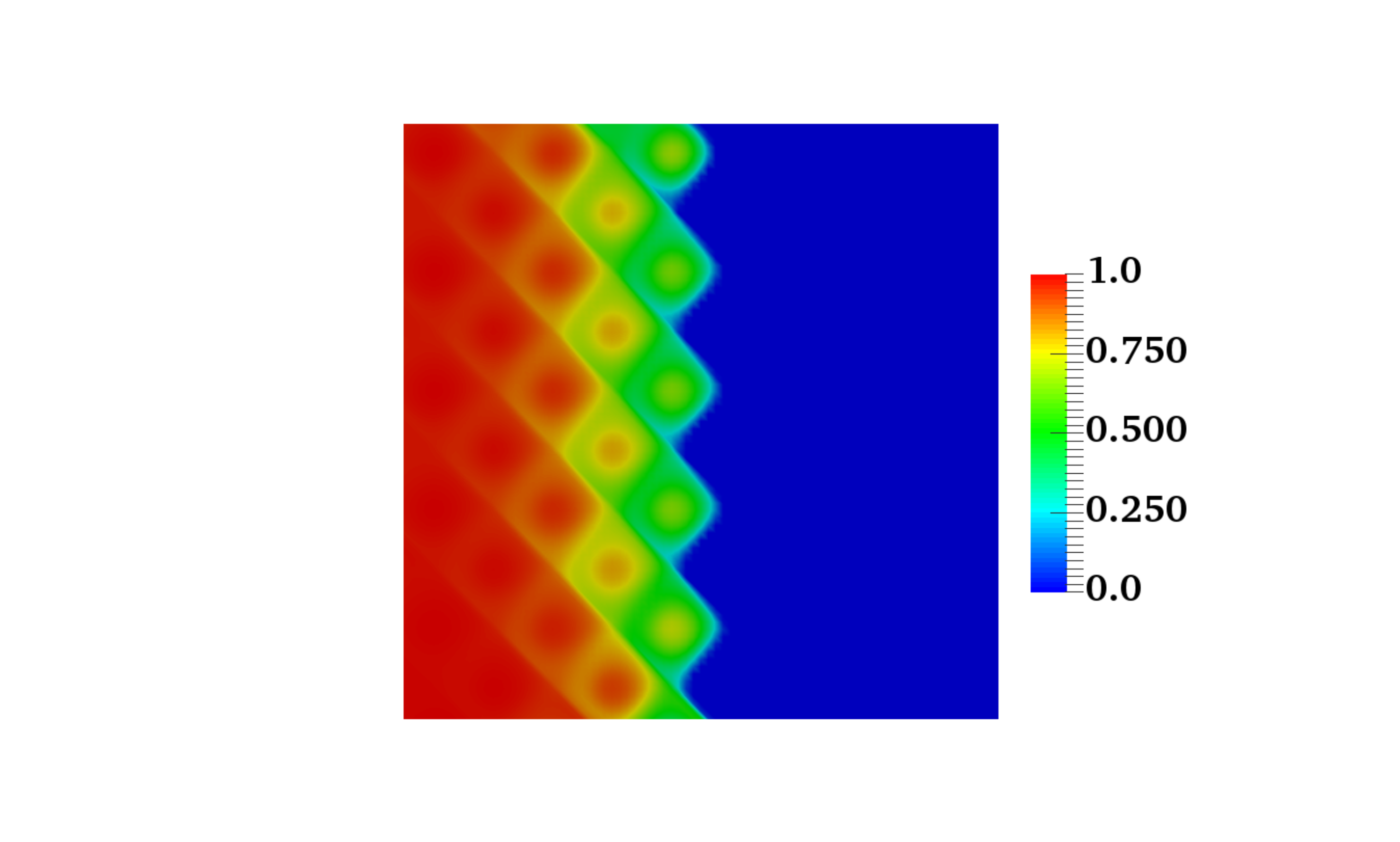}}
  \hspace{-0.5in}
  \subfigure[$\kappa_fL = 5$ and $t = 0.5$]
    {\includegraphics[clip=true,width = 0.37\textwidth]
    {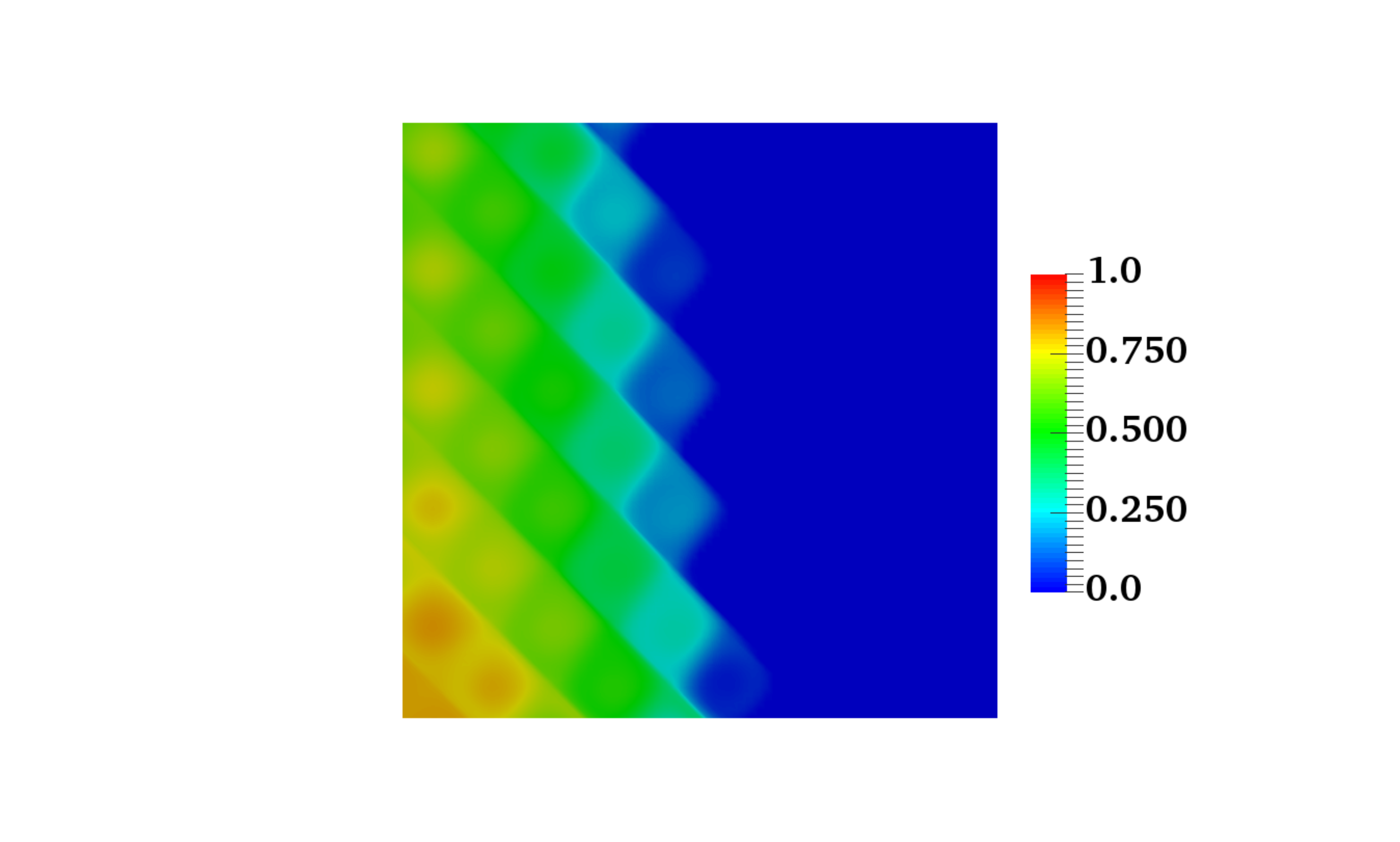}}
  \hspace{-0.5in}
  \subfigure[$\kappa_fL = 5$ and $t = 1.0$]
    {\includegraphics[clip=true,width = 0.37\textwidth]
    {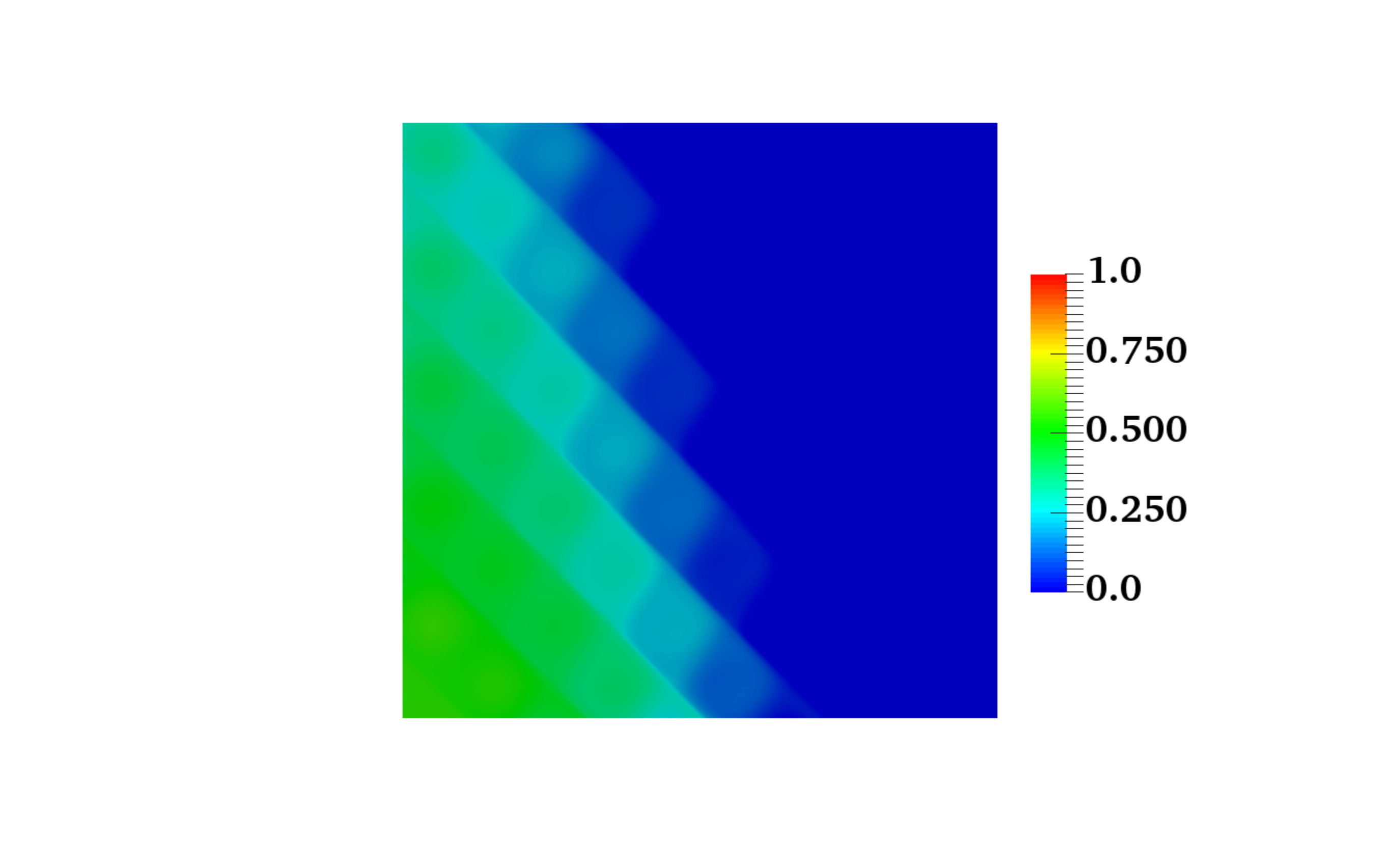}}
  \caption{\textsf{\textbf{Concentration contours of species $A$:}}~These 
    figures show the concentration of species $A$ at times $t = 0.1, \, 0.5,$ 
    and $1.0$. Other input parameters are $\frac{\alpha_L}{\alpha_T} = 10^{4}$, 
    $v_o = 10^{-1}$, $T = 0.1$, and $D_m = 10^{-3}$. From the above figures, 
    it is clear that species $A$ is not consumed in its entirety for $t \in [0,1]$. 
    Moreover, at lower values of $\kappa_fL$ considerable amount of species $A$ 
    remains in the left half of the domain as compared to higher values of 
    $\kappa_fL$.
  \label{Fig:Contours_A_Difftimes}}
\end{figure}

\begin{figure}
  \centering
  \subfigure[$\kappa_fL = 2$ and $t = 0.1$]
    {\includegraphics[clip=true,width = 0.37\textwidth]
    {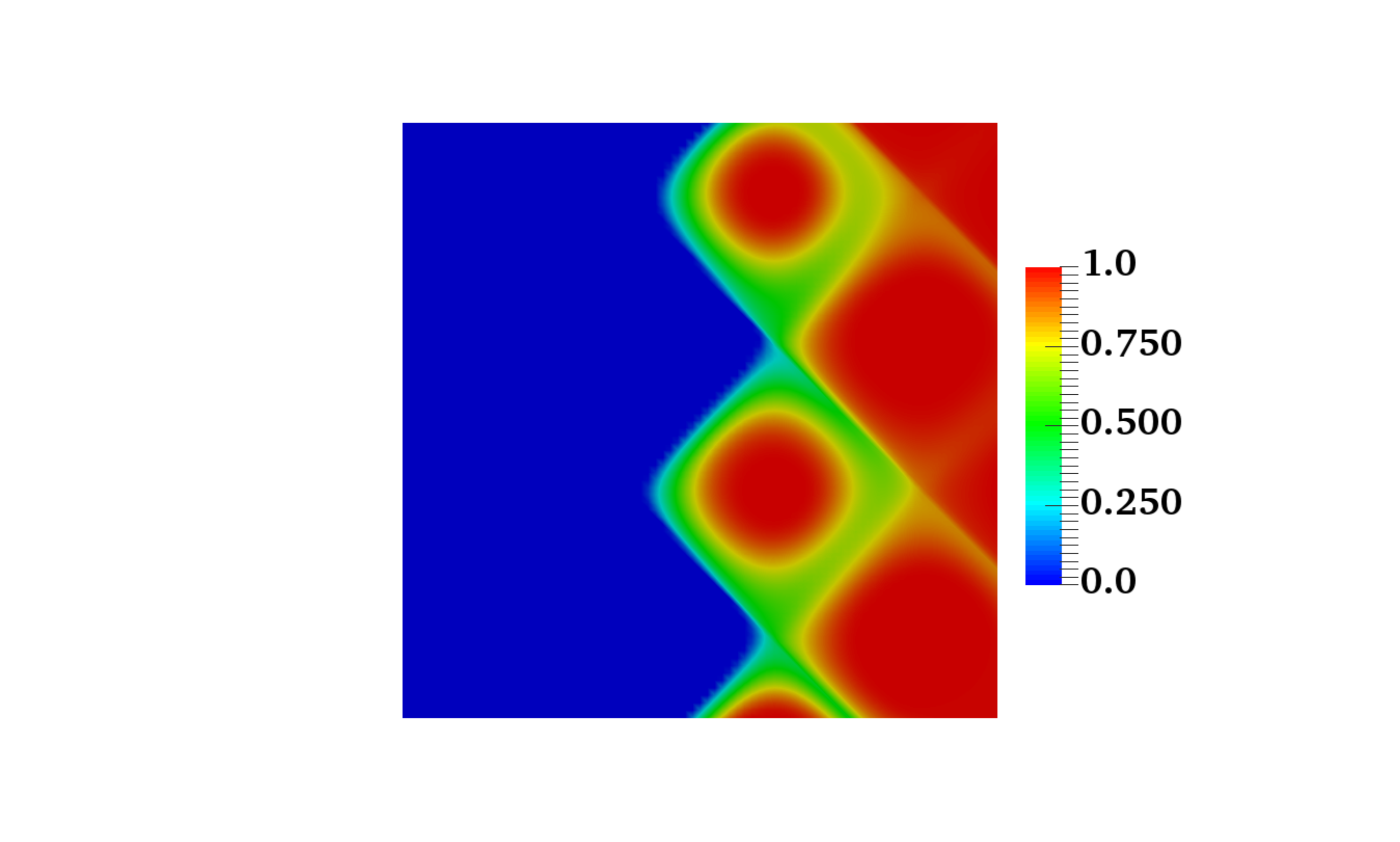}}
  \hspace{-0.5in}
  \subfigure[$\kappa_fL = 2$ and $t = 0.5$]
    {\includegraphics[clip=true,width = 0.37\textwidth]
    {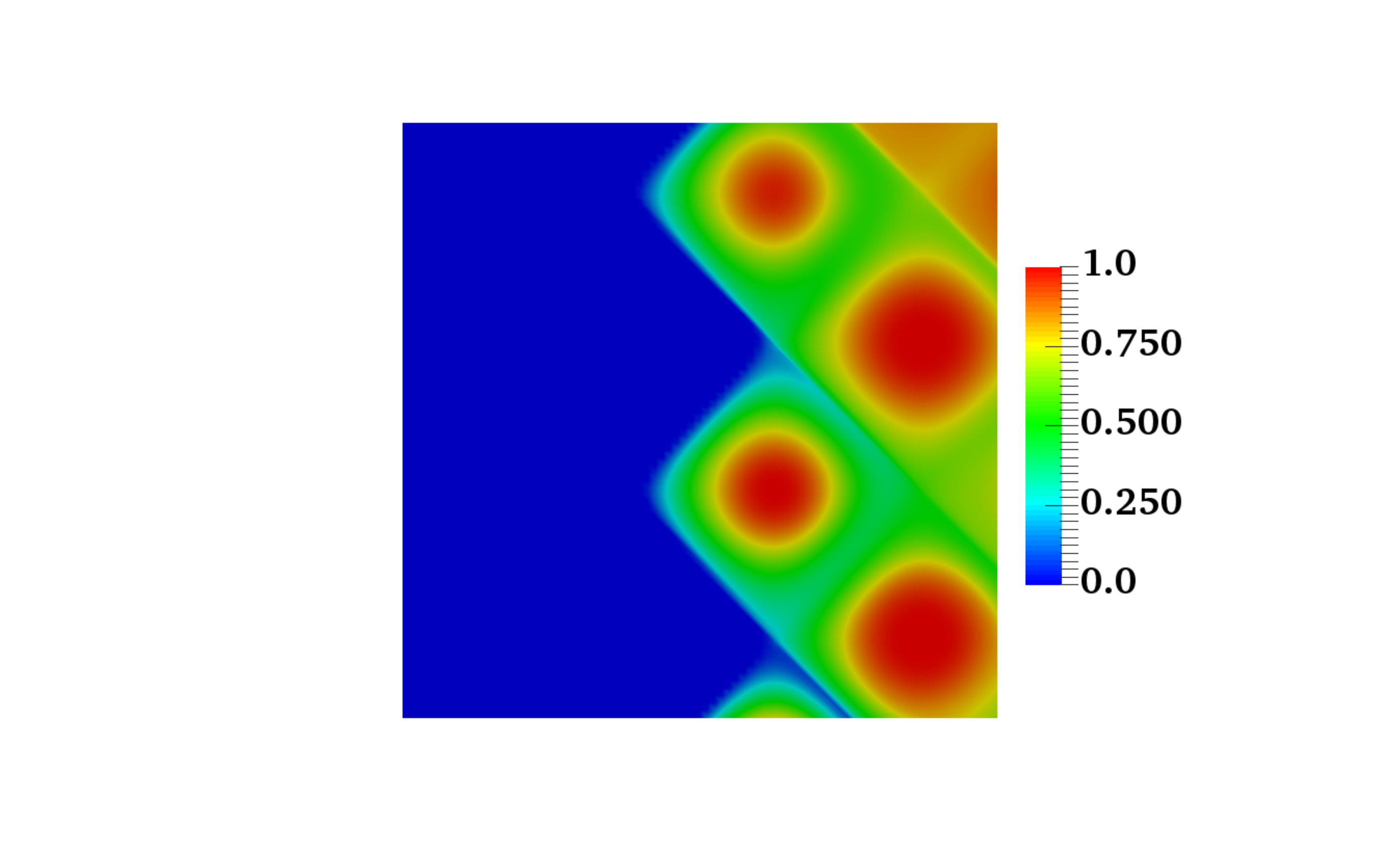}}
  \hspace{-0.5in}
  \subfigure[$\kappa_fL = 2$ and $t = 1.0$]
    {\includegraphics[clip=true,width = 0.37\textwidth]
    {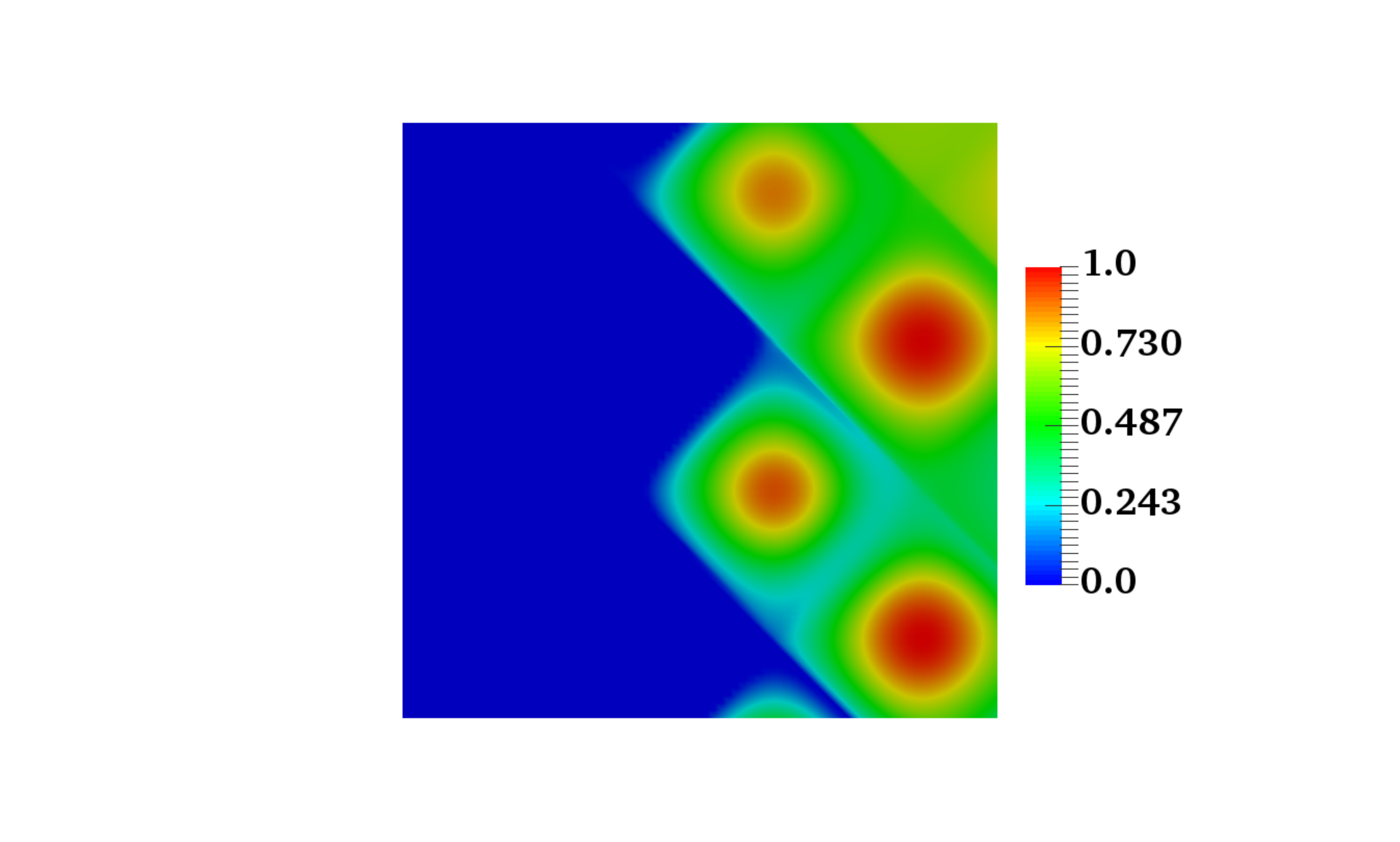}}
  \subfigure[$\kappa_fL = 3$ and $t = 0.1$]
    {\includegraphics[clip=true,width = 0.37\textwidth]
    {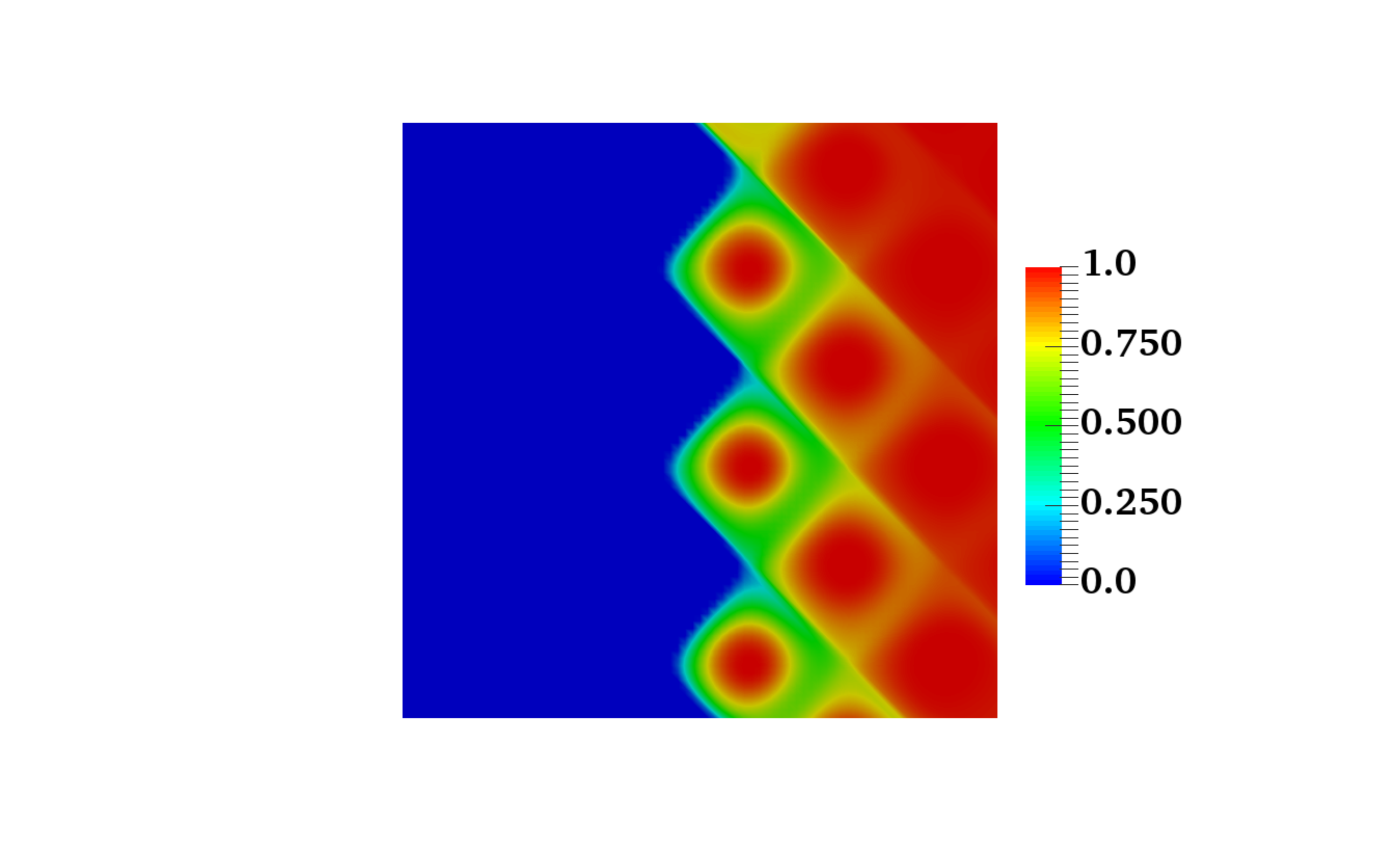}}
  \hspace{-0.5in}
  \subfigure[$\kappa_fL = 3$ and $t = 0.5$]
    {\includegraphics[clip=true,width = 0.37\textwidth]
    {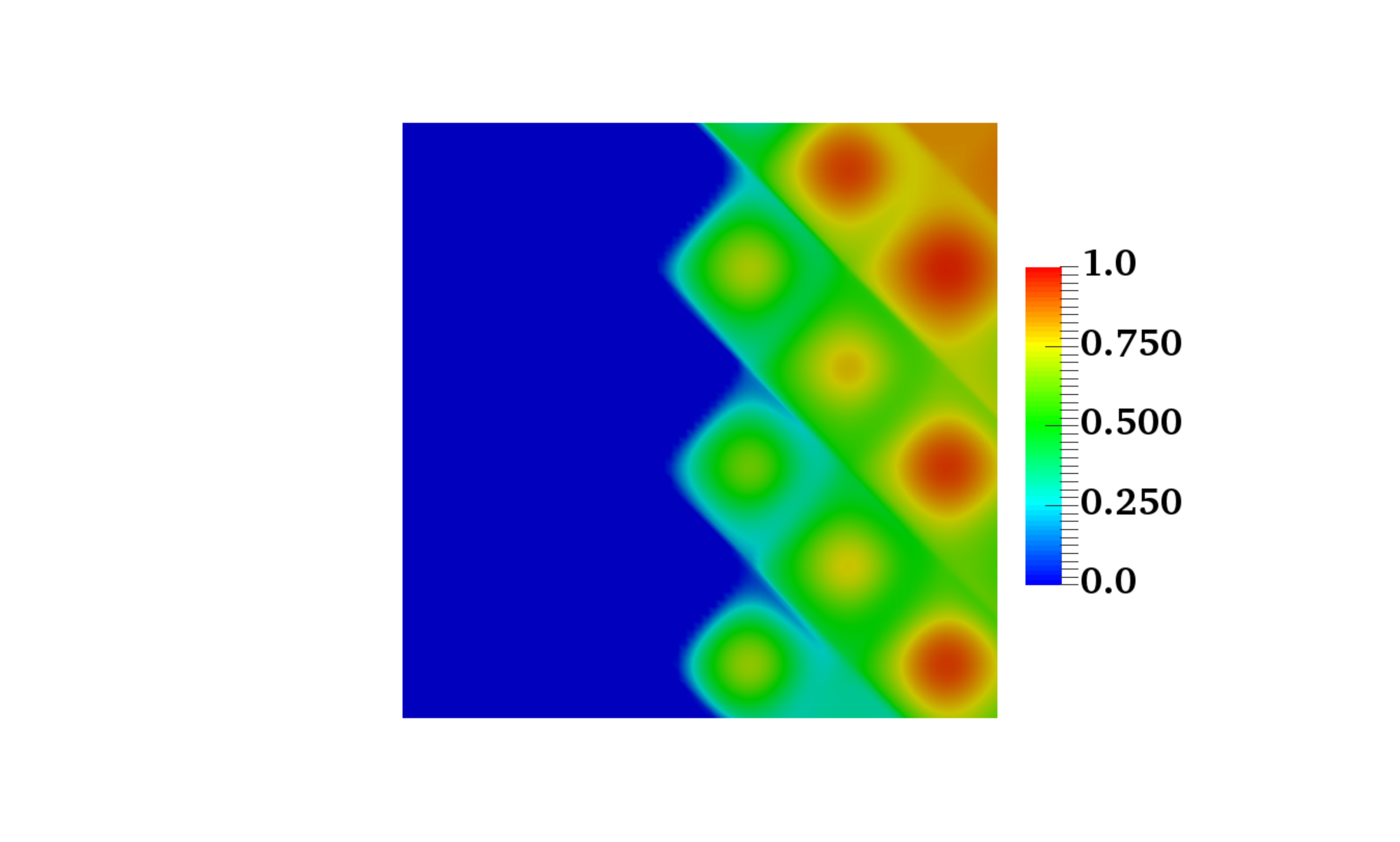}}
  \hspace{-0.5in}
  \subfigure[$\kappa_fL = 3$ and $t = 1.0$]
    {\includegraphics[clip=true,width = 0.37\textwidth]
    {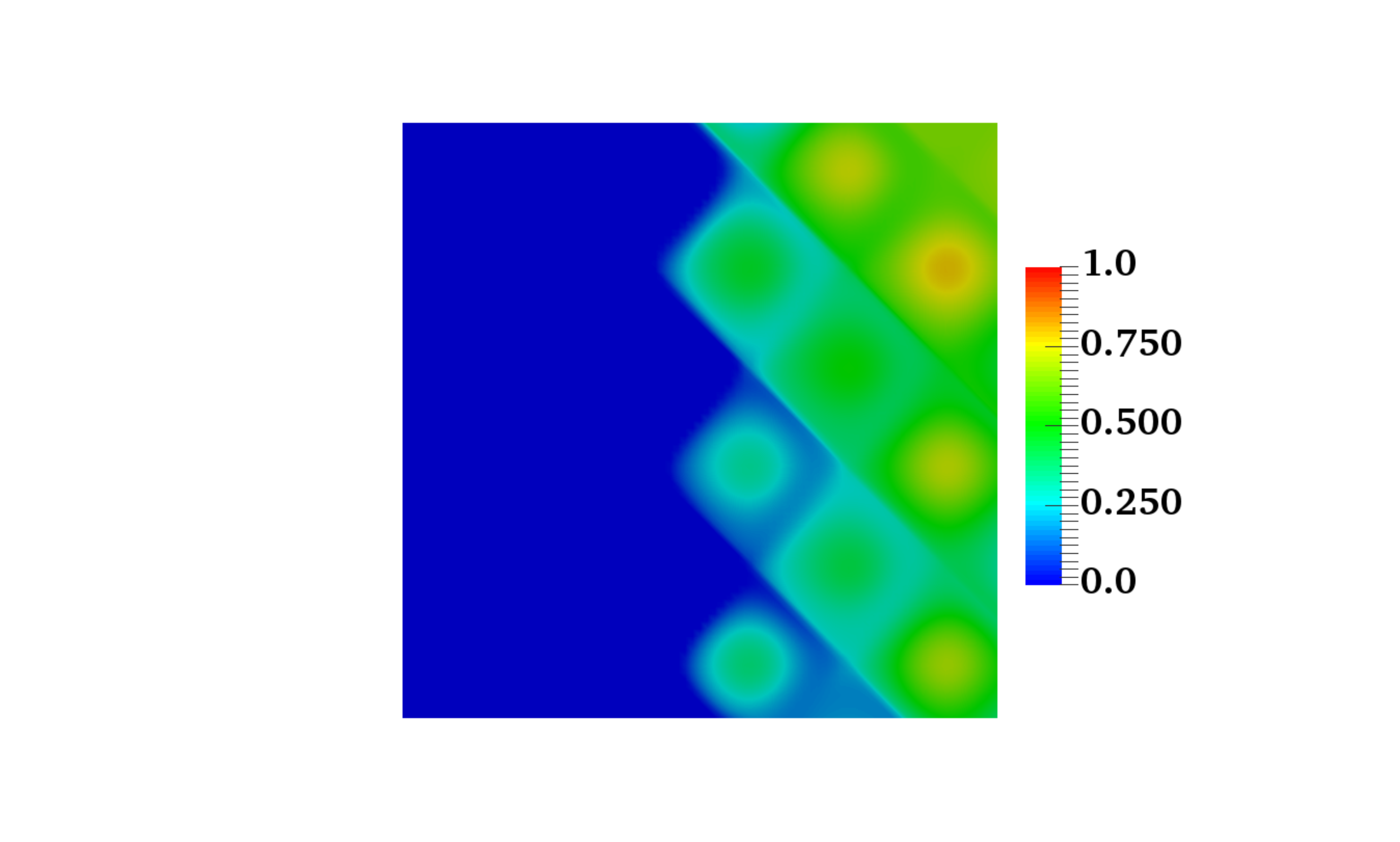}}
  \subfigure[$\kappa_fL = 4$ and $t = 0.1$]
    {\includegraphics[clip=true,width = 0.37\textwidth]
    {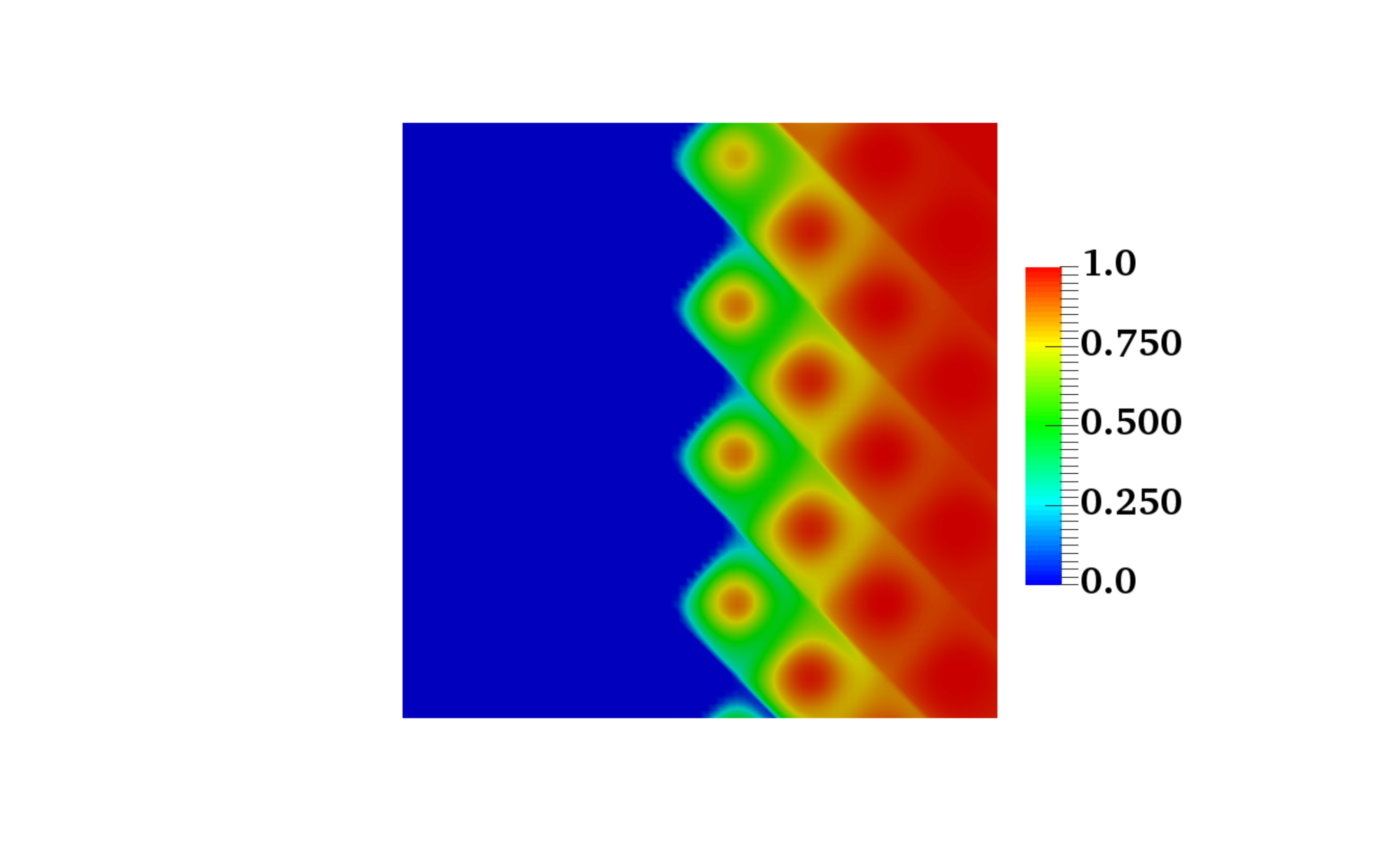}}
  \hspace{-0.5in}
  \subfigure[$\kappa_fL = 4$ and $t = 0.5$]
    {\includegraphics[clip=true,width = 0.37\textwidth]
    {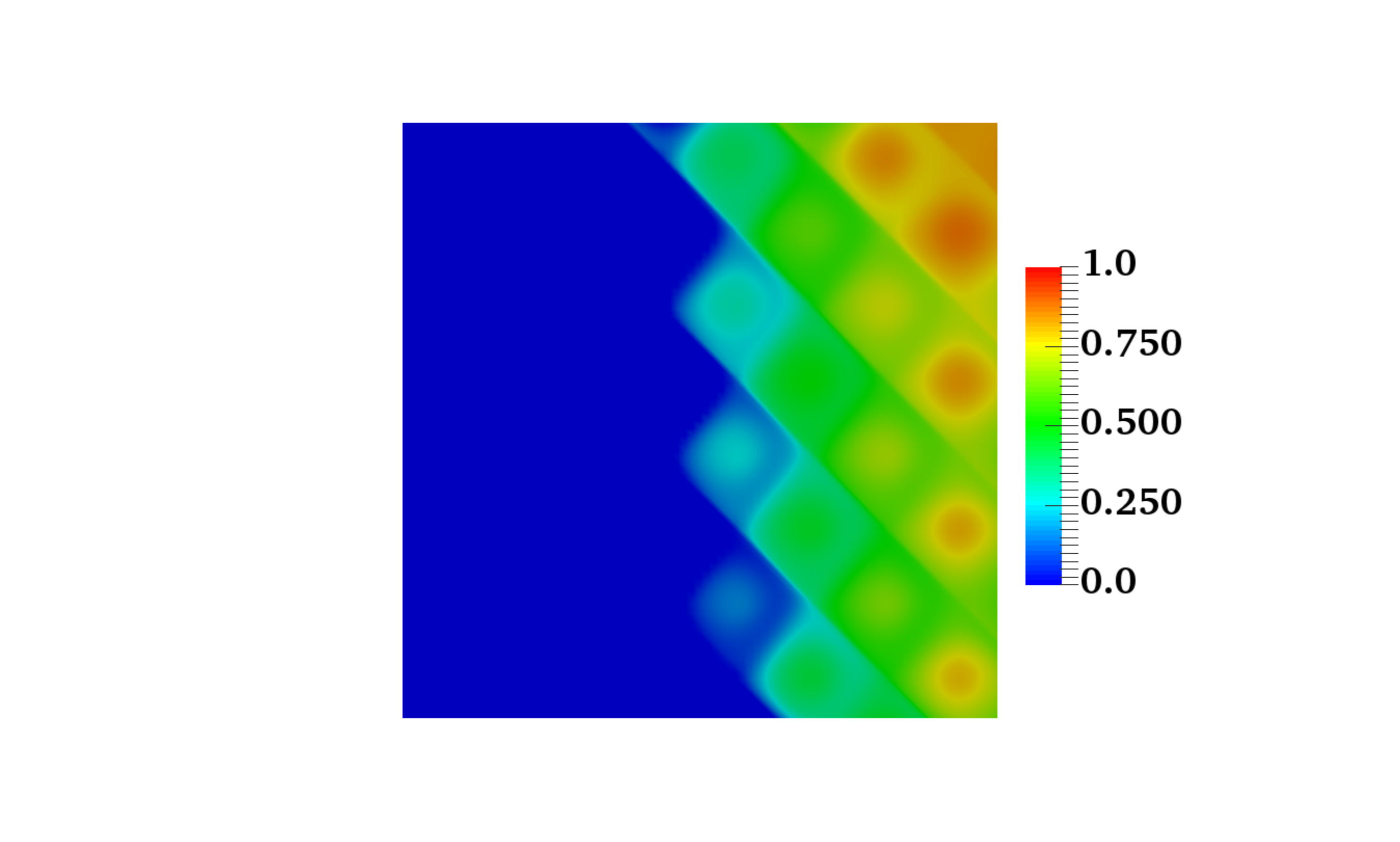}}
  \hspace{-0.5in}
  \subfigure[$\kappa_fL = 4$ and $t = 1.0$]
    {\includegraphics[clip=true,width = 0.37\textwidth]
    {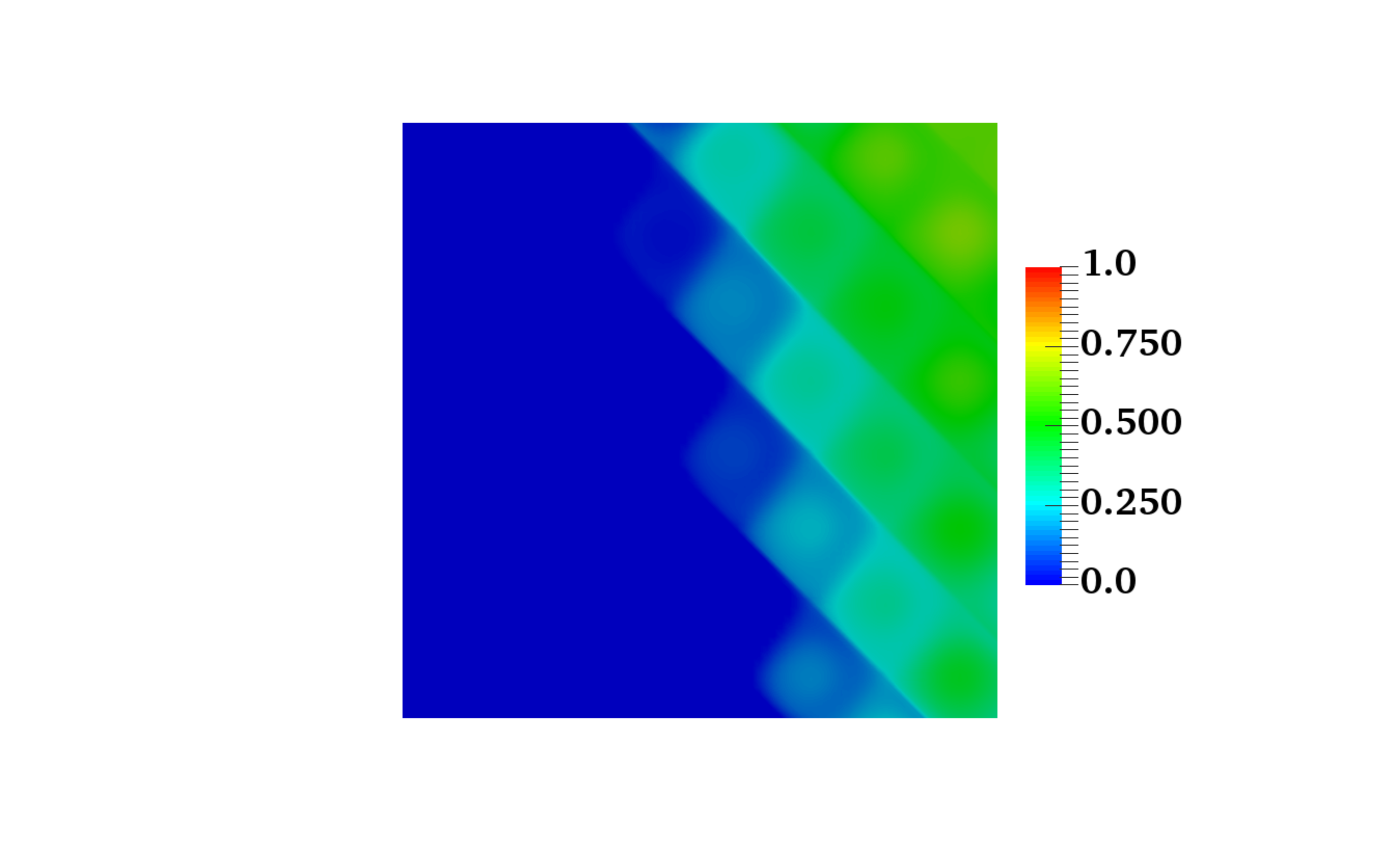}}
  \subfigure[$\kappa_fL = 5$ and $t = 0.1$]
    {\includegraphics[clip=true,width = 0.37\textwidth]
    {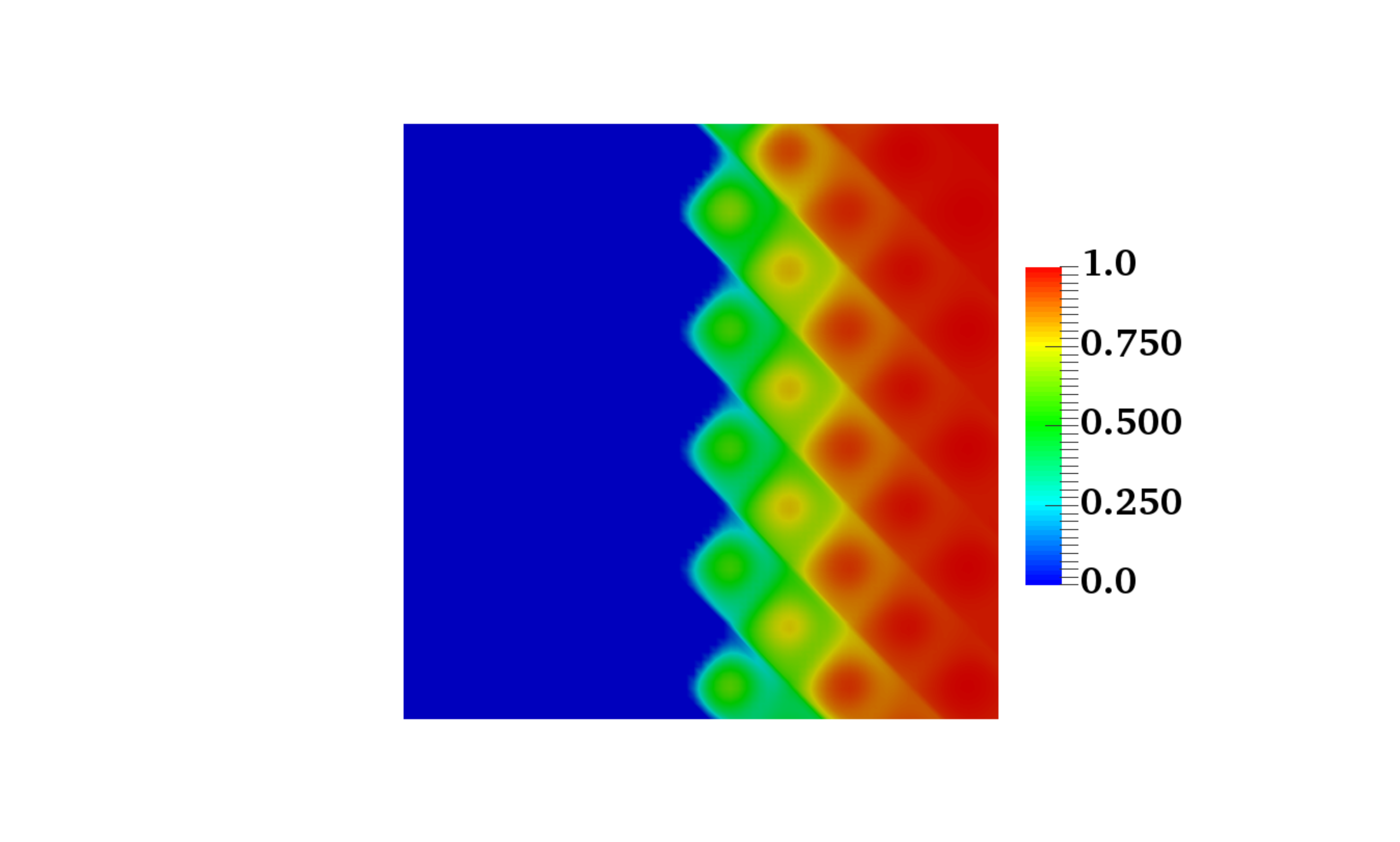}}
  \hspace{-0.5in}
  \subfigure[$\kappa_fL = 5$ and $t = 0.5$]
    {\includegraphics[clip=true,width = 0.37\textwidth]
    {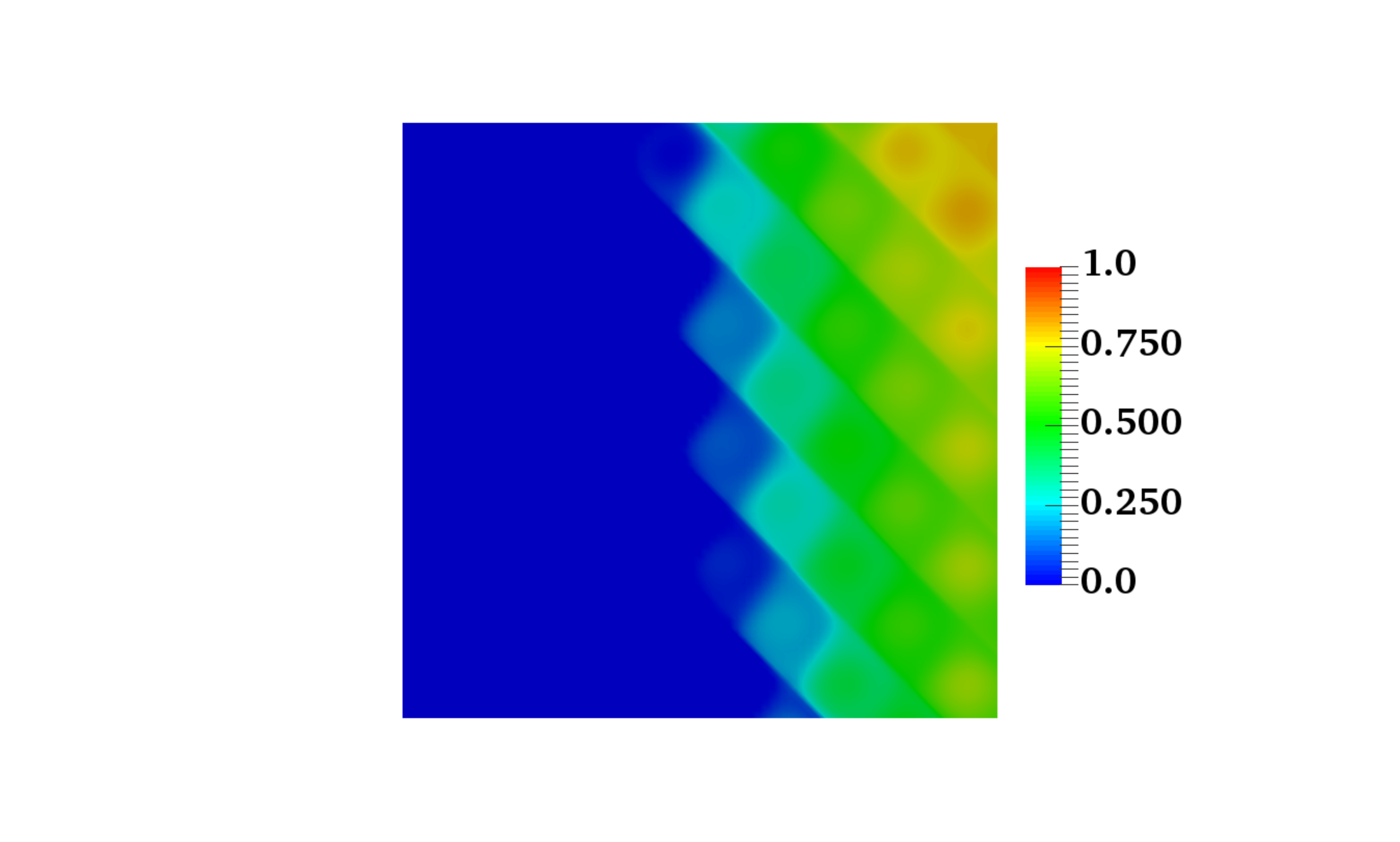}}
  \hspace{-0.5in}
  \subfigure[$\kappa_fL = 5$ and $t = 1.0$]
    {\includegraphics[clip=true,width = 0.37\textwidth]
    {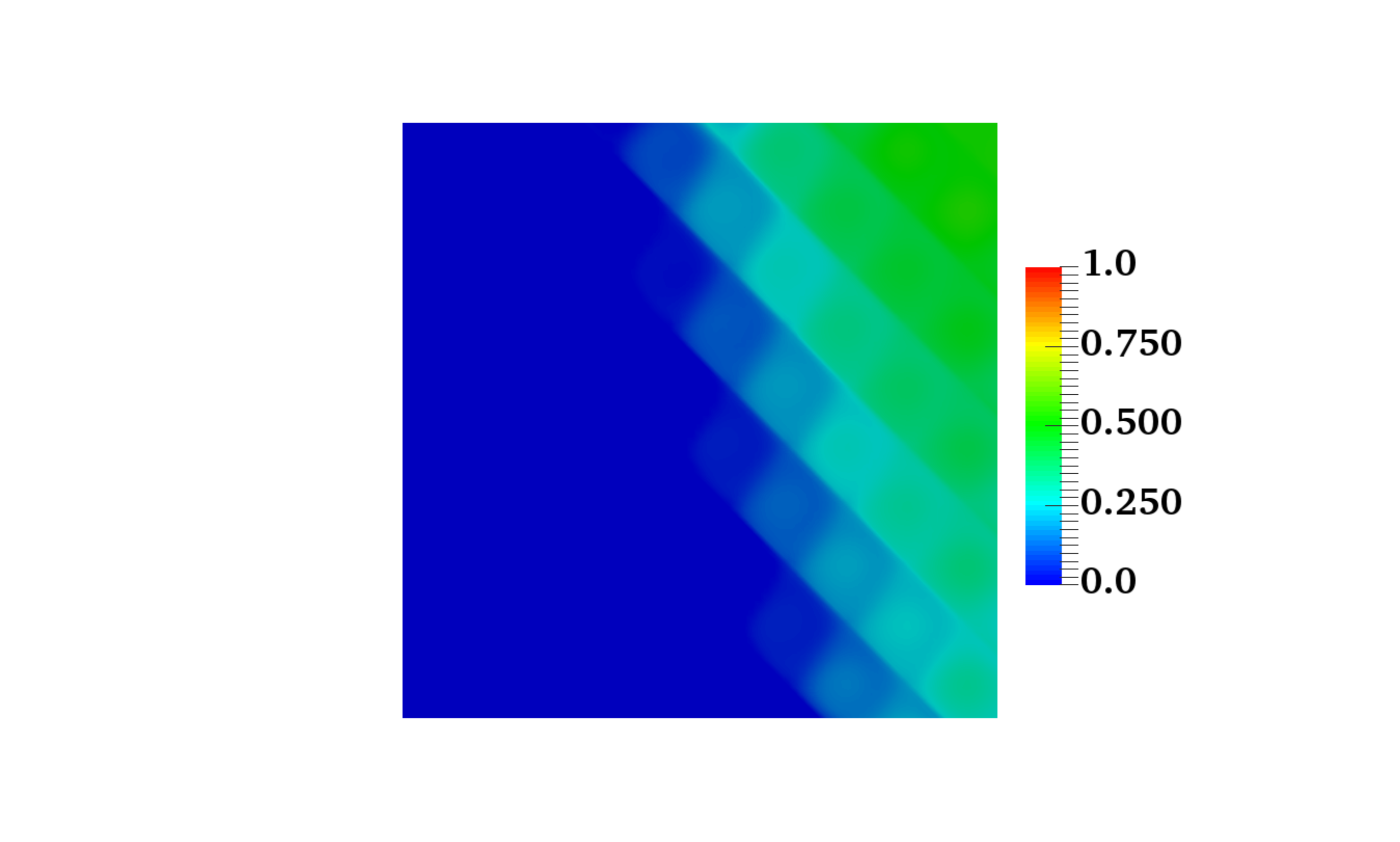}}
  \caption{\textsf{\textbf{Concentration contours of species $B$:}}~These 
    figures show the concentration of species $B$ at times $t = 0.1, \, 0.5,$ 
    and $1.0$. Other input parameters are $\frac{\alpha_L}{\alpha_T} = 10^{4}$, 
    $v_o = 10^{-1}$, $T = 0.1$, and $D_m = 10^{-3}$. Similar to species $A$,
    species $B$ is not consumed in its entirety for $t \in [0,1]$. For higher 
    $\kappa_fL$ values, we see more of species $B$ being consumed in right half 
    of the domain. From this figure, it is clear that for enhanced mixing, we need 
    small-scale structures in the velocity field (meaning higher values of $\kappa_fL$).
  \label{Fig:Contours_B_Difftimes}}
\end{figure}

\begin{figure}
  \centering
  \subfigure[$\kappa_fL = 2$ and $t = 0.1$]
    {\includegraphics[clip=true,width = 0.37\textwidth]
    {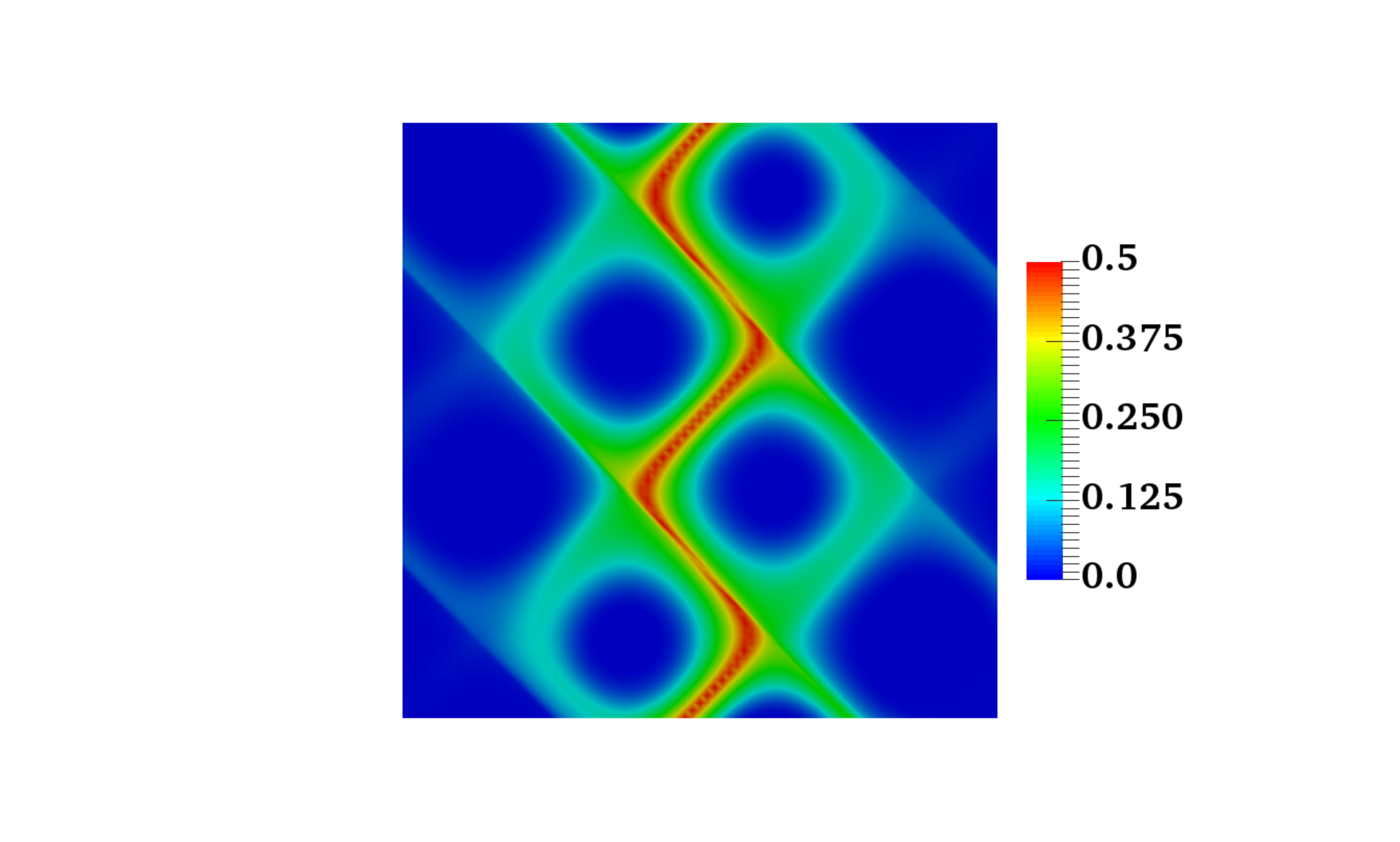}}
  \hspace{-0.5in}
  \subfigure[$\kappa_fL = 2$ and $t = 0.5$]
    {\includegraphics[clip=true,width = 0.37\textwidth]
    {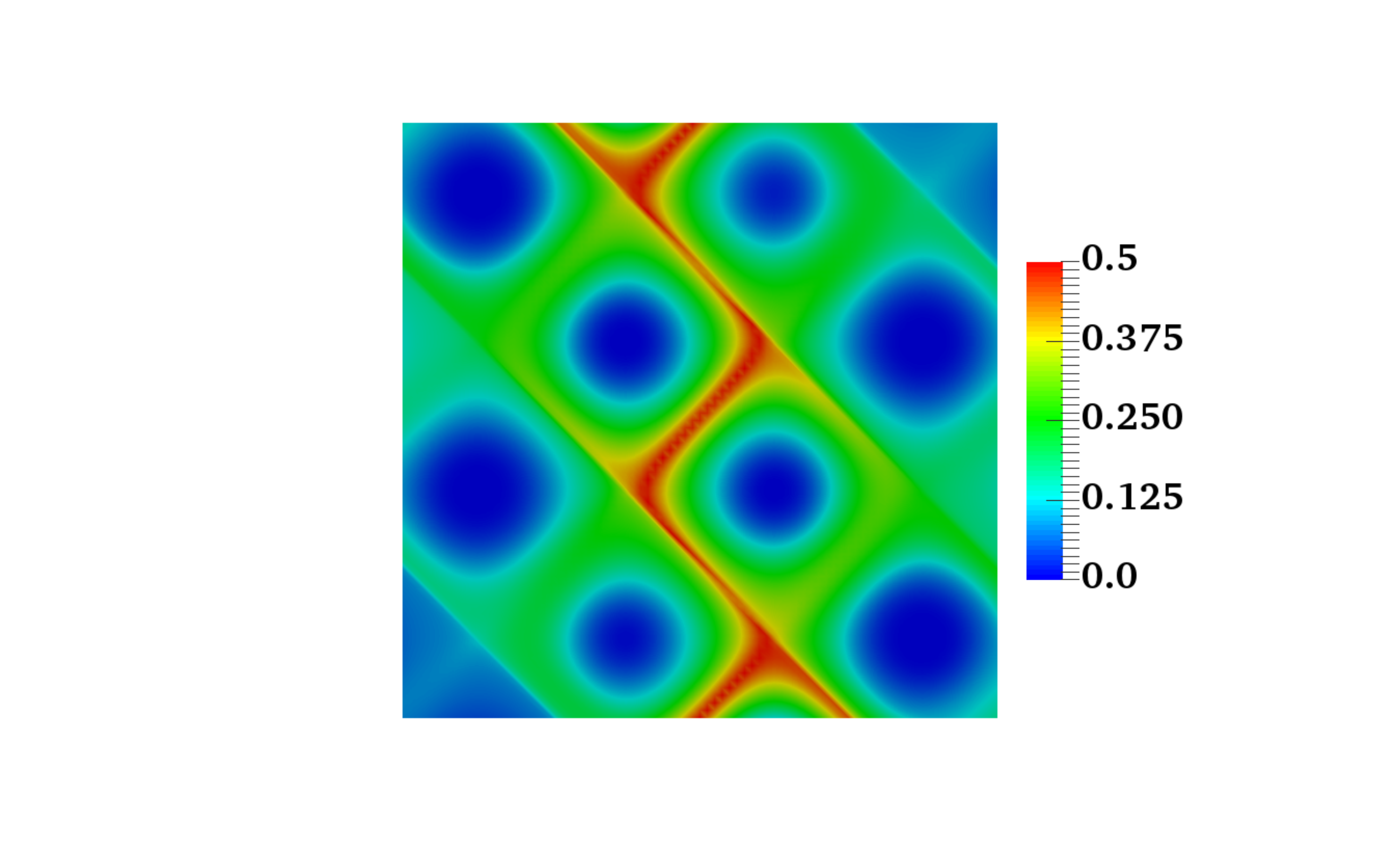}}
  \hspace{-0.5in}
  \subfigure[$\kappa_fL = 2$ and $t = 1.0$]
    {\includegraphics[clip=true,width = 0.37\textwidth]
    {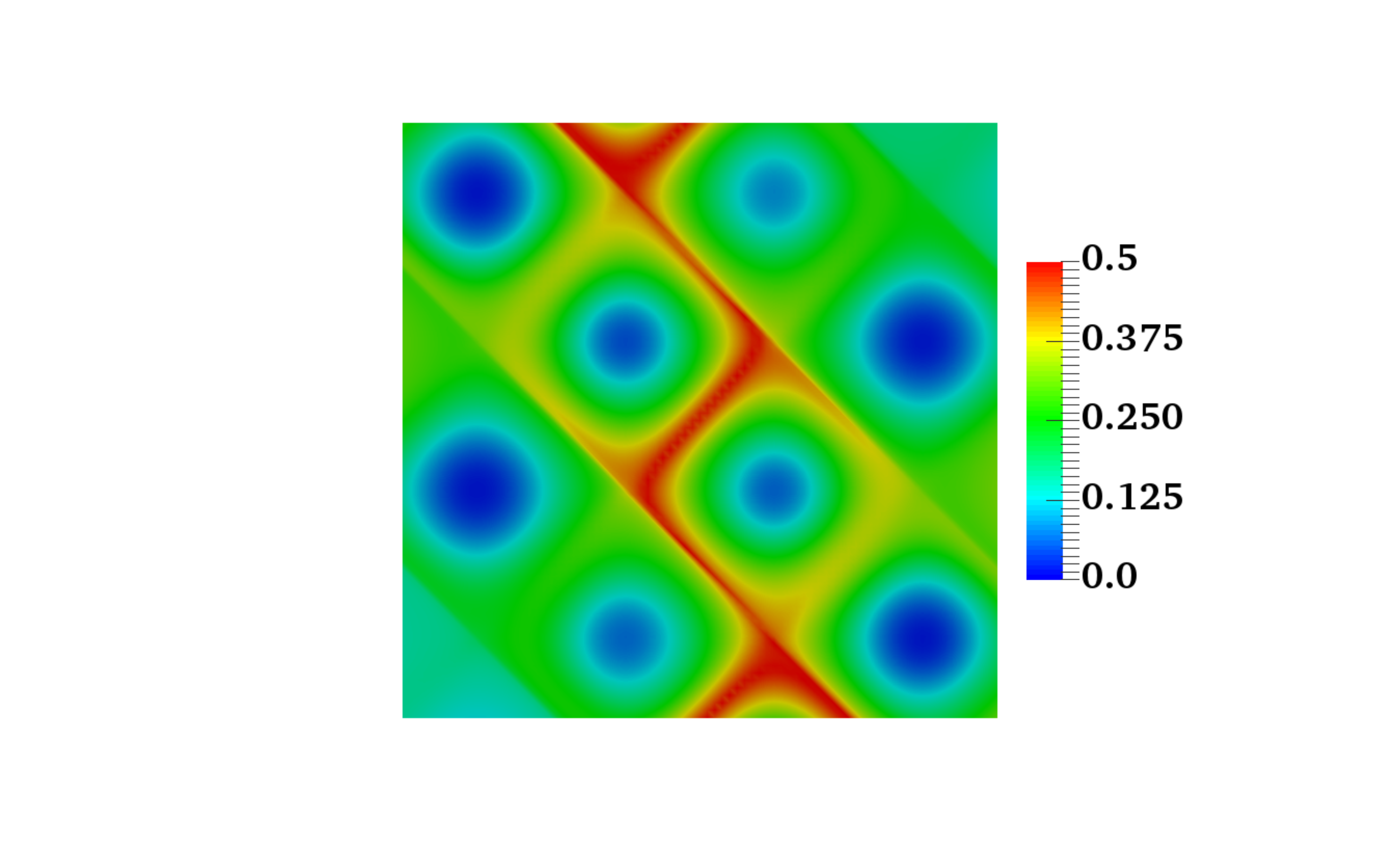}}
  \subfigure[$\kappa_fL = 3$ and $t = 0.1$]
    {\includegraphics[clip=true,width = 0.37\textwidth]
    {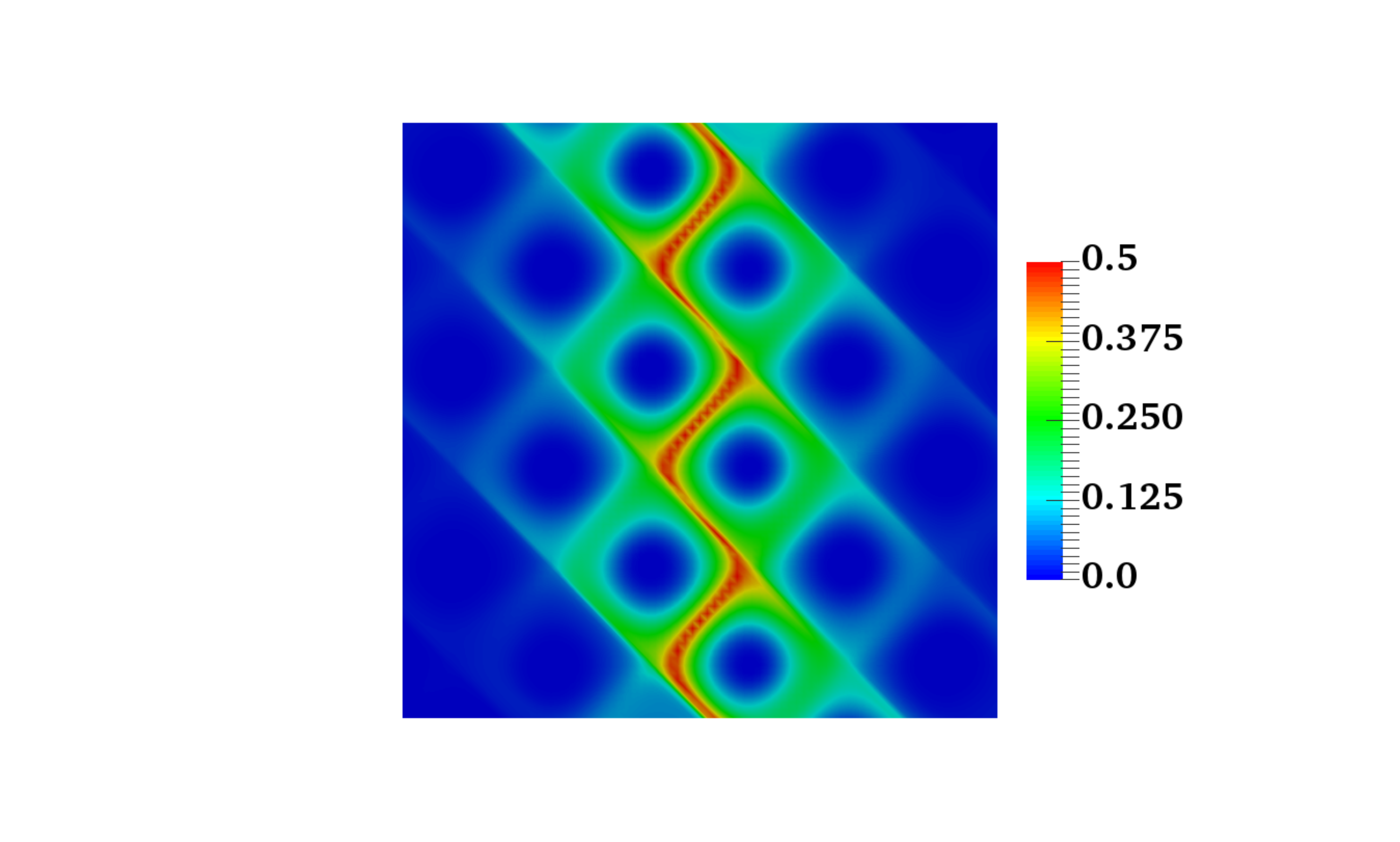}}
  \hspace{-0.5in}
  \subfigure[$\kappa_fL = 3$ and $t = 0.5$]
    {\includegraphics[clip=true,width = 0.37\textwidth]
    {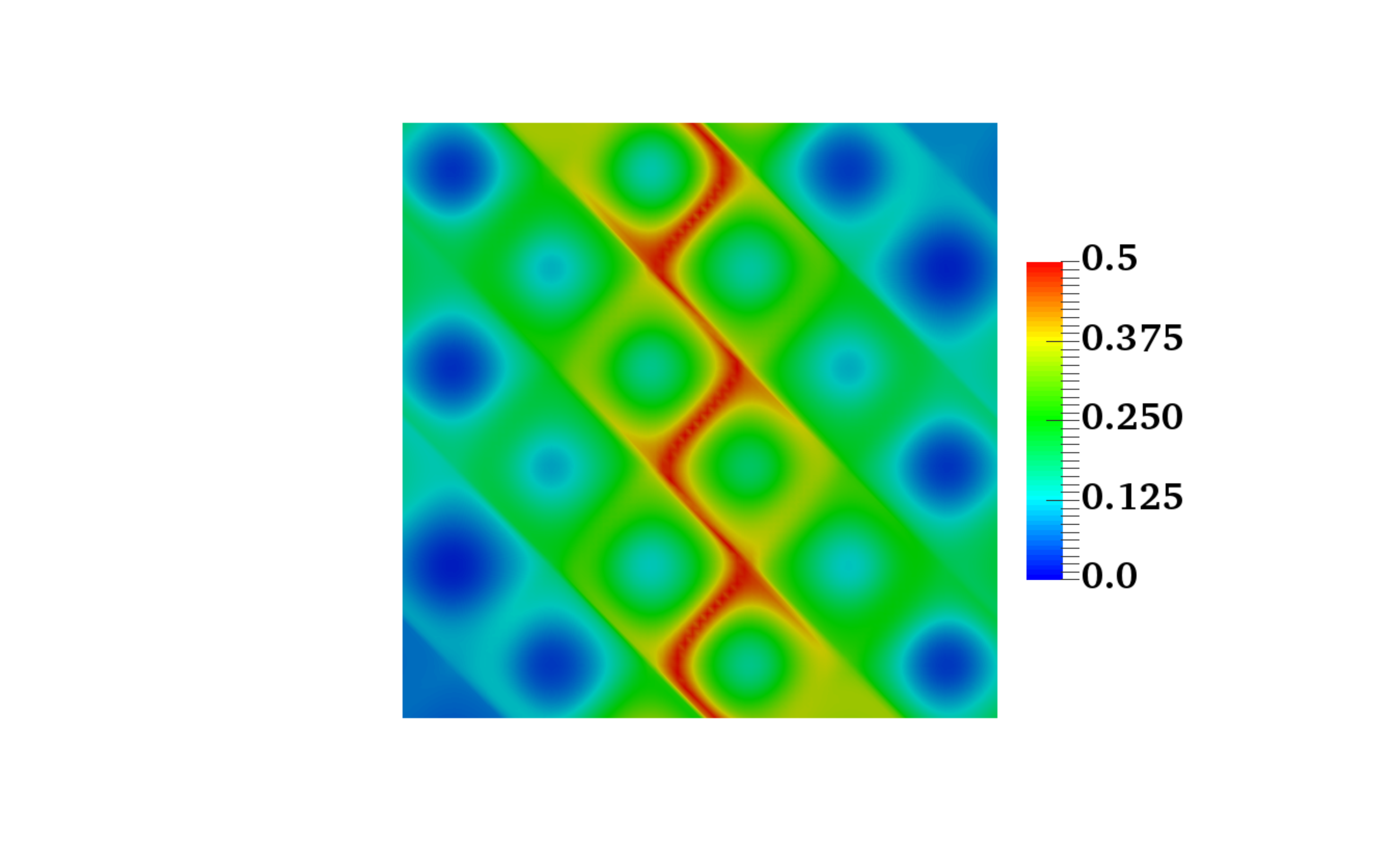}}
  \hspace{-0.5in}
  \subfigure[$\kappa_fL = 3$ and $t = 1.0$]
    {\includegraphics[clip=true,width = 0.37\textwidth]
    {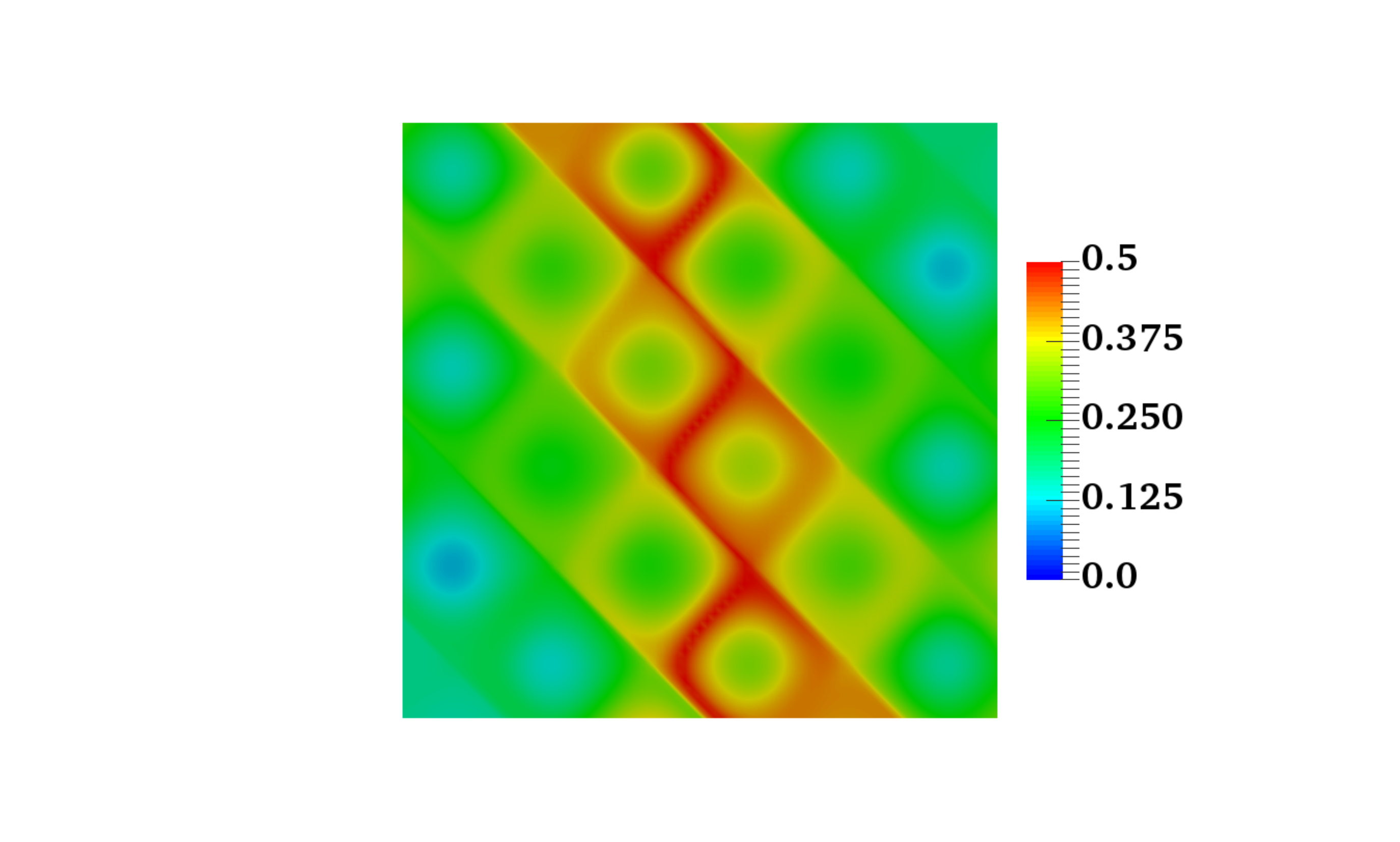}}
  \subfigure[$\kappa_fL = 4$ and $t = 0.1$]
    {\includegraphics[clip=true,width = 0.37\textwidth]
    {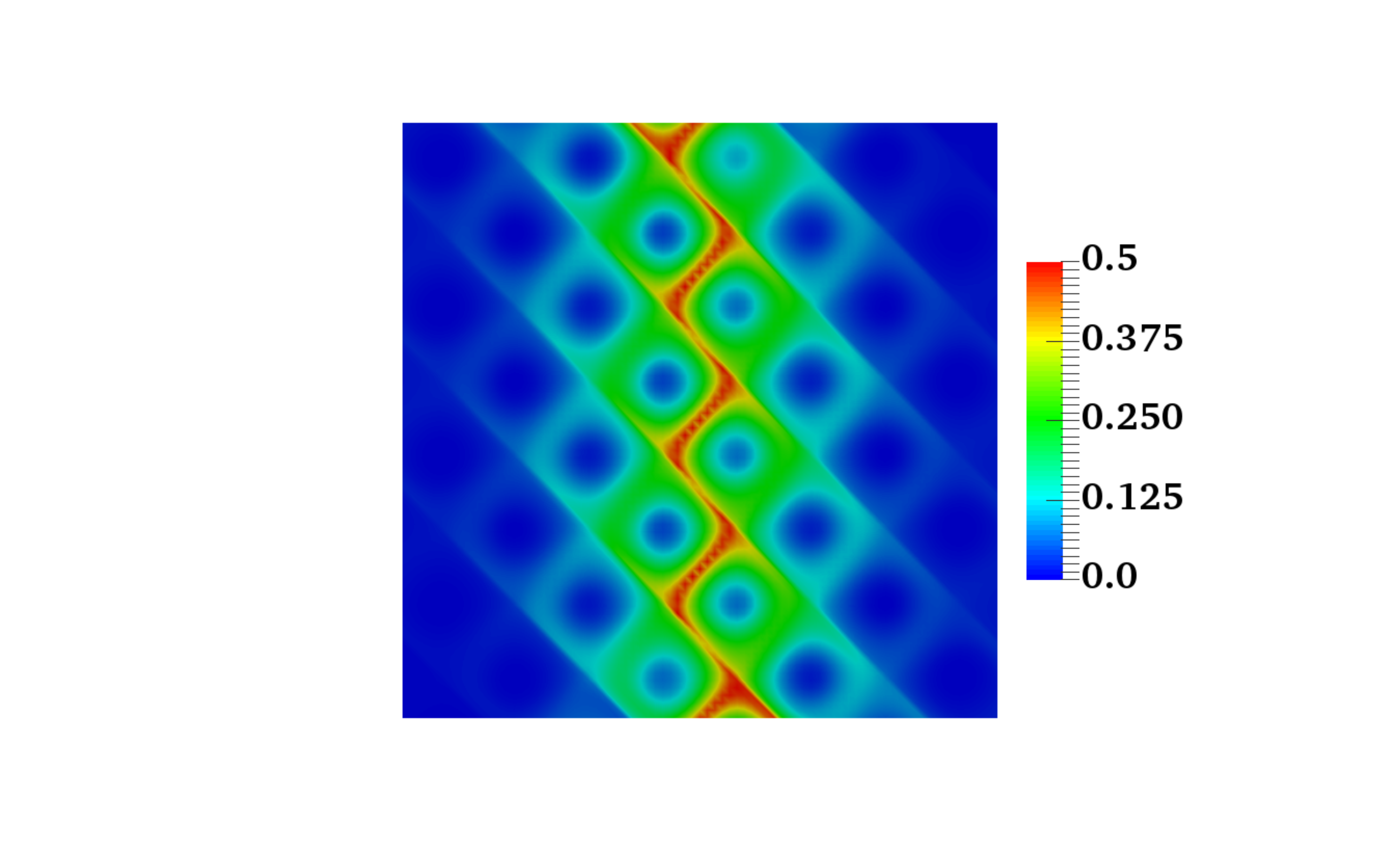}}
  \hspace{-0.5in}
  \subfigure[$\kappa_fL = 4$ and $t = 0.5$]
    {\includegraphics[clip=true,width = 0.37\textwidth]
    {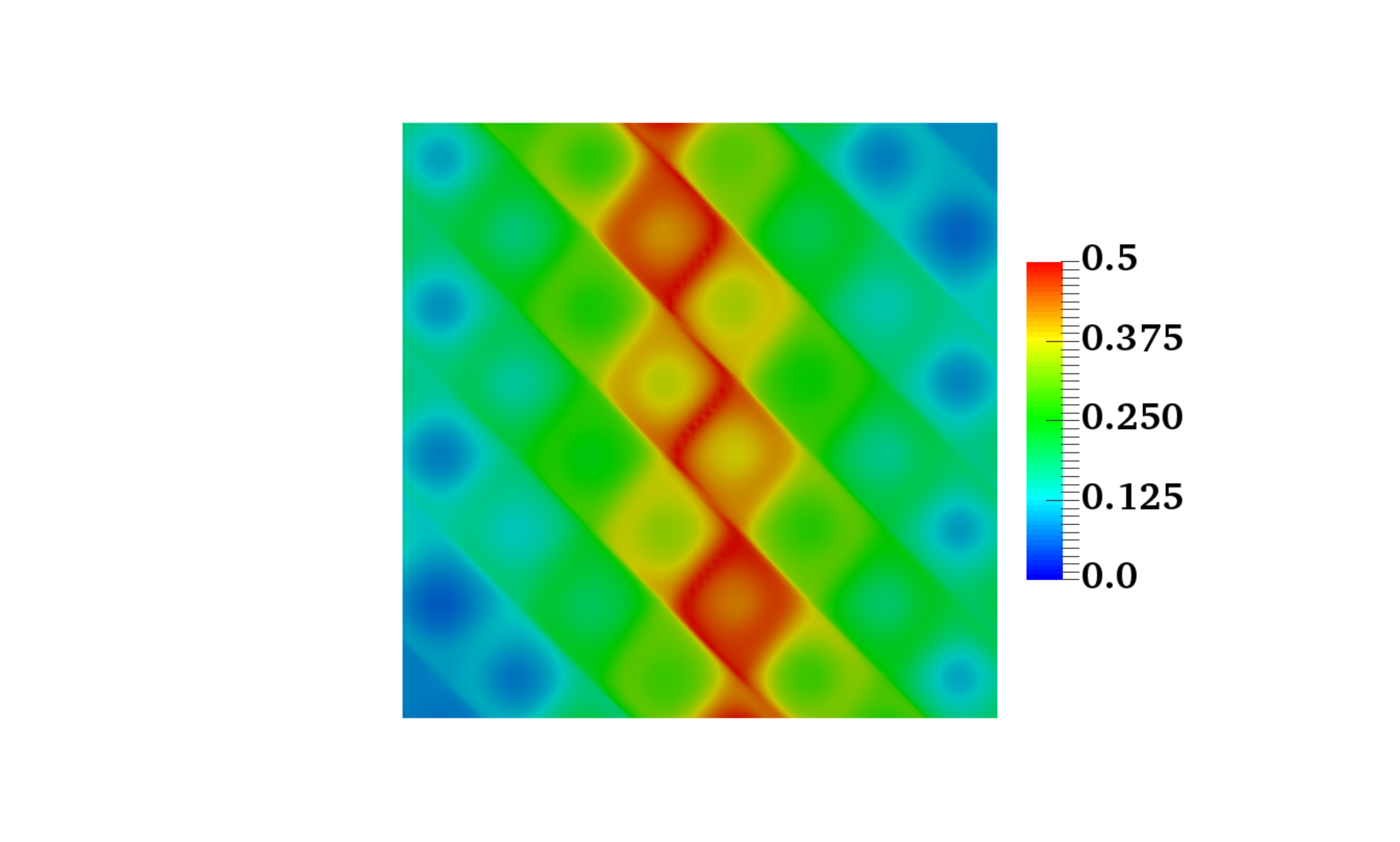}}
  \hspace{-0.5in}
  \subfigure[$\kappa_fL = 4$ and $t = 1.0$]
    {\includegraphics[clip=true,width = 0.37\textwidth]
    {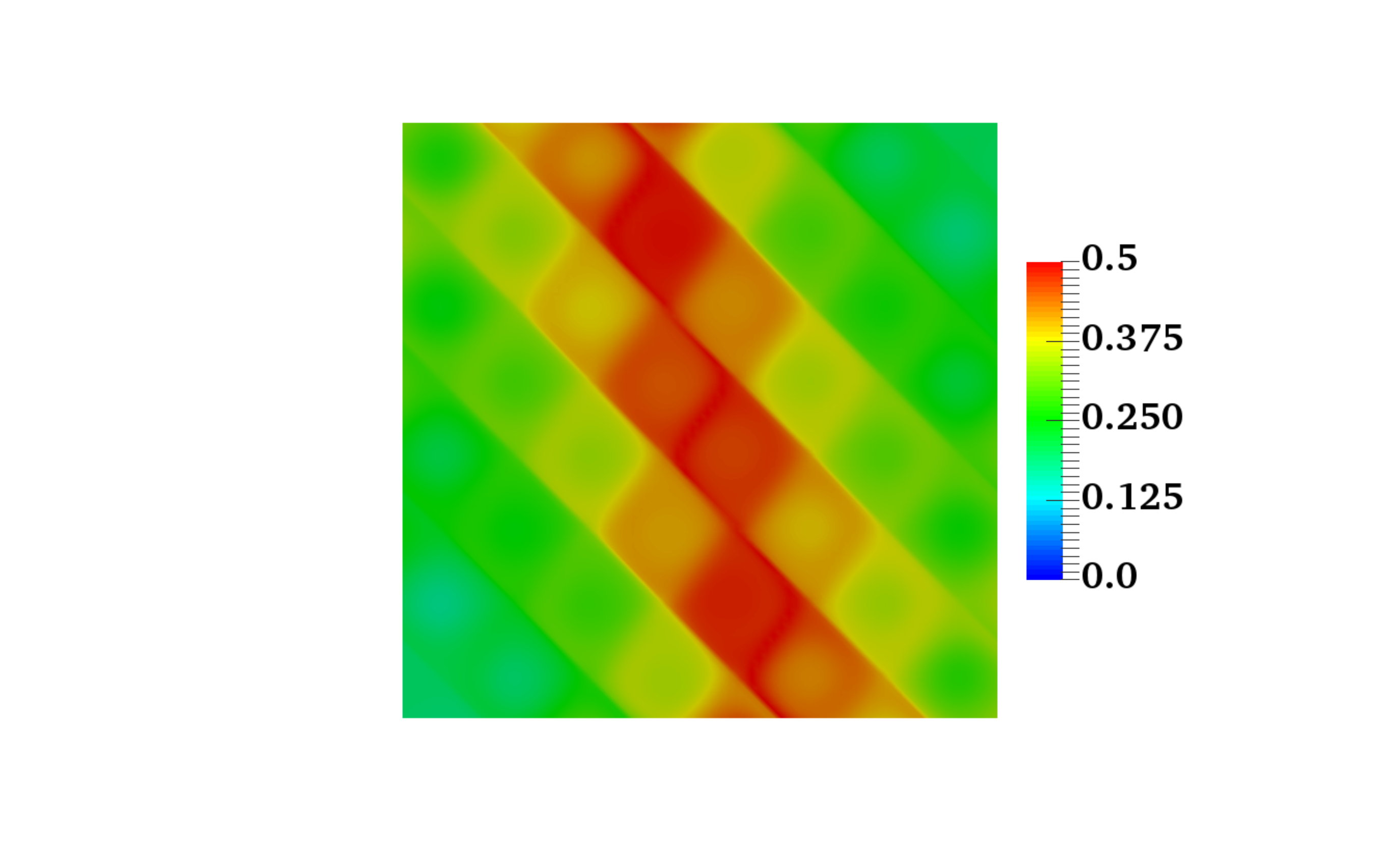}}
  \subfigure[$\kappa_fL = 5$ and $t = 0.1$]
    {\includegraphics[clip=true,width = 0.37\textwidth]
    {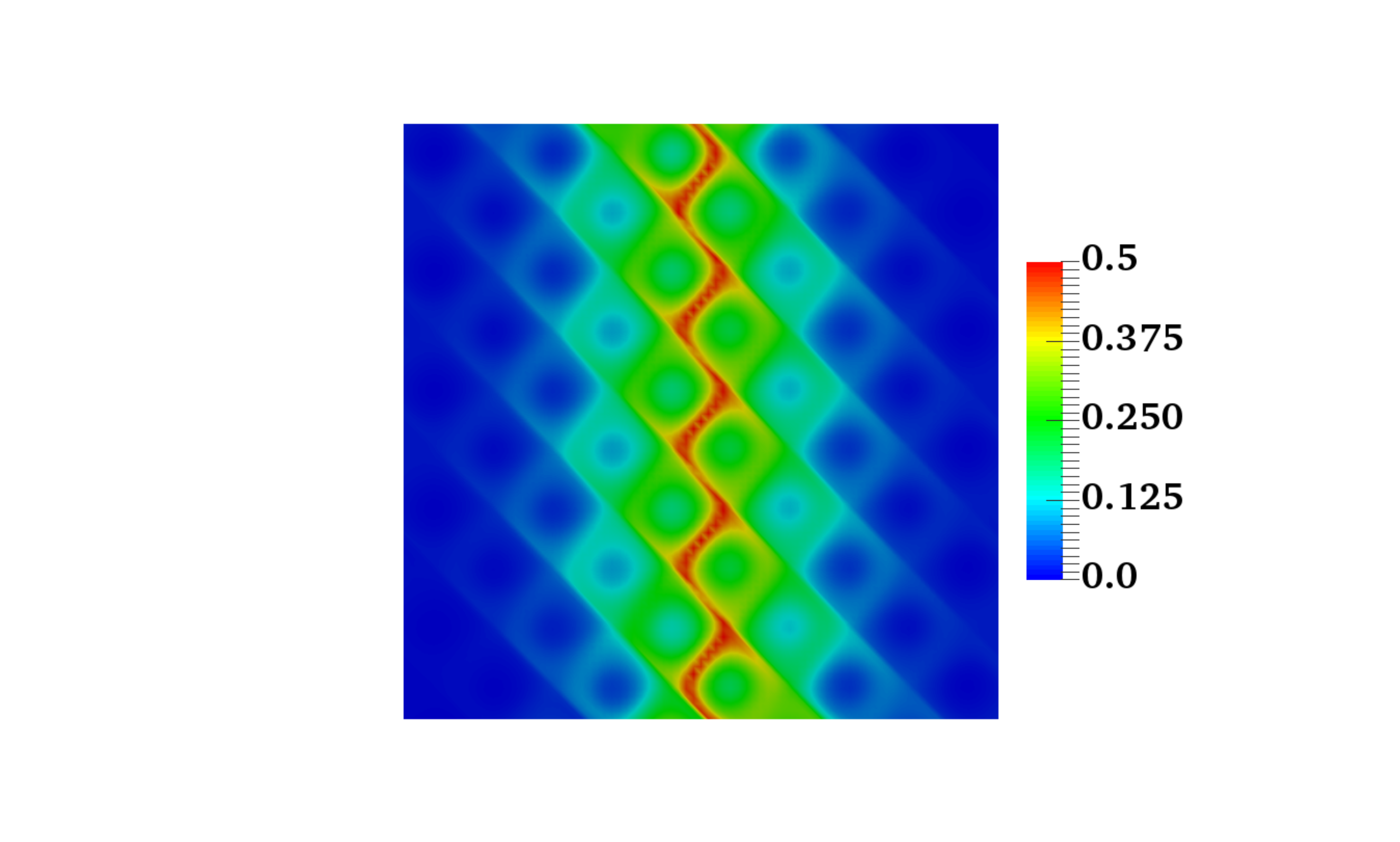}}
  \hspace{-0.5in}
  \subfigure[$\kappa_fL = 5$ and $t = 0.5$]
    {\includegraphics[clip=true,width = 0.37\textwidth]
    {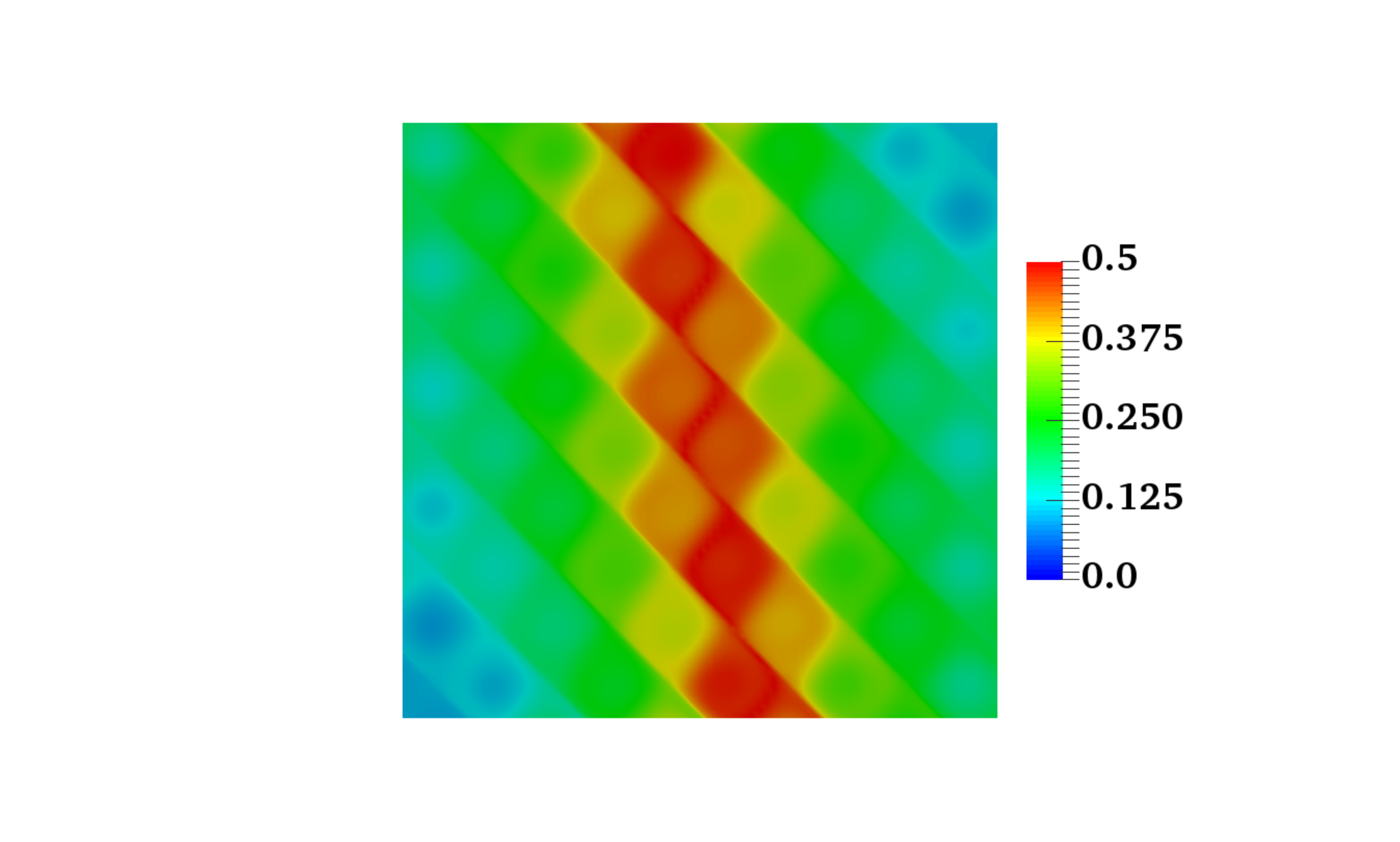}}
  \hspace{-0.5in}
  \subfigure[$\kappa_fL = 5$ and $t = 1.0$]
    {\includegraphics[clip=true,width = 0.37\textwidth]
    {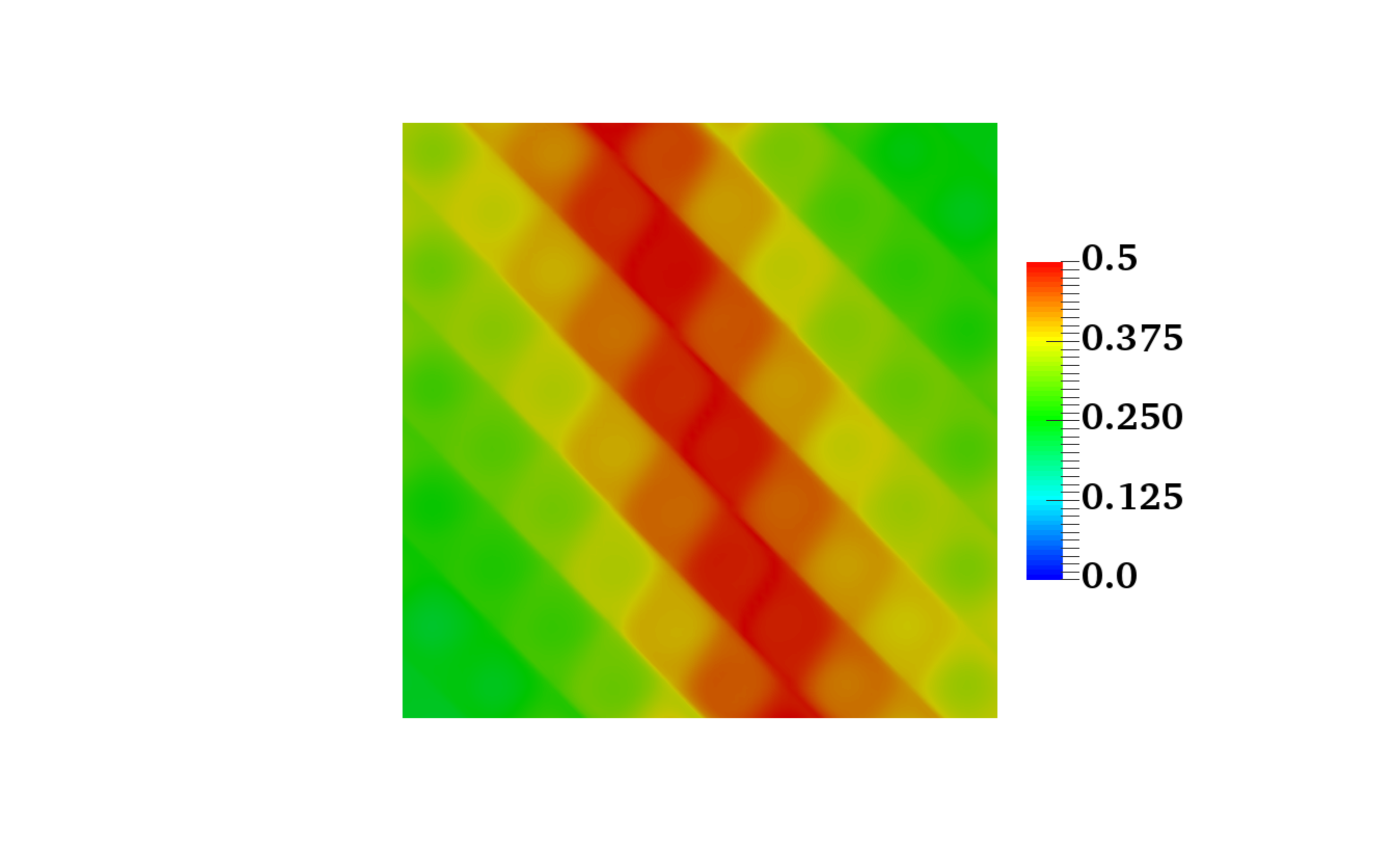}}
  \caption{\textsf{\textbf{Concentration contours of species $C$:}}~These 
    figures show the concentration of product $C$ at times $t = 0.1, \, 0.5,$ 
    and $1.0$. Other input parameters are $\frac{\alpha_L}{\alpha_T} = 10^{4}$, 
    $v_o = 10^{-1}$, $T = 0.1$, and $D_m = 10^{-3}$. From the above figures 
    it is clear that for smaller values of $\kappa_fL$, product $C$ concentration 
    is zero in certain regions. As $\kappa_fL$ increases, we can see that 
    number of regions with zero product $C$ concentration decrease. 
    This is because the velocity field contains a lot of small-scale 
    vortices that are spread across the domain. These small-scale 
    vortex structures present in the velocity field enhance mixing 
    even when the values of molecular diffusivity is small and anisotropic 
    dispersivity is large.
  \label{Fig:Contours_C_Difftimes}}
\end{figure}

\begin{figure}
  \centering
  \subfigure[$\mathfrak{c}_A:= \frac{\langle c_A \rangle}{
    \mathrm{max}\left[\langle c_A \rangle \right]}$]
    {\includegraphics[width = 0.325\textwidth]
    {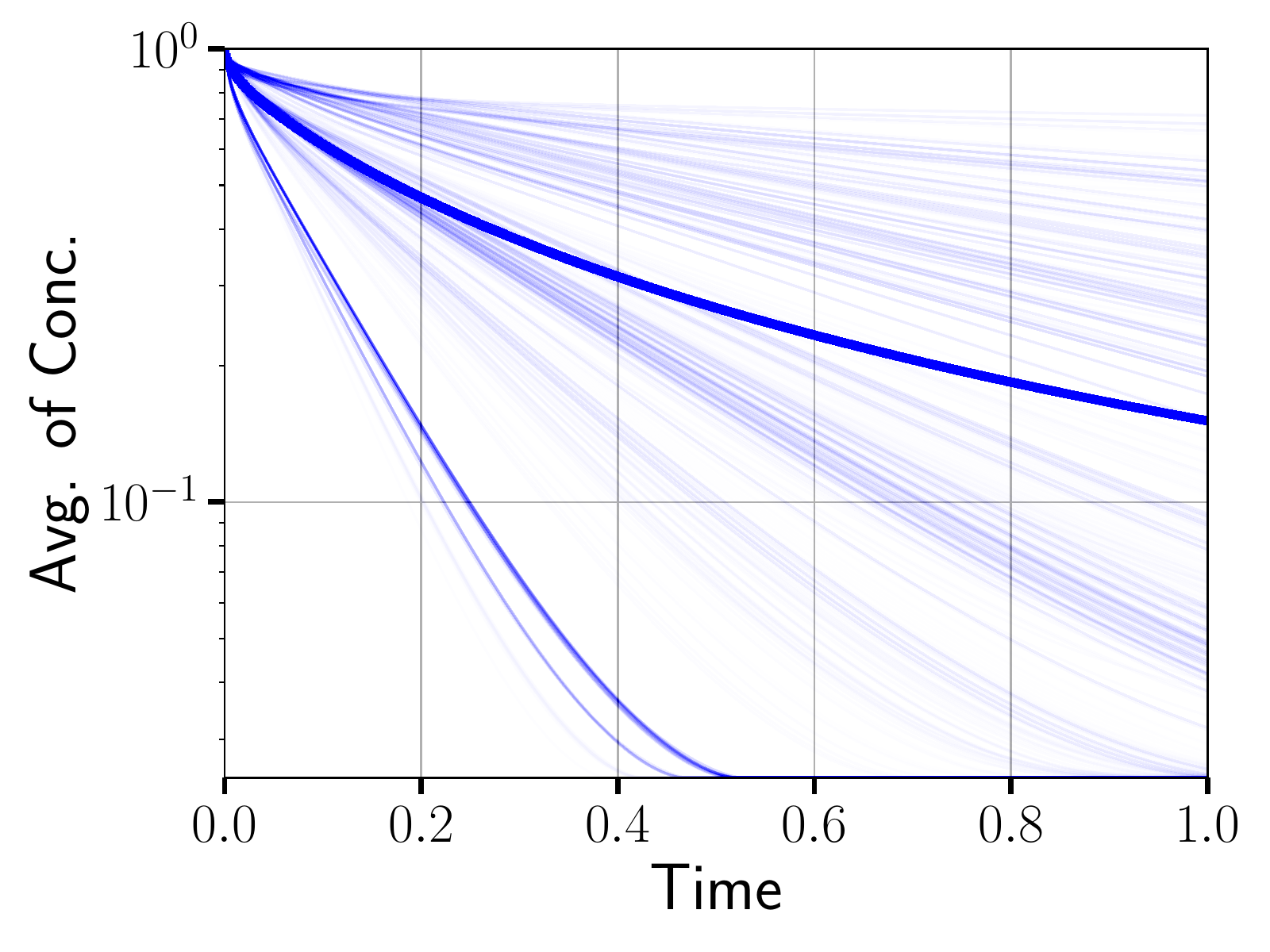}}
  \subfigure[$\mathfrak{c}_B:= \frac{\langle c_B \rangle}{
    \mathrm{max}\left[\langle c_B \rangle \right]}$]
    {\includegraphics[width = 0.325\textwidth]
    {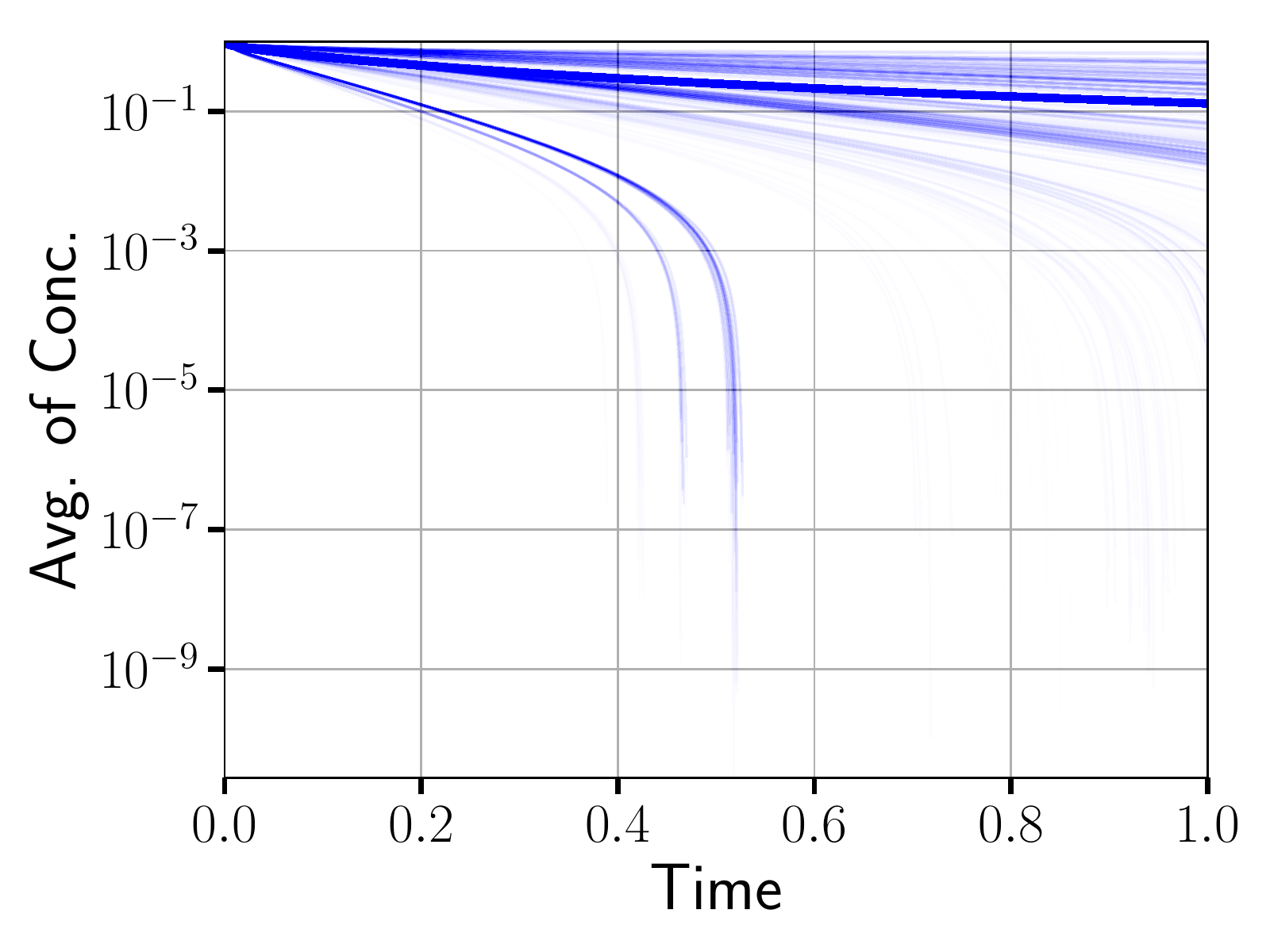}}
  \subfigure[$\mathfrak{c}_C:= \frac{\langle c_C \rangle}{
    \mathrm{max}\left[\langle c_C \rangle \right]}$]
    {\includegraphics[width = 0.325\textwidth]
    {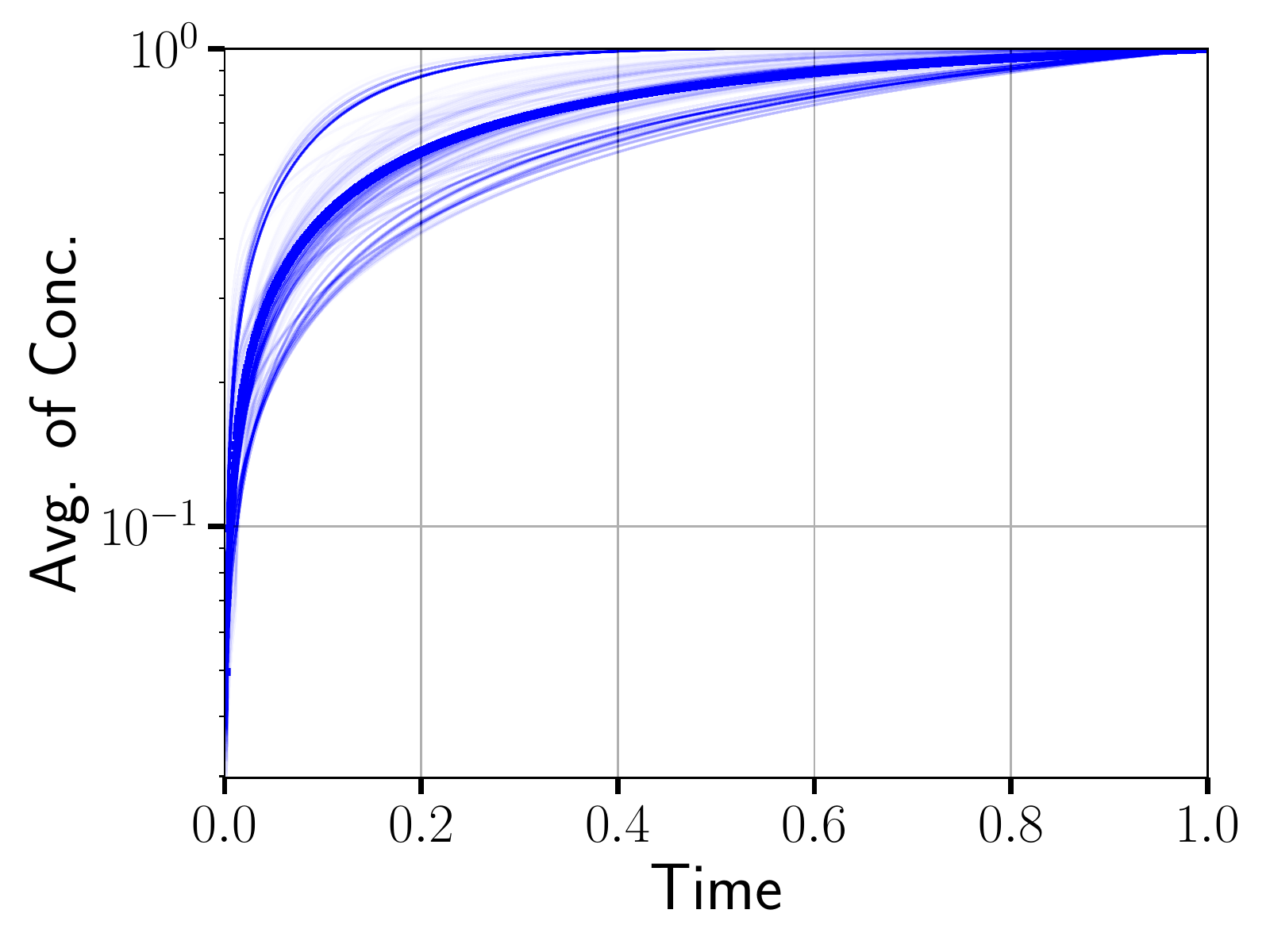}}
  \subfigure[$\mathbb{c}_A:= \frac{\langle c^2_A \rangle}{
    \mathrm{max}\left[\langle c^2_A \rangle \right]}$]
    {\includegraphics[width = 0.325\textwidth]
    {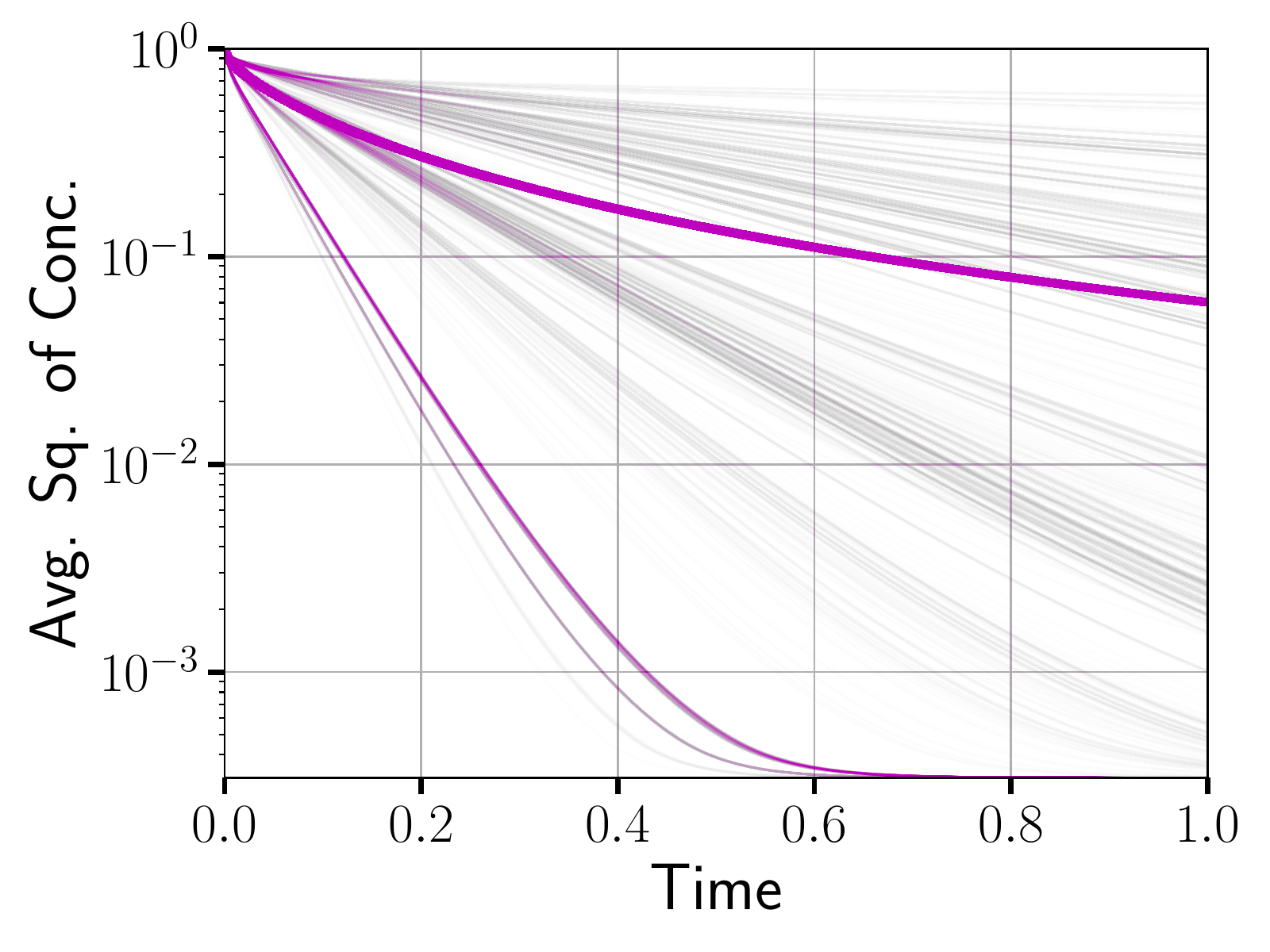}}
  \subfigure[$\mathbb{c}_B:= \frac{\langle c^2_B \rangle}{
    \mathrm{max}\left[\langle c^2_B \rangle \right]}$]
    {\includegraphics[width = 0.325\textwidth]
    {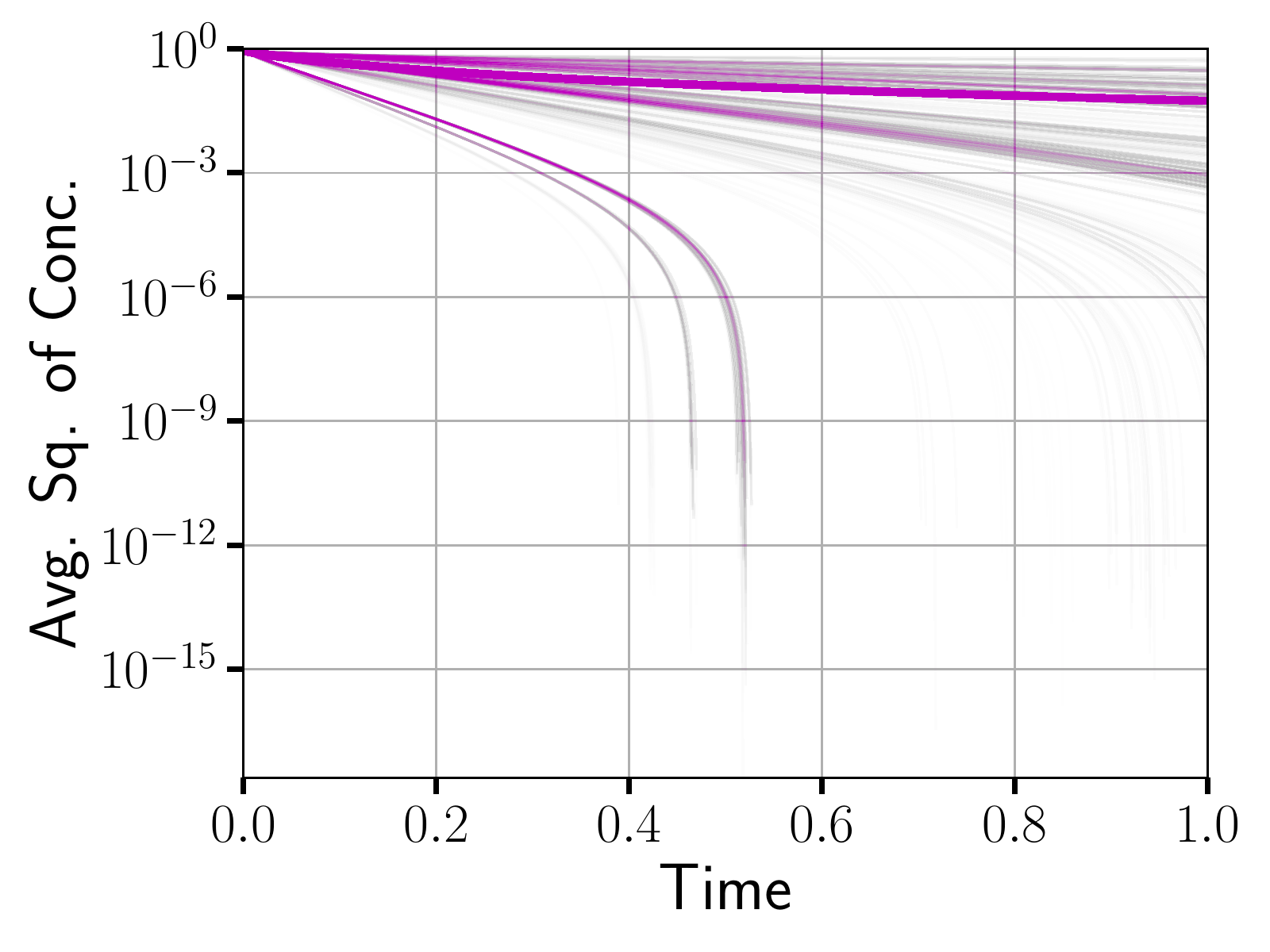}}
  \subfigure[$\mathbb{c}_C:= \frac{\langle c^2_C \rangle}{
    \mathrm{max}\left[\langle c^2_C \rangle \right]}$]
    {\includegraphics[width = 0.325\textwidth]
    {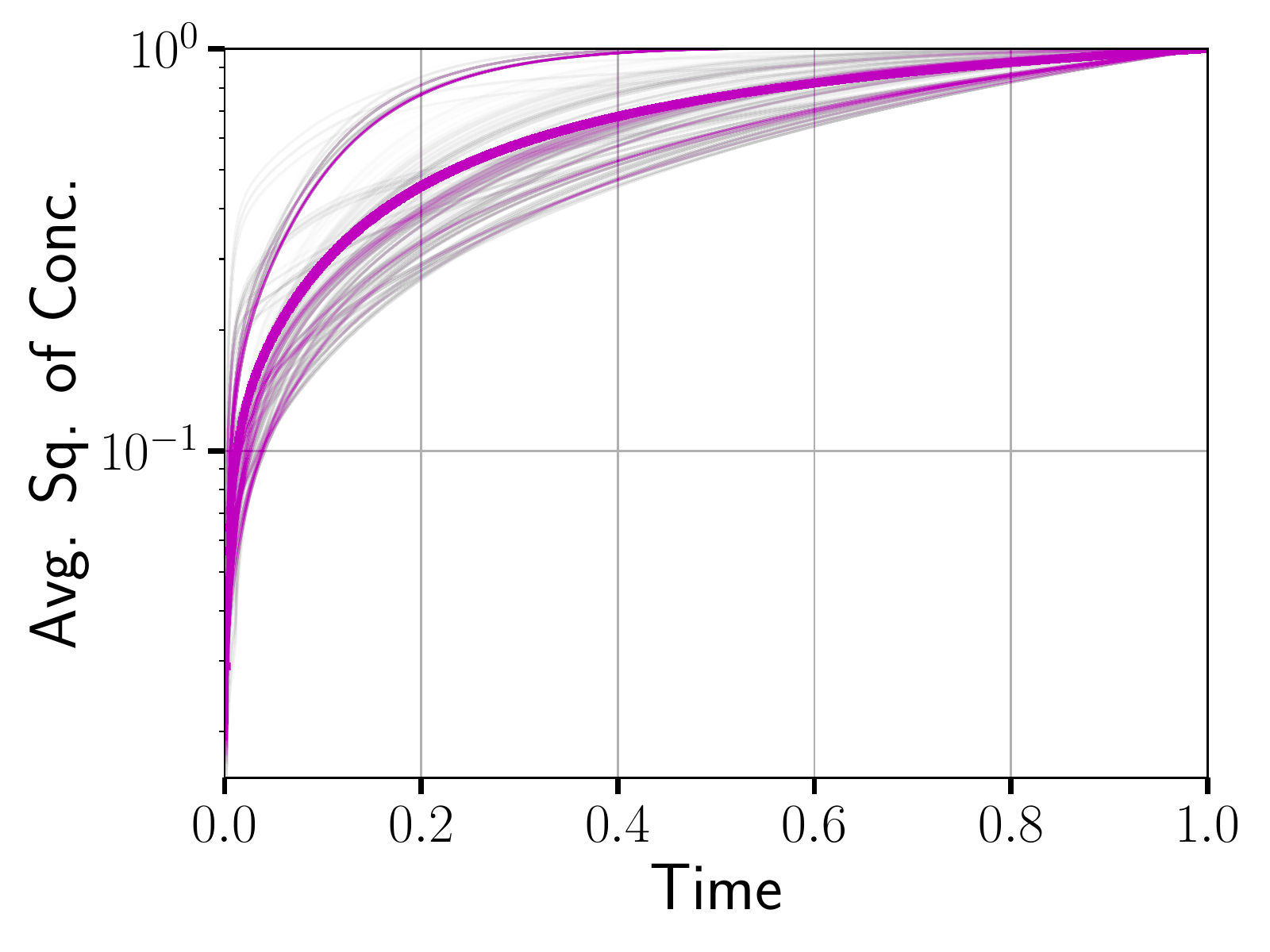}}
  \subfigure[$\sigma^2_A = \frac{\langle c^2_A \rangle 
    - \langle c_A \rangle^2}{\mathrm{max} \left[\langle 
    c^2_A \rangle - \langle c_A \rangle^2 \right]}$]
    {\includegraphics[width = 0.325\textwidth]
    {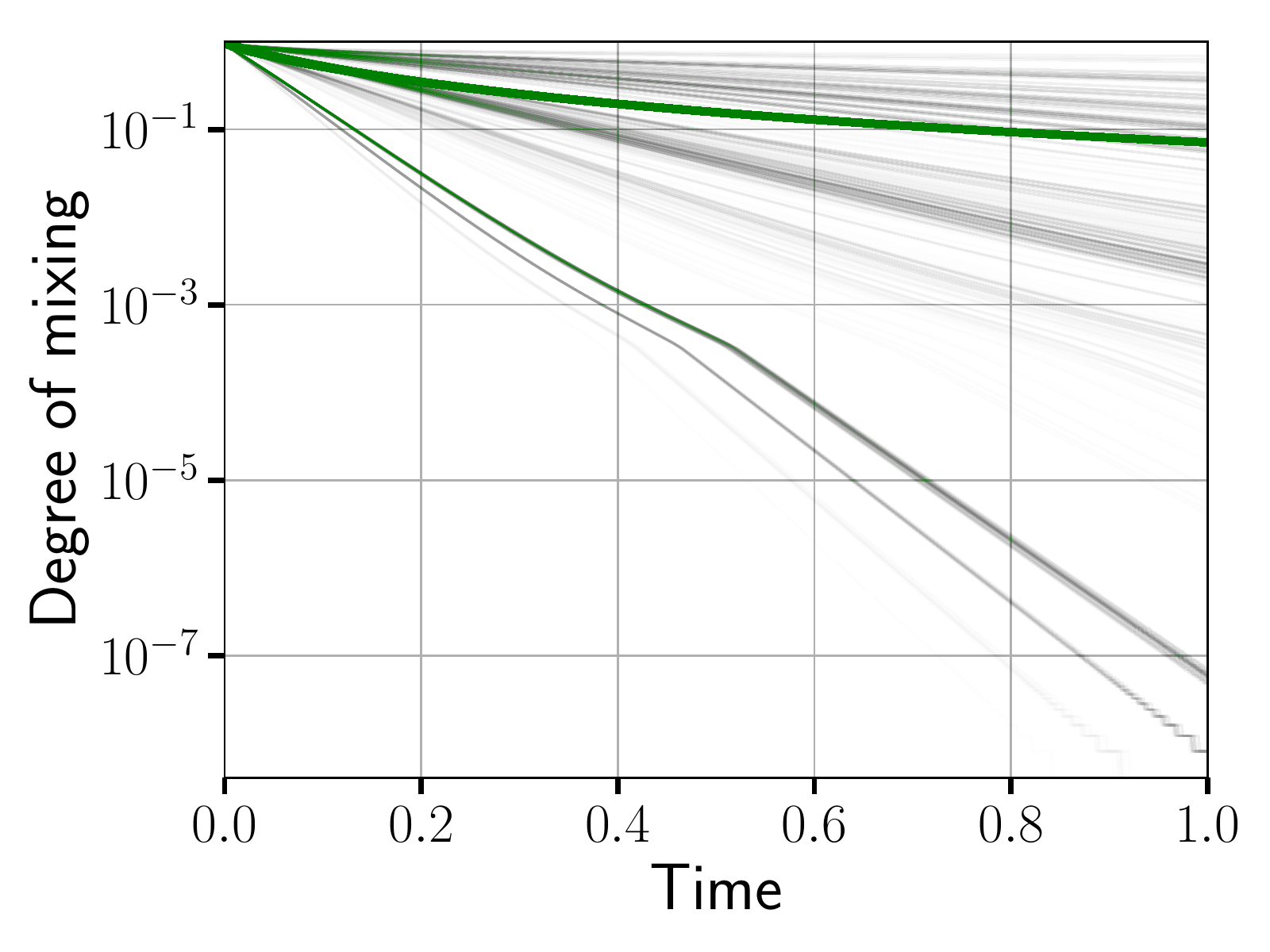}}
  \subfigure[$\sigma^2_B = \frac{\langle c^2_B \rangle 
    - \langle c_B \rangle^2}{\mathrm{max} \left[\langle 
    c^2_B \rangle - \langle c_B \rangle^2 \right]}$]
    {\includegraphics[width = 0.325\textwidth]
    {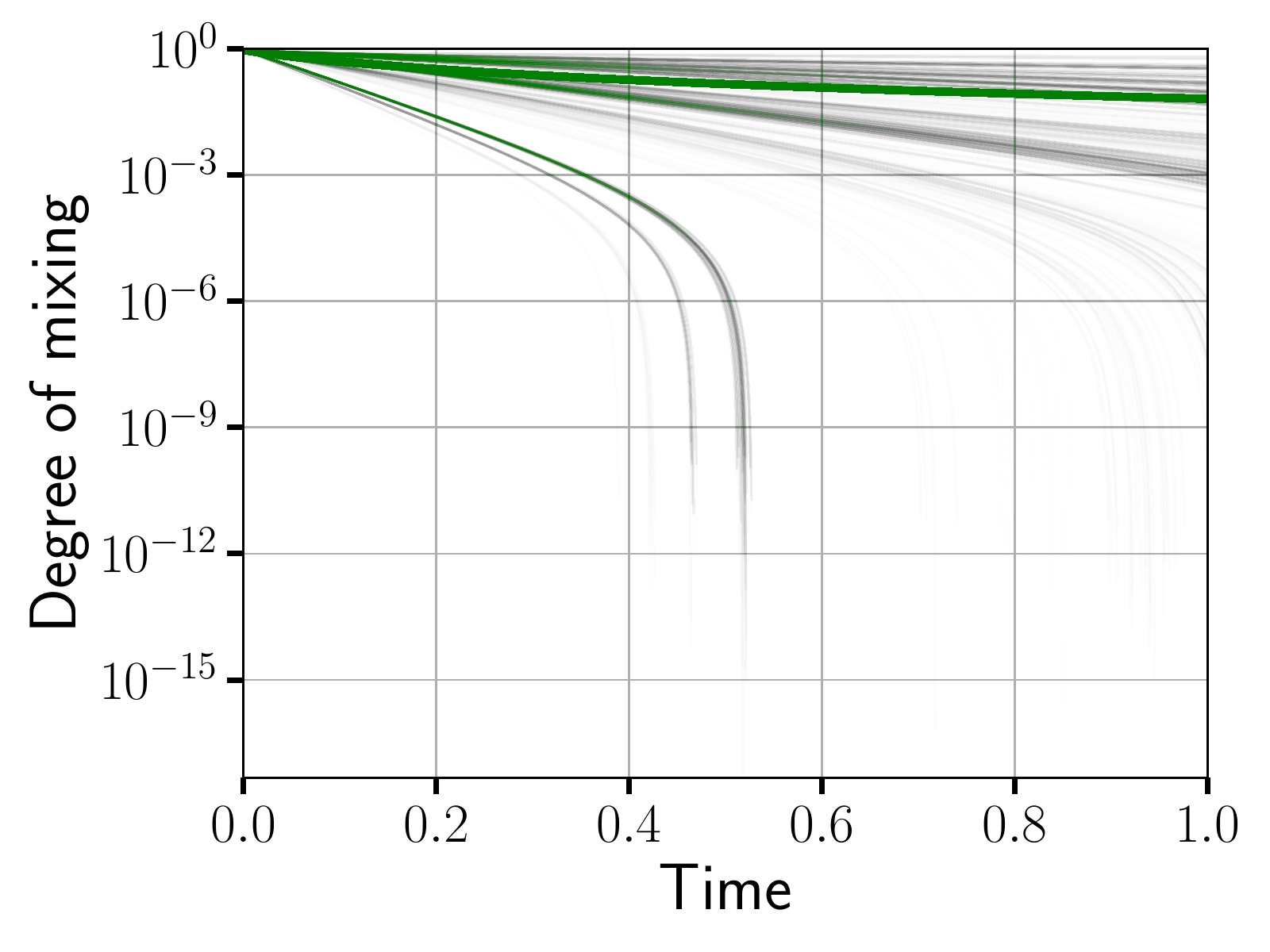}}
  \subfigure[$\sigma^2_C = \frac{\langle c^2_C \rangle 
    - \langle c_C \rangle^2}{\mathrm{max} \left[\langle 
    c^2_C \rangle - \langle c_C \rangle^2 \right]}$]
    {\includegraphics[width = 0.325\textwidth]
    {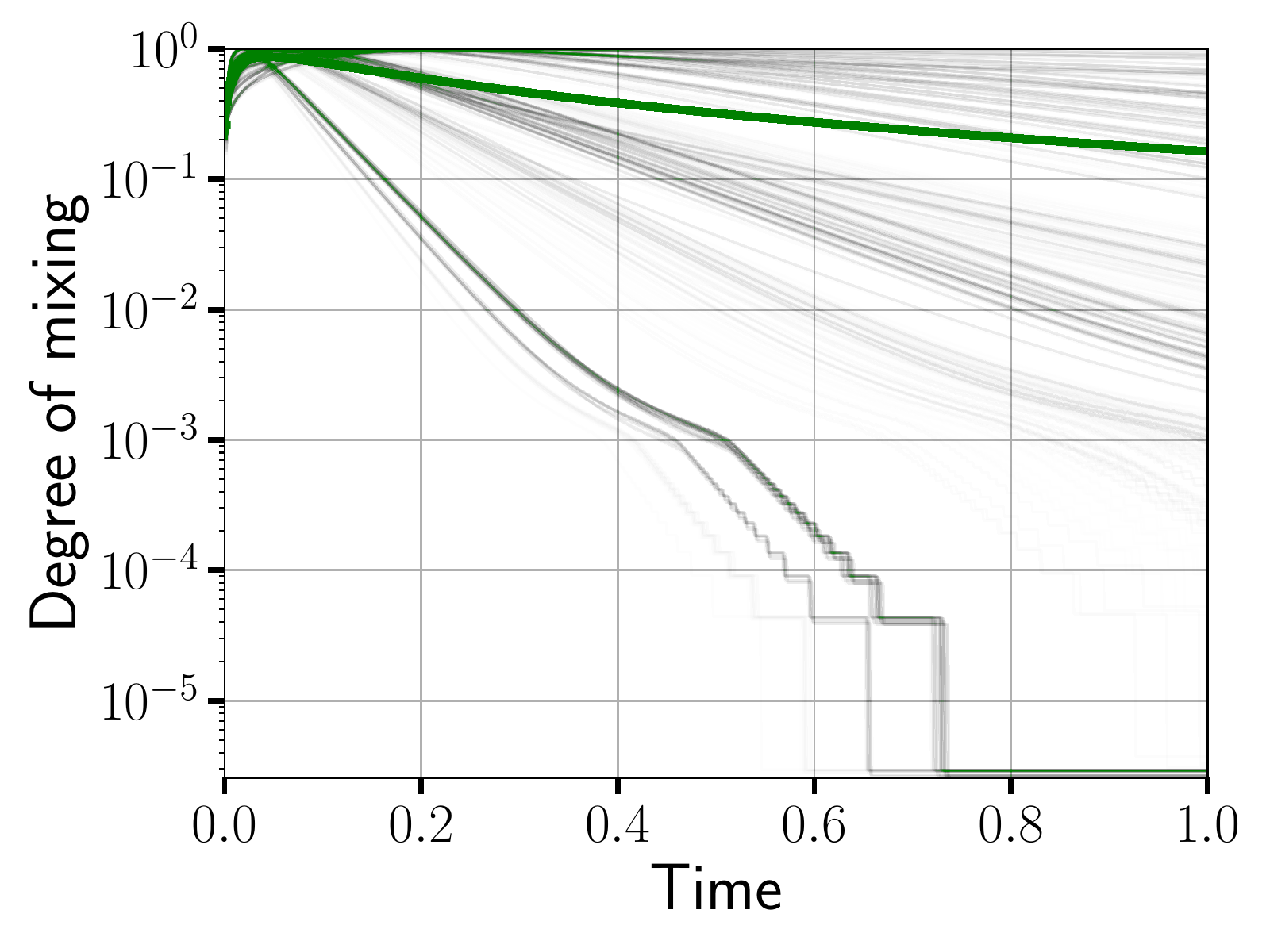}}
  \caption{\textsf{\textbf{Quantities of Interest for species $A$, $B$, and $C$:}}~These 
    figures show average of concentration $\mathfrak{c}_i$, average of square 
    of concentration $\mathbb{c}_i$, and degree of mixing $\sigma^2_i$ as a 
    function of time for all three species $A, B, C$. Plots are shown for 
    all the realizations, which are depicted in light blue, light magenta, 
    and light green colors. The solid blue, magenta, and green lines are the 
    average of all the realization at each time level. From these figures, 
    it is clear that $\ln[\mathfrak{c}_i] \propto t$, $\ln[\mathbb{c}_i] 
    \propto t$, and $\ln[\sigma^2_i] \propto t$ for $t \in [0.2, 1.0]$. 
    This means that the QoIs decrease with time in an exponential manner, 
    approximately. For $t \in [0.2, 1.0]$, in general, the average of 
    concentration and average of square of concentration decreases at 
    a much faster rate, resulting in high product yield. To see which 
    input parameters contribute for significant decrease in degree of 
    mixing of species $A$, $B$, and $C$, we perform feature importance 
    using random forests, F-test, mutual information criteria, and $k$-means 
    clustering on scaling exponent with respect to input parameters. 
  \label{Fig:SpecA_Data}}
\end{figure}

\begin{figure}
  \centering
  \subfigure[$\sigma^2_A$:~Exponent vs. $T$]
    {\includegraphics[width = 0.25\textwidth]
    {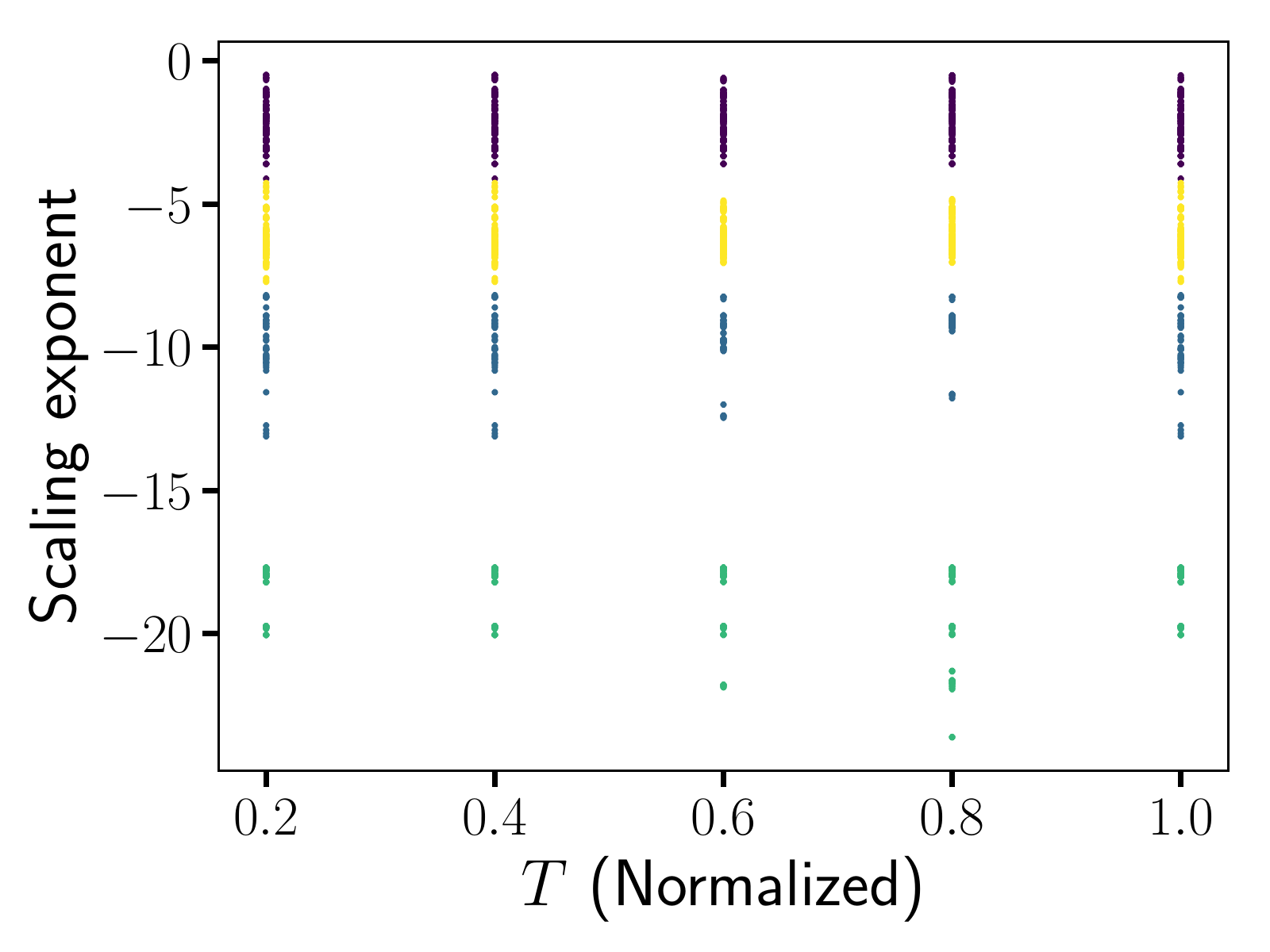}}
  \subfigure[$\sigma^2_B$:~Exponent vs. $T$]
    {\includegraphics[width = 0.25\textwidth]
    {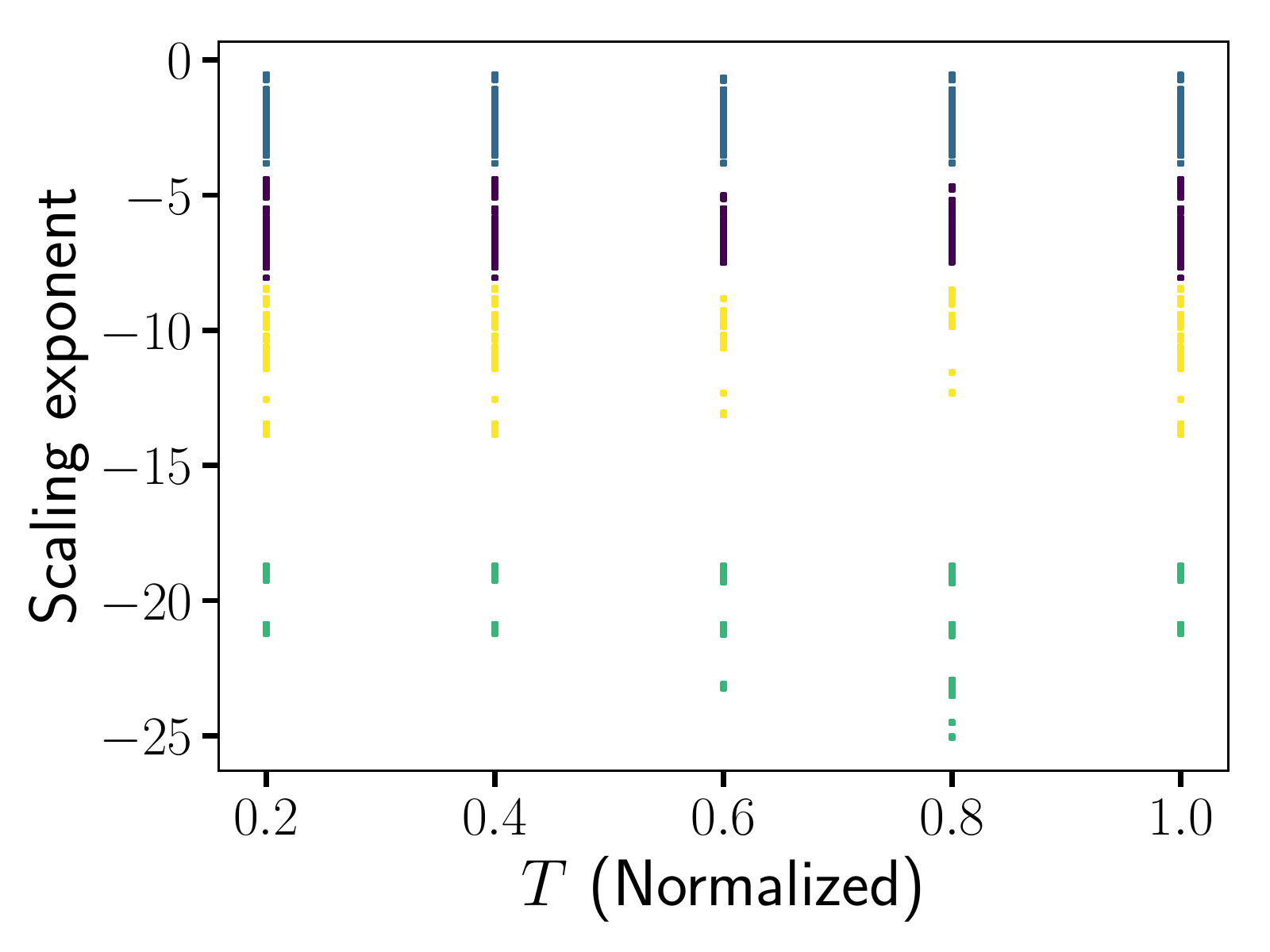}}
  \subfigure[$\sigma^2_C$:~Exponent vs. $T$]
    {\includegraphics[width = 0.25\textwidth]
    {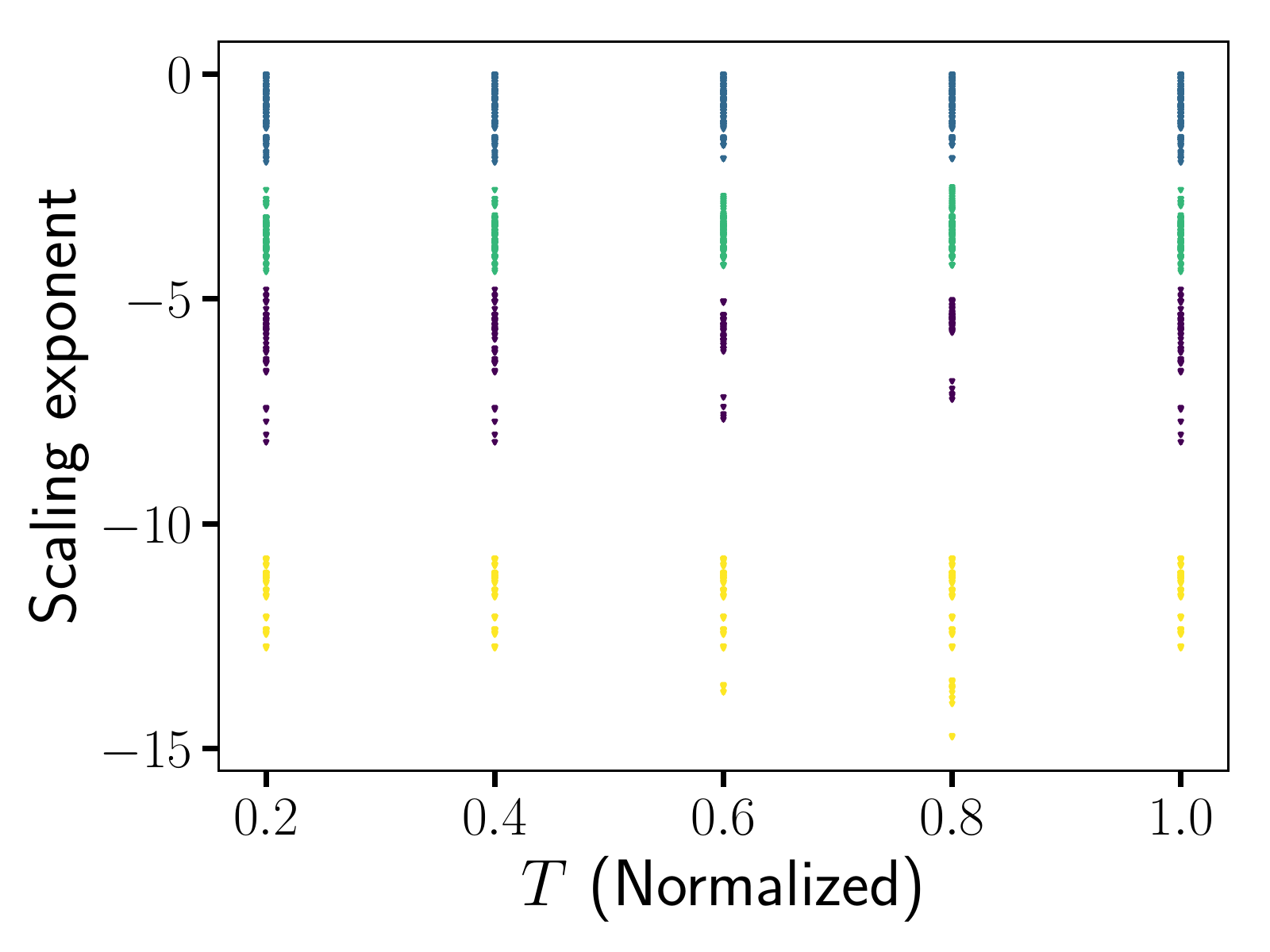}}
  \subfigure[$\sigma^2_A$:~Exponent vs. $\log{[\frac{\alpha_L}{\alpha_T}]}$]
    {\includegraphics[width = 0.25\textwidth]
    {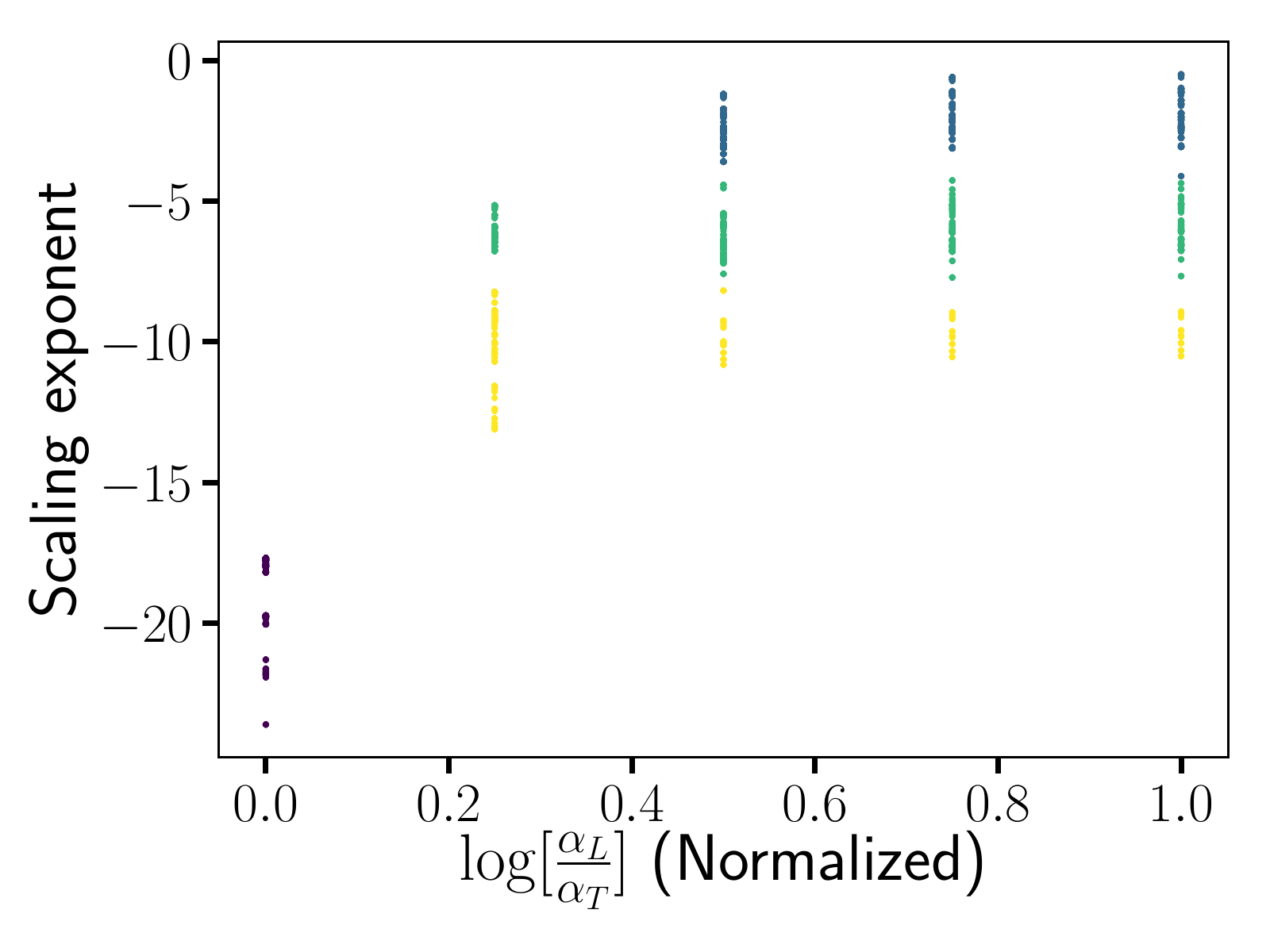}}
  \subfigure[$\sigma^2_B$:~Exponent vs. $\log{[\frac{\alpha_L}{\alpha_T}]}$]
    {\includegraphics[width = 0.25\textwidth]
    {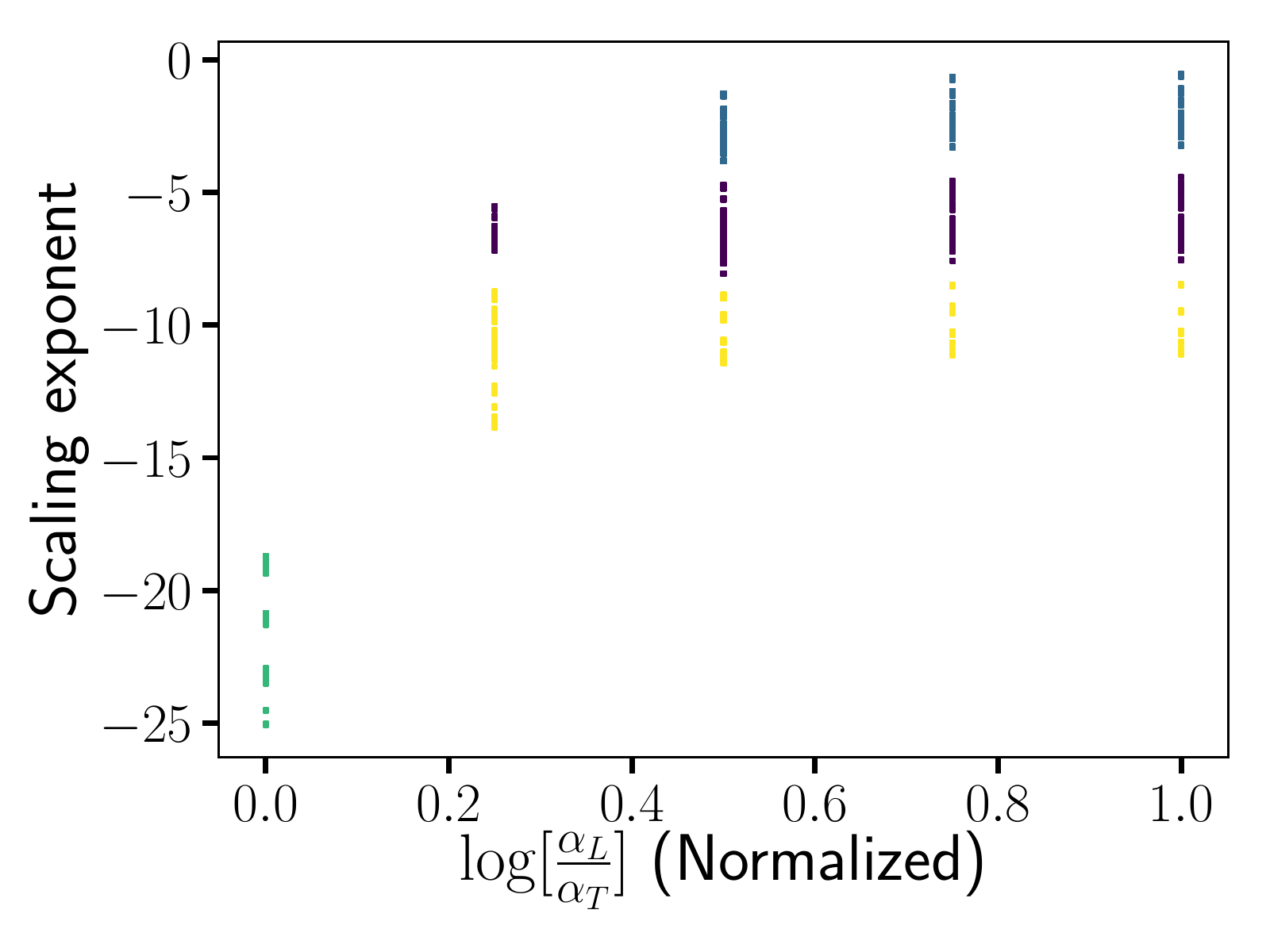}}
  \subfigure[$\sigma^2_C$:~Exponent vs. $\log{[\frac{\alpha_L}{\alpha_T}]}$]
    {\includegraphics[width = 0.25\textwidth]
    {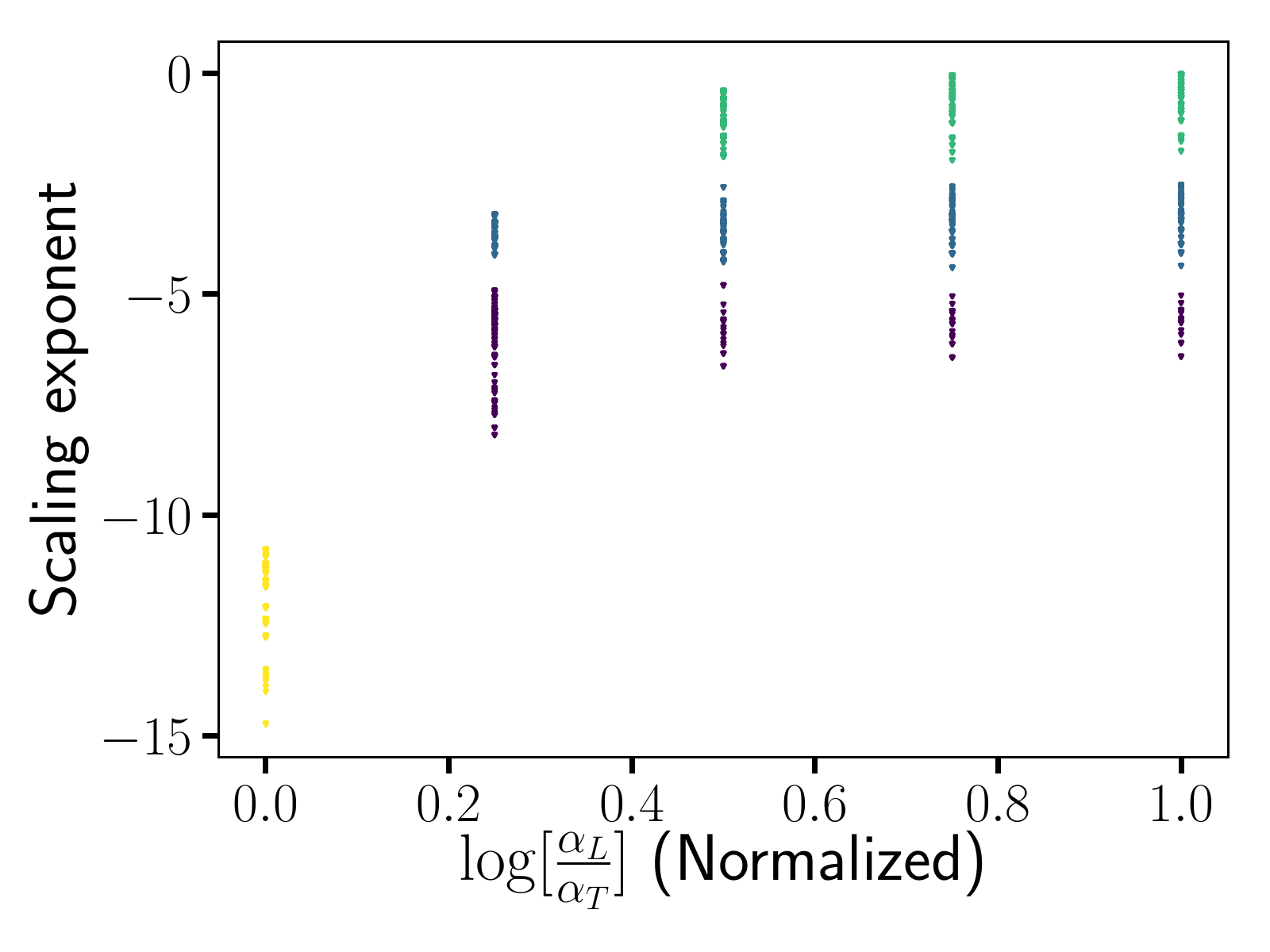}}
  \subfigure[$\sigma^2_A$:~Exponent vs. $\kappa_fL$]
    {\includegraphics[width = 0.25\textwidth]
    {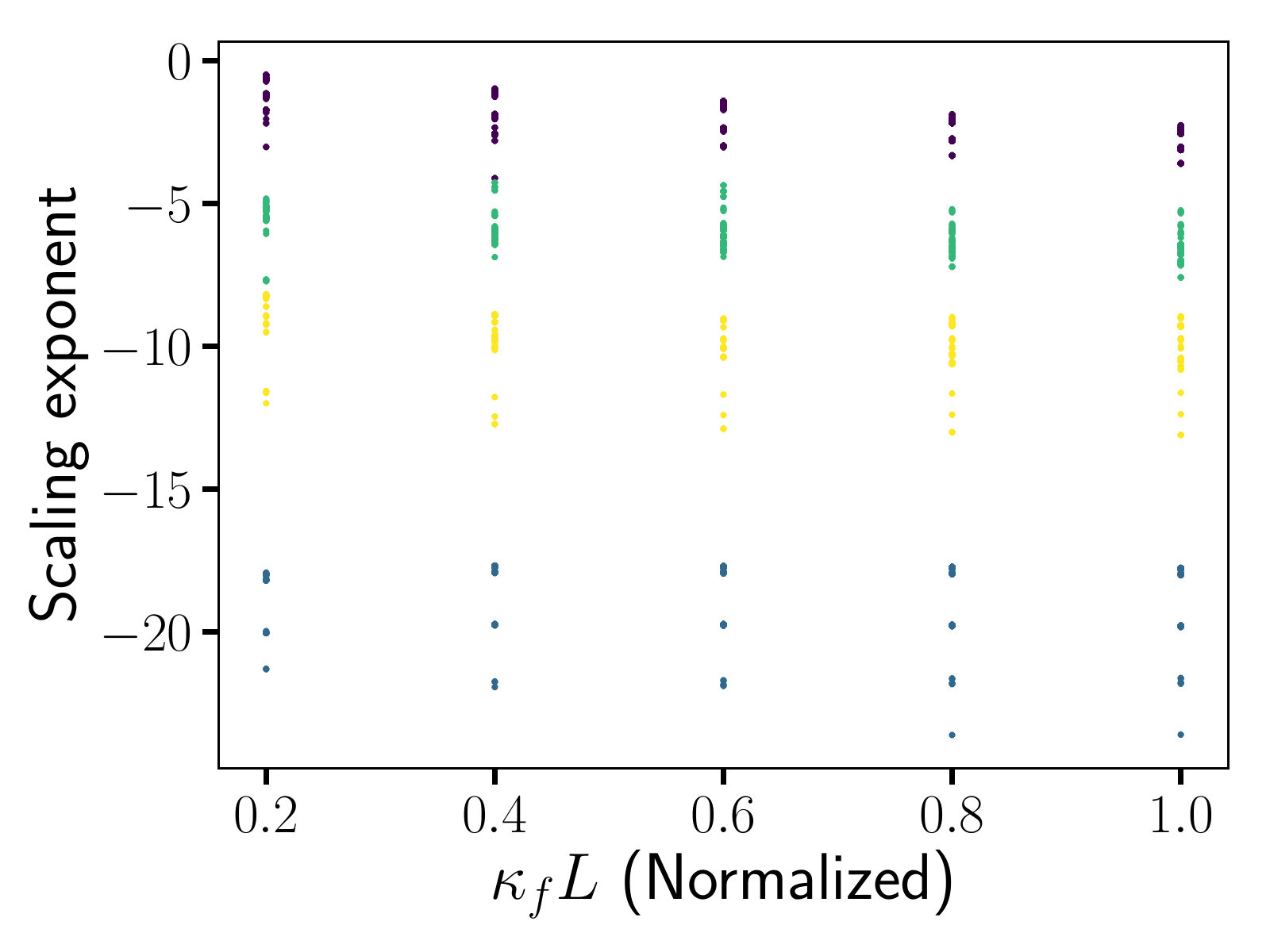}}
  \subfigure[$\sigma^2_B$:~Exponent vs. $\kappa_fL$]
    {\includegraphics[width = 0.25\textwidth]
    {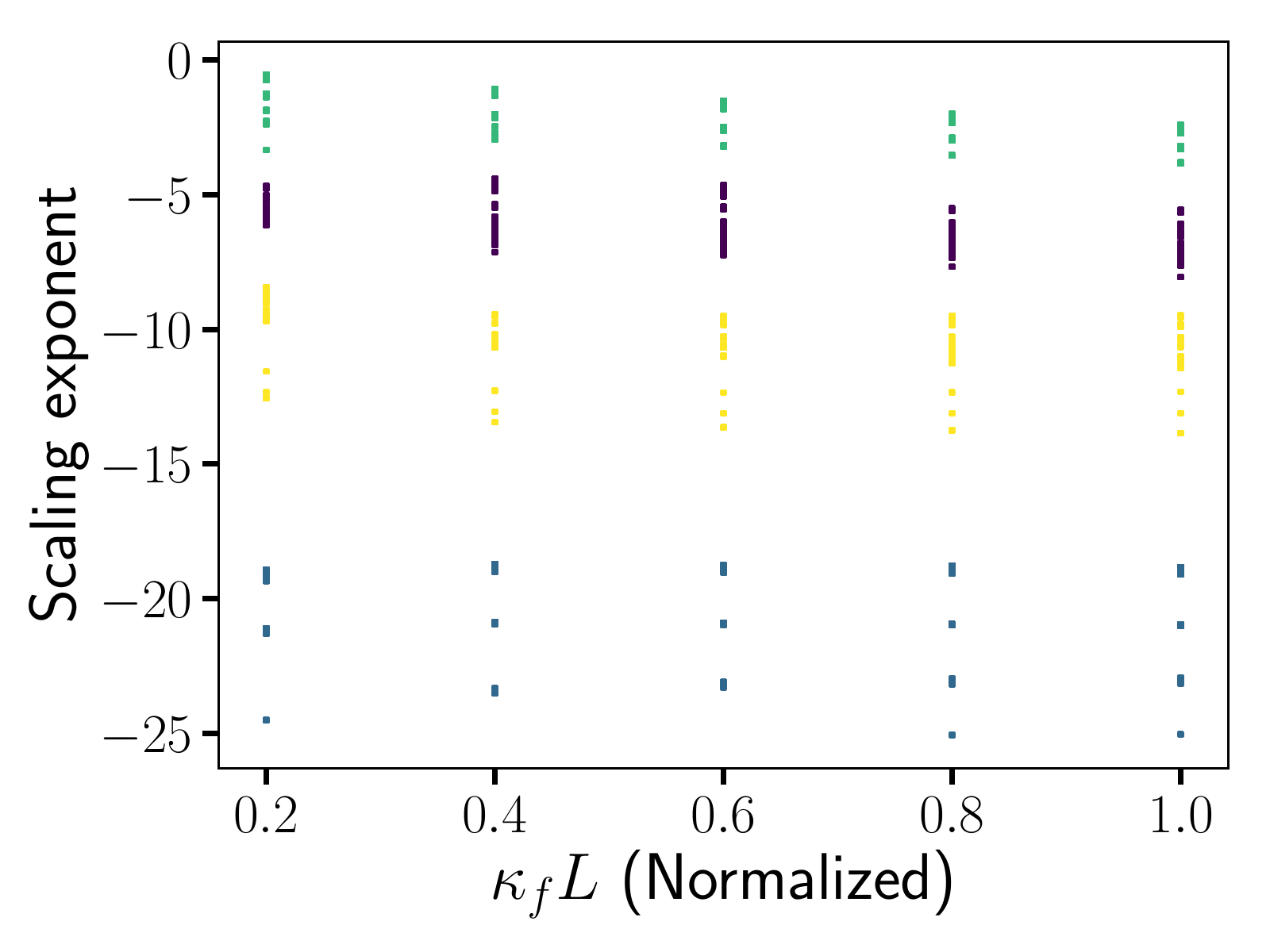}}
  \subfigure[$\sigma^2_C$:~Exponent vs. $\kappa_fL$]
    {\includegraphics[width = 0.25\textwidth]
    {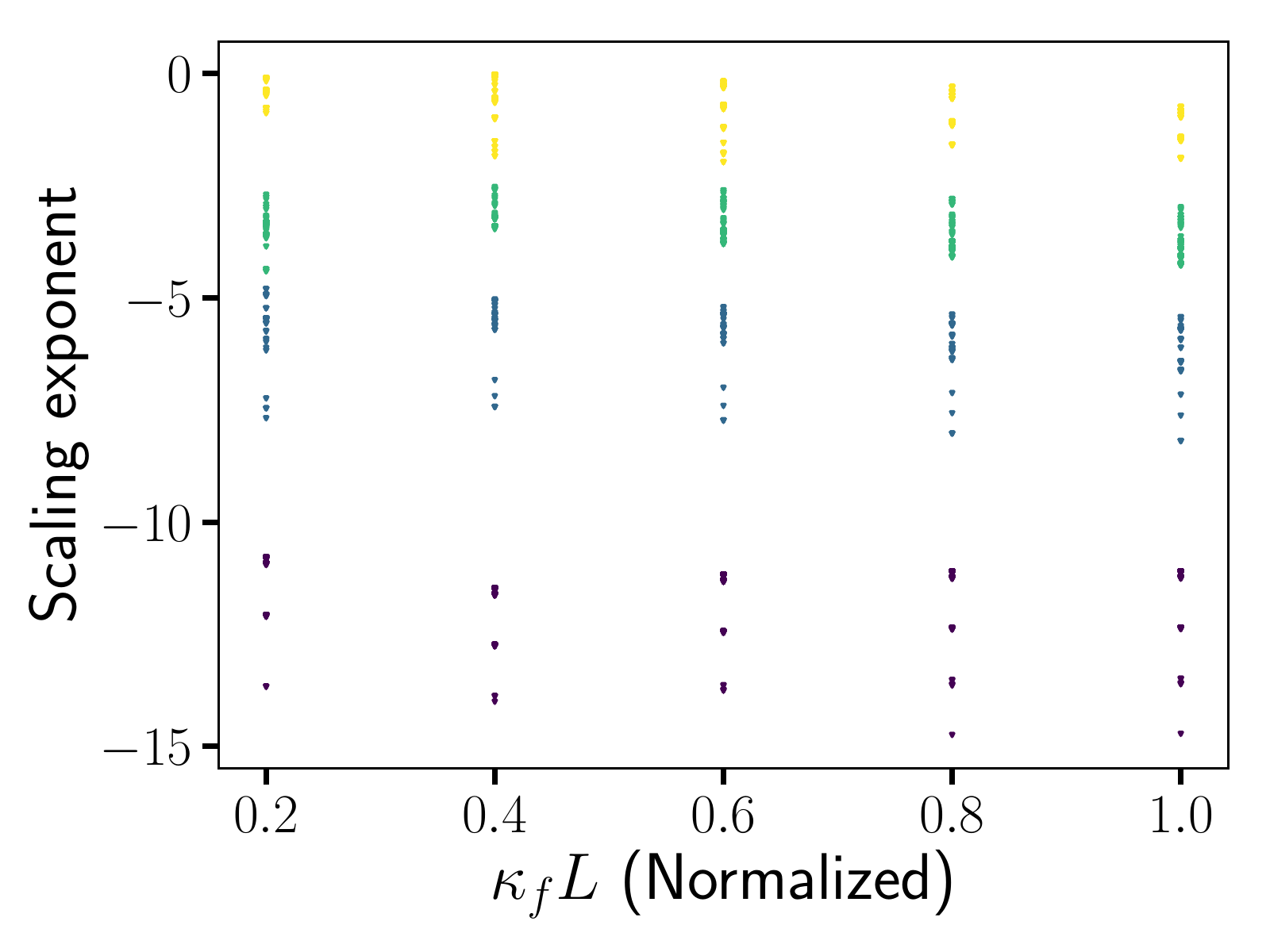}}
  \subfigure[$\sigma^2_A$:~Exponent vs. $\log{\left[v_o \right]}$]
    {\includegraphics[width = 0.25\textwidth]
    {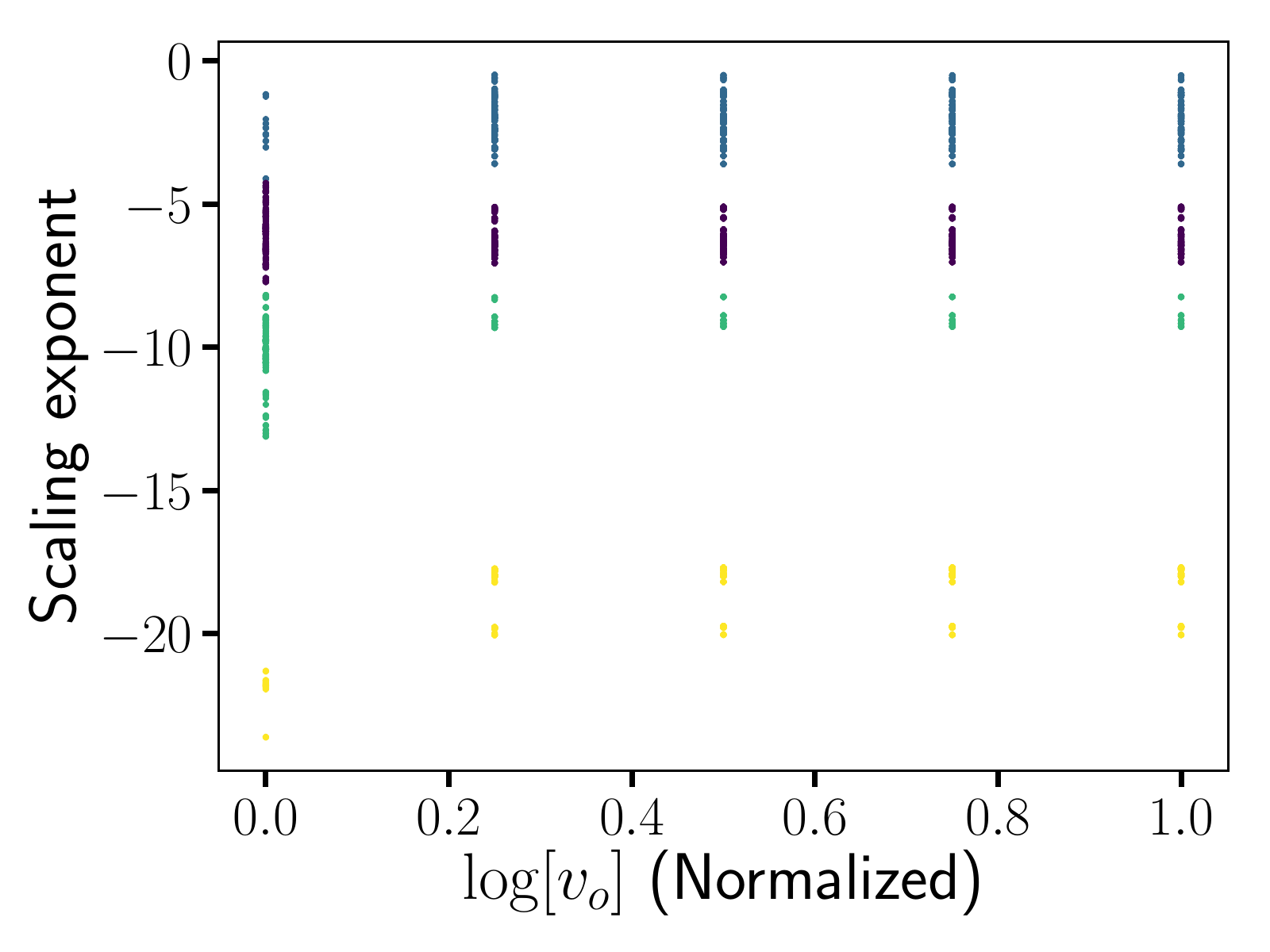}}
  \subfigure[$\sigma^2_B$:~Exponent vs. $\log{\left[v_o \right]}$]
    {\includegraphics[width = 0.25\textwidth]
    {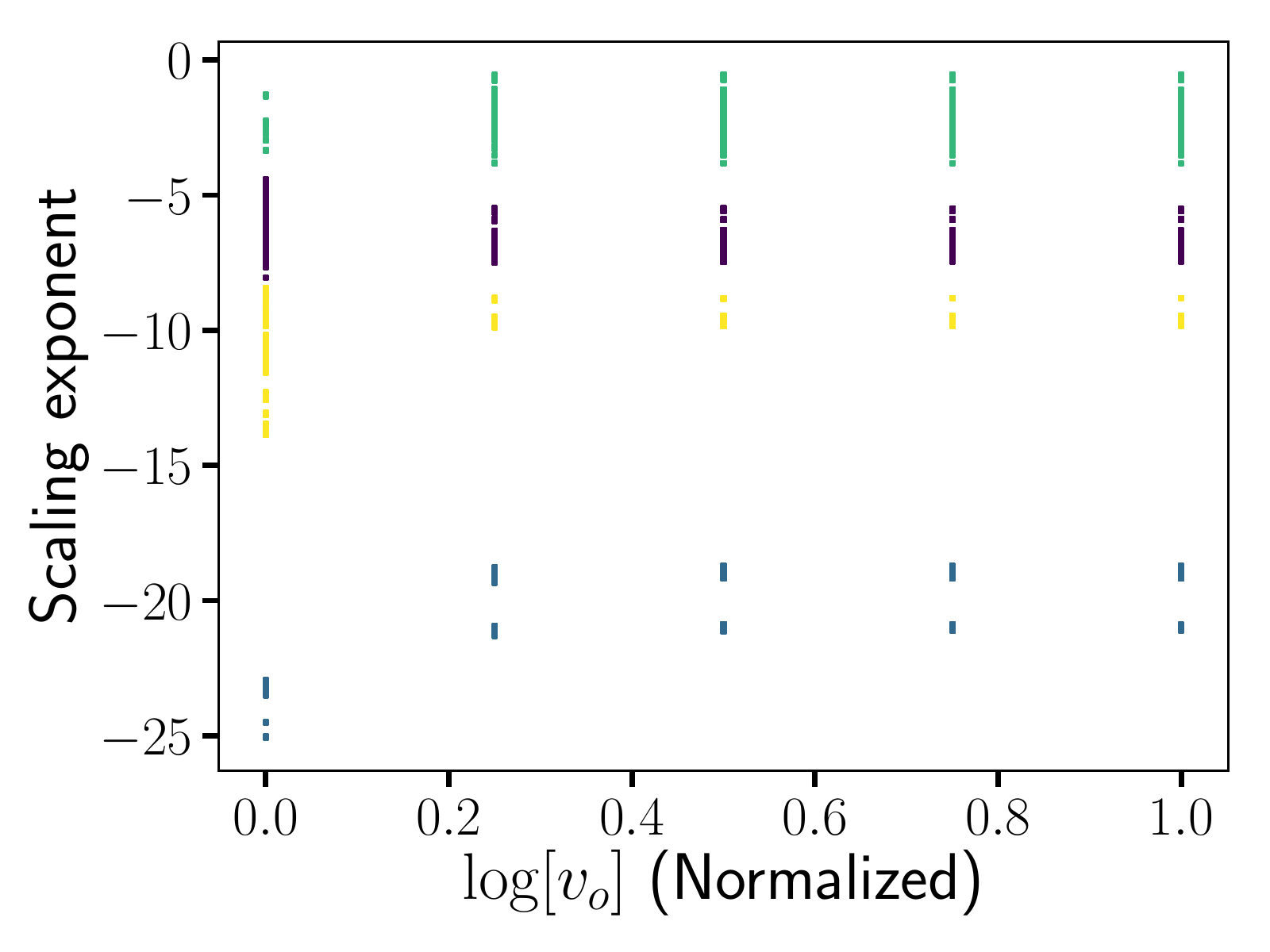}}
  \subfigure[$\sigma^2_C$:~Exponent vs. $\log{\left[v_o \right]}$]
    {\includegraphics[width = 0.25\textwidth]
    {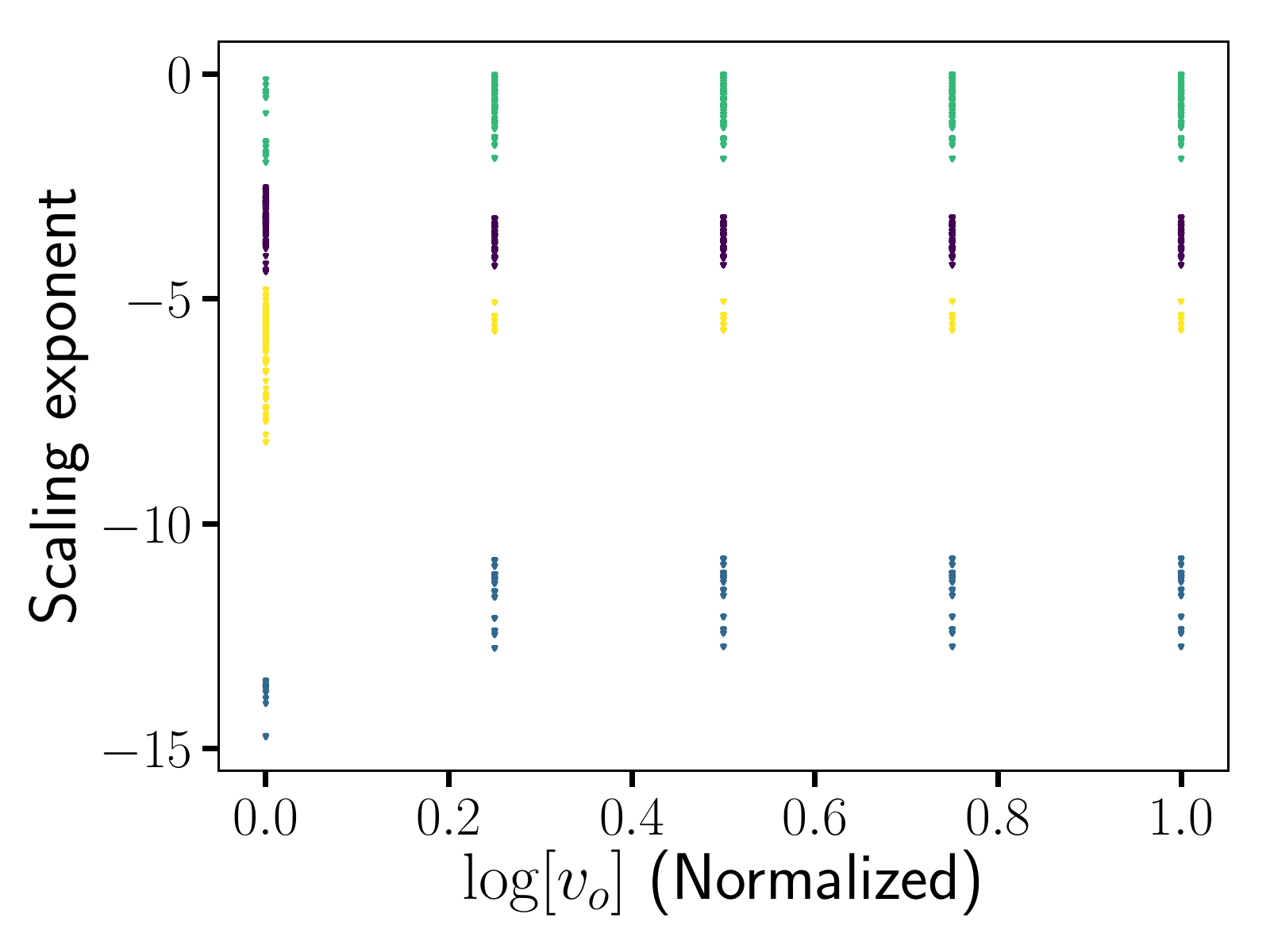}}
  \subfigure[$\sigma^2_A$:~Exponent vs. $D_m$]
    {\includegraphics[width = 0.25\textwidth]
    {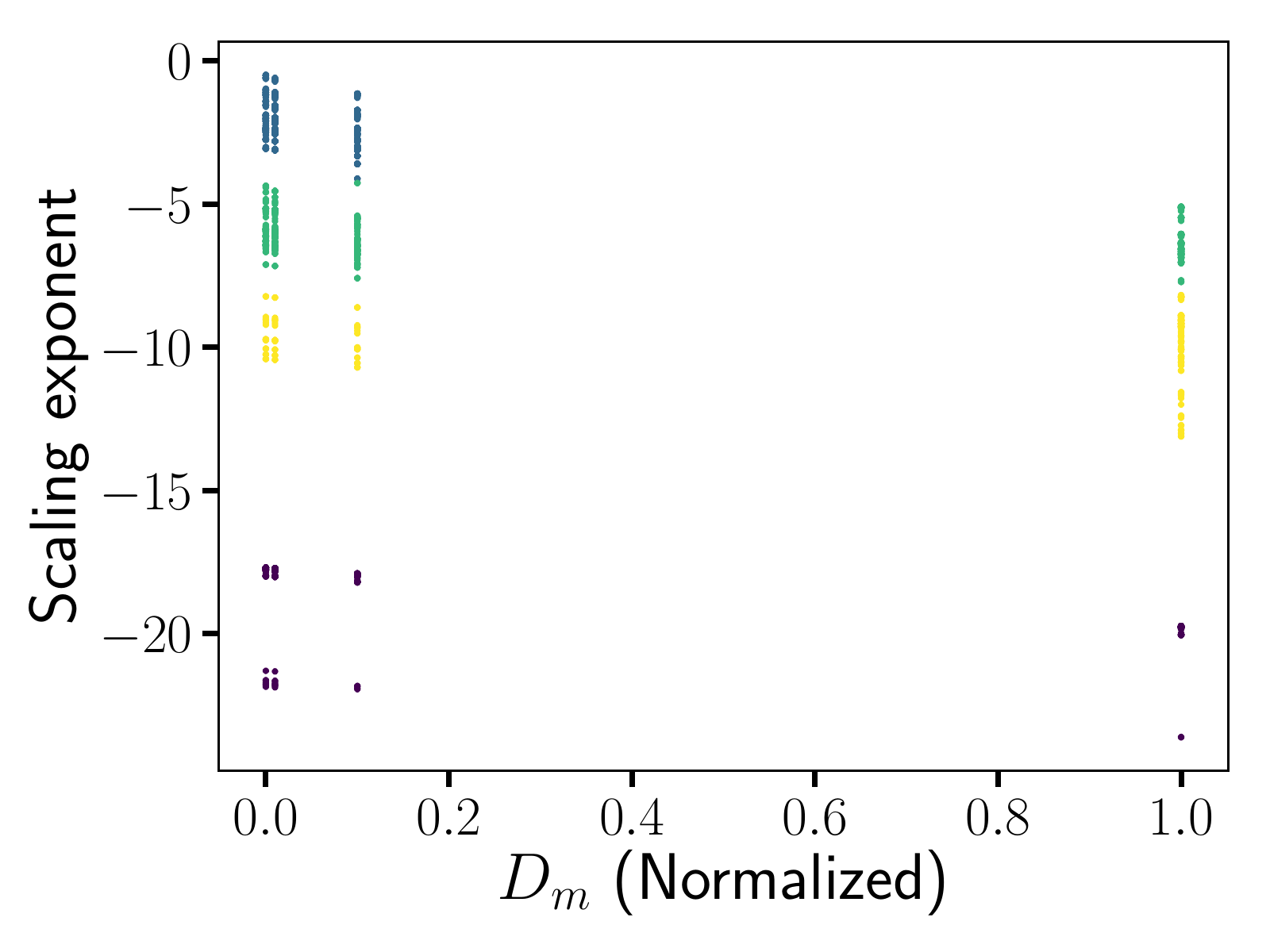}}
  \subfigure[$\sigma^2_B$:~Exponent vs. $D_m$]
    {\includegraphics[width = 0.25\textwidth]
    {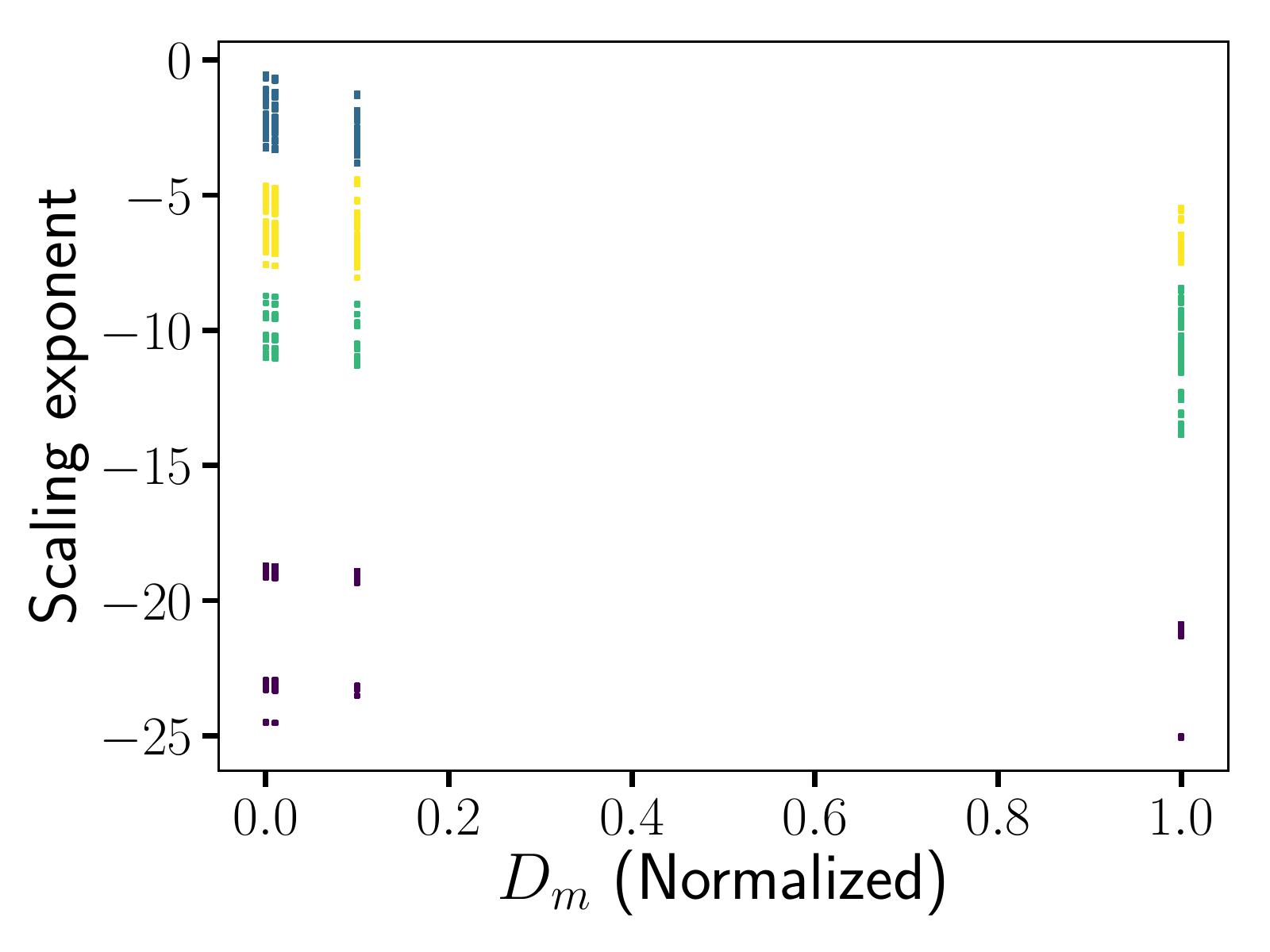}}
  \subfigure[$\sigma^2_C$:~Exponent vs. $D_m$]
    {\includegraphics[width = 0.25\textwidth]
    {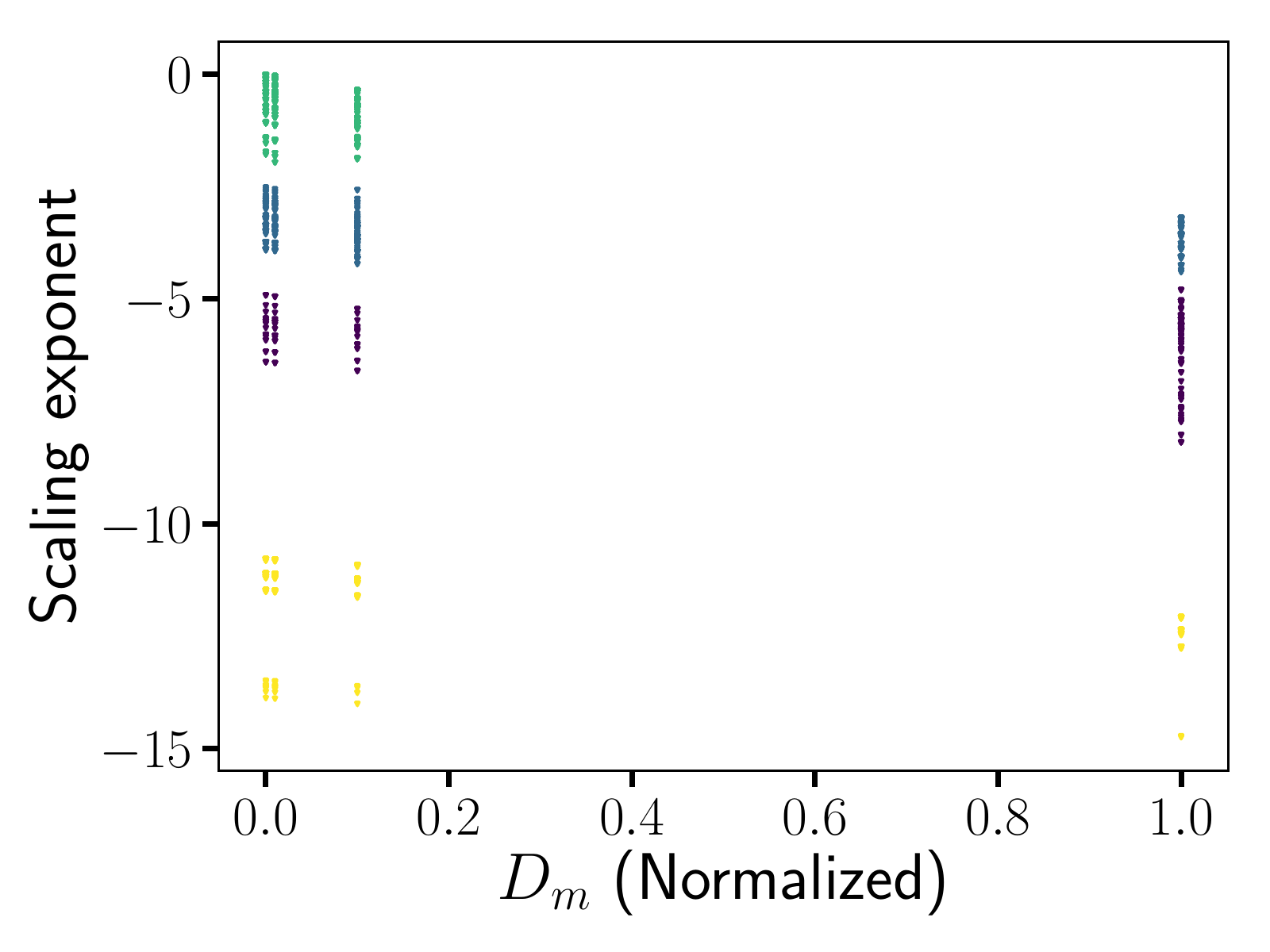}}
  \caption{\textsf{\textbf{\textit{k}-means clustering to identify salient 
    (high and low) mixing features:}}~These figures show the relationship 
    between scaling exponent and input parameters. A main inference from 
    these figures is that for lower values of $\log{[\frac{\alpha_L}{
    \alpha_T}]}$, the scaling exponent is low. Meaning that, enhanced 
    mixing occurs when $\alpha_L \approx \alpha_T$. As the ratio of 
    $\frac{\alpha_L}{\alpha_T}$ increases, we see incomplete mixing due 
    to high anisotropy.
  \label{Fig:Clustering_Analysis}}
\end{figure}

\begin{figure}
  \centering
  \subfigure[Species $A$:~No. of trees = 5]
    {\includegraphics[width = 0.325\textwidth]
    {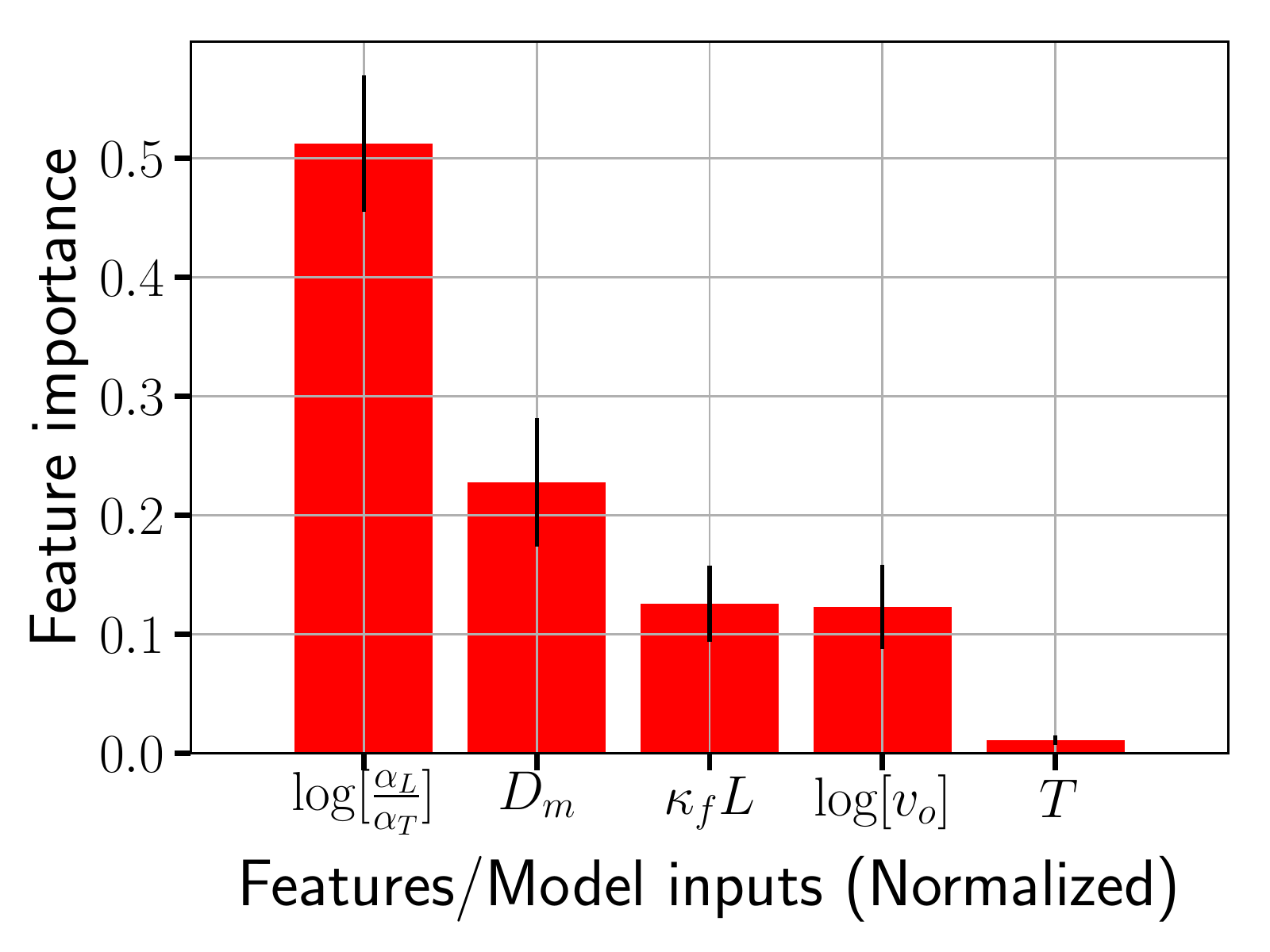}}
  \subfigure[Species $A$:~No. of trees = 100]
    {\includegraphics[width = 0.325\textwidth]
    {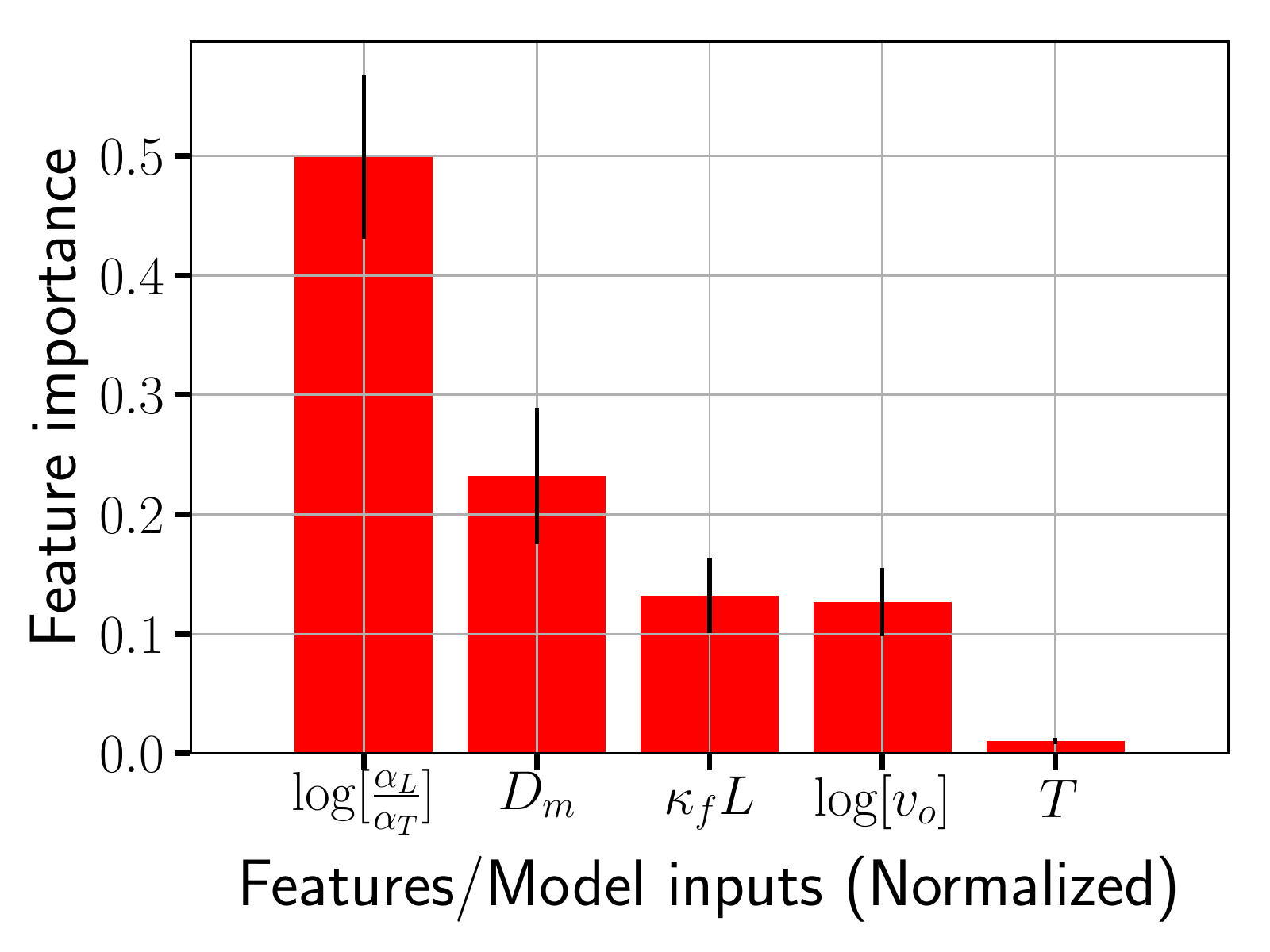}}
  \subfigure[Species $A$:~No. of trees = 250]
    {\includegraphics[width = 0.325\textwidth]
    {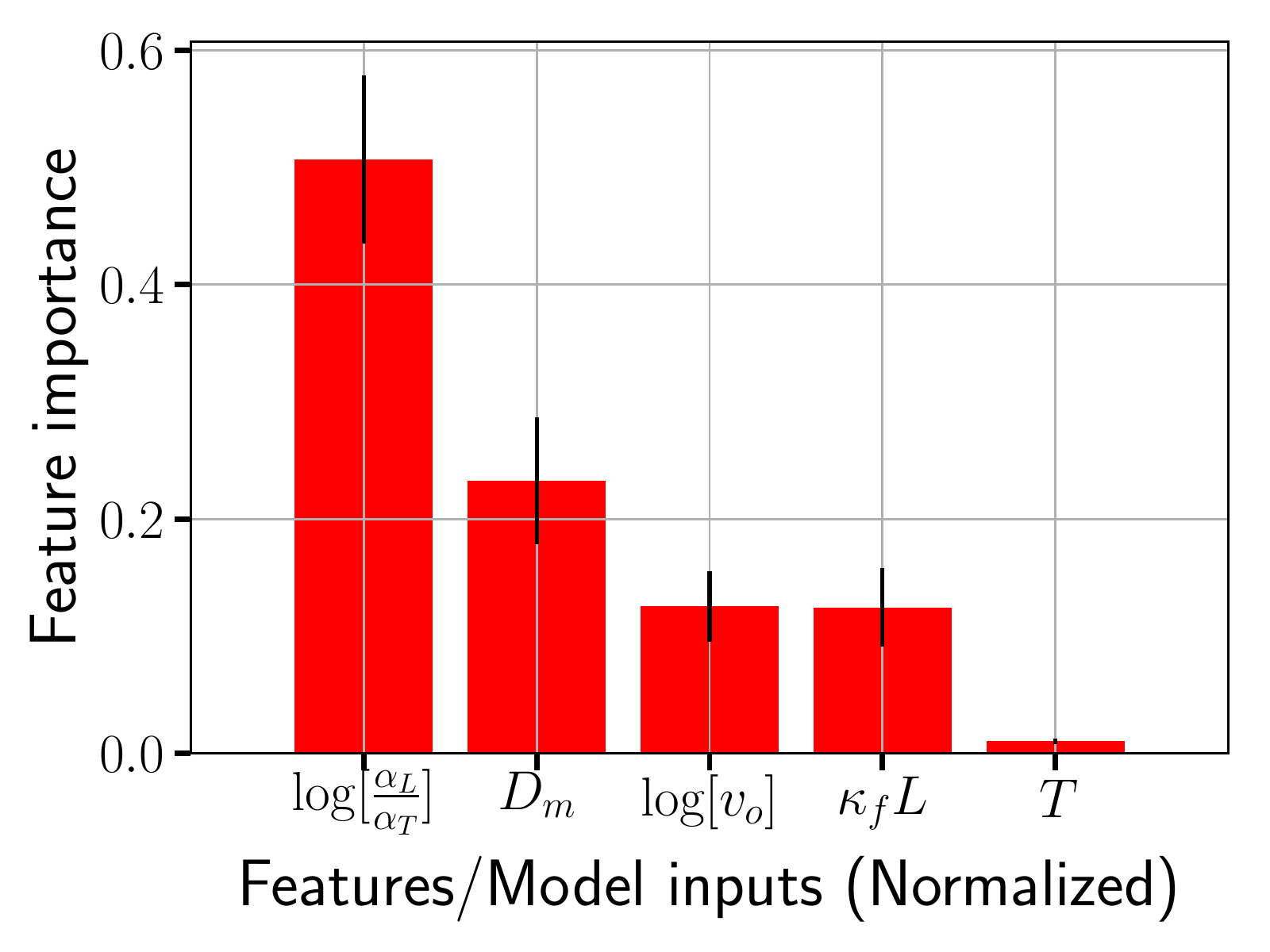}}
  \subfigure[Species $B$:~No. of trees = 5]
    {\includegraphics[width = 0.325\textwidth]
    {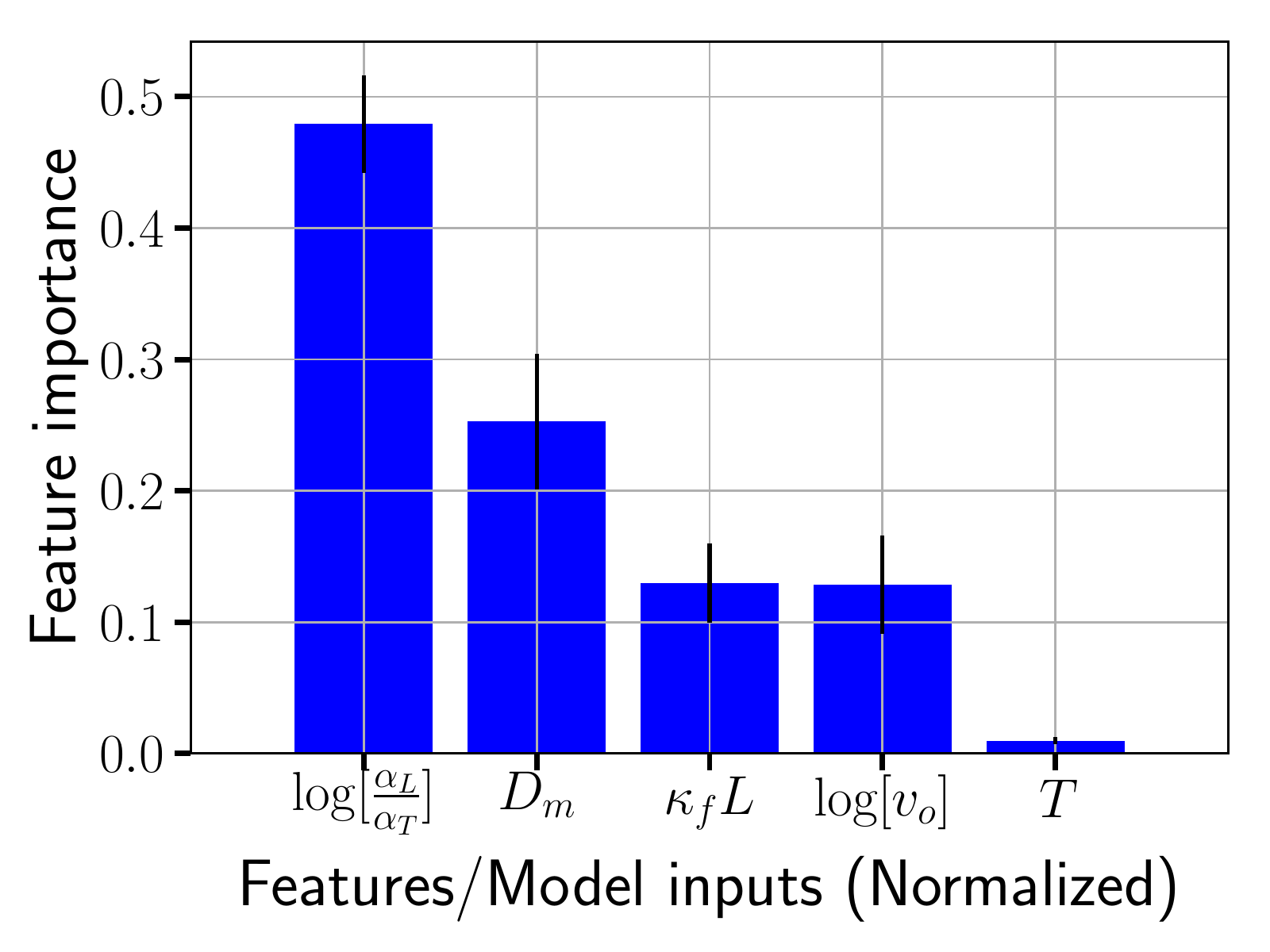}}
  \subfigure[Species $B$:~No. of trees = 100]
    {\includegraphics[width = 0.325\textwidth]
    {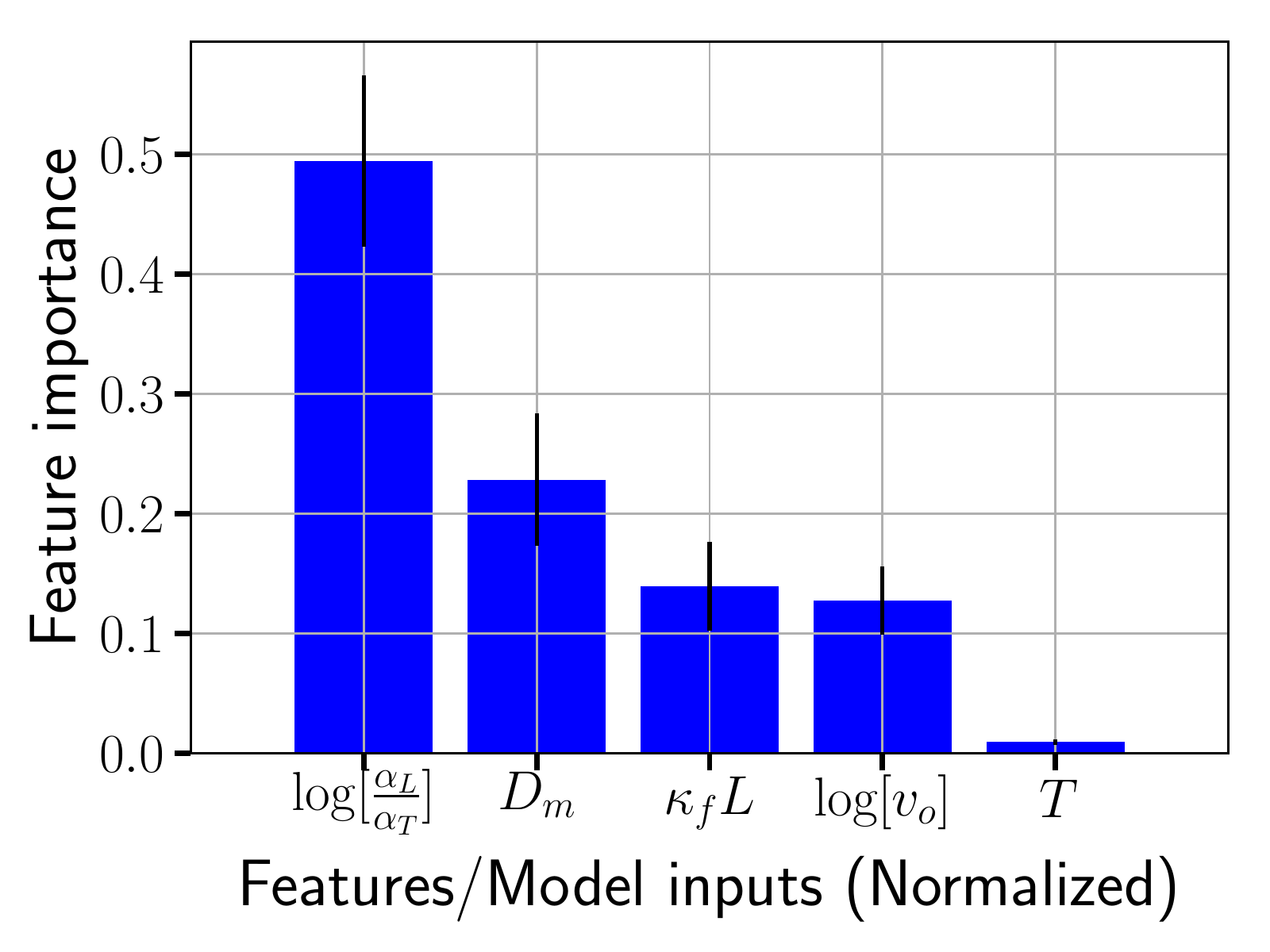}}
  \subfigure[Species $B$:~No. of trees = 250]
    {\includegraphics[width = 0.325\textwidth]
    {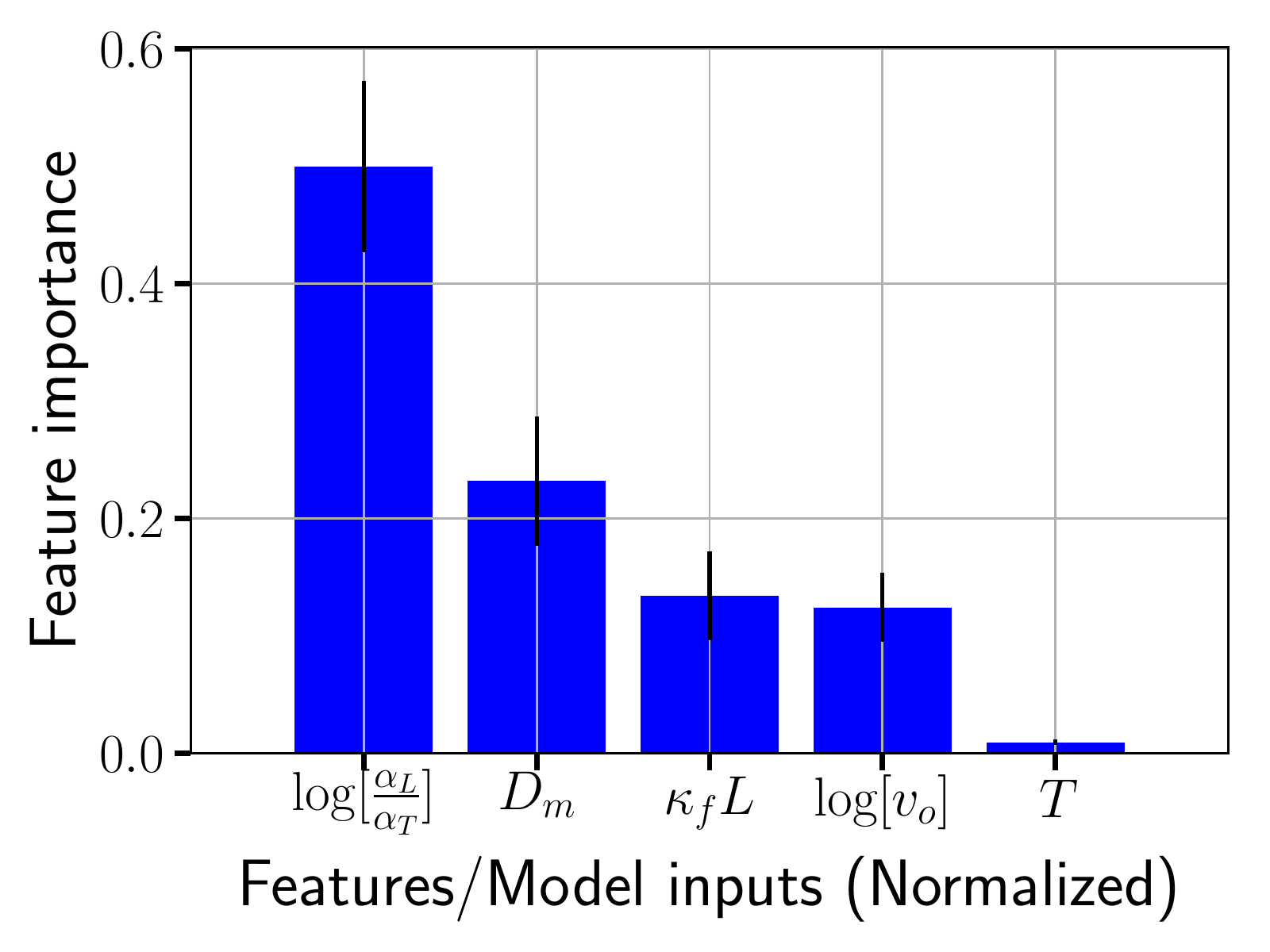}}
  \subfigure[Species $C$:~No. of trees = 5]
    {\includegraphics[width = 0.325\textwidth]
    {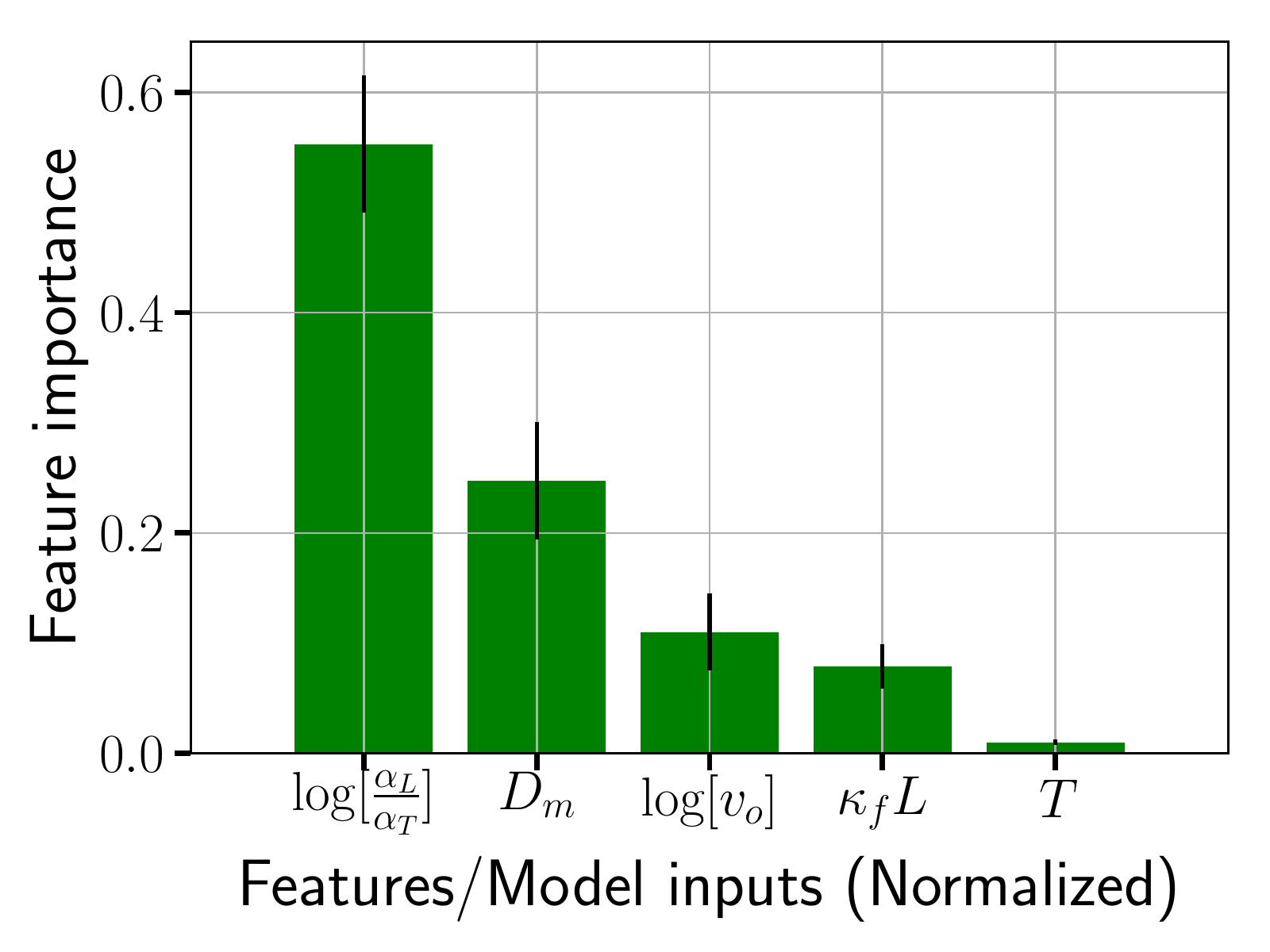}}
  \subfigure[Species $C$:~No. of trees = 100]
    {\includegraphics[width = 0.325\textwidth]
    {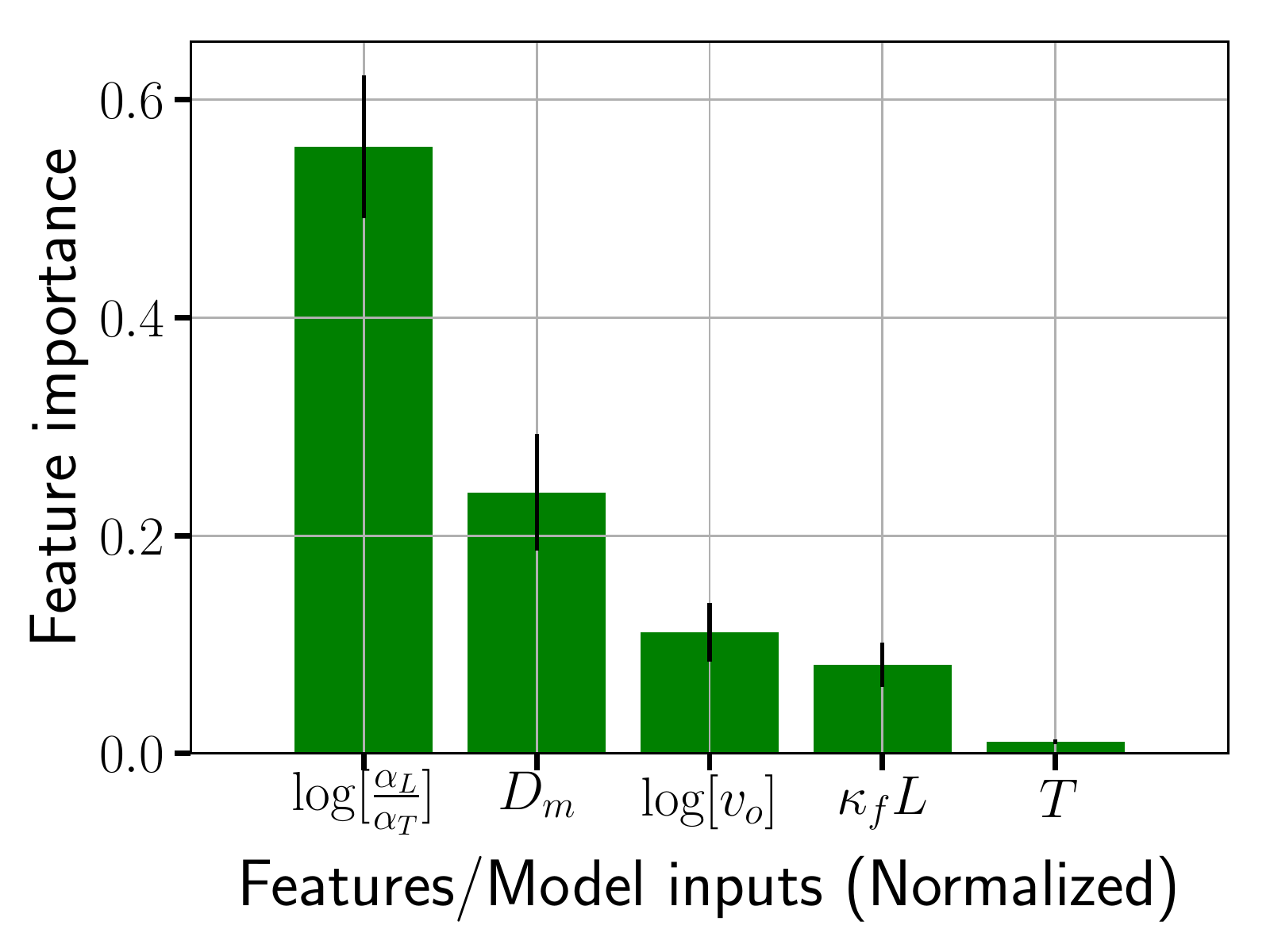}}
  \subfigure[Species $C$:~No. of trees = 250]
    {\includegraphics[width = 0.325\textwidth]
    {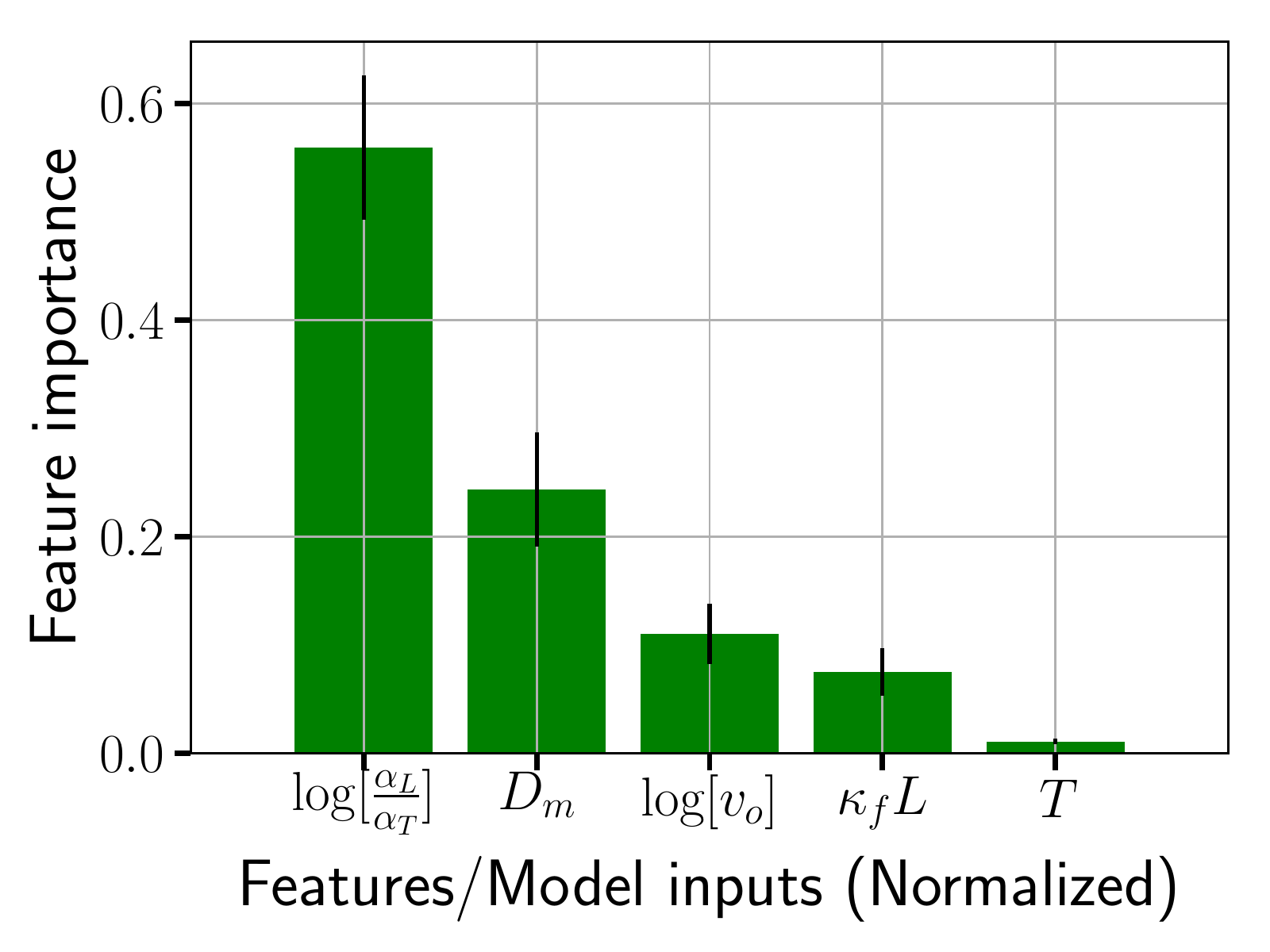}}
  \caption{\textsf{\textbf{Feature importance using Random Forests:}}~These 
    figures show the use of forests of trees, which fit a number of randomized 
    decision trees on different sub-samples of the data to evaluate the relative 
    importance of model inputs. The color bars indicated the feature importances 
    of the random forests, along with their inter-tree variability. It is clear 
    that $\log{[\frac{\alpha_L}{\alpha_T}]}$ is the most important feature and 
    $T$ being the least important. After $\log{[\frac{\alpha_L}{\alpha_T}]}$, 
    $D_m$ is the next informative feature. The parameters $\kappa_fL$ and 
    $\log{\left[v_o \right]}$ are important features after $D_m$.
  \label{Fig:RF_DOM_Feature_Importance}}
\end{figure}

\begin{figure}
  \centering
  \subfigure[Species $A$:~$\mathfrak{c}_A$]
    {\includegraphics[width = 0.325\textwidth]
    {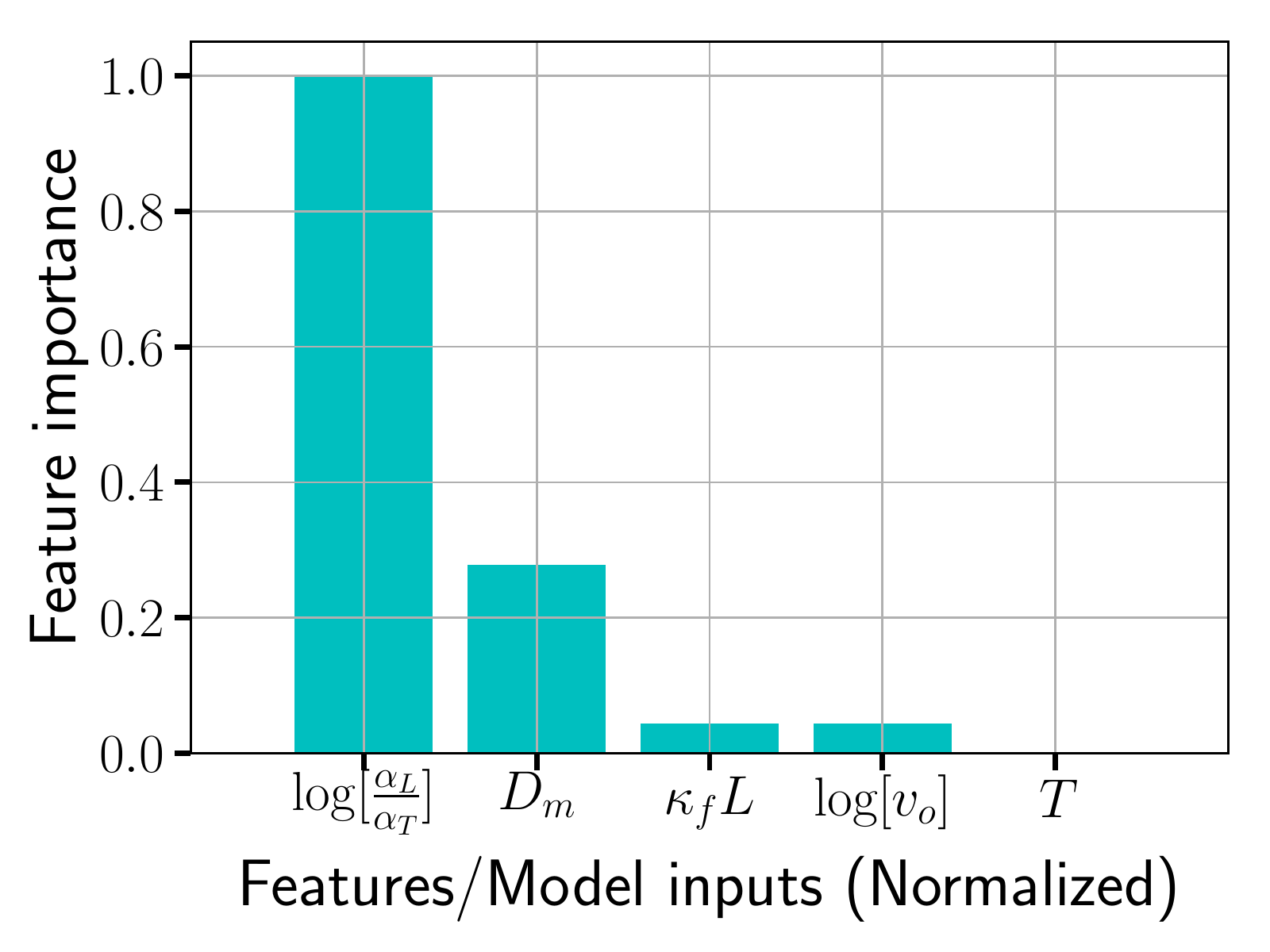}}
  \subfigure[Species $A$:~$\mathbb{c}_A$]
    {\includegraphics[width = 0.325\textwidth]
    {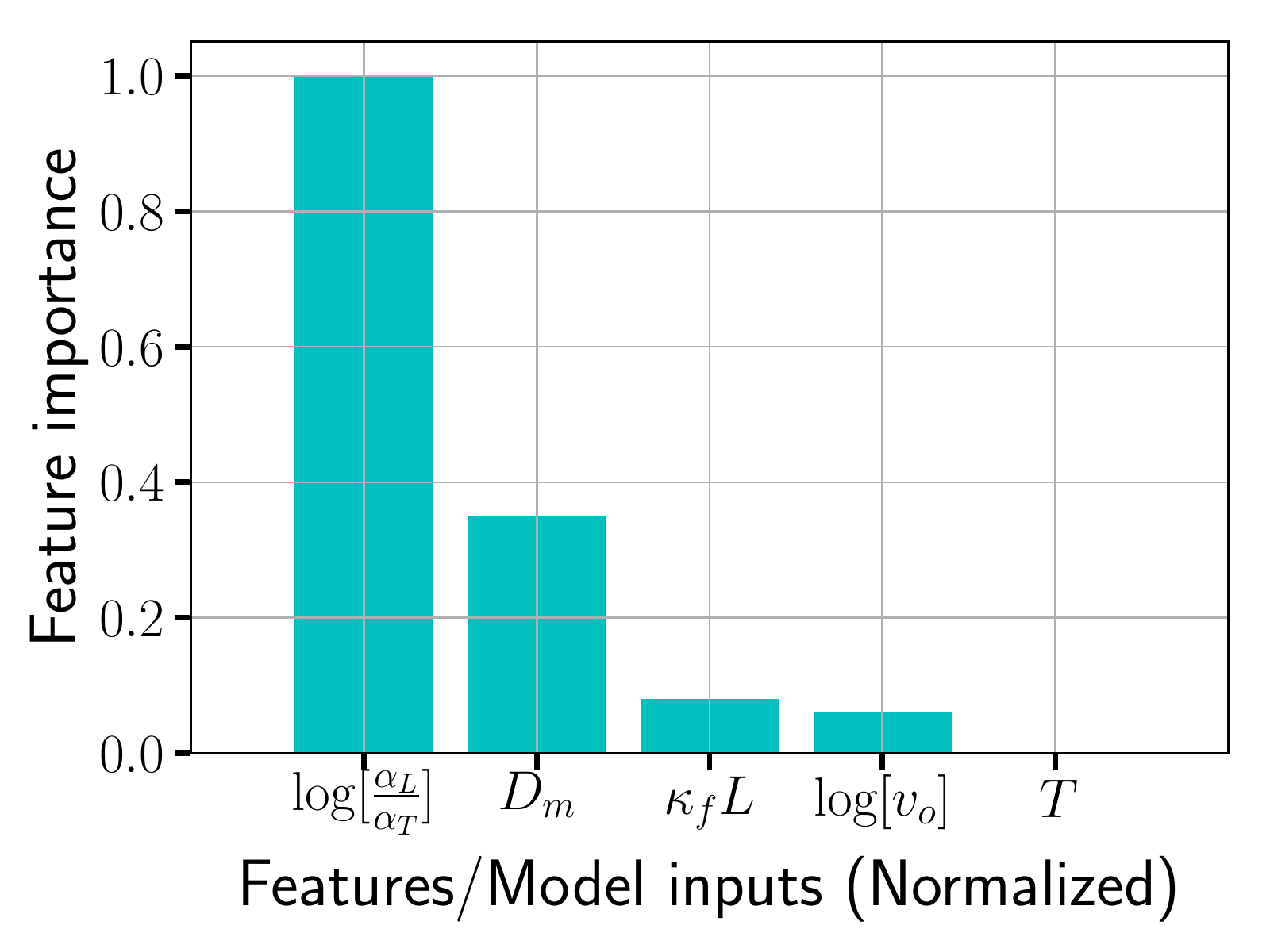}}
  \subfigure[Species $A$:~$\sigma^2_A$]
    {\includegraphics[width = 0.325\textwidth]
    {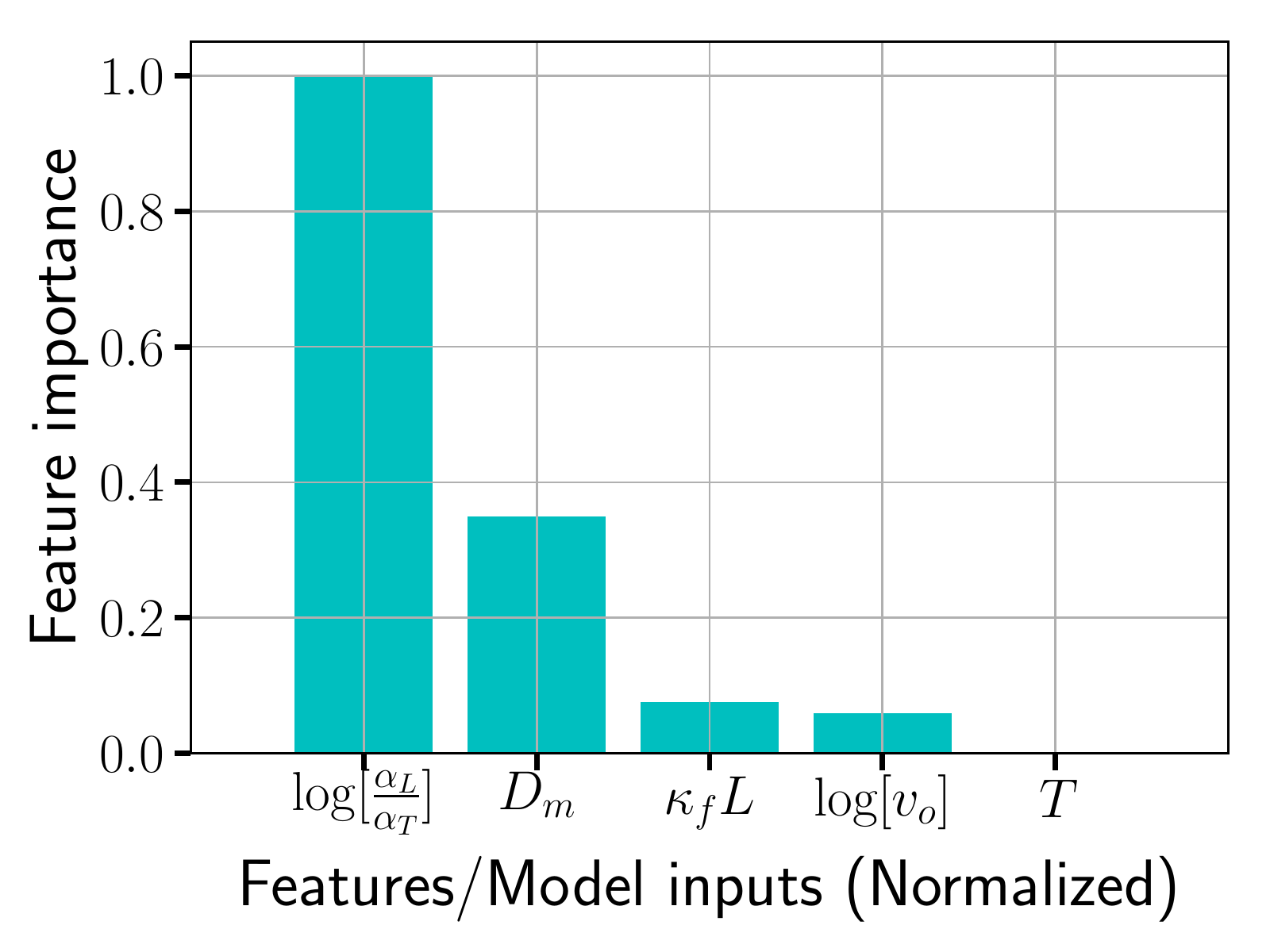}}
  \subfigure[Species $B$:~$\mathfrak{c}_B$]
    {\includegraphics[width = 0.325\textwidth]
    {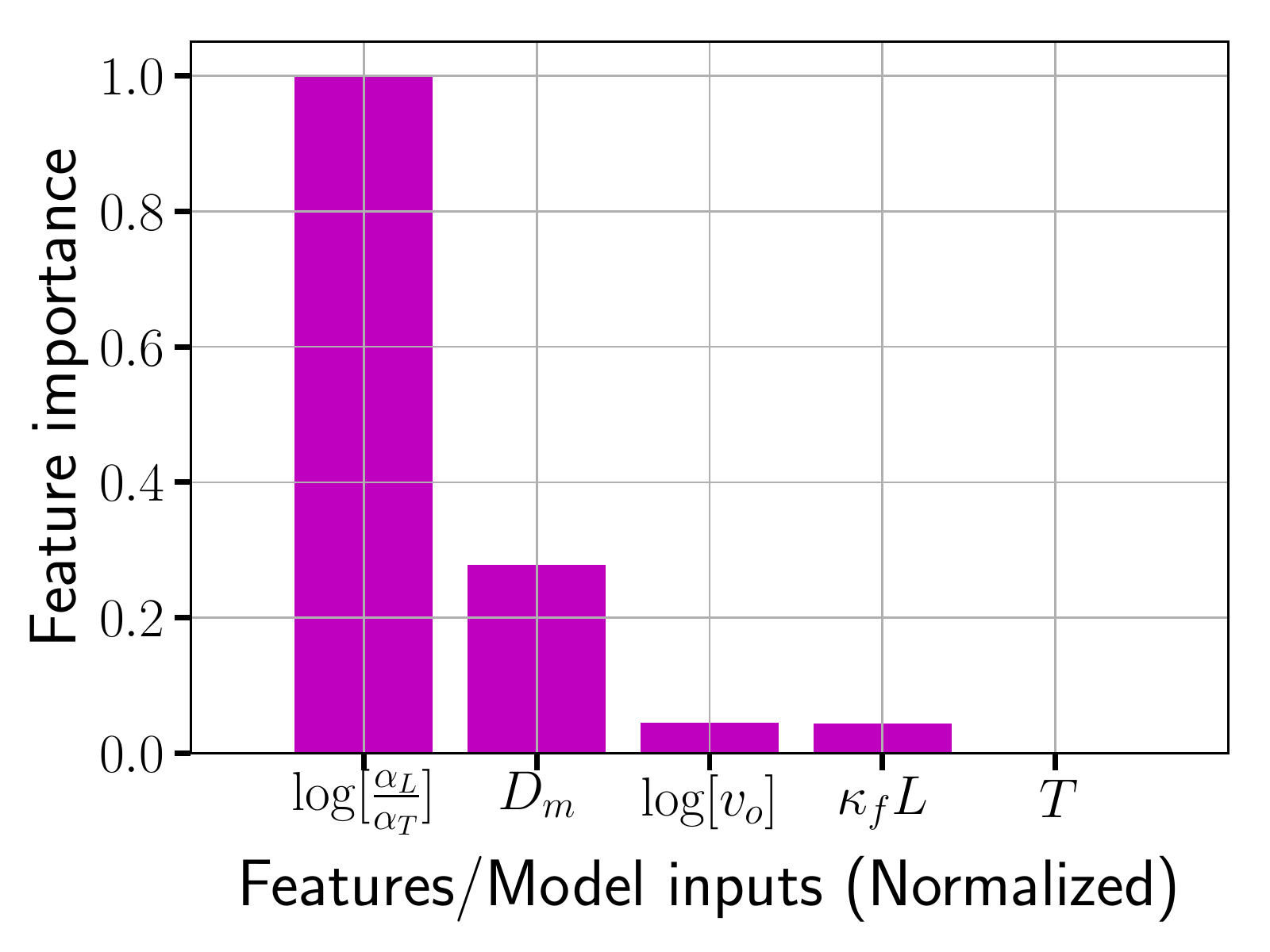}}
  \subfigure[Species $B$:~$\mathbb{c}_B$]
    {\includegraphics[width = 0.325\textwidth]
    {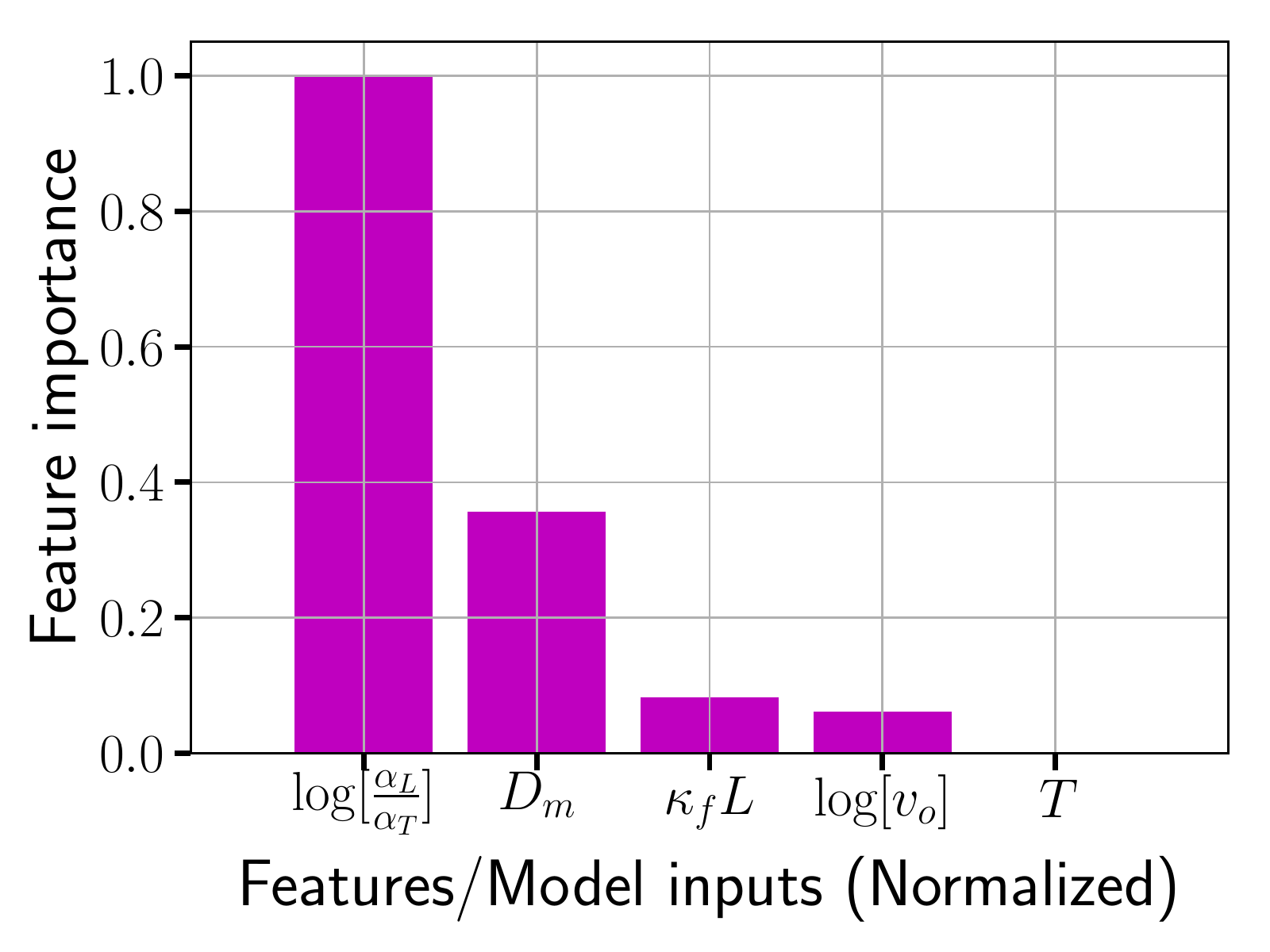}}
  \subfigure[Species $B$:~$\sigma^2_B$]
    {\includegraphics[width = 0.325\textwidth]
    {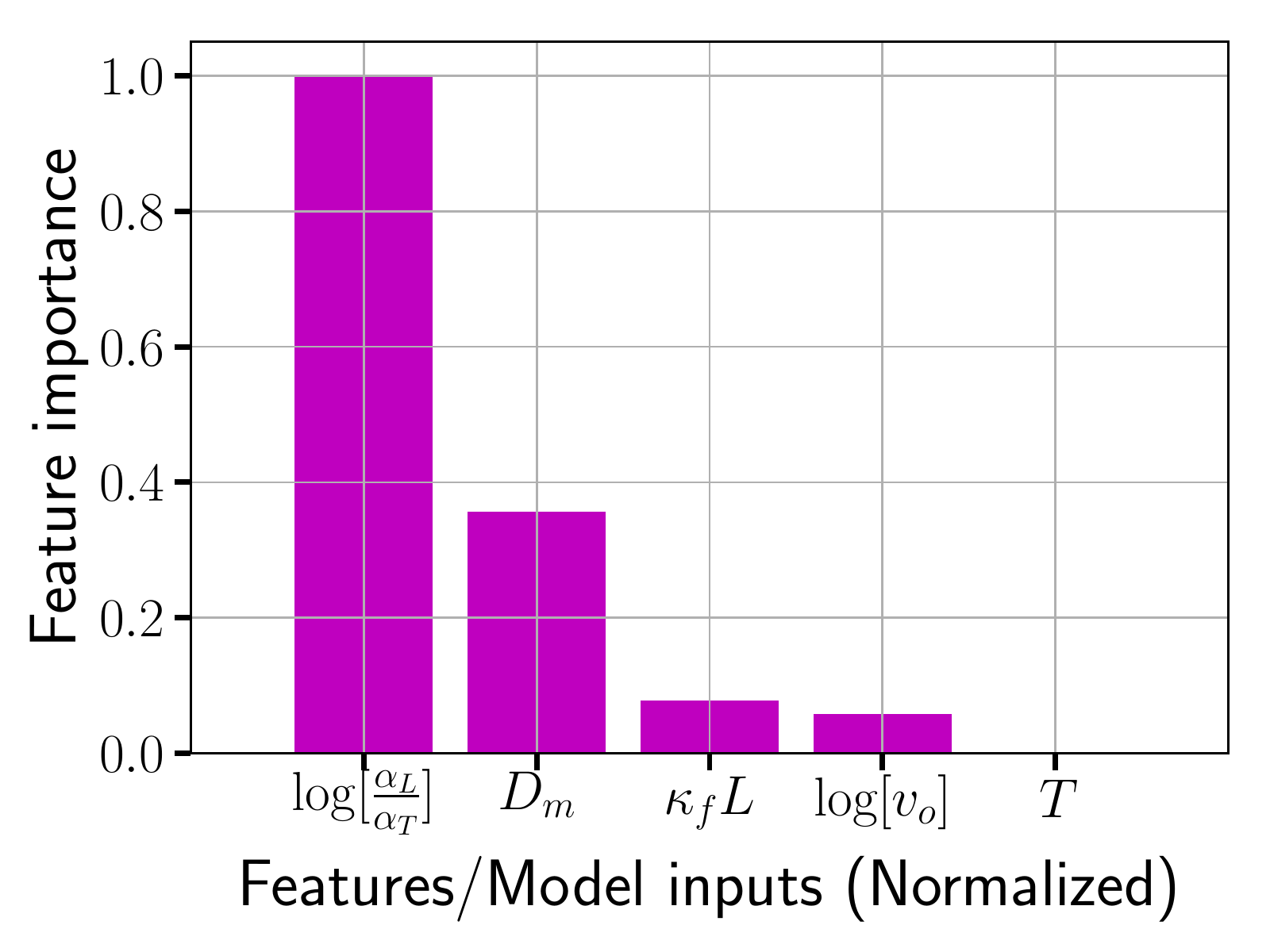}}
  \subfigure[Species $C$:~$\mathfrak{c}_C$]
    {\includegraphics[width = 0.325\textwidth]
    {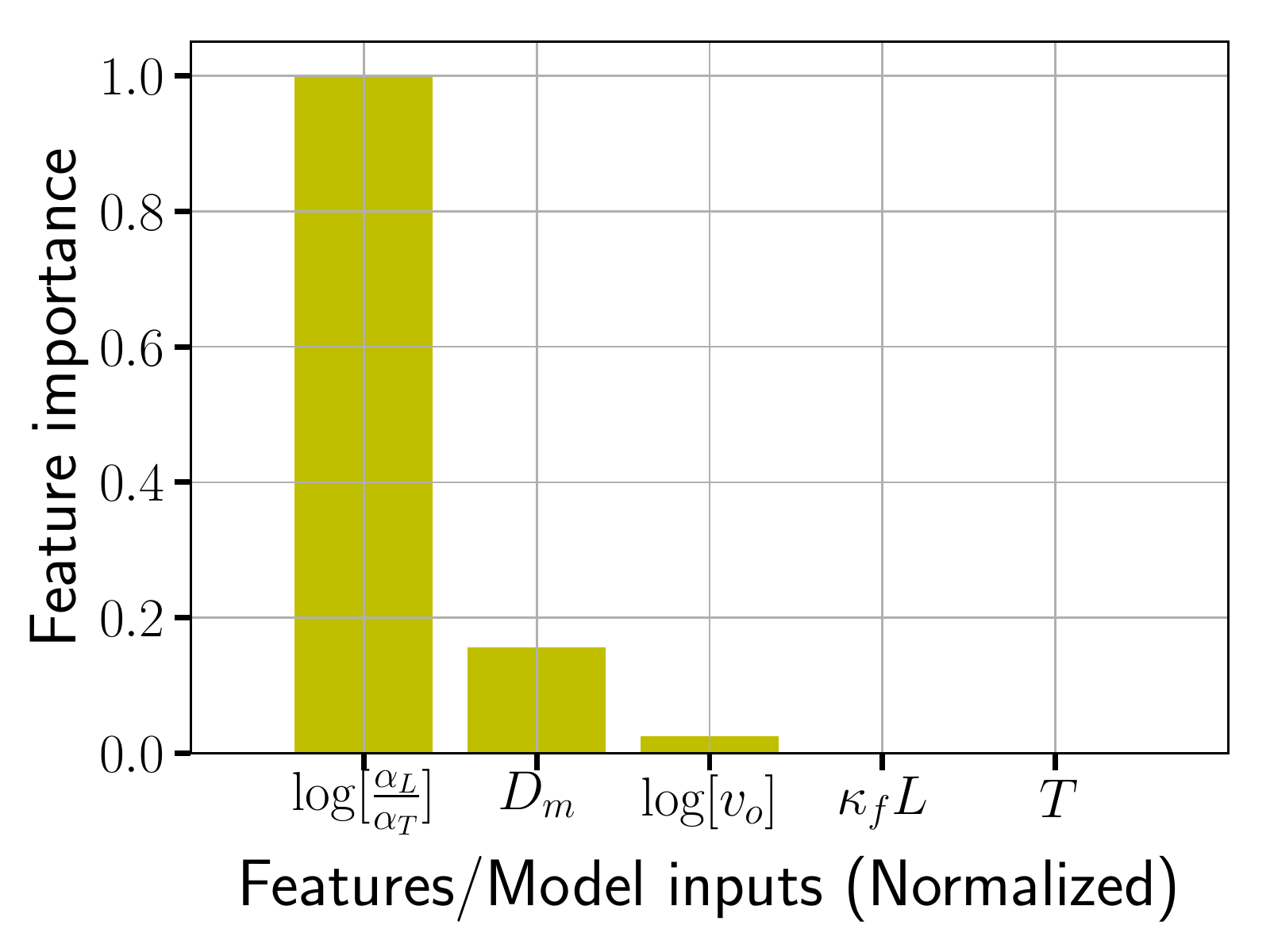}}
  \subfigure[Species $C$:~$\mathbb{c}_C$]
    {\includegraphics[width = 0.325\textwidth]
    {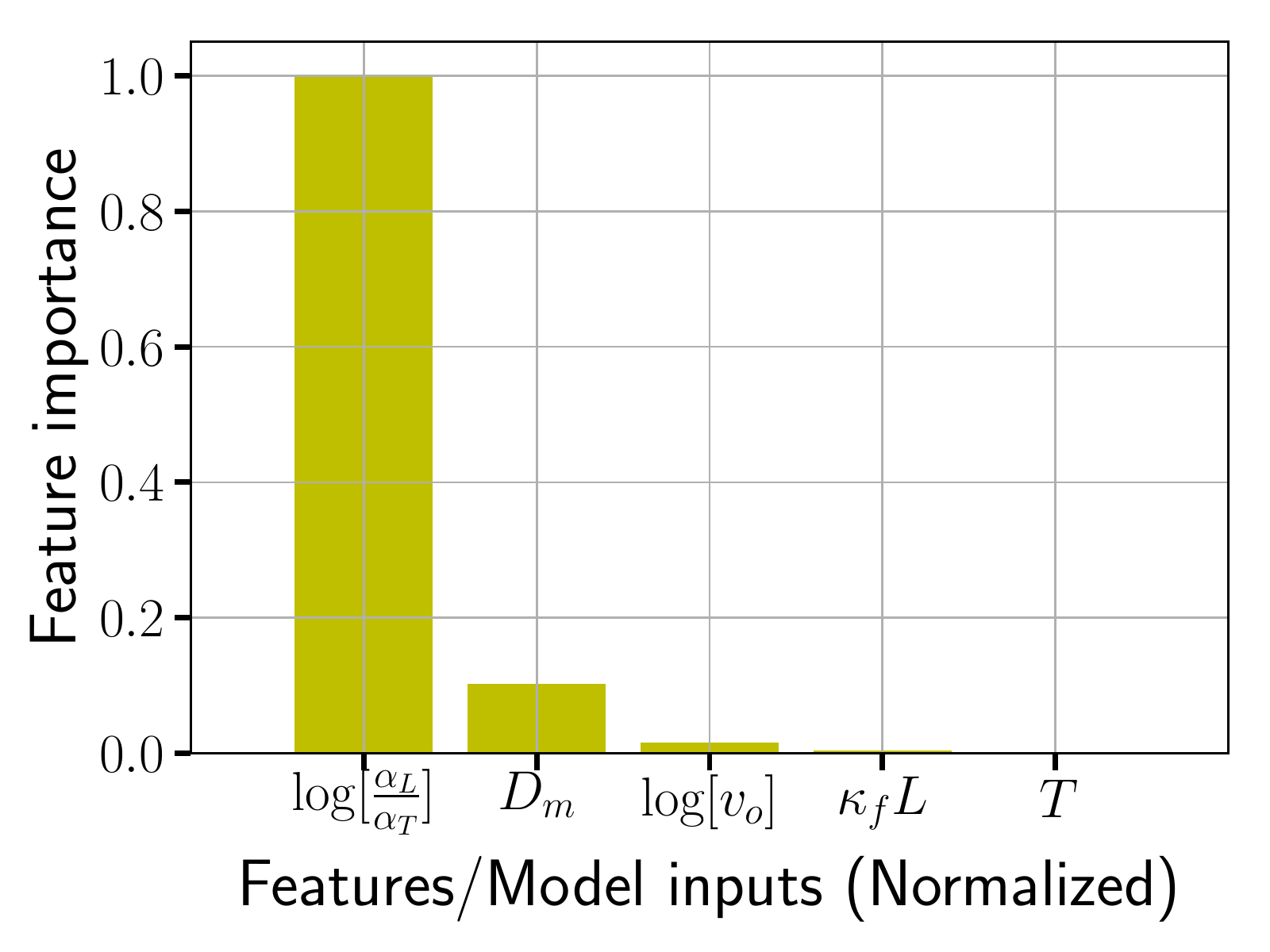}}
  \subfigure[Species $C$:~$\sigma^2_C$]
    {\includegraphics[width = 0.325\textwidth]
    {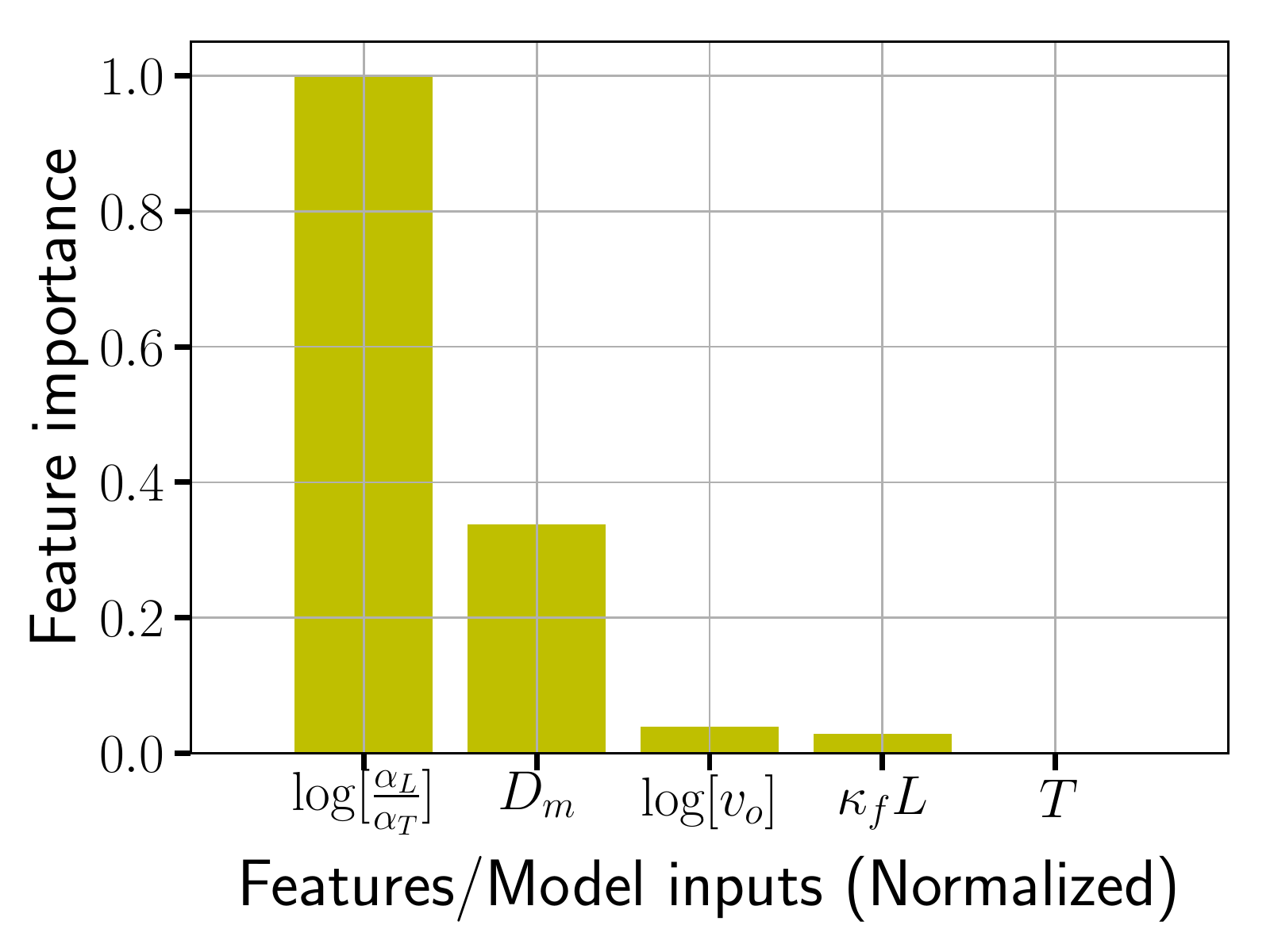}}
  \caption{\textsf{\textbf{Feature importance using F-test:}}~These 
    figures show relative importance of five different input parameters 
    using F-test. Similar to random forests, F-test indicates that $\log{[\frac{\alpha
    _L}{\alpha_T}]}$ and $T$ are the most important and least important 
    features. On the other hand, compared to random forests, F-test shows that 
    $\kappa_fL$ and $\log{\left[v_o \right]}$ are also not relevant.
  \label{Fig:Ftest_Feature_Importance}}
\end{figure}

\begin{figure}
  \centering
  \subfigure[Species $A$:~$\sigma^2_A$]
    {\includegraphics[width = 0.325\textwidth]
    {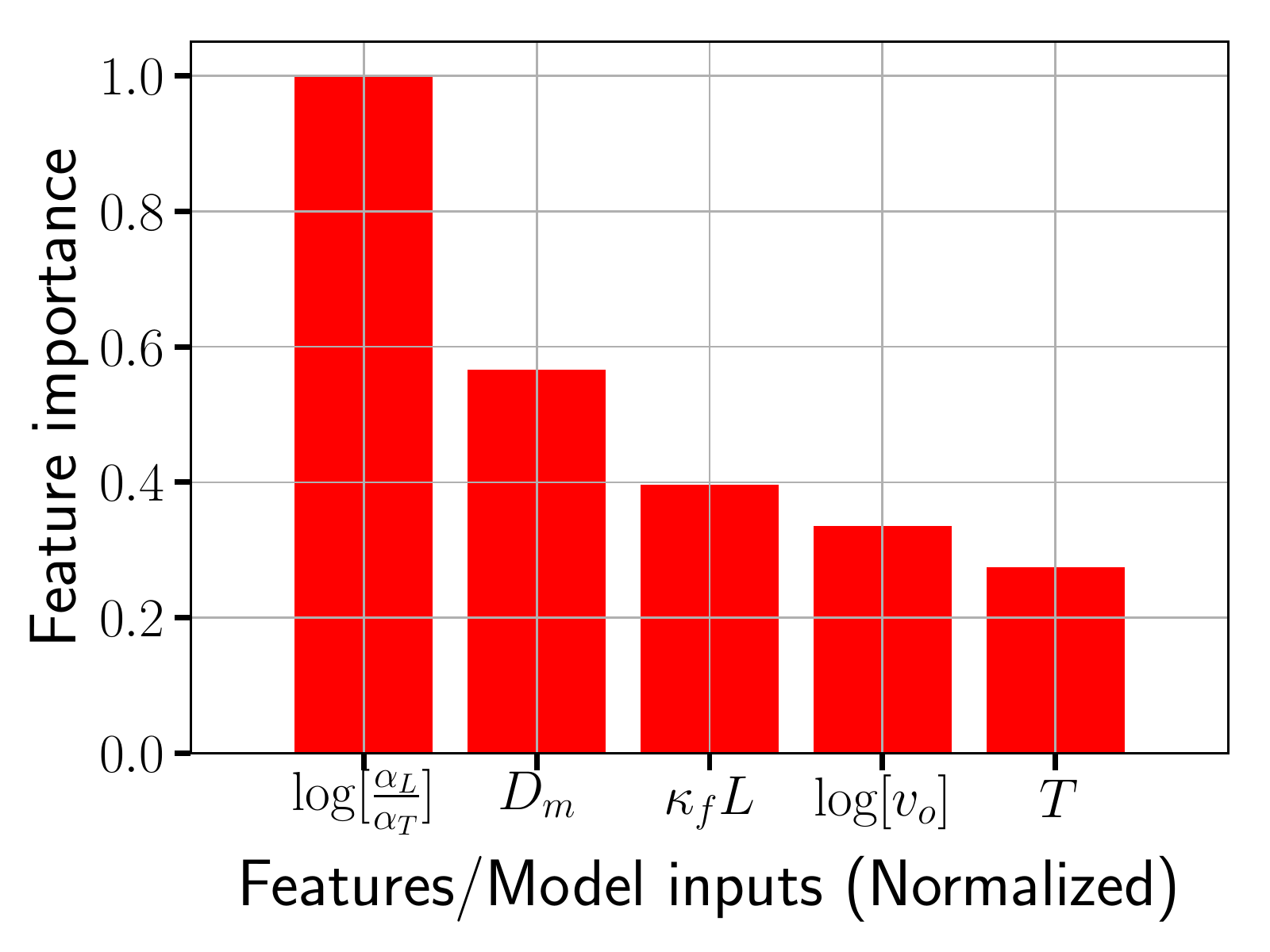}}
  \subfigure[Species $B$:~$\sigma^2_B$]
    {\includegraphics[width = 0.325\textwidth]
    {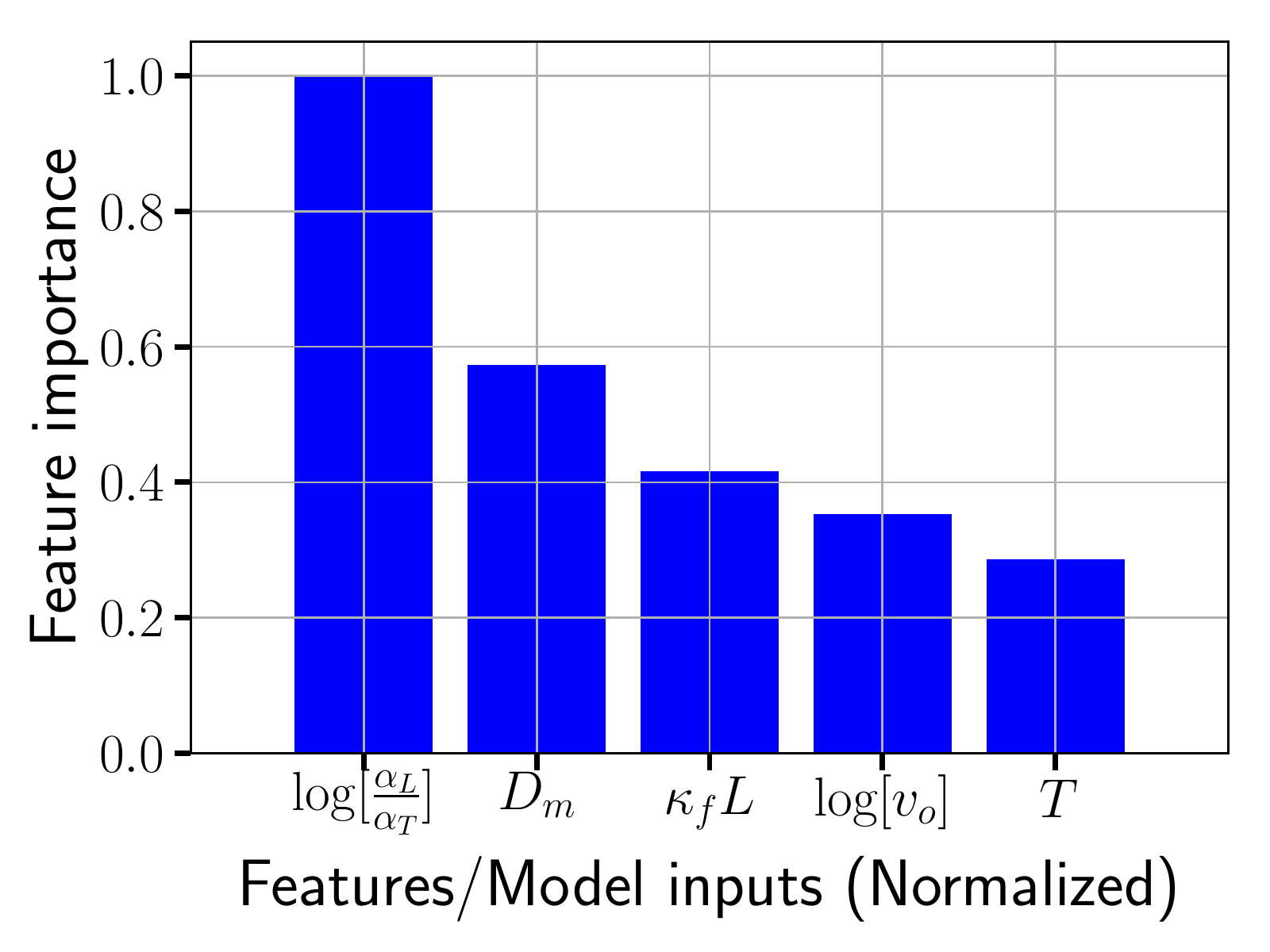}}
  \subfigure[Species $C$:~$\sigma^2_C$]
    {\includegraphics[width = 0.325\textwidth]
    {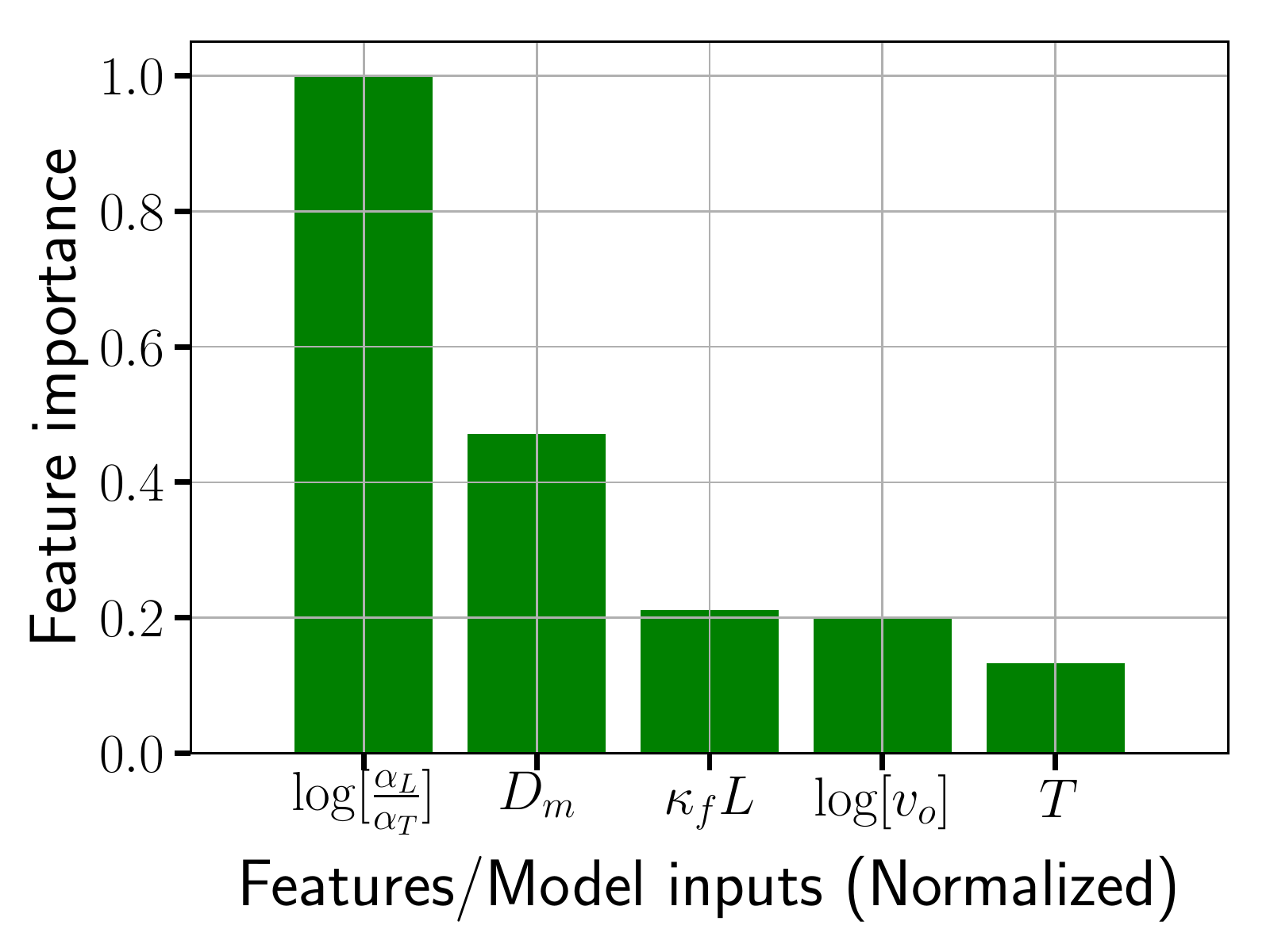}}
  \caption{\textsf{\textbf{Feature importance using MI-criteria:}}~These 
    figures show the relative importance of model inputs using mutual 
    information criteria for degree of mixing of species $A$, $B$, and $C$.
    Similar to random forests and F-test, MI-criteria shows that $\log{[\frac{
    \alpha_L}{\alpha_T}]}$ is the most important feature among all input parameters. 
    Next important features are in the following order $D_m$, $\kappa_fL$, $\log{\left[v_o 
    \right]}$, and $T$. One can observe a steep decrease in the feature importance 
    from $\log{[\frac{\alpha_L}{\alpha_T}]}$ to $D_m$, which is accordance with 
    random forests and F-test. However, there is a gradual decrease in feature 
    importance from $D_m$ to $T$, which is in contrast with feature importances 
    from random forests and F-test.
  \label{Fig:MIC_Feature_Importance}}
\end{figure}

\begin{figure}
  \centering
  \subfigure[Species $A$:~$\mathfrak{c}_A$]
    {\includegraphics[width = 0.325\textwidth]
    {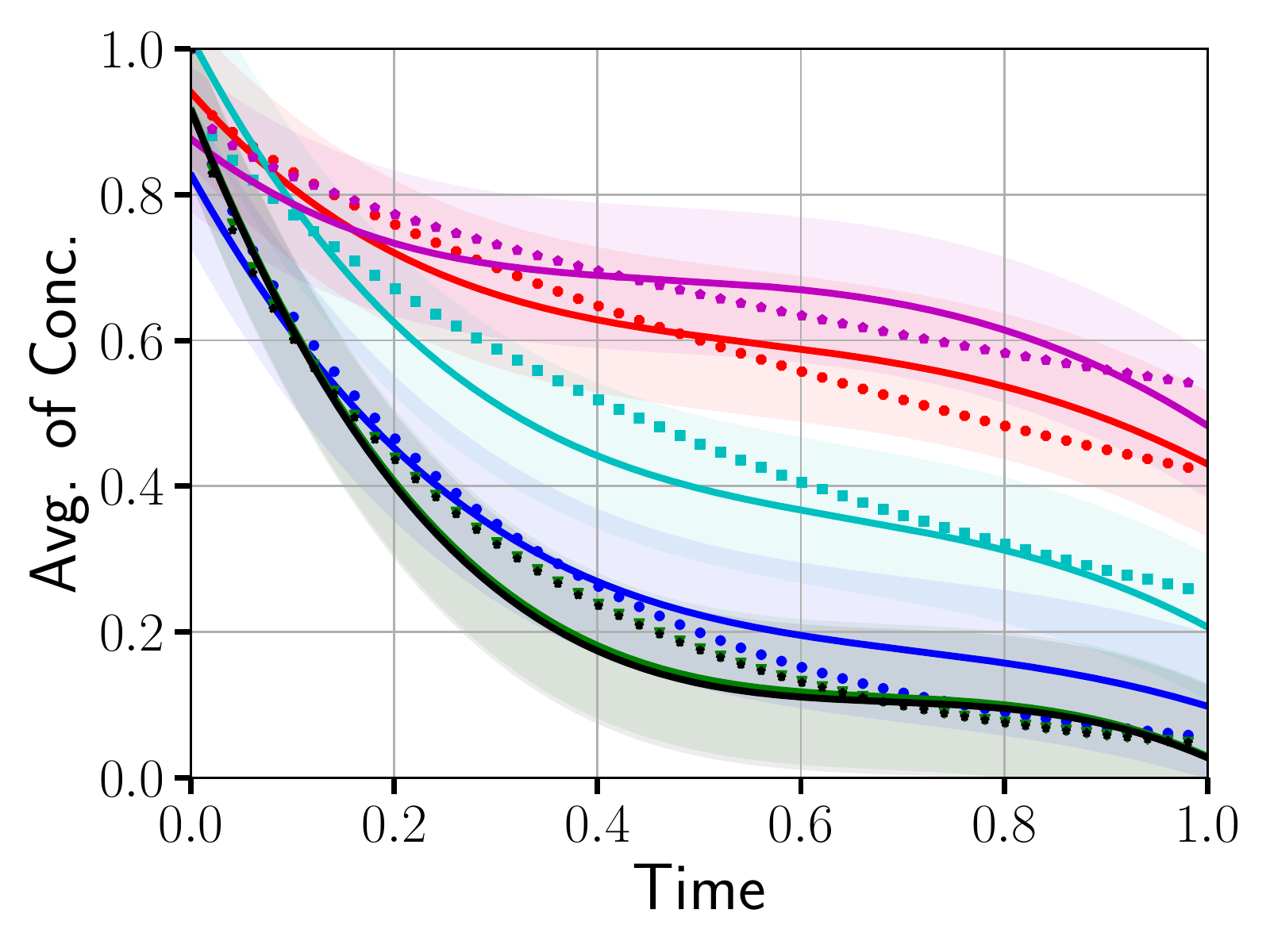}}
  \subfigure[Species $A$:~$\mathbb{c}_A$]
    {\includegraphics[width = 0.325\textwidth]
    {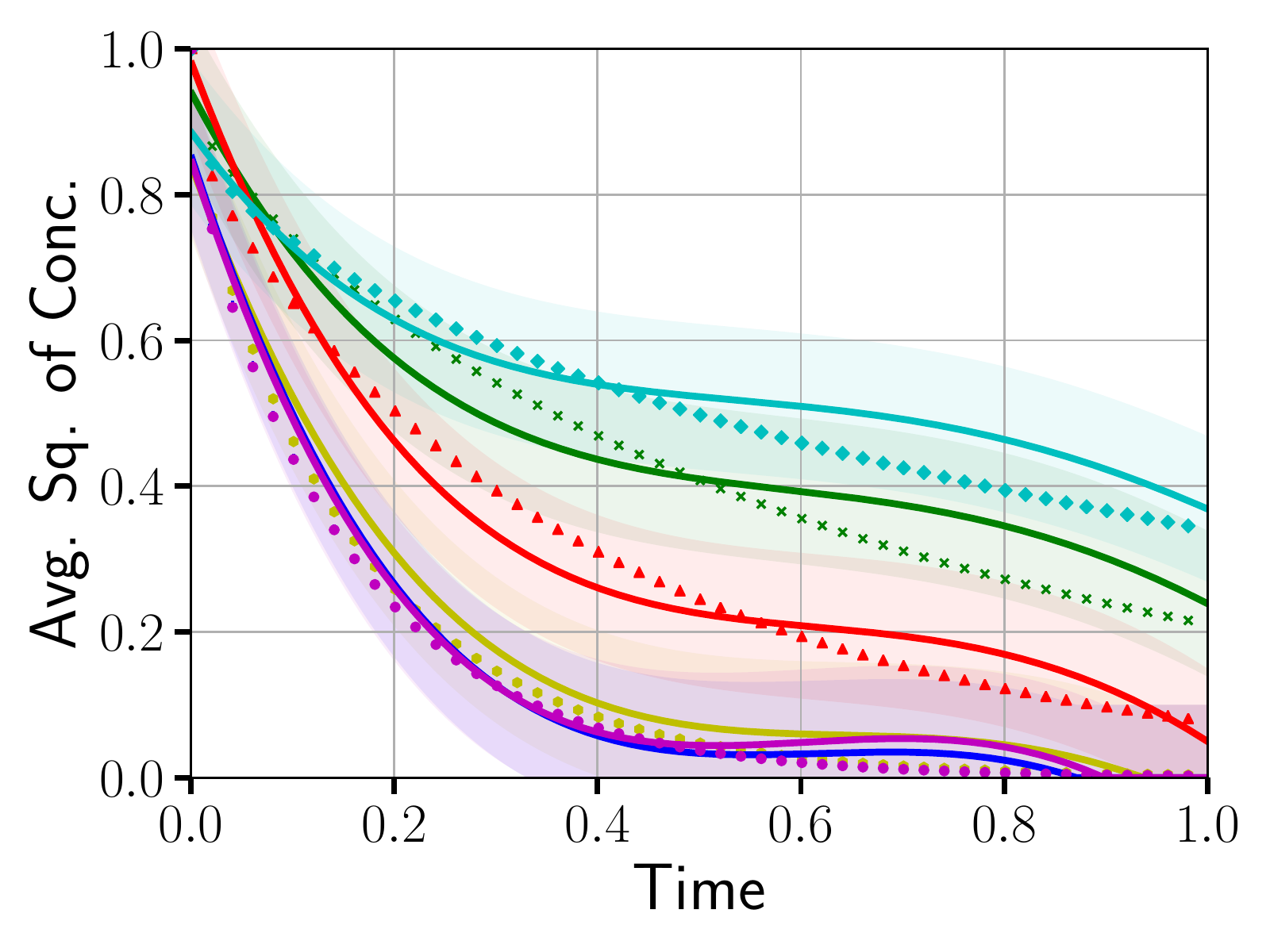}}
  \subfigure[Species $A$:~$\sigma^2_A$]
    {\includegraphics[width = 0.325\textwidth]
    {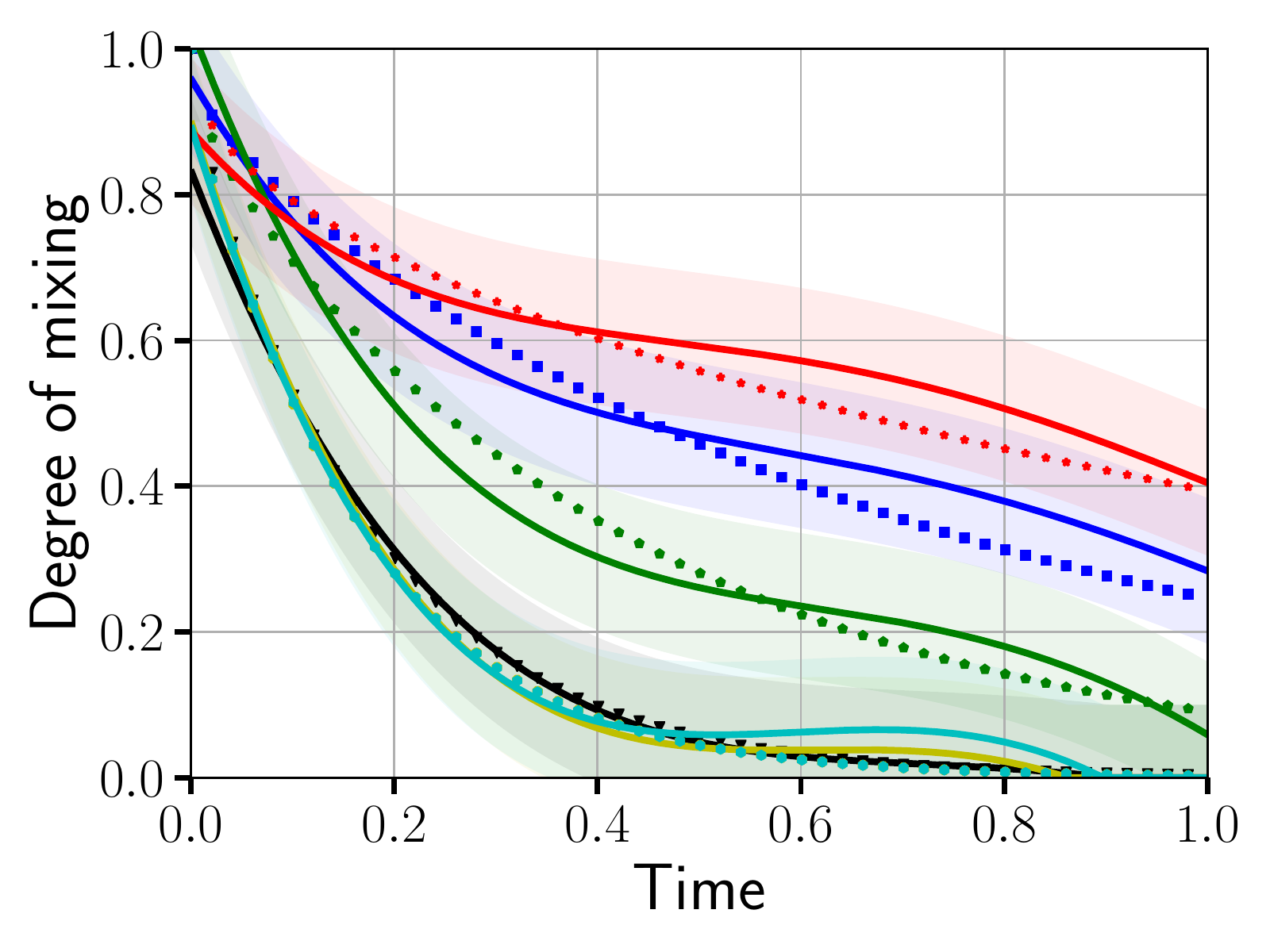}}
  \subfigure[Species $B$:~$\mathfrak{c}_B$]
    {\includegraphics[width = 0.325\textwidth]
    {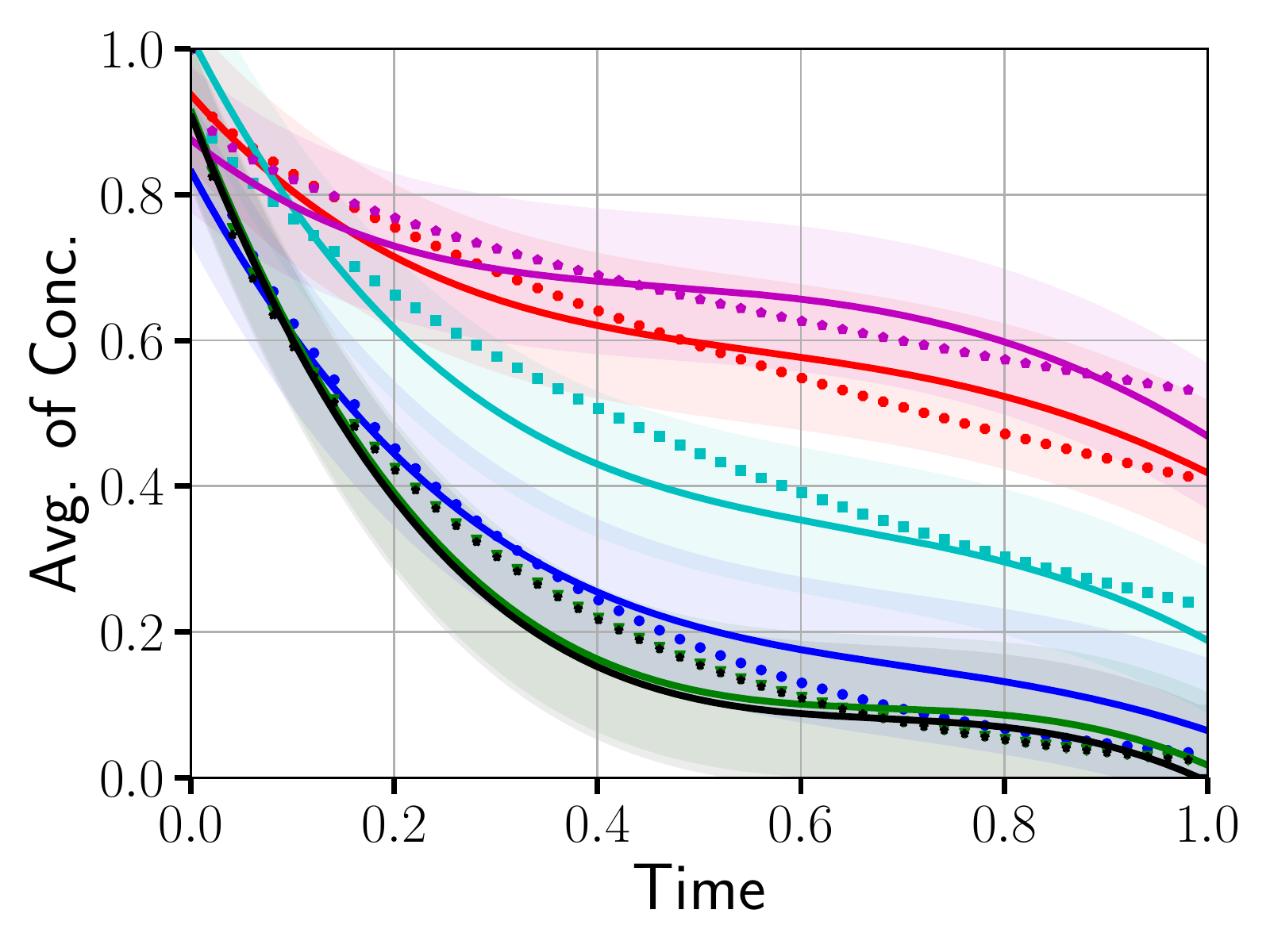}}
  \subfigure[Species $B$:~$\mathbb{c}_B$]
    {\includegraphics[width = 0.325\textwidth]
    {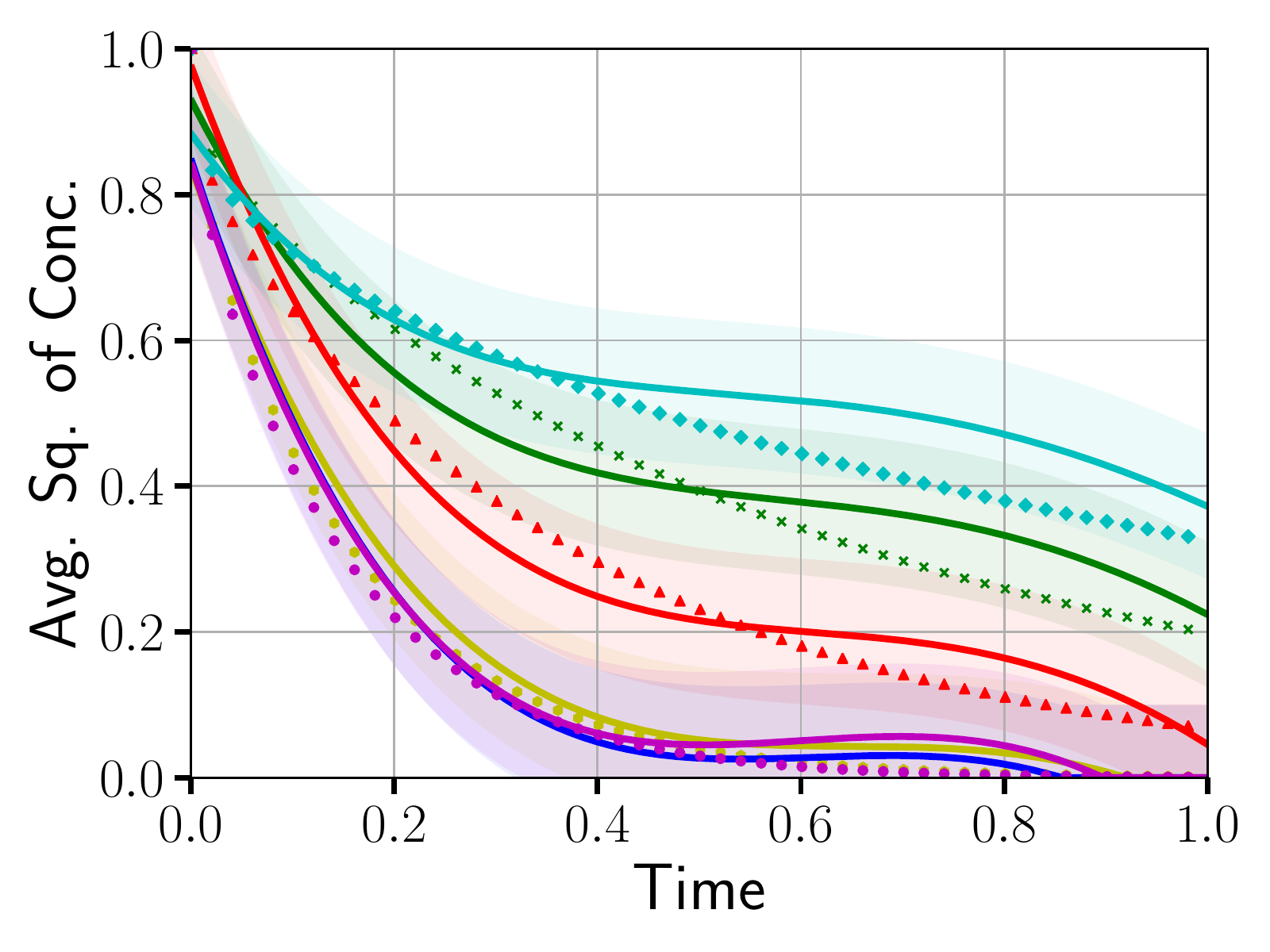}}
  \subfigure[Species $B$:~$\sigma^2_B$]
    {\includegraphics[width = 0.325\textwidth]
    {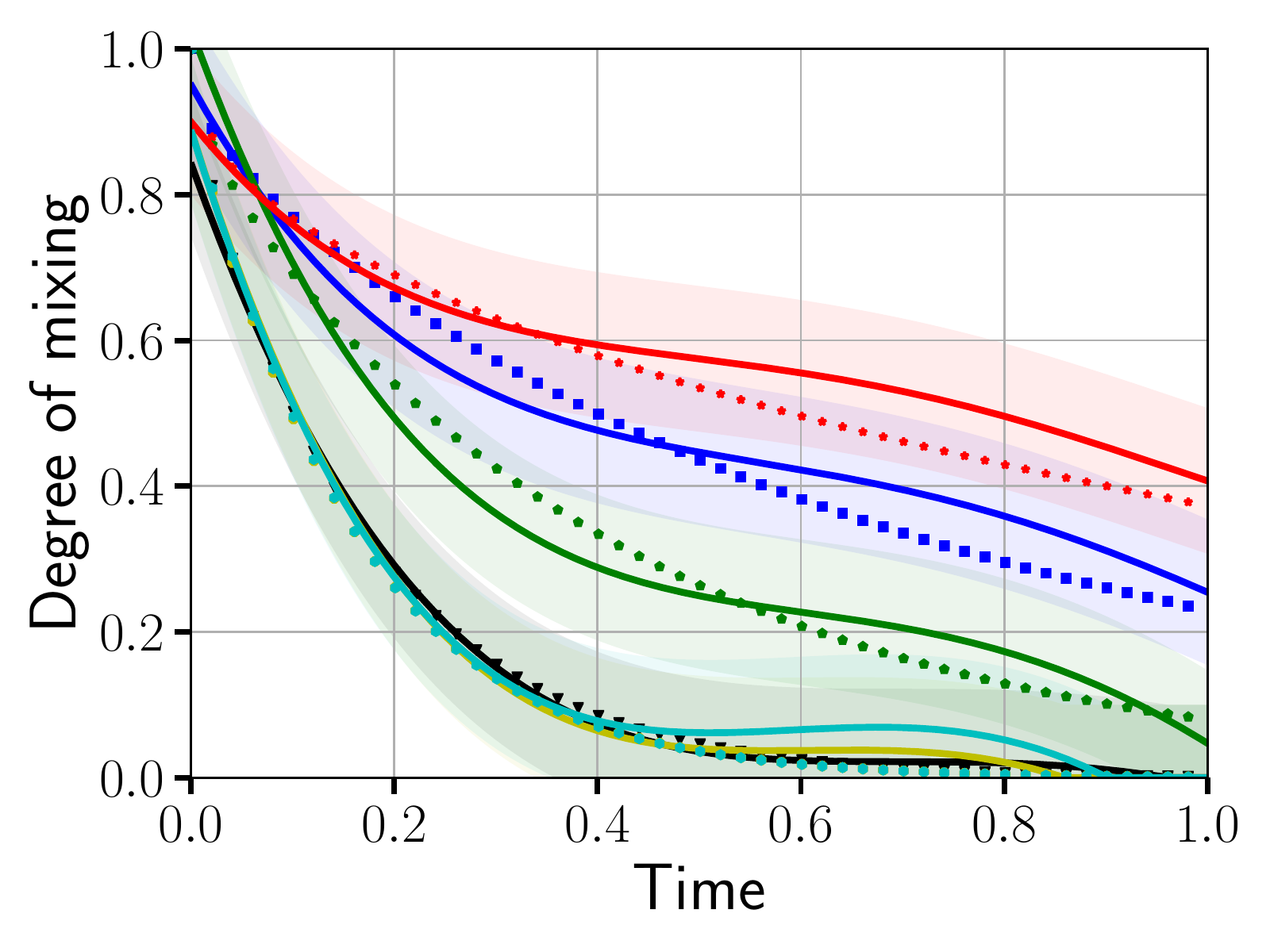}}
  \subfigure[Species $C$:~$\mathfrak{c}_C$]
    {\includegraphics[width = 0.325\textwidth]
    {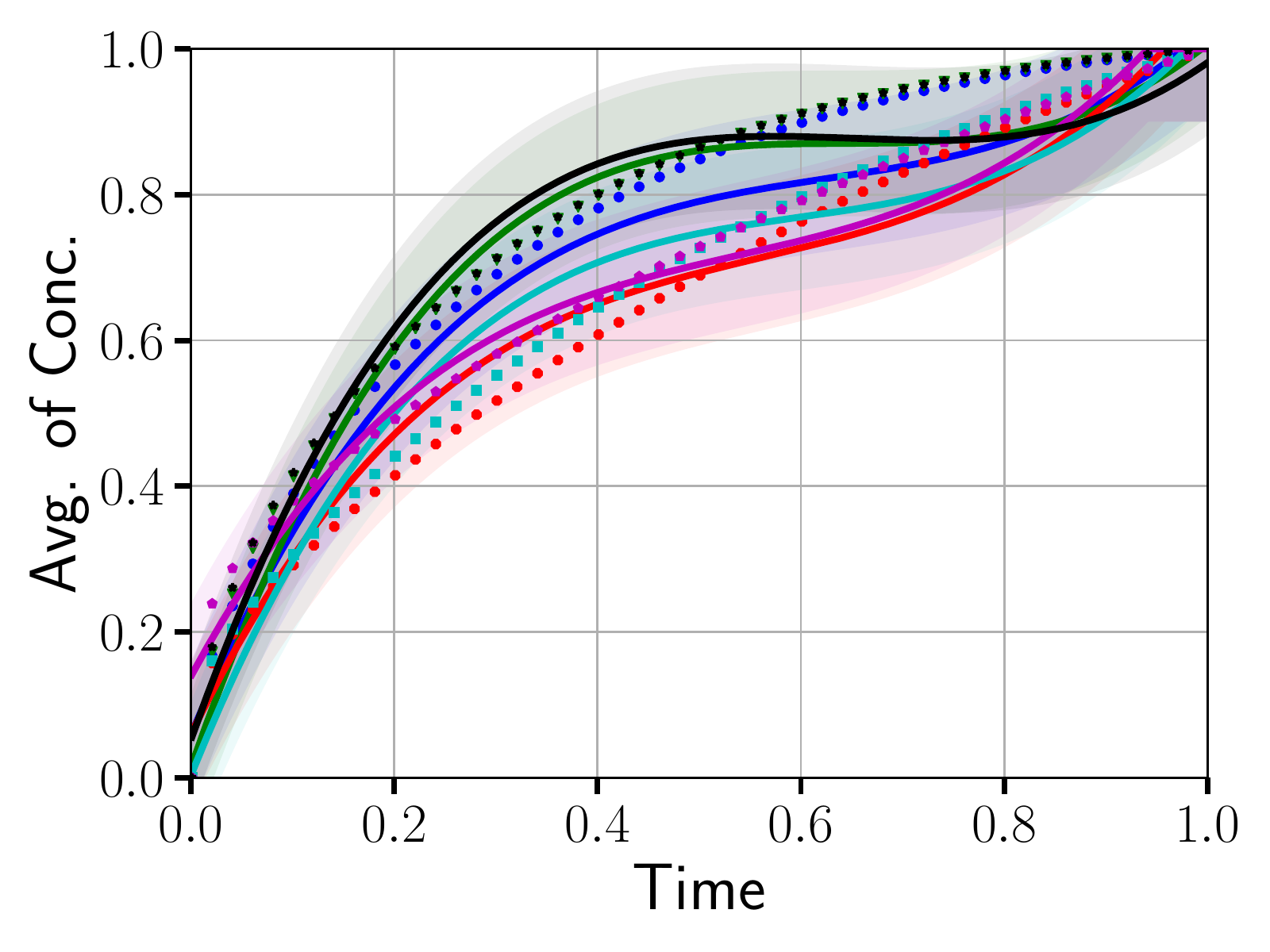}}
  \subfigure[Species $C$:~$\mathbb{c}_C$]
    {\includegraphics[width = 0.325\textwidth]
    {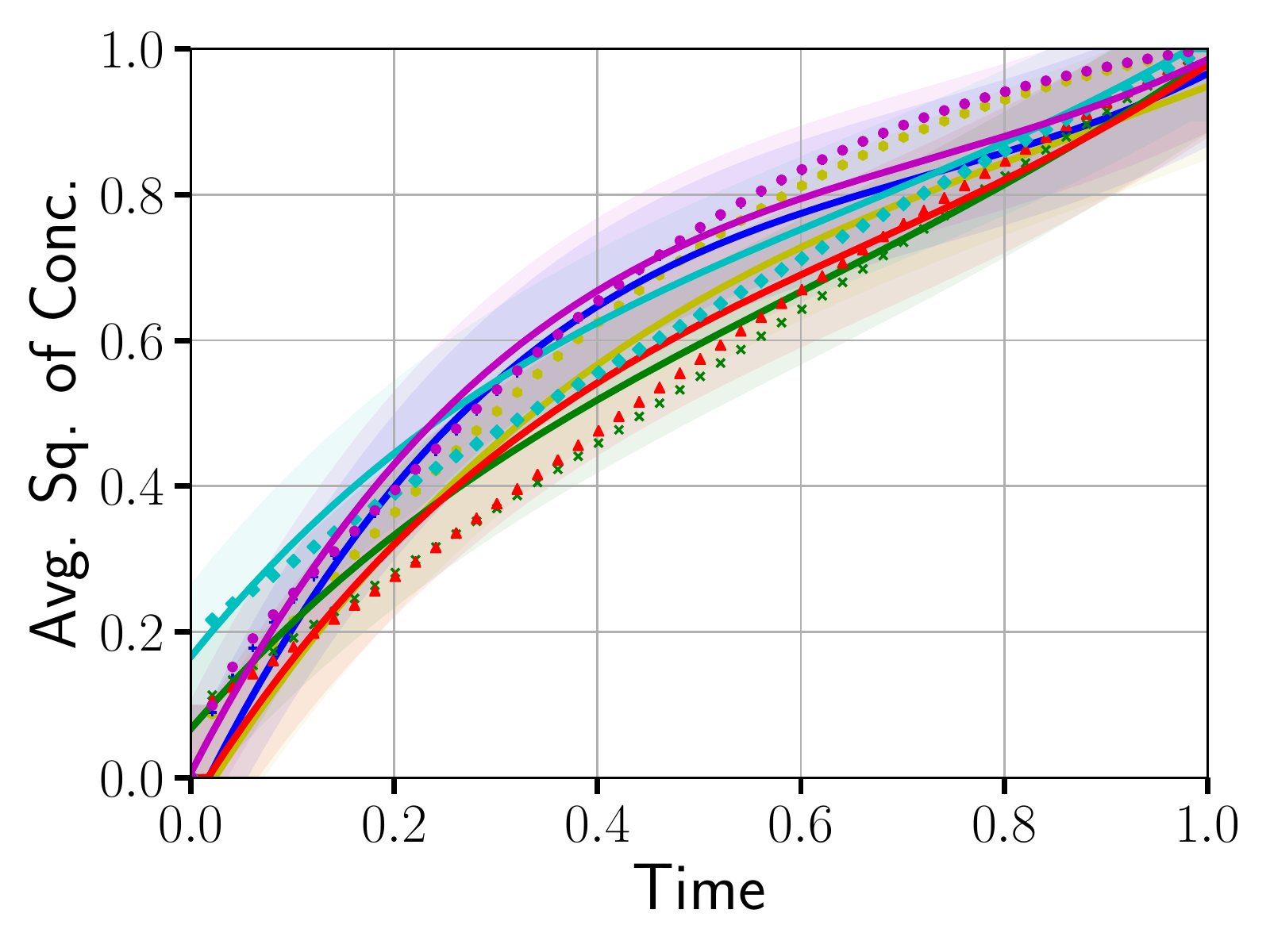}}
  \subfigure[Species $C$:~$\sigma^2_C$]
    {\includegraphics[width = 0.325\textwidth]
    {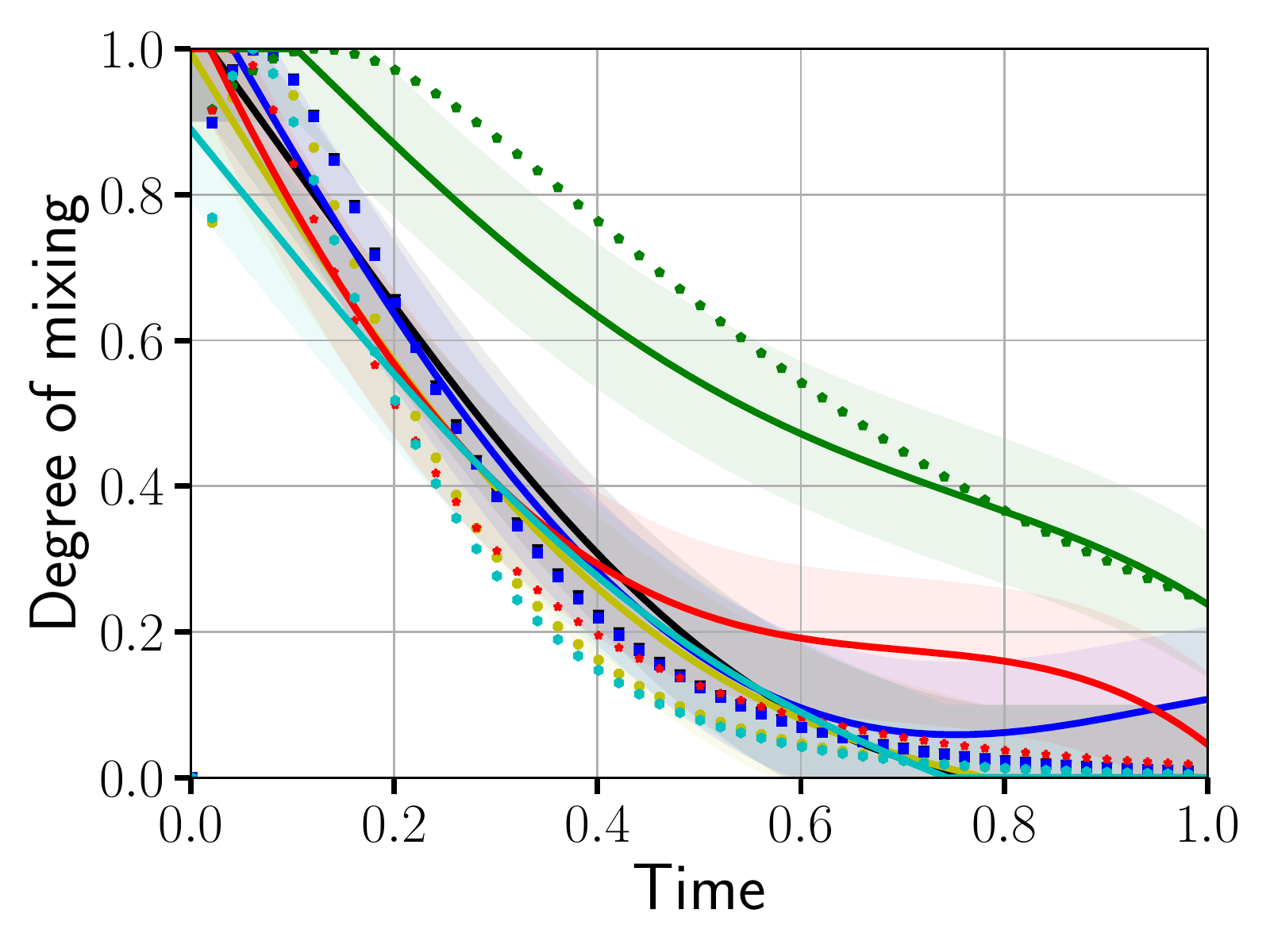}}
  \caption{\textsf{\textbf{Prediction of SVR-ROMs for 
    various QoIs:}}~These figures shows the prediction of SVR-ROMs on randomly 
    sampled unused realizations.
    The $R^2$-score for ROMs prediction on unseen data 
    is greater than 0.9, which instills confidence in 
    our SVR-ROMs capability for QoIs prediction.
  \label{Fig:SVR_vs_Unseen_data}}
\end{figure}

\begin{figure}
  \centering
  \subfigure[Species $A$:~$\mathfrak{c}_A$]
    {\includegraphics[width = 0.325\textwidth]
    {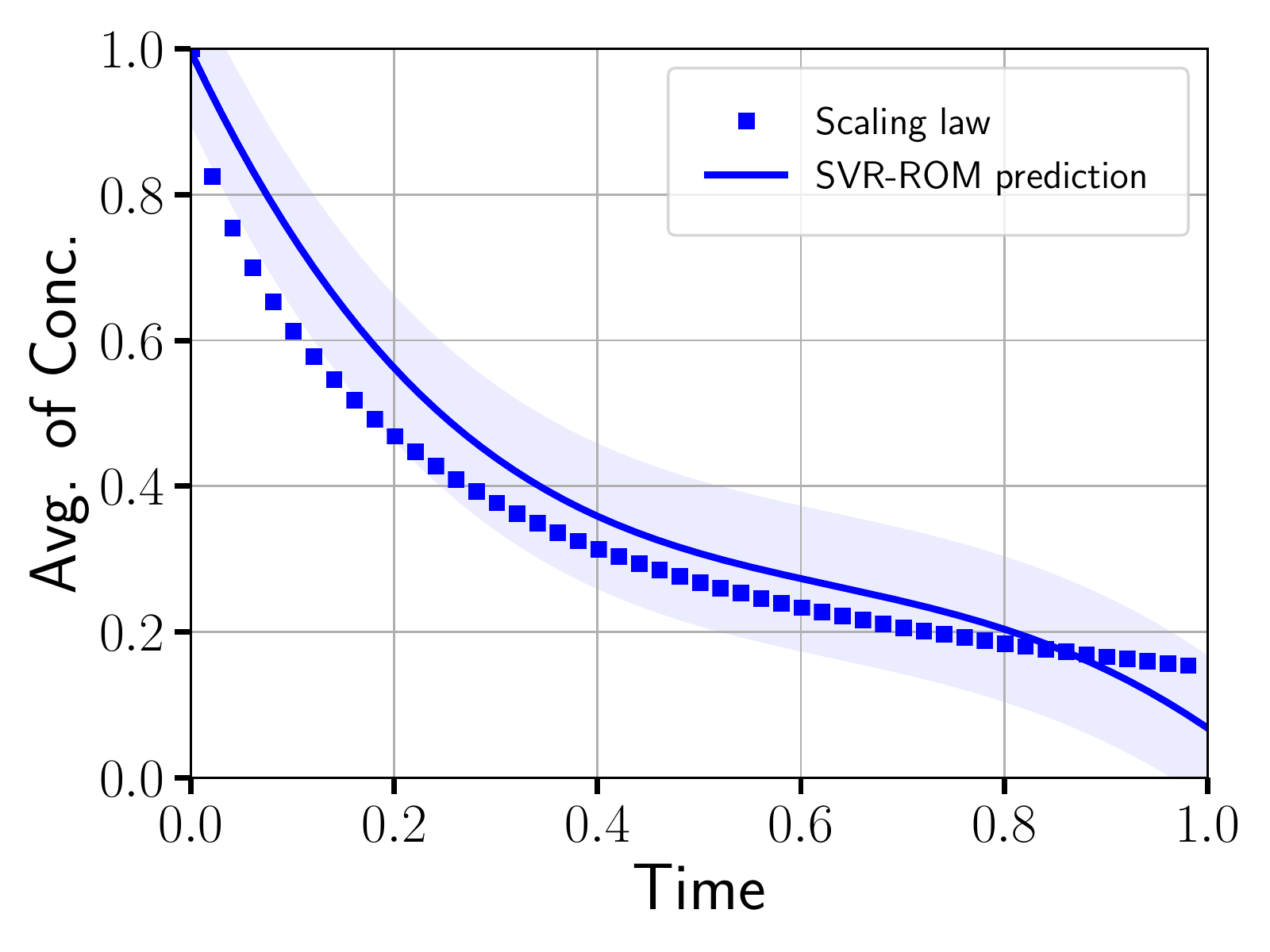}}
  \subfigure[Species $A$:~$\mathbb{c}_A$]
    {\includegraphics[width = 0.325\textwidth]
    {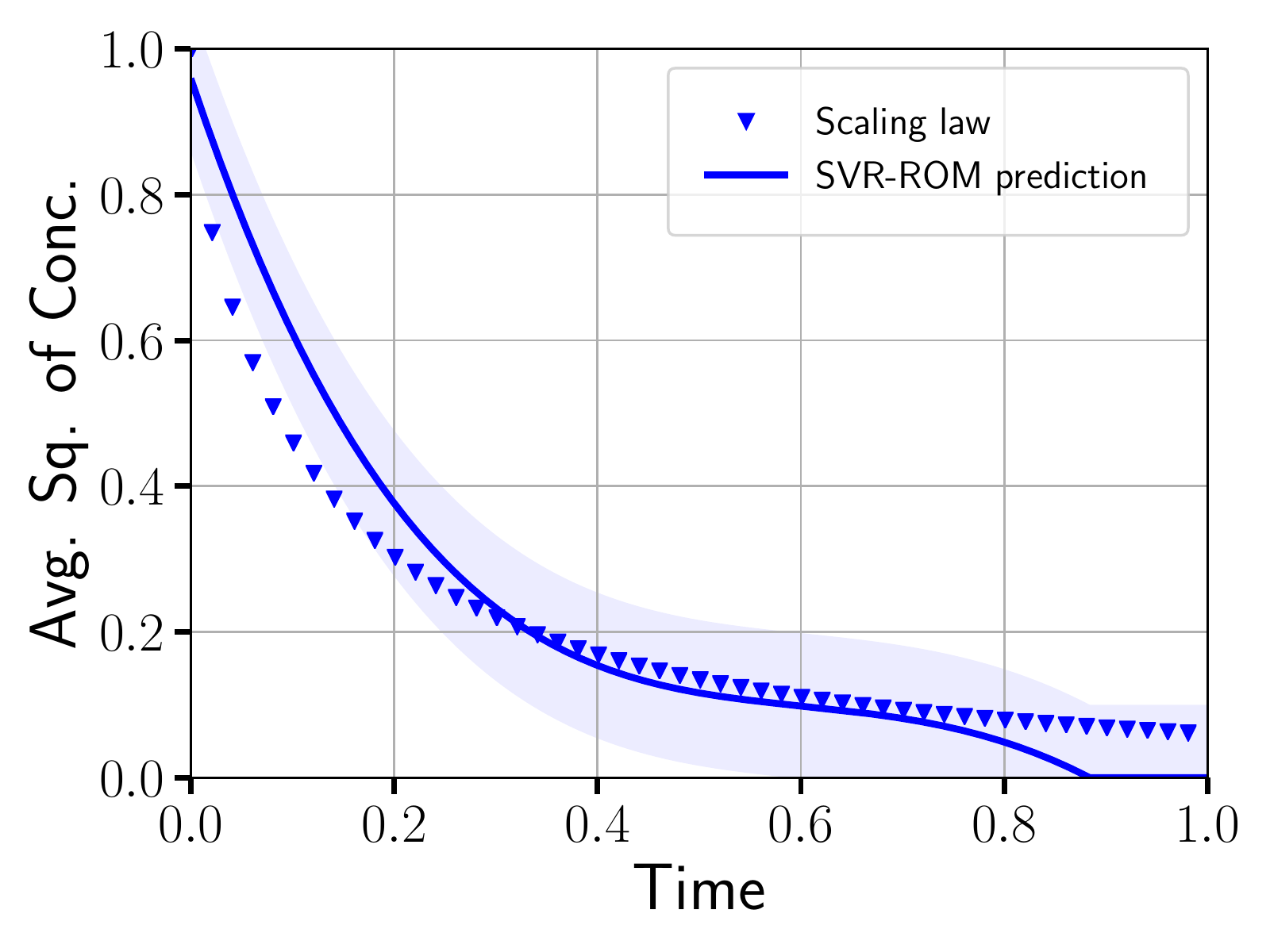}}
  \subfigure[Species $A$:~$\sigma^2_A$]
    {\includegraphics[width = 0.325\textwidth]
    {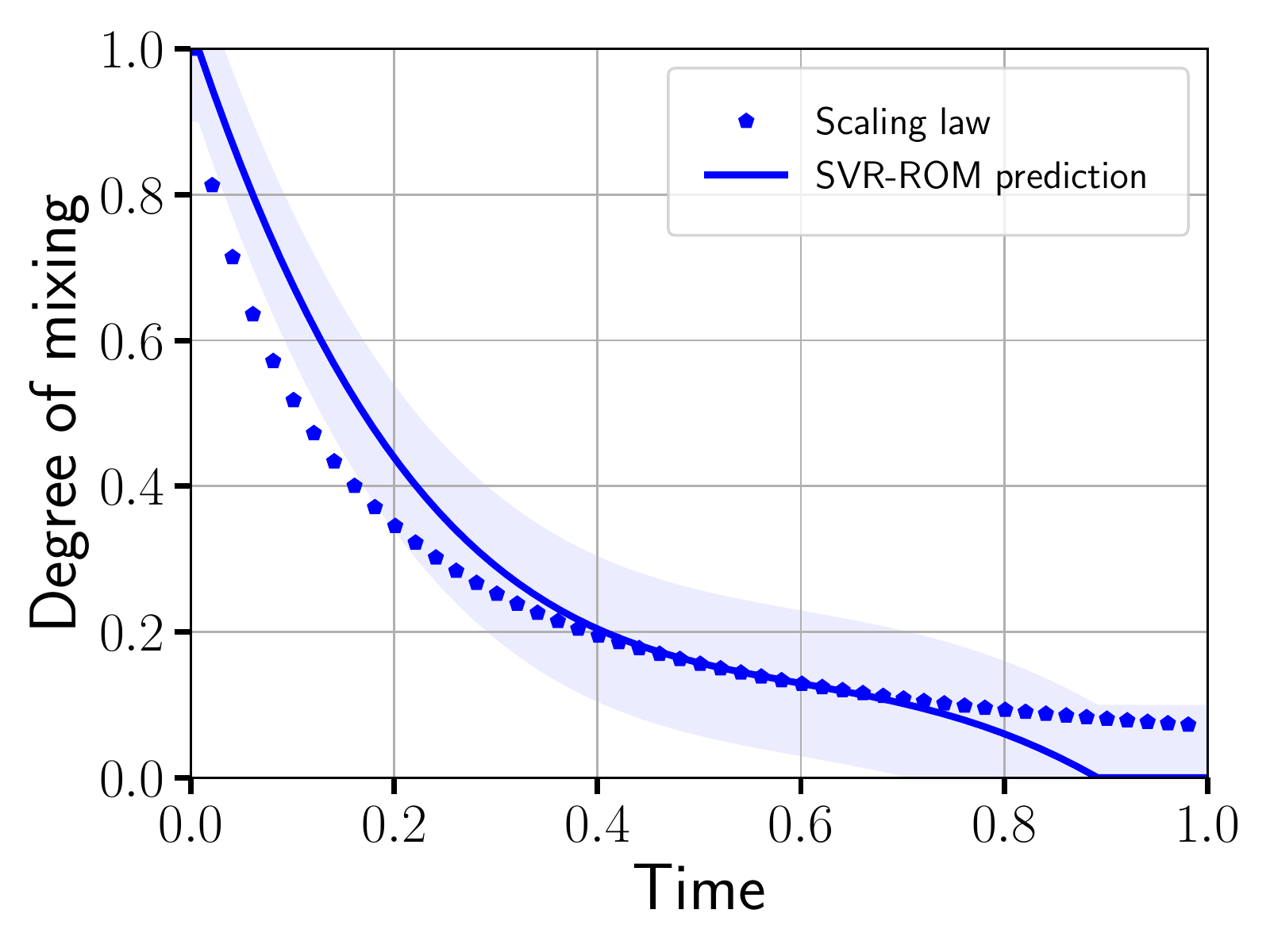}}
  \subfigure[Species $B$:~$\mathfrak{c}_B$]
    {\includegraphics[width = 0.325\textwidth]
    {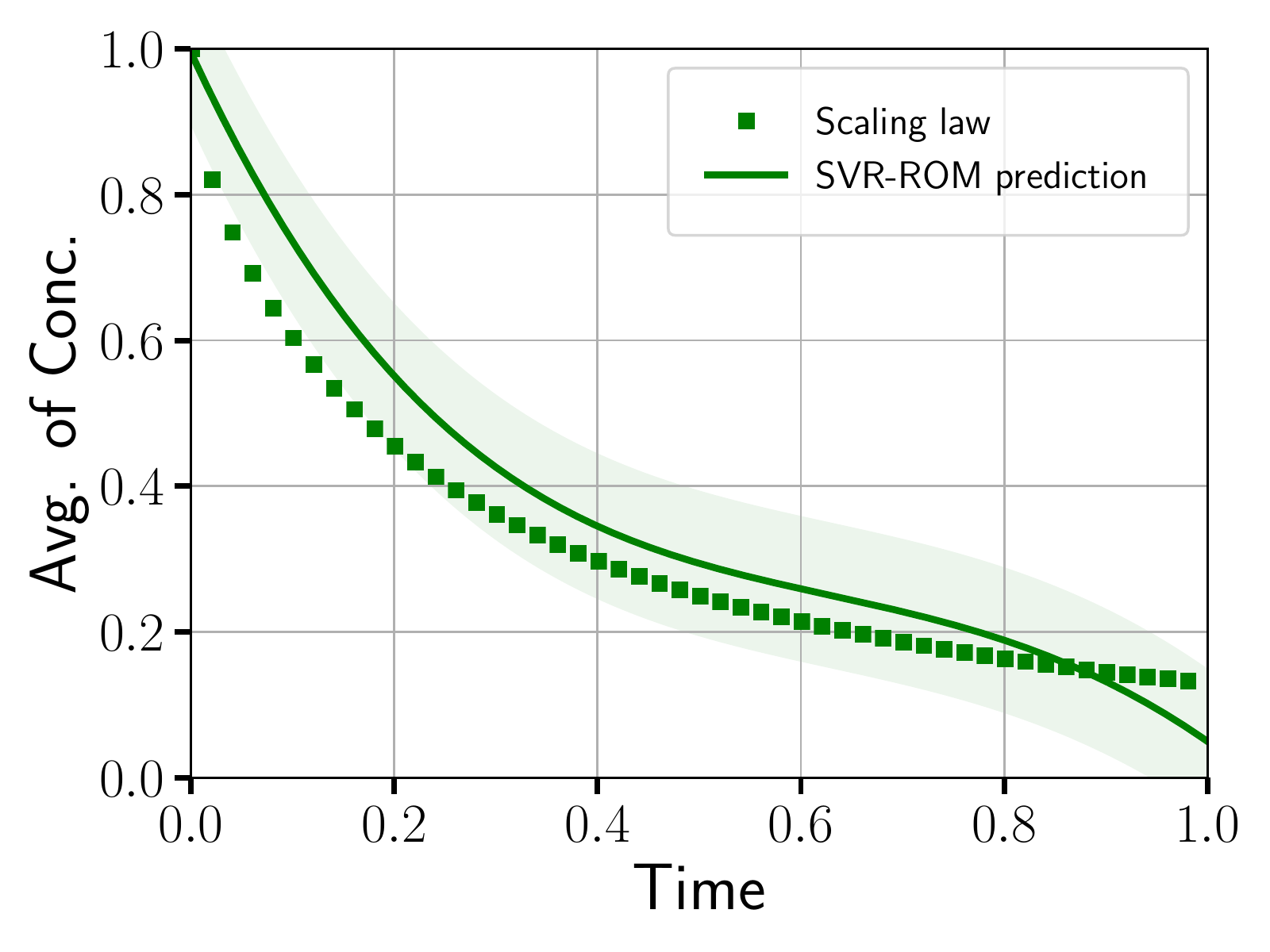}}
  \subfigure[Species $B$:~$\mathbb{c}_B$]
    {\includegraphics[width = 0.325\textwidth]
    {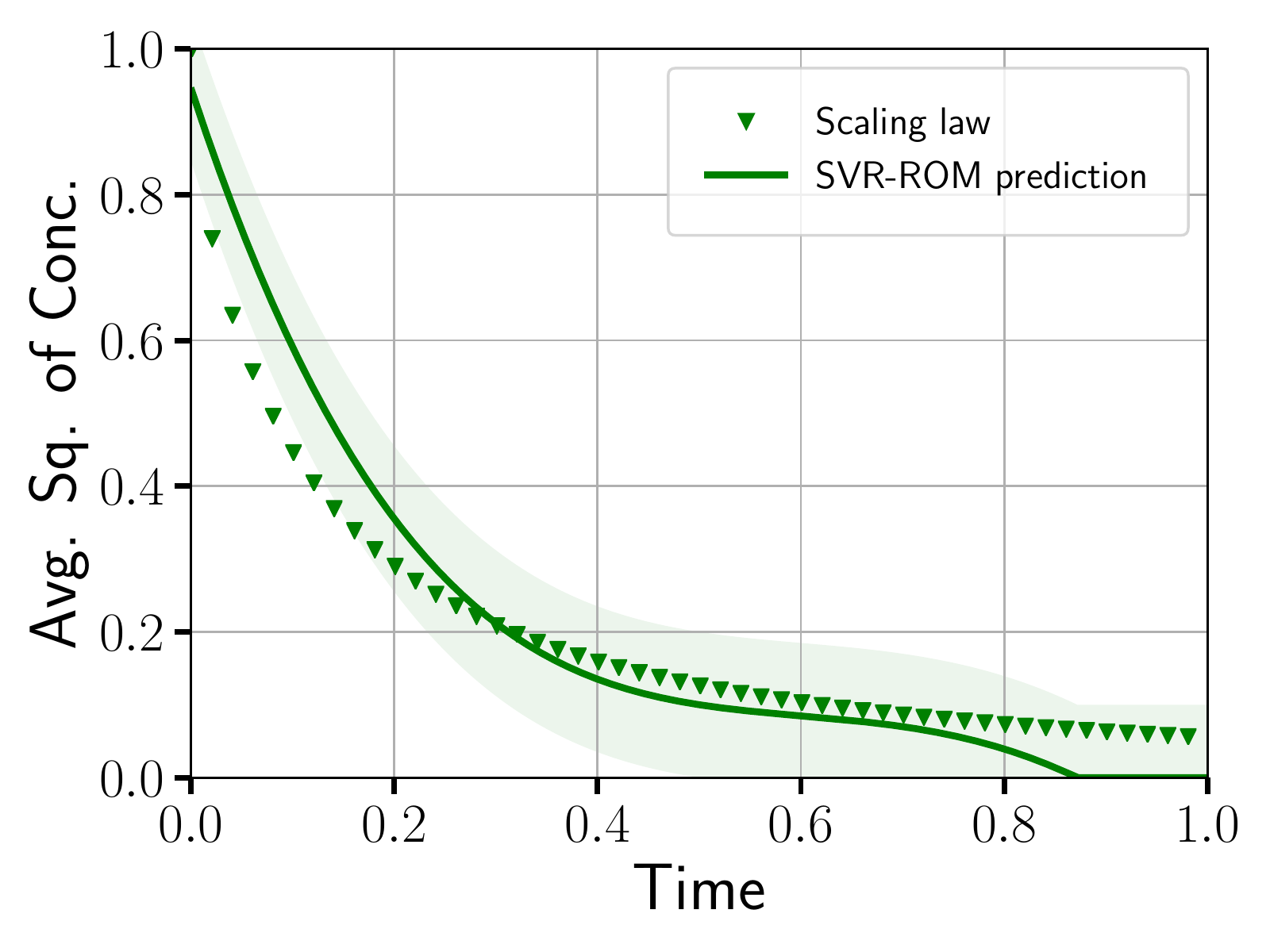}}
  \subfigure[Species $B$:~$\sigma^2_B$]
    {\includegraphics[width = 0.325\textwidth]
    {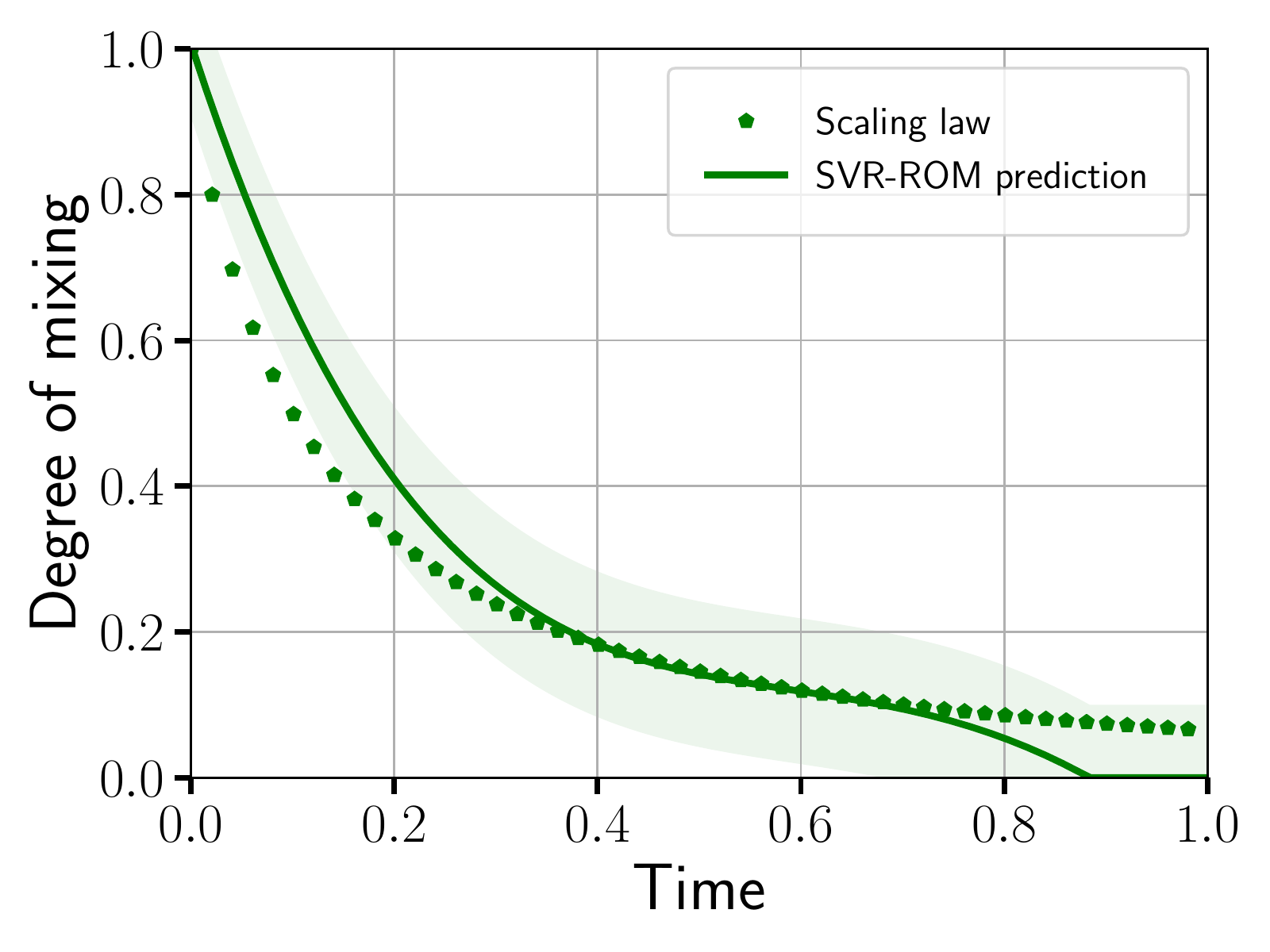}}
  \subfigure[Species $C$:~$\mathfrak{c}_C$]
    {\includegraphics[width = 0.325\textwidth]
    {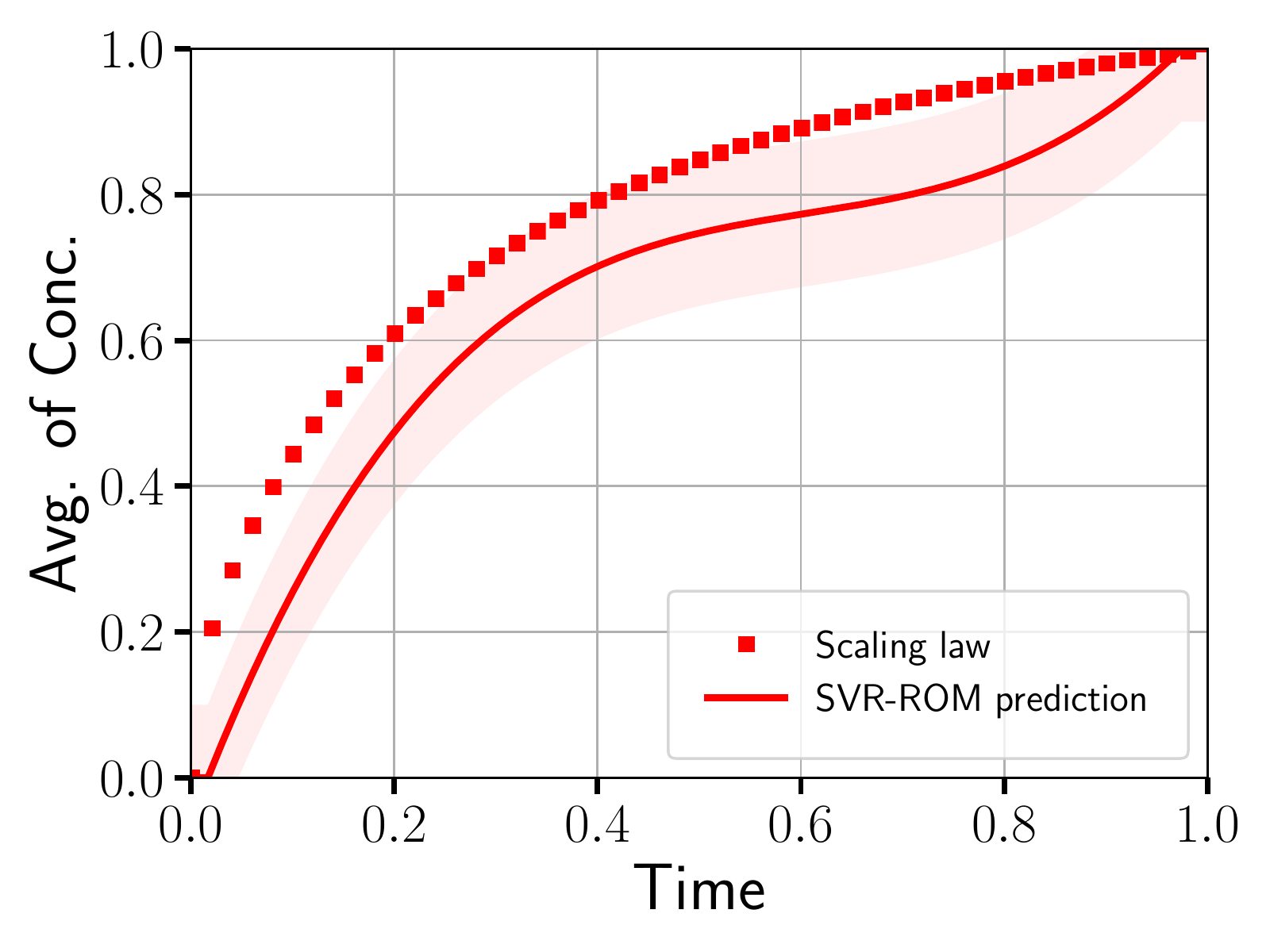}}
  \subfigure[Species $C$:~$\mathbb{c}_C$]
    {\includegraphics[width = 0.325\textwidth]
    {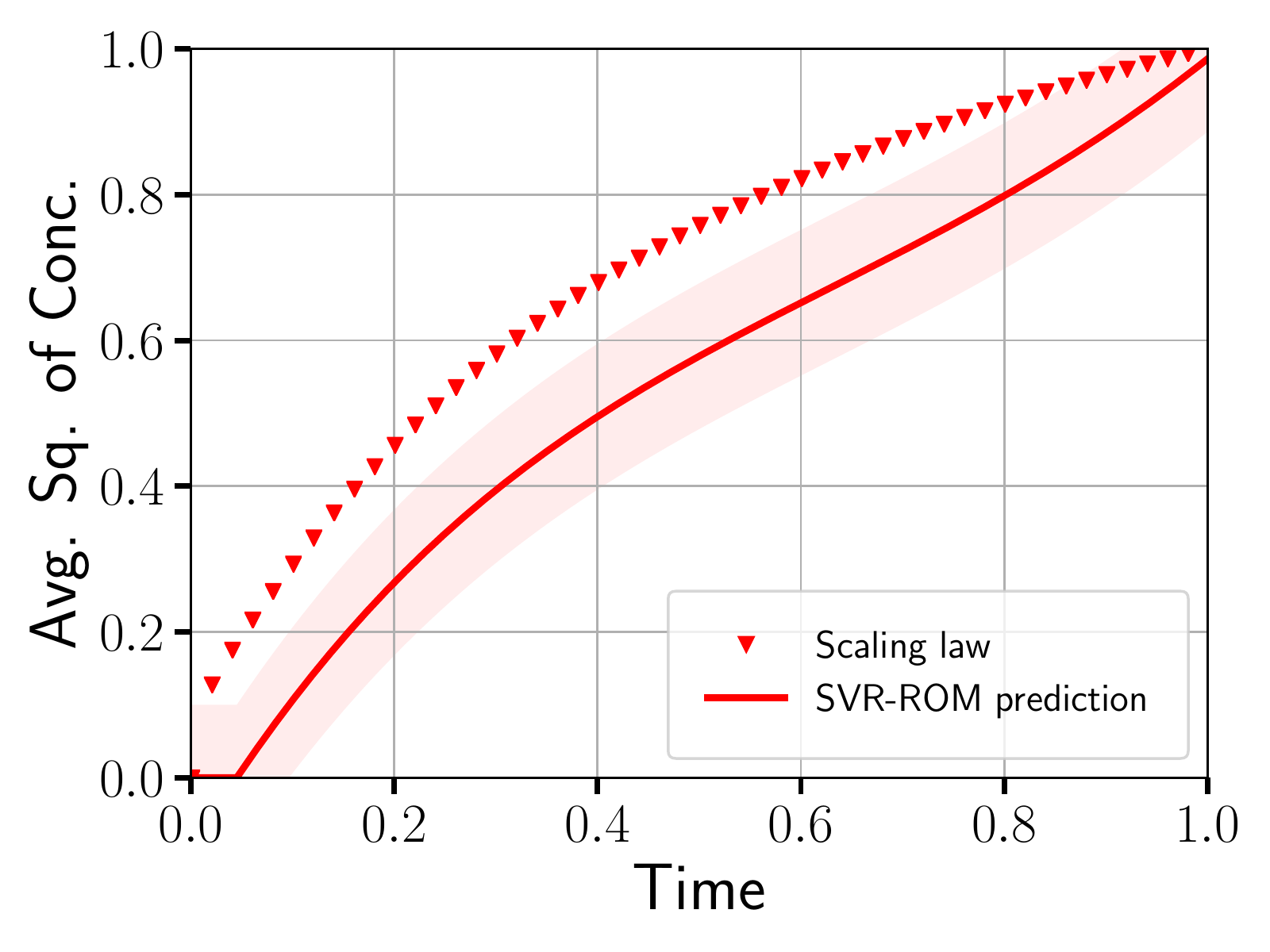}}
  \subfigure[Species $C$:~$\sigma^2_C$]
    {\includegraphics[width = 0.325\textwidth]
    {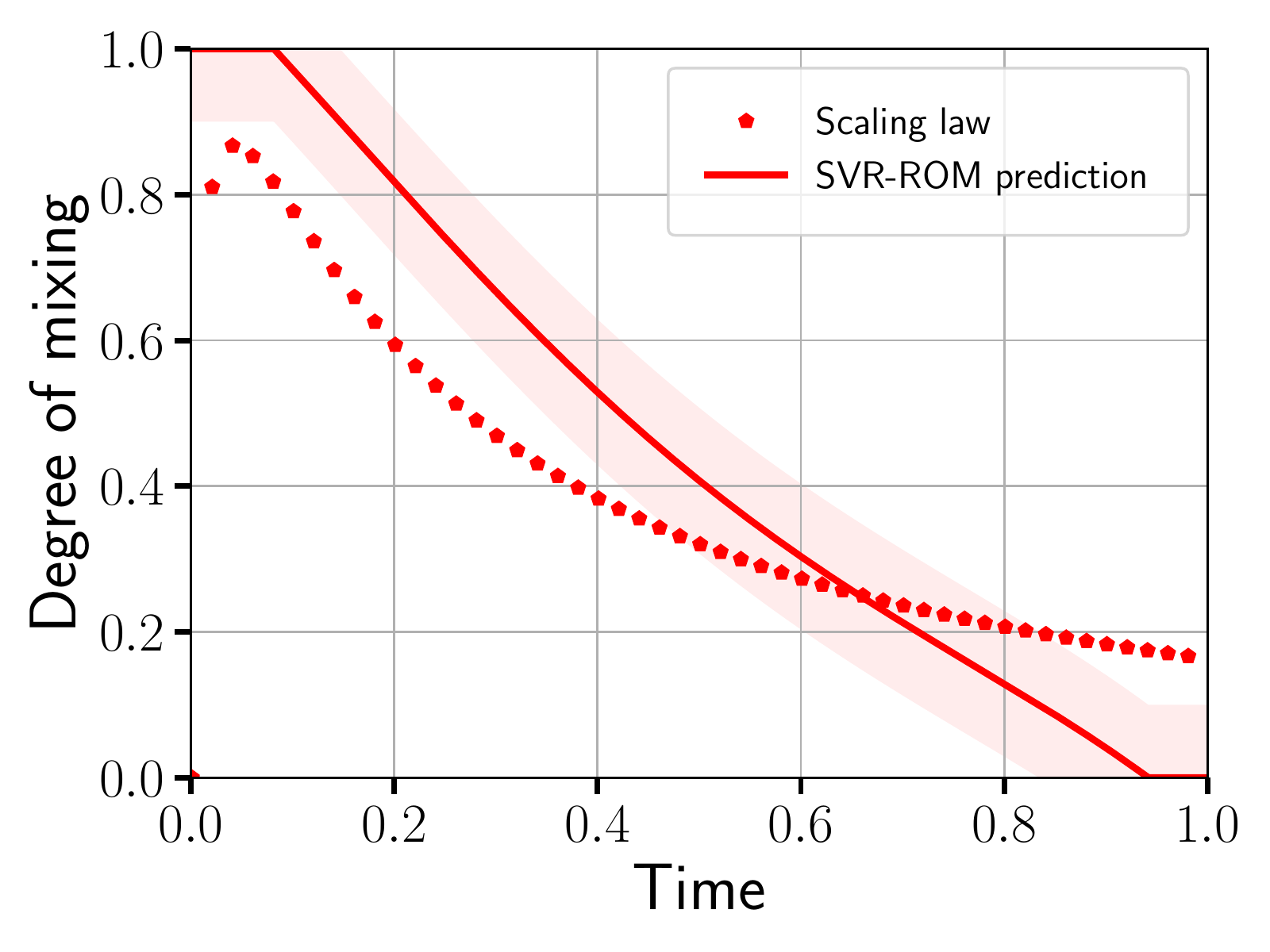}}
  \caption{\textsf{\textbf{Scaling law prediction using 
    SVR-ROMs:}}~These figures show the prediction of 
    SVR-ROMs for scaling quantities of interest, which 
    are average of concentration, average of square of 
    concentration, and degree of mixing of species $A$, 
    $B$, and $C$. From the above figures, qualitatively 
    and quantitatively, it clear that SVR-ROMs are able 
    to predict the scaling law trends of QoIs for species 
    $A$ and $B$, accurately. Quantitatively, SVR-ROMs 
    deviate when predicting scaling QoIs for species $C$. 
    However, qualitatively, SVR-ROMs are able to predict 
    the trends observed in scaling quantities of interest 
    for species $C$.
  \label{Fig:SVR_Scaling_Law}}
\end{figure}
\end{document}